\documentclass[final,3p]{elsarticle}
\usepackage{lineno,hyperref}
\usepackage{amsmath}
\usepackage{amssymb}
\usepackage[linesnumbered,ruled,vlined]{algorithm2e}
\usepackage{caption}
\usepackage{mathtools}
\usepackage{changes}
\usepackage{multirow}
\usepackage{verbatim}
\usepackage{mathrsfs}
\usepackage{graphicx}
\usepackage{subcaption}
\usepackage{amsmath,amsthm,bm,mathrsfs}

\theoremstyle{plain}
\newtheorem{theorem}{Theorem}[section]

\newtheorem{proposition}[theorem]{Proposition}
\newtheorem{assumption}[theorem]{Assumption}

\theoremstyle{definition}

\theoremstyle{remark}
\newtheorem{remark}{Remark}

\modulolinenumbers[5]

\journal{ArXiv.org}

\begin{document}

\begin{frontmatter}

\title{A Unified Low-rank ADI Framework with Shared Linear Solves for Simultaneously Solving Multiple Lyapunov, Sylvester, and Riccati Equations}
\author[uz]{Umair~Zulfiqar}
\author[shu]{Zhong-Yi Huang}
\author[qs]{Qiu-Yan~Song\corref{mycorrespondingauthor}}
\cortext[mycorrespondingauthor]{Corresponding author}
\ead{qysong@hbnu.edu.cn}
\author[shu]{Zhi-Yuan Gao}
\address[uz]{School of Electronic Information and Electrical Engineering, Yangtze University, Jingzhou, Hubei, 434023, China}
\address[shu]{School of Mechatronic Engineering and Automation, Shanghai University, Shanghai, 200444, China}
\address[qs]{Huangshi Key Laboratory of Metaverse and Virtual Simulation, School of Mathematics and Statistics, Hubei Normal University, Huangshi, Hubei 435002, China.}
\begin{abstract}
The alternating direction implicit (ADI) methods are computationally efficient and numerically effective tools for computing low-rank solutions of large-scale linear matrix equations. It is known in the literature that the low-rank ADI method for Lyapunov equations is a Petrov–Galerkin projection algorithm that implicitly performs model order reduction. It recursively enforces interpolation at the mirror images of the ADI shifts and places the poles of the reduced-order models at the ADI shifts. In this paper, we show that the low-rank ADI methods for Sylvester and Riccati equations are also Petrov–Galerkin projection algorithms that implicitly perform model order reduction. These methods likewise enforce interpolation at the mirror images of the ADI shifts; however, they do not place the poles at the mirror images of the interpolation points. Instead, their pole placement ensures that the projected Sylvester and Riccati equations they implicitly solve admit a unique solution.

By observing that the ADI methods for Lyapunov, Sylvester, and Riccati equations differ only in pole placement and not in their interpolatory nature, we show that the shifted linear solves—which constitute the bulk of the computational cost—can be shared. The pole-placement step involves only small-scale operations and is therefore inexpensive. We propose a unified ADI framework that requires only two shifted linear solves per iteration to simultaneously solve six Lyapunov equations, one Sylvester equation, and ten Riccati equations, thus substantially increasing the return on investment for the computational cost spent on the linear solves. All operations needed to extract the individual solutions from these shared linear solves are small-scale and inexpensive. Three numerical examples are presented to demonstrate the effectiveness of the proposed ADI framework.
\end{abstract}


\end{frontmatter}

\section{Problem Setting}
Let $G_1(s)$ and $G_2(s)$ be linear time-invariant (LTI) systems with state-space realizations:
\begin{align}
G_1(s)&=C_1(sE_1-A_1)^{-1}B_1+D_1,\label{eq1}\\
G_2(s)&=C_2(sE_2-A_2)^{-1}B_2+D_2.\label{eq2}
\end{align}The matrices in \eqref{eq1} and \eqref{eq2} have the following dimensions: $A_1\in\mathbb{R}^{n_1\times n_1}$, $B_1\in\mathbb{R}^{n_1\times m_1}$, $C_1\in\mathbb{R}^{p_1\times n_1}$, $D_1\in\mathbb{R}^{p_1\times m_1}$, $E_1\in\mathbb{R}^{n_1\times n_1}$, $A_2\in\mathbb{R}^{n_2\times n_2}$, $B_2\in\mathbb{R}^{n_2\times m_2}$, $C_2\in\mathbb{R}^{p_2\times n_2}$, $D_2\in\mathbb{R}^{p_2\times m_2}$, and $E_2\in\mathbb{R}^{n_2\times n_2}$.

The corresponding state-space equations are:
\begin{align}
E_1\dot{x}_1(t)&=A_1x_1(t)+B_1u_1(t),\label{ss1}\\
y_1(t)&=C_1x_1(t)+D_1u_1(t),\label{ss2}\\
E_2\dot{x}_2(t)&=A_2x_2(t)+B_2u_2(t),\label{ss3}\\
y_2(t)&=C_2x_2(t)+D_2u_2(t).\label{ss4}
\end{align}Here, $x_1(t) \in \mathbb{R}^{n_1\times 1}$, $u_1(t) \in \mathbb{R}^{m_1\times 1}$, and $y_1(t) \in \mathbb{R}^{p_1\times 1}$ are the state, input, and output vectors for $(A_1,B_1,C_1,D_1,E_1)$. Similarly, $x_2(t) \in \mathbb{R}^{n_2\times 1}$, $u_2(t) \in \mathbb{R}^{m_2\times 1}$, and $y_2(t) \in \mathbb{R}^{p_2\times 1}$ are the state, input, and output vectors for $(A_2,B_2,C_2,D_2,E_2)$. 
\begin{assumption}\label{assum}
The following assumptions are made throughout the paper:
\begin{enumerate}
  \item The matrices $E_1$ and $E_2$ are invertible.
  \item $G_1(s)$ and $G_2(s)$ are stable; that is, all eigenvalues of $E_1^{-1}A_1$ and $E_2^{-1}A_2$ have negative real parts.
  \item\label{assum_3} The state dimensions $n_1$ and $n_2$ are large, while the input and output dimensions satisfy $m_1, p_1 \ll n_1$ and $m_2, p_2 \ll n_2$.
\end{enumerate}
\end{assumption}
This paper focuses on low-rank approximations of several linear matrix equations associated with the state-space realizations in \eqref{ss1}–\eqref{ss4}, as described below.

The controllability Gramian $P_1$ for the realization in \eqref{ss1}–\eqref{ss2} and the observability Gramian $Q_2$ for the realization in \eqref{ss3}–\eqref{ss4} are the unique solutions of the Lyapunov equations:
\begin{align}
A_1P_1E_1^\top+E_1P_1A_1^\top+B_1B_1^\top&=0,\label{lyap_p}\\
A_2^\top Q_2E_2+E_2^\top Q_2A_2+C_2^\top C_2&=0.\label{lyap_q}
\end{align}
These Lyapunov equations are central to balanced truncation (BT) \cite{moore2003principal}, one of the most widely used model order reduction (MOR) methods.

Let $S_1\in\mathbb{R}^{m_1\times m_1}$ and $S_2\in\mathbb{R}^{p_2\times p_2}$ be indefinite matrices, and let $P_s$ and $Q_s$ be the unique solutions of the Lyapunov equations:
\begin{align}
A_1P_sE_1^\top+E_1P_sA_1^\top+B_1S_1B_1^\top&=0,\label{lyap_p_ldl}\\
A_2^\top Q_sE_2+E_2^\top Q_sA_2+C_2^\top S_2C_2&=0,\label{lyap_q_ldl}
\end{align}where $B_1S_1B_1^\top$ and $C_2^\top S_2C_2$ are indefinite. Lyapunov equations of the form \eqref{lyap_p_ldl} and \eqref{lyap_q_ldl} arise in frequency-limited BT (FLBT) \cite{benner2016frequency} and time-limited BT (TLBT) \cite{kurschner2018balanced}.

If $D_1$ and $D_2$ are square and invertible, and the matrices $A_1-B_1D_1^{-1}C_1$ and $A_2-B_2D_2^{-1}C_2$ are Hurwitz, then $P_{\mathrm{mp}}$ and $Q_{\mathrm{mp}}$ uniquely solve the Lyapunov equations:
\begin{align}
(A_1-B_1D_1^{-1}C_1)P_{\mathrm{mp}}E_1^\top+E_1P_{\mathrm{mp}}(A_1-B_1D_1^{-1}C_1)^\top+B_1(D_1^\top D_1)^{-1}B_1^\top&=0,\label{lyap_p_mp}\\
(A_2-B_2D_2^{-1}C_2)^\top Q_{\mathrm{mp}}E_2+E_2^\top Q_{\mathrm{mp}}(A_2-B_2D_2^{-1}C_2)+C_2^\top (D_2 D_2^\top)^{-1}C_2&=0.\label{lyap_q_mp}
\end{align}
These Lyapunov equations appear in relative-error BT \cite{zhou1995frequency} for preserving the minimum phase properties of $G_1(s)$ and $G_2(s)$.

Assuming $m_1=p_2$, the Sylvester equation
\begin{align}
A_1X_{\mathrm{sylv}}E_2+E_1X_{\mathrm{sylv}}A_2+B_1C_2=0\label{x_sylv}
\end{align} 
has a unique solution $X_{\mathrm{sylv}}$. A special case of this Sylvester equation arises in cross Gramian-based MOR \cite{baur2008cross}.
 
The Riccati equations for Kalman filtering and the Linear Quadratic Regulator (LQR) \cite{zhou1998essentials} associated with \eqref{eq1} and \eqref{eq2}, respectively, are:
\begin{align}
A_1P_{\mathrm{ricc}}E_1^\top+E_1P_{\mathrm{ricc}}A_1^\top+B_1B_1^\top-E_1P_{\mathrm{ricc}}C_1^\top C_1P_{\mathrm{ricc}}E_1^\top&=0,\label{ricc_p}\\
A_2^\top Q_{\mathrm{ricc}}E_2+E_2^\top Q_{\mathrm{ricc}}A_2+C_2^\top C_2-E_2^\top Q_{\mathrm{ricc}}B_2B_2^\top Q_{\mathrm{ricc}}E_2&=0, \label{ricc_q}
\end{align}where $A_1-E_1P_{\mathrm{ricc}}C_1^\top C_1$ and $A_2-B_2B_2^\top Q_{\mathrm{ricc}}E_2$ are Hurwitz. These Riccati equations also appear in LQG BT \cite{jonckheere2003new}.

For $\mathcal{H}_\infty$ filtering and control \cite{fortuna2021optimal} with $\gamma_1>0$ and $\gamma_2>0$, the corresponding Riccati equations are:
\begin{align}
A_1P_{\infty}E_1^\top+E_1P_{\infty}A_1^\top+B_1B_1^\top-(1-\gamma_1^{-2})E_1P_{\infty}C_1^\top C_1P_{\infty}E_1^\top&=0,\label{ricc_p_inf}\\
A_2^\top Q_{\infty}E_2+E_2^\top Q_{\infty}A_2+C_2^\top C_2-(1-\gamma_2^{-2})E_2^\top Q_{\infty}B_2B_2^\top Q_{\infty}E_2&=0, \label{ricc_q_inf}
\end{align}with $A_1-E_1P_{\infty}C_1^\top C_1$ and $A_2-B_2B_2^\top Q_{\infty}E_2$ Hurwitz. These Riccati equations also appear in $\mathcal{H}_\infty$ BT \cite{mustafa2002controller}.

If $D_1+D_1^\top$ and $D_2+D_2^\top$ are square and invertible, the positive-real lemma \cite{green1988balanced} applied to $G_1(s)$ and $G_2(s)$ yields the Riccati equations:
\begin{align}
A_1P_{\mathrm{pr}}E_1^\top+E_1P_{\mathrm{pr}}A_1^\top+(B_1-E_1P_{\mathrm{pr}}C_1^\top)(D_1+D_1^\top)^{-1}(B_1-E_1P_{\mathrm{pr}}C_1^\top)^\top=0,\label{ricc_p_pr}\\
A_2^\top Q_{\mathrm{pr}}E_2+E_2^\top Q_{\mathrm{pr}}A_2+(C_2-B_2^\top Q_{\mathrm{pr}}E_2)^\top(D_2+D_2^\top)^{-1}(C_2-B_2^\top Q_{\mathrm{pr}}E_2)=0,\label{ricc_q_pr}
\end{align}where $A_1-B_1(D_1+D_1^\top)^{-1}C_1-E_1P_{\mathrm{pr}}C_1^\top(D_1+D_1^\top)^{-1} C_1$ and $A_2-B_2(D_2+D_2^\top)^{-1}C_2-B_2(D_2+D_2^\top)^{-1}B_2^\top Q_{\mathrm{pr}}E_2$ are Hurwitz. These Riccati equations are also used in positive-real BT \cite{phillips2002guaranteed}.

If $I-D_1D_1^\top$ and $I-D_2^\top D_2$ are invertible, the bounded-real lemma \cite{phillips2002guaranteed} applied to $G_1(s)$ and $G_2(s)$ yields the Riccati equations:
\begin{align}
A_1P_{\mathrm{br}}E_1^\top+E_1P_{\mathrm{br}}A_1^\top+B_1B_1^\top+(E_1P_{\mathrm{br}}C_1^\top+B_1D_1^\top)(I-D_1D_1^\top)^{-1}(E_1P_{\mathrm{br}}C_1^\top+B_1D_1^\top)^\top&=0,\label{ricc_p_br}\\
A_2^\top Q_{\mathrm{br}}E_2+E_2^\top Q_{\mathrm{br}}A_2+C_2^\top C_2+(B_2^\top Q_{\mathrm{br}}E_2+D_2^\top C_2)^\top(I-D_2^\top D_2)^{-1}(B_2^\top Q_{\mathrm{br}}E_2+D_2^\top C_2)&=0,\label{ricc_q_br}
\end{align}
with $A_1+B_1D_1^\top(I-D_1D_1^\top)^{-1}C_1-E_1P_{\mathrm{br}}C_1^\top(I-D_1D_1^\top)^{-1} C_1$ and $A_2+B_2(I-D_2^\top D_2)^{-1}D_2^\top C_2-B_2(I-D_2^\top D_2)^{-1}B_2^\top Q_{\mathrm{br}}E_2$ Hurwitz. These Riccati equations also appear in bounded-real BT \cite{phillips2002guaranteed}.

Assuming $A_1=A_2$, $B_1=B_2$, $C_1=C_2$, $D_1=D_2$, $E_1=E_2$, and $D_1^\top D_1$ and $D_2D_2^\top$ invertible, the Riccati equations for computing spectral factorizations of $G_1(s)$ and $G_2(s)$ \cite{zhoubook} are given by:
\begin{align}
A_1 P_{\mathrm{sf}}E_1^\top+E_1 P_{\mathrm{sf}}A_1^\top+\big(B_1-E_1P_{\mathrm{sf}}(Q_2B_1+C_1^\top D_1)\big)(D_1^\top D_1)^{-1}\big(B_1-E_1P_{\mathrm{sf}}(Q_2B_1+C_1^\top D_1)\big)^\top=0,\label{ricc_p_sf}\\
A_2^\top Q_{\mathrm{sf}}E_2+E_2^\top Q_{\mathrm{sf}}A_2+\big(C_2-(C_2P_1+D_2B_2^\top)Q_{\mathrm{sf}}E_2\big)^\top(D_2D_2^\top)^{-1}\big(C_2-(C_2P_1+D_2B_2^\top)Q_{\mathrm{sf}}E_2\big)=0,\label{ricc_q_sf}
\end{align}where the matrices
\[A_1-B_1(D_1^\top D_1)^{-1}\big(B_1^\top Q_2+D_1^\top C_1\big)-E_1P_{\mathrm{sf}}\big(B_1^\top Q_2+D_1^\top C_1\big)^\top(D_1^\top D_1)^{-1}\big(B_1^\top Q_2+D_1^\top C_1\big)\]
and
\[A_2-\big(P_1C_2^\top+B_2 D_2^\top\big)(D_2 D_2^\top)^{-1}C_2- \big(P_1C_2^\top+B_2D_2^\top\big)(D_2 D_2^\top)^{-1}\big(P_1C_2^\top+B_2D_2^\top\big)^\top Q_{\mathrm{sf}}E_2\]
are Hurwitz. These Riccati equations are also used in Stochastic BT (BST) \cite{green1988balanced}.

In large-scale settings, solving the linear matrix equations \eqref{lyap_p}–\eqref{ricc_q_sf} is computationally expensive. However, under assumption \ref{assum} (\ref{assum_3}), their solutions are typically low-rank. This allows efficient computation via Krylov-subspace or alternating direction implicit (ADI) methods; see \cite{benner2013numerical} for a survey. This paper presents a single ADI algorithm to compute low-rank solutions of \eqref{lyap_p}–\eqref{ricc_q_sf} simultaneously.

The motivation for solving several linear matrix equations simultaneously with shared linear solves is as follows. BT generally provides higher accuracy than positive-real BT and bounded-real BT; however, it does not guarantee that the resulting reduced-order model (ROM) is passive. Thus, \cite{phillips2002guaranteed} recommends first computing a ROM using BT and checking its passivity. If the ROM is not passive, one should then apply positive-real BT or bounded-real BT. Since BT often does yield a passive ROM, \cite{phillips2002guaranteed} argues that accepting the risk of BT producing a non-passive ROM—and thereby potentially wasting significant computational effort—is reasonable given the likelihood of achieving good accuracy. The results in this paper eliminate this issue entirely, as the solutions to the Riccati equations \eqref{ricc_p}–\eqref{ricc_q_br} can be obtained efficiently from the solutions of the Lyapunov equations \eqref{lyap_p} and \eqref{lyap_q} already computed for BT. A second motivation for solving multiple linear matrix equations together is that reducing $G_1(s)$ or $G_2(s)$ using BT and then designing a controller for the ROM does not account for closed-loop stability. Relative-error BT, LQG BT, and $\mathcal{H}_\infty$ BT explicitly incorporate closed-loop stability in their performance criteria. Again, the results of this paper overcome this limitation of BT, since the solutions of \eqref{lyap_p_mp}–\eqref{ricc_q_br} can be extracted from the solutions of \eqref{lyap_p} and \eqref{lyap_q} computed during BT. Consequently, reusing the solutions of the Lyapunov equations \eqref{lyap_p} and \eqref{lyap_q} to obtain the remaining solutions \eqref{lyap_p_ldl}–\eqref{ricc_q_sf} significantly reduces the computational effort needed for analyzing $G_1(s)$ or $G_2(s)$ and for carrying out various design tasks associated with these dynamical systems.
\section{Literature Review}
This section briefly reviews the ADI methods used to compute low-rank approximations of Lyapunov, Sylvester, and Riccati equations. It also summarizes the interpolation theory relevant to the subsequent discussion. The key mathematical notations used throughout the paper are listed in Table \ref{tab00}.
\begin{table}[!t]
\centering
\caption{Mathematical Notations}\label{tab00}
\begin{tabular}{|c|c|}\hline
Notation & Meaning\\\hline
$Z^*$ & Hermitian transpose of matrix $Z$\\
$Z^{-*}$ & Hermitian transpose of matrix $Z^{-1}$\\
Re($Z$) & Real part of matrix $Z$\\
Im($Z$) & Imaginary part of matrix $Z$\\
$Z(:,k)$ & $k^{th}$ column of matrix $Z$\\
$Z(k,:)$ & $k^{th}$ row of matrix $Z$\\
tr($Z$) & Trace of matrix $Z$\\
$\lambda_i(Z)$& The eigenvalues of matrix $Z$\\
span$\{Z\}$ & Span of the columns of matrix $Z$  \\
Ran$(Z)$ & Range of matrix $Z$  \\
orth$(Z)$ & Orthonormal basis for range of matrix $Z$\\
$j$ & $\sqrt{-1}$\\
$\alpha\in\mathbb{R}_{-}$& Negative real number $\alpha$\\
$\alpha\in\mathbb{C}_{-}$& The real part of $\alpha$ is negative\\
$\land$&Logical and operator\\\hline
\end{tabular}
\end{table}
\subsection{Cholesky Factor ADI Method for Lyapunov Equations (CF-ADI)}
The CF-ADI algorithm \cite{wachspress1988iterative,benner2013efficient,benner2013reformulated} computes a low-rank solution of \eqref{lyap_p}, approximating $P_1$ as $P_1\approx Z_1^{(k)}\big(Z_1^{(k)}\big)^*$. The corresponding residual $R_1^{(k)}$ is given by:
\begin{align}
R_1^{(k)}=A_1Z_1^{(k)}\big(Z_1^{(k)}\big)^*E_1^\top+E_1Z_1^{(k)}\big(Z_1^{(k)}\big)^*A_1^\top+B_1B_1^\top.
\end{align}
\begin{algorithm}
\DontPrintSemicolon
\caption{CF-ADI}\label{cfadi_alg}
\KwIn{Matrices of Lyapunov equation \eqref{lyap_p}: $A_1, B_1, E_1$; ADI shifts: $\{\alpha_i\}_{i=1}^k \in \mathbb{C}_{-}$.}
\KwOut{Approximation of $P_1$: $P_1 \approx Z_1^{(k)}\big(Z_1^{(k)}\big)^*$; Residual: $R_1^{(k)} = B_{\perp,k}^{\mathrm{lyap}} \big(B_{\perp,k}^{\mathrm{lyap}}\big)^{*}$.}

Initialize: $Z_1^{(0)} = [\hspace*{0.15cm}]$, \ $B_{\perp,0}^{\mathrm{lyap}} = B_1$.\;

\For{$i = 1$ \KwTo $k$}{
    Solve $(A_1 + \alpha_i E_1) v_i^{\mathrm{lyap}} = B_{\perp,i-1}^{\mathrm{lyap}}$ for $v_i^{\mathrm{lyap}}$.\;
    
    Expand $Z_1^{(i)} = \left[ Z_1^{(i-1)} \quad \sqrt{-2\operatorname{Re}(\alpha_i)} \, v_i^{\mathrm{lyap}} \right]$.\;
    
    Update $B_{\perp,i}^{\mathrm{lyap}} = B_{\perp,i-1}^{\mathrm{lyap}} - 2\operatorname{Re}(\alpha_i) \, E_1 v_i^{\mathrm{lyap}}$.\;
}
\end{algorithm}
\begin{remark}
CF-ADI can also compute a low-rank solution of \eqref{lyap_p_ldl}, approximating $P_s$ as $P_s\approx Z_1^{(k)}\big(I \otimes S_1\big)\big(Z_1^{(k)}\big)^*$. The residual $R_{s_1}^{(k)}$,
\begin{align}
R_{s_1}^{(k)}=A_1Z_1^{(k)}\big(I \otimes S_1\big)\big(Z_1^{(k)}\big)^*E_1^\top+E_1Z_1^{(k)}\big(I \otimes S_1\big)\big(Z_1^{(k)}\big)^*A_1^\top+B_1S_1B_1^\top,
\end{align} can be computed as $R_{s_1}^{(k)}=B_{\perp,k}^{\mathrm{lyap}}S_1\big(B_{\perp,k}^{\mathrm{lyap}}\big)^*$. Although the original $LDL^T$-ADI formulation in \cite{lang2015benefits} differs slightly, the approximation $P_s\approx Z_1^{(k)}\big(I \otimes S_1\big)\big(Z_1^{(k)}\big)^*$ and the residual expression $R_{s_1}^{(k)}=B_{\perp,k}^{\mathrm{lyap}}S_1\big(B_{\perp,k}^{\mathrm{lyap}}\big)^*$ are equivalent after a rearrangement of variables.
\end{remark}
Since $m_1 \ll n_1$, matrix norms such as the $L_2$-norm and Frobenius norm of $R_1$ and $R_{s_1}$ can be computed efficiently.
\subsection{Factorized ADI Method for Sylvester Equations (FADI)}
The FADI algorithm \cite{benner2009adi,benner2014computing} computes a low-rank solution of \eqref{x_sylv}, approximating $X_{\mathrm{sylv}}$ as $X_{\mathrm{sylv}}\approx V_{\mathrm{fadi}}^{(k)} D_{\mathrm{fadi}}^{(k)} \big(W_{\mathrm{fadi}}^{(k)}\big)^*$. The corresponding residual $R_{\mathrm{sylv}}^{(k)}$ is defined by:
\begin{align}
R_{\mathrm{sylv}}^{(k)}=A_1V_{\mathrm{fadi}}^{(k)} D_{\mathrm{fadi}}^{(k)} \big(W_{\mathrm{fadi}}^{(k)}\big)^*E_2+E_1V_{\mathrm{fadi}}^{(k)} D_{\mathrm{fadi}}^{(k)} \big(W_{\mathrm{fadi}}^{(k)}\big)^*A_2+B_1C_2.
\end{align}
\begin{algorithm}
\DontPrintSemicolon
\caption{FADI}\label{fadi_alg}
\KwIn{Matrices of Sylvester Equation \eqref{x_sylv}: $A_1$, $B_1$, $E_1$, $A_2$, $C_2$, $E_2$; \\
       ADI shifts: $\{\alpha_i\}_{i=1}^k$ and $\{\beta_i\}_{i=1}^k$ such that $\alpha_i \neq -\beta_i$, $\operatorname{Re}(\alpha_i) \neq 0$, and $\operatorname{Re}(\beta_i) \neq 0$.}
\KwOut{Approximation of $X_{\mathrm{sylv}}$: $X_{\mathrm{sylv}}\approx V_{\mathrm{fadi}}^{(k)} D_{\mathrm{fadi}}^{(k)} \big(W_{\mathrm{fadi}}^{(k)}\big)^*$; Residual: $R_{\mathrm{sylv}}^{(k)} = B_{\perp,k}^{\mathrm{sylv}} C_{\perp,k}^{\mathrm{sylv}}$.}

Initialize: $B_{\perp,0}^{\mathrm{sylv}} = B_1$, $V_{\mathrm{fadi}}^{(0)} = [\hspace*{0.15cm}]$, $C_{\perp,0}^{\mathrm{sylv}} = C_2$, $W_{\mathrm{fadi}}^{(0)} = [\hspace*{0.15cm}]$, $D_{\mathrm{fadi}}^{(0)} = [\hspace*{0.15cm}]$.\;

\For{$i = 1$ \KwTo $k$}{
    Solve $(A_1 + \alpha_i E_1) v_i^{\mathrm{sylv}} = B_{\perp,i-1}^{\mathrm{sylv}}$ for $v_i^{\mathrm{sylv}}$.\; \label{fadi_step_3}
    
    Solve $(A_2^\top + \overline{\beta_i} E_2^\top) w_i^{\mathrm{sylv}} = (C_{\perp,i-1}^{\mathrm{sylv}})^*$ for $w_i^{\mathrm{sylv}}$.\; \label{fadi_step_4}
    
    Expand $V_{\mathrm{fadi}}^{(i)} = [V_{\mathrm{fadi}}^{(i-1)} \quad v_i^{\mathrm{sylv}}]$, $W_{\mathrm{fadi}}^{(i)} = [W_{\mathrm{fadi}}^{(i-1)} \quad w_i^{\mathrm{sylv}}]$, and $D_{\mathrm{fadi}}^{(i)} = \operatorname{blkdiag}\bigl(D_{\mathrm{fadi}}^{(i-1)}, -(\alpha_i + \beta_i) I\bigr)$.\;
    
    Update $B_{\perp,i}^{\mathrm{sylv}} = B_{\perp,i-1}^{\mathrm{sylv}} - (\alpha_i + \beta_i) E_1 v_i^{\mathrm{sylv}}$ and $C_{\perp,i}^{\mathrm{sylv}} = C_{\perp,i-1}^{\mathrm{sylv}} - (\alpha_i + \beta_i) (w_i^{\mathrm{sylv}})^* E_2$.\;
}
\end{algorithm}
\subsection{ADI-type Method for Riccati Equations (RADI)}
The RADI \cite{benner2018radi,bertram2024family} algorithm computes a low-rank solution of \eqref{ricc_p}, approximating $P_{\mathrm{ricc}}$ as $P_{\mathrm{ricc}}\approx V_{\mathrm{radi}}^{(k)} \hat{P}_{\mathrm{ricc}}^{(k)} \big(V_{\mathrm{radi}}^{(k)}\big)^{*}$. The corresponding residual $R_{\mathrm{p,ricc}}^{(k)}$ is defined by:
\begin{align}
R_{\mathrm{p,ricc}}^{(k)}&=A_1V_{\mathrm{radi}}^{(k)} \hat{P}_{\mathrm{ricc}}^{(k)} \big(V_{\mathrm{radi}}^{(k)}\big)^{*}E_1^\top+E_1V_{\mathrm{radi}}^{(k)} \hat{P}_{\mathrm{ricc}}^{(k)} \big(V_{\mathrm{radi}}^{(k)}\big)^{*}A_1^\top+B_1B_1^\top\nonumber\\
&\hspace*{4cm}-E_1V_{\mathrm{radi}}^{(k)} \hat{P}_{\mathrm{ricc}}^{(k)} \big(V_{\mathrm{radi}}^{(k)}\big)^{*}C_1^\top C_1V_{\mathrm{radi}}^{(k)} \hat{P}_{\mathrm{ricc}}^{(k)} \big(V_{\mathrm{radi}}^{(k)}\big)^{*}E_1^\top.
\end{align}
\begin{algorithm}
\caption{RADI}\label{radi_alg}
\DontPrintSemicolon
\KwIn{Matrices of Riccati equation \eqref{ricc_p}: $A_1$, $B_1$, $C_1$, $E_1$; ADI shifts: $\{\alpha_i\}_{i=1}^k \in \mathbb{C}_{-}$.}
\KwOut{Approximation of $P_{\mathrm{ricc}}$: $P_{\mathrm{ricc}} \approx V_{\mathrm{radi}}^{(k)} \hat{P}_{\mathrm{ricc}}^{(k)} \big(V_{\mathrm{radi}}^{(k)}\big)^{*}$; Residual: $R_{\mathrm{p,ricc}}^{(k)} = B_{\perp,k}^{\mathrm{ricc}} (B_{\perp,k}^{\mathrm{ricc}})^{*}$.}

Initialize: $V_{\mathrm{radi}}^{(0)} = [\hspace*{0.15cm}]$, $\hat{P}_{\mathrm{ricc}}^{(0)} = [\hspace*{0.15cm}]$, $B_{\perp,0}^{\mathrm{ricc}} = B_1$.\;

\For{$i = 1$ \KwTo $k$}{
    Solve $\Big( A_1 - E_1 V_{\mathrm{radi}}^{(i-1)} \hat{P}_{\mathrm{ricc}}^{(i-1)} \big(V_{\mathrm{radi}}^{(i-1)}\big)^{*} C_1^\top C_1 + \alpha_i E_1 \Big) v_i^{\mathrm{ricc}} = B_{\perp,i-1}^{\mathrm{ricc}}$ for $v_i^{\mathrm{ricc}}$.\; \label{radi_step6}
    
    Compute $\hat{p}_{\mathrm{ricc}}^{(i)} = -2\operatorname{Re}(\alpha_i) \big[ I + (v_i^{\mathrm{ricc}})^* C_1^\top C_1 v_i^{\mathrm{ricc}} \big]^{-1}$.\;
    
    Expand $V_{\mathrm{radi}}^{(i)} = [V_{\mathrm{radi}}^{(i-1)} \quad v_i^{\mathrm{ricc}}]$ and $\hat{P}_{\mathrm{ricc}}^{(i)} = \operatorname{blkdiag}\bigl(\hat{P}_{\mathrm{ricc}}^{(i-1)},\, \hat{p}_{\mathrm{ricc}}^{(i)}\bigr)$.\;
    
    Update $B_{\perp,i}^{\mathrm{ricc}} = B_{\perp,i-1}^{\mathrm{ricc}} - 2\operatorname{Re}(\alpha_i) E_1 v_i^{\mathrm{ricc}} \big[ I + (v_i^{\mathrm{ricc}})^* C_1^\top C_1 v_i^{\mathrm{ricc}} \big]^{-1}$.\;
}
\end{algorithm}
\begin{remark}
In large-scale settings, explicitly forming $V_{\mathrm{radi}}^{(i-1)} \hat{P}_{\mathrm{ricc}}^{(i-1)} \big(V_{\mathrm{radi}}^{(i-1)}\big)^{*}$
in Step~\ref{radi_step6} of RADI is infeasible due to the size of the matrices $(n_1 \times n_1)$. It is therefore recommended in \cite{benner2018radi} to rewrite the linear system in Step~\ref{radi_step6} as $\mathcal{A} v_i^{\text{ricc}} = \mathcal{B}$, where the right-hand side has $m_1 + p_1$ columns instead of $m_1$, using the Sherman–Morrison–Woodbury (SMW) formula; see \cite{benner2018radi,golub2013matrix} for details. As a result, the linear solves in RADI are more expensive than in CF-ADI, whose right-hand sides have only $m_1$ columns.
\end{remark}
\subsection{Interpolation Theory}
Let ${\sigma_1, \dots, \sigma_k}$ and ${\mu_1, \dots, \mu_k}$ be two sets of interpolation points. Define the projection matrices:
\begin{align}
    V_1 &= \left[ (\sigma_1 E_1 - A_1)^{-1} B_1, \dots, (\sigma_k E_1 - A_1)^{-1} B_1 \right], \label{v_kry}\\
    W_2 &= \left[ (\overline{\mu_1} E_2^\top - A_2^\top)^{-1} C_2^\top, \dots, (\overline{\mu_k} E_2^\top - A_2^\top)^{-1} C_2^\top \right].\label{w_kry}
\end{align}
These matrices satisfy the Sylvester equations:
\begin{align}
    A_1 V_1 - E_1 V_1 S_v + B_1 L_v &= 0, \label{v_kry_sylv}\\
    A_2^\top W_2 - E_2^\top W_2 S_w + C_2^\top L_w &= 0,\label{w_kry_sylv}
\end{align}
where
\begin{align}
    S_v &= \operatorname{blkdiag}(\sigma_1I, \dots, \sigma_kI), \quad L_v = \begin{bmatrix}I& \dots& I\end{bmatrix}, \label{SvLv_kry}\\
    S_w &= \operatorname{blkdiag}(\overline{\mu_1}I, \dots, \overline{\mu_k}I), \quad L_w = \begin{bmatrix}I& \dots& I\end{bmatrix}.\label{SwLw_kry}
\end{align}
Compute the reduced matrices as:
\begin{align}
    \hat{E}_1 &= W_1^* E_1 V_1, \quad
    \hat{A}_1 = W_1^* A_1 V_1, \quad
    \hat{B}_1 = W_1^* B_1, \\
    \hat{E}_2 &= W_2^* E_2 V_2, \quad
    \hat{A}_2 = W_2^* A_2 V_2, \quad
    \hat{C}_2 = C_2 V_2,
\end{align}
with $W_1 \in \mathbb{C}^{n_1 \times k m_1}$ and $V_2 \in \mathbb{C}^{n_2 \times k p_2}$.

Assuming $V_1$, $W_1$, $V_2$, and $W_2$ have full column rank and that $\hat{E}_1$ and $\hat{E}_2$ are invertible, the following interpolation conditions hold for $i = 1, \dots, k$:
\begin{align}
    (\sigma_i E_1 - A_1)^{-1} B_1 &= V_1 (\sigma_i \hat{E}_1 - \hat{A}_1)^{-1} \hat{B}_1, \label{eq_int1}\\
    C_2 (\mu_i E_2 - A_2)^{-1} &= \hat{C}_2 (\mu_i \hat{E}_2 - \hat{A}_2)^{-1} W_2^*;\label{eq_int2}
\end{align}cf. \cite{beattie2017chapter}.

Moreover,
\begin{align}
\hat{A}_1 = \hat{E}_1 S_v - \hat{B}_1 L_v
    \quad \mathrm{and} \quad
    \hat{A}_2 = S_w^* \hat{E}_2 - L_w^\top \hat{C}_2.\end{align}
Thus, $\hat{A}_1$ can be computed without explicitly constructing $W_1$, as it is parameterized by $\hat{E}_1$ and $\hat{B}_1$; see \cite{wolfthesis,panzerthesis}. Similarly, $\hat{A}_2$ can be computed without $V_2$, being parameterized by $\hat{E}_2$ and $\hat{C}_2$.

Define the matrices $B_\perp$ and $C_\perp$ as:
\begin{align}
    B_\perp &= B_1 - E_1 V_1 \hat{E}_1^{-1} \hat{B}_1, \\
    C_\perp &= C_2 - \hat{C}_2 \hat{E}_2^{-1} W_2^* E_2.
\end{align}
Then $V_1$ and $W_2$ also satisfy the Sylvester equations:
\begin{align}
    A_1 V_1 - E_1 V_1 \hat{E}_1^{-1} \hat{A}_1 + B_\perp L_v &= 0, \label{v_kry_sylv_2}\\
    A_2^\top W_2 - E_2^\top W_2 \hat{E}_2^{-*}\hat{A}_2^* + C_\perp^* L_w &= 0;\label{w_kry_sylv_2}
\end{align}cf. \cite{wolfthesis,panzerthesis}.
\subsection{Recursive Interpolation Framework}
The recursive interpolation framework in \cite{panzer2013greedy} iteratively enforces the interpolation conditions \eqref{eq_int1} and \eqref{eq_int2}. The reduced matrices $\hat{E}_1$, $\hat{A}_1$, $\hat{B}_1$, $\hat{E}_2$, $\hat{A}_2$, and $\hat{C}_2$ are constructed recursively, and the projection matrices $V_1$ and $W_2$ are updated in the same manner.

Initialize the algorithm with:
\begin{align}
    B_\perp^{(0)} &= B_1, \quad
    V_1^{(0)} = [\hspace*{0.15cm}], \quad
    S_v^{(0)} = [\hspace*{0.15cm}], \quad
    L_v^{(0)} = [\hspace*{0.15cm}], \quad
    \hat{E}_1^{(0)} = [\hspace*{0.15cm}], \quad
    \hat{A}_1^{(0)} = [\hspace*{0.15cm}], \quad
    \hat{B}_1^{(0)} = [\hspace*{0.15cm}], \nonumber\\
    C_\perp^{(0)} &= C_2, \quad
    W_2^{(0)} = [\hspace*{0.15cm}], \quad
    S_w^{(0)} = [\hspace*{0.15cm}], \quad
    L_w^{(0)} = [\hspace*{0.15cm}], \quad
    \hat{E}_2^{(0)} = [\hspace*{0.15cm}], \quad
    \hat{A}_2^{(0)} = [\hspace*{0.15cm}], \quad
    \hat{C}_2^{(0)} = [\hspace*{0.15cm}].
\end{align}
For $i = 1, \dots, k$, set $s_v^{(i)}=\sigma_i I$, $l_v^{(i)}=I$, $s_w^{(i)}=\overline{\mu_i} I$, and $l_w^{(i)}=I$. Then compute:
\begin{align}
    v_i &= (\sigma_i E_1 - A_1)^{-1} B_\perp^{(i-1)}l_v^{(i)},\label{recc_v}\\
    w_i &= (\overline{\mu_i} E_2^\top - A_2^\top)^{-1} \big(C_\perp^{(i-1)}\big)^*l_w^{(i)},\label{recc_w}
\end{align}
with
\begin{align}
    B_\perp^{(i)} &= B_\perp^{(i-1)} - E_1 v_i (\hat{e}_1^{(i)})^{-1} \hat{b}_1^{(i)}, \label{b_perp_rec}\\
    C_\perp^{(i)} &= C_\perp^{(i-1)} - \hat{c}_2^{(i)} (\hat{e}_2^{(i)})^{-1} w_i^* E_2.\label{c_perp_rec}
\end{align}
Here, $\hat{e}_1^{(i)}\in\mathbb{C}^{m_1\times m_1}$, $\hat{b}_1^{(i)}\in\mathbb{C}^{m_1\times m_1}$, $\hat{e}_2^{(i)}\in\mathbb{C}^{p_2\times p_2}$, $\hat{c}_2^{(i)}\in\mathbb{C}^{p_2\times p_2}$ are free parameters, with $\hat{e}_1^{(i)}$ and $\hat{e}_2^{(i)}$ invertible. Moreover, $v_i$ and $w_i$ satisfy the Sylvester equations:
\begin{align}
A_1v_i-E_1v_is_v^{(i)}+B_\perp^{(i-1)}l_v^{(i)}&=0,\\
A_2^\top w_i-E_2^\top w_i s_w^{(i)}+(C_\perp^{(i-1)})^*l_w^{(i)}&=0.
\end{align}
Setting $\hat{a}_1^{(i)}=\hat{e}_1^{(i)}s_v^{(i)}-\hat{b}_1^{(i)}l_v^{(i)}$ and $\hat{a}_2^{(i)}=(s_w^{(i)})^*\hat{e}_2^{(i)}-(l_w^{(i)})^\top \hat{c}_2^{(i)}$, the matrices are updated recursively as:
\begin{align}
    V_1^{(i)} &= \left[ V_1^{(i-1)} \quad v_i \right], &S_v^{(i)} &= \begin{bmatrix}
        S_v^{(i-1)} & (\hat{E}_1^{(i-1)})^{-1}\hat{B}_1^{(i-1)} l_v^{(i)} \\
        0 & s_v^{(i)}
    \end{bmatrix},& 
    L_v^{(i)}&= \left[ L_v^{(i-1)} \quad l_v^{(i)} \right], \nonumber\\
    \hat{E}_1^{(i)} &= \begin{bmatrix}
        \hat{E}_1^{(i-1)} & 0 \\
        0 & \hat{e}_1^{(i)}
    \end{bmatrix},& 
     \hat{A}_1^{(i)} &= \begin{bmatrix}
        \hat{A}_1^{(i-1)} & 0 \\
        -\hat{b}_1^{(i)}L_v^{(i-1)} & \hat{a}_1^{(i)}
    \end{bmatrix},&
    \hat{B}_1^{(i)}& = \begin{bmatrix}
        \hat{B}_1^{(i-1)} \\
        \hat{b}_1^{(i)}
    \end{bmatrix}, \label{rec1}\\
    W_2^{(i)}& = \left[ W_2^{(i-1)} \quad w_i \right],&
    S_w^{(i)} &= \begin{bmatrix}
        S_w^{(i-1)} & (\hat{E}_2^{(i-1)})^{-*}(\hat{C}_2^{(i-1)})^* l_w^{(i)} \\
        0 & s_w^{(i)}
    \end{bmatrix}, &
    L_w^{(i)} &= \left[ L_w^{(i-1)} \quad l_w^{(i)} \right],\nonumber\\
    \hat{E}_2^{(i)} &= \begin{bmatrix}
        \hat{E}_2^{(i-1)} & 0 \\
        0 & \hat{e}_2^{(i)}
    \end{bmatrix},&
    \hat{A}_2^{(i)}&=\begin{bmatrix}\hat{A}_2^{(i-1)}& -(L_w^{(i-1)})^\top \hat{c}_2^{(i)}\\0 &\hat{a}_2^{(i)}\end{bmatrix},&
    \hat{C}_2^{(i)} &= \left[ \hat{C}_2^{(i-1)} \quad \hat{c}_2^{(i)} \right].\label{rec2}
\end{align}

The matrices $B_\perp^{(i)}$, $C_\perp^{(i)}$, $\hat{A}_1^{(i)}$, and $\hat{A}_2^{(i)}$ can also be written as:
\begin{align}
    B_\perp^{(i)} &= B_1 - E_1 V_1^{(i)} (\hat{E}_1^{(i)})^{-1} \hat{B}_1^{(i)}, \\
    C_\perp^{(i)} &= C_2 - \hat{C}_2^{(i)} (\hat{E}_2^{(i)})^{-1} (W_2^{(i)})^* E_2,\\
        \hat{A}_1^{(i)} &= \hat{E}_1^{(i)} S_v^{(i)} - \hat{B}_1^{(i)} L_v^{(i)}, \\
    \hat{A}_2^{(i)} &= (S_w^{(i)})^* \hat{E}_2^{(i)}- (L_w^{(i)})^\top \hat{C}_2^{(i)}.
\end{align}
The projection matrices satisfy the Sylvester equations:
\begin{align}
    A_1 V_1^{(i)} - E_1 V_1^{(i)} S_v^{(i)} + B_1 L_v^{(i)} &= 0, \\
    A_1 V_1^{(i)} - E_1 V_1^{(i)} (\hat{E}_1^{(i)})^{-1} \hat{A}_1^{(i)} + B_\perp^{(i)} L_v^{(i)} &= 0, \\
    A_2^\top W_2^{(i)} - E_2^\top W_2^{(i)} S_w^{(i)} + C_2^\top L_w^{(i)} &= 0, \\
    A_2^\top W_2^{(i)} - E_2^\top W_2^{(i)} (\hat{E}_2^{(i)})^{-*}(\hat{A}_2^{(i)})^* + (C_\perp^{(i)})^* L_w^{(i)} &= 0.
\end{align}
The block triangular structure of $S_v^{(i)}$ and $S_w^{(i)}$ implies that their eigenvalues are $\sigma_i$ (with multiplicity $m_1$) and $\overline{\mu}_i$ (with multiplicity $p_2$), respectively. Given that the interpolation points are the eigenvalues of these matrices, and assuming $V_1^{(i)}$ and $W_2^{(i)}$ have full column rank, the following interpolation conditions hold:
\begin{align}
    (\sigma_i E_1 - A_1)^{-1} B_1 &= V_1^{(i)} (\sigma_i \hat{E}_1^{(i)} - \hat{A}_1^{(i)})^{-1} \hat{B}_1^{(i)}, \label{rec_int_1}\\
    C_2 (\mu_i E_2 - A_2)^{-1} &= \hat{C}_2^{(i)} (\mu_i \hat{E}_2^{(i)} - \hat{A}_2^{(i)})^{-1} (W_2^{(i)})^*, \label{rec_int_2}
\end{align}
for $i = 1, \dots, k$.
\subsection{CF-ADI as a Recursive Interpolation Algorithm}
It is noted in \cite{wolfthesis,wolf2016adi} that CF-ADI can be interpreted as a recursive interpolation algorithm.
Let $\alpha_i$ be the shifts used in CF-ADI. By setting $\sigma_i = -\alpha_i$, $s_v^{(i)}=\sigma_i I$, and $l_v^{(i)}=-I$ within the recursive interpolation framework, we define the matrix $q_{v,i}^{\mathrm{lyap}}$ as the solution to the Lyapunov equation:
\begin{align}
(-s_v^{(i)})^* q_{v,i}^{\mathrm{lyap}} + q_{v,i}^{\mathrm{lyap}} (-s_v^{(i)}) + (l_v^{(i)})^\top l_v^{(i)}= 0.
\end{align}
Now, set $\hat{e}_1^{(i)} = I$ and $\hat{b}_1^{(i)} = - (q_{v,i}^{\mathrm{lyap}})^{-1}(l_v^{(i)})^\top$ in \eqref{rec1}.
As shown in \cite{wolfthesis}, this setting satisfies the Lyapunov equation:
\[
    (-S_v^{(i)})^* Q_{v,i}^{\mathrm{lyap}} + Q_{v,i}^{\mathrm{lyap}} (-S_v^{(i)}) + (L_v^{(i)})^\top L_v^{(i)} = 0,
\]
where
\[
    Q_{v,i}^{\mathrm{lyap}} = \operatorname{blkdiag}\left( q_{v,1}^{\mathrm{lyap}}, \dots, q_{v,i}^{\mathrm{lyap}} \right).
\]
Since $q_{v,i}^{\mathrm{lyap}} = -\frac{1}{2\operatorname{Re}(\alpha_i)} I$, $(Q_{v,i}^{\mathrm{lyap}})^{-1} = \operatorname{blkdiag}\left( -2\operatorname{Re}(\alpha_1) I, \dots, -2\operatorname{Re}(\alpha_i) I \right)$, and $\hat{a}_1^{(i)}=s_v^{(i)}-\hat{b}_1^{(i)}l_v^{(i)}=-(q_{v,1}^{\mathrm{lyap}})^{-1}(s_v^{(i)})^*q_{v,1}^{\mathrm{lyap}}$, it follows that $(Q_{v,i}^{\mathrm{lyap}})^{-1}$ satisfies the projected Lyapunov equation:
\begin{align}
    \hat{A}_1^{(i)} (Q_{v,i}^{\mathrm{lyap}})^{-1} (\hat{E}_1^{(i)})^\top+ \hat{E}_1^{(i)}(Q_{v,i}^{\mathrm{lyap}})^{-1} (\hat{A}_1^{(i)})^* + \hat{B}_1^{(i)} (\hat{B}_1^{(i)})^* = 0,
\end{align}with
\[
\hat{E}_1^{(i)}=I,\quad  \hat{A}_1^{(i)}=-(Q_{v,i}^{\mathrm{lyap}})^{-1}(S_v^{(i)})^*Q_{v,i}^{\mathrm{lyap}},\quad \hat{B}_1^{(i)}=(Q_{v,i}^{\mathrm{lyap}})^{-1} (L_v^{(i)})^\top.
\]
Using $V_1^{(i)}$ from \eqref{rec1}, a projection-based approximation of $P_1$ is given by:
\[
    P_1 \approx V_1^{(i)} (Q_{v,i}^{\mathrm{lyap}})^{-1} (V_1^{(i)})^*.
\]
This approximation is identical to the one produced by the CF-ADI method, where:
\[
    Z_1^{(k)} = V_1^{(k)} \sqrt{(Q_{v,k}^{\mathrm{lyap}})^{-1}}\quad \mathrm{and}\quad B_{\perp,k}^{\mathrm{lyap}}=B_\perp^{(k-1)}-E_1v_k(q_{v,k}^{\mathrm{lyap}})^{-1}(l_v^{(k)})^\top.
\]
Finally, if $V_1^{(i)}$ in \eqref{rec1} has full column rank, the interpolation condition \eqref{rec_int_1} holds. Hence, the CF-ADI method interpolates at the mirror images of the ADI shifts $\alpha_i$.
\section{Main Work}
This section first demonstrates that FADI and RADI, like CF-ADI, are recursive interpolation algorithms that perform interpolation at the mirror images of the ADI shifts. It then shows how CF-ADI can be extended to solve Sylvester and Riccati equations, resulting in a unified ADI algorithm. The core property enabling this unification is the pole-placement property of CF-ADI, which allows the eigenvalues of the projected matrices to be placed at specific locations to solve multiple Lyapunov, Sylvester, and Riccati equations simultaneously. Furthermore, it is shown that the ROMs implicitly constructed by the proposed unified ADI algorithm can be accumulated. These ROMs preserve important properties of the original systems $G_1(s)$ and $G_2(s)$, such as stability, passivity, and the minimum-phase property. Finally, the self-generating shift strategies, in which the unified ADI algorithm produces the ADI shifts without user involvement, are discussed.
\subsection{FADI as a Recursive Interpolation Algorithm}
This subsection shows that FADI performs recursive interpolation at $-\alpha_i$ and $-\beta_i$.
\begin{proposition}\label{prop1}
Let $\{\alpha_i\}_{i=1}^{k}\in\mathbb{C}$ and $\{\beta_i\}_{i=1}^{k}\in\mathbb{C}$ be the ADI shifts used in FADI, with $\alpha_i\neq -\beta_i$. Define $s_v^{(i)}=-\alpha_i I$, $l_v^{(i)}=-I$, $s_w^{(i)}=-\overline{\beta}_i I$, and $l_w^{(i)}=-I$. Set the free parameters $\hat{e}_1^{(i)}$, $\hat{b}_1^{(i)}$, $\hat{e}_2^{(i)}$, and $\hat{c}_2^{(i)}$ as:
\begin{align}
\hat{e}_1^{(i)}=I,\quad \hat{b}_1^{(i)}=(d_{\mathrm{sylv}}^{(i)})^{-1}(l_w^{(i)})^\top,\quad \hat{e}_2^{(i)}=I,\quad \hat{c}_2^{(i)}=l_v^{(i)}(d_{\mathrm{sylv}}^{(i)})^{-1},\label{fadi_free}
\end{align}where $d_{\mathrm{sylv}}^{(i)}$ solves the Sylvester equation:
\begin{align}
   (-s_w^{(i)})^* d_{\mathrm{sylv}}^{(i)} + d_{\mathrm{sylv}}^{(i)} (-s_v^{(i)}) +(l_w^{(i)})^\top l_v^{(i)}  = 0.\label{eq_sylv}
\end{align}
Then the following hold:
\begin{enumerate}
  \item\label{1_of_prop1} $\hat{a}_1^{(i)}=\hat{e}_1^{(i)}s_v^{(i)}-\hat{b}_1^{(i)}l_v^{(i)}=-(d_{\mathrm{sylv}}^{(i)})^{-1}(s_w^{(i)})^*d_{\mathrm{sylv}}^{(i)}$.
  \item\label{2_of_prop1} $\hat{a}_2^{(i)}=(s_w^{(i)})^*\hat{e}_2^{(i)}-(l_w^{(i)})^\top\hat{c}_2^{(i)}=-d_{\mathrm{sylv}}^{(i)}s_v^{(i)}(d_{\mathrm{sylv}}^{(i)})^{-1}$.
  \item\label{3_of_prop1} $\hat{a}_1(d_{\mathrm{sylv}}^{(i)})^{-1}\hat{e}_2^{(i)}+\hat{e}_1^{(i)}(d_{\mathrm{sylv}}^{(i)})^{-1}\hat{a}_2^{(i)}+\hat{b}_1^{(i)}\hat{c}_2^{(i)}=0$.
  \item $d_{\mathrm{sylv}}^{(i)} = -\frac{1}{\alpha_i + \beta_i} I$ and $(d_{\mathrm{sylv}}^{(i)})^{-1}=-(\alpha_i + \beta_i)I$.
\end{enumerate}
\end{proposition}
\begin{proof}
The proof is given in Appendix A.
\end{proof}
\begin{proposition}\label{prop2}
Let Proposition~\ref{prop1} hold. Define $S_v^{(i)}$, $L_v^{(i)}$, $\hat{E}_1^{(i)}$, $\hat{A}_1^{(i)}$, $\hat{B}_1^{(i)}$, $S_w^{(i)}$, $L_w^{(i)}$, $\hat{E}_2^{(i)}$, $\hat{A}_2^{(i)}$, and $\hat{C}_2^{(i)}$ using the recursive formulas \eqref{rec1} and \eqref{rec2}. Let $D_{\mathrm{sylv}}^{(i)} = \mathrm{blkdiag}\left(-\frac{1}{\alpha_1 + \beta_1} I,\cdots,-\frac{1}{\alpha_i + \beta_i} I\right)$.
Then the following hold:
\begin{enumerate}
\item $\hat{B}_1^{(i)}=(D_{\mathrm{sylv}}^{(i)})^{-1}(L_w^{(i)})^\top$ and $\hat{C}_2^{(i)}=L_v^{(i)}(D_{\mathrm{sylv}}^{(i)})^{-1}$.
\item The matrix $D_{\mathrm{sylv}}^{(i)}$ satisfies the Sylvester equation:
\begin{align}(-S_w^{(i)})^* D_{\mathrm{sylv}}^{(i)} + D_{\mathrm{sylv}}^{(i)} (-S_v^{(i)}) +(L_w^{(i)})^\top L_v^{(i)}  = 0.\end{align}
  \item\label{3_of_prop2} $\hat{A}_1^{(i)}=\hat{E}_1^{(i)}S_v^{(i)}-\hat{B}_1^{(i)}L_v^{(i)}=-(D_{\mathrm{sylv}}^{(i)})^{-1}(S_w^{(i)})^*D_{\mathrm{sylv}}^{(i)}$.
  \item\label{4_of_prop2}  $\hat{A}_2^{(i)}=(S_w^{(i)})^*\hat{E}_2^{(i)}-(L_w^{(i)})^\top\hat{C}_2^{(i)}=-D_{\mathrm{sylv}}^{(i)}S_v^{(i)}(D_{\mathrm{sylv}}^{(i)})^{-1}$.
  \item\label{5_of_prop2} The following projected Sylvester equation holds:
      \begin{align}
      \hat{A}_1(D_{\mathrm{sylv}}^{(i)})^{-1}\hat{E}_2^{(i)}+\hat{E}_1^{(i)}(D_{\mathrm{sylv}}^{(i)})^{-1}\hat{A}_2^{(i)}+\hat{B}_1^{(i)}\hat{C}_2^{(i)}=0.\label{D_fadi_proj}
      \end{align}
  \end{enumerate}
\end{proposition}
\begin{proof}
The proof is given in Appendix B.
\end{proof}
Using the recursive formulas \eqref{b_perp_rec} and \eqref{c_perp_rec}, the matrices $B_\perp^{(i)}$ and $C_\perp^{(i)}$ are given by:
\begin{align}
B_\perp^{(i)}&=B_\perp^{(i-1)}-(\alpha_i+\beta_i)E_1v_i,\\
C_\perp^{(i)}&=C_\perp^{(i-1)}-(\alpha_i+\beta_i)w_i^*E_2.
\end{align}
Furthermore, from \eqref{recc_v} and \eqref{recc_w}, $v_i$ and $w_i$ are computed as:
\begin{align}
v_i&=(A_1+\alpha_iE_1)^{-1}B_\perp^{(i-1)}\\
w_i&=(A_2^\top+\overline{\beta}_iE_2^\top)^{-1}(C_\perp^{(i-1)})^*.
\end{align}
Thus, when the free parameters $\hat{e}_1^{(i)}$, $\hat{b}_1^{(i)}$, $\hat{e}_2^{(i)}$, and $\hat{c}_2^{(i)}$ are selected according to Proposition \ref{prop1} in the recursive interpolation framework, we obtain:
\[
V_{\mathrm{fadi}}^{(i)} = V_1^{(i)}, \quad W_{\mathrm{fadi}}^{(i)} = W_2^{(i)}, \quad B_{\perp,i-1}^\mathrm{sylv}=B_\perp^{(i-1)}, \quad C_{\perp,i-1}^\mathrm{sylv}=C_\perp^{(i-1)}, \quad D_{\mathrm{fadi}}^{(i)} = (D_{\mathrm{sylv}}^{(i)})^{-1}.
\] Consequently, the projection-based approximation $V_1^{(k)} (D_{\mathrm{sylv}}^{(k)})^{-1} (W_2^{(k)})^*$ from the recursive interpolation framework matches the FADI approximation:
\[
    X_{\mathrm{sylv}} \approx V_1^{(k)} \big(D_{\mathrm{sylv}}^{(k)}\big)^{-1} \big(W_2^{(k)}\big)^*=V_{\mathrm{fadi}}^{(k)}D_{\mathrm{fadi}}^{(k)}\big(W_{\mathrm{fadi}}^{(k)}\big)^*.
\]
Assuming $V_1^{(i)}$ in \eqref{recc_v} and $W_2^{(i)}$ in \eqref{recc_w} have full column rank, the interpolation conditions \eqref{rec_int_1} and \eqref{rec_int_2} hold, meaning FADI interpolates at the mirror images of the ADI shifts $\alpha_i$ and $\beta_i$.
\subsection{RADI as a Recursive Interpolation Algorithm}
We now demonstrate that RADI also performs recursive interpolation at the mirror images of the ADI shifts $\alpha_i$. Although RADI's recursive interpolation framework differs slightly from that in \cite{panzer2013greedy}, the overall approach remains conceptually similar.
\begin{theorem}\label{th_radi}
Assume all variables in Algorithm \ref{radi_alg} are well-defined. For $i=1,\dots,k$, define the following quantities:
\begin{align}
s_{v,\mathrm{ricc}}^{(i)}&=-\alpha_i I,\quad l_{v,\mathrm{ricc}}^{(i)}=-I,\quad \hat{e}_1^{(i)}=I, \quad \hat{b}_1^{(i)}=\hat{p}_{\mathrm{ricc}}^{(i)}(l_{v,\mathrm{ricc}}^{(i)})^\top,\quad \hat{a}_1^{(i)}=\hat{e}_1^{(i)}s_{v,\mathrm{ricc}}^{(i)}-\hat{b}_1^{(i)}l_{v,\mathrm{ricc}}^{(i)},\quad \hat{c}_1^{(i)}=C_1v_i^{\mathrm{ricc}}\nonumber\\
 S_{v,\mathrm{ricc}}^{(i)} &= \begin{bmatrix}
        S_{v,\mathrm{ricc}}^{(i-1)} & \hat{P}_{\mathrm{ricc}}^{(i-1)}\Big((L_{v,\mathrm{ricc}}^{(i-1)})^\top l_{v,\mathrm{ricc}}^{(i)}\ + (\hat{C}_1^{(i-1)})^* \hat{c}_1^{(i)}\Big) \\
        0 & s_{v,\mathrm{ricc}}^{(i)}\end{bmatrix},\quad L_{v,\mathrm{ricc}}^{(i)}=\begin{bmatrix}L_{v,\mathrm{ricc}}^{(i-1)}&l_{v,\mathrm{ricc}}^{(i)}\end{bmatrix},\nonumber\\
\hat{E}_1^{(i)}&=\begin{bmatrix}\hat{E}_1^{(i-1)}&0\\0& \hat{e}_1^{(i)}\end{bmatrix}=I,\quad
\hat{A}_1^{(i)}=\begin{bmatrix}\hat{A}_1^{(i-1)}&\hat{P}_{\mathrm{ricc}}^{(i-1)}(\hat{C}_1^{(i-1)})^* \hat{c}_1^{(i)}\\-\hat{b}_1^{(i)}L_{v,\mathrm{ricc}}^{(i-1)}&\hat{a}_1^{(i)}\end{bmatrix},\nonumber\\ \hat{B}_1^{(i)}&=\begin{bmatrix}\hat{B}_1^{(i-1)}\\\hat{b}_1^{(i)}\end{bmatrix}=\begin{bmatrix}\hat{P}_{\mathrm{ricc}}^{(i-1)}(L_{v,\mathrm{ricc}}^{(i-1)})^\top\\\hat{p}_{\mathrm{ricc}}^{(i)}(l_{v,\mathrm{ricc}}^{(i)})^\top\end{bmatrix}=\hat{P}_{\mathrm{ricc}}^{(i)}(L_{v,\mathrm{ricc}}^{(i)})^\top,\nonumber\\
\hat{C}_1^{(i)}&=\begin{bmatrix}\hat{C}_1^{(i-1)}&\hat{c}_1^{(i)}\end{bmatrix}=\begin{bmatrix}\hat{C}_1^{(i-1)}&C_1v_i^{\mathrm{ricc}}\end{bmatrix}=C_1\begin{bmatrix}V_{\mathrm{radi}}^{(i-1)}&v_i^{\mathrm{ricc}}\end{bmatrix}=C_1V_{\mathrm{radi}}^{(i)}.\label{radi_free}
\end{align}
Then the following statements hold:
\begin{enumerate}
  \item $(\hat{p}_{\mathrm{ricc}}^{(i)})^{-1}$ solves the Lyapunov equation:
  \begin{align}
  (-s_{v,\mathrm{ricc}}^{(i)})^*(\hat{p}_{\mathrm{ricc}}^{(i)})^{-1}+(\hat{p}_{\mathrm{ricc}}^{(i)})^{-1}(-s_{v,\mathrm{ricc}}^{(i)})+(l_{v,\mathrm{ricc}}^{(i)})^\top l_{v,\mathrm{ricc}}^{(i)}+(\hat{c}_1^{(i)})^*\hat{c}_1^{(i)}=0.\label{lyap_p_ricc}
  \end{align}
  \item $(\hat{P}_{\mathrm{ricc}}^{(i)})^{-1}$ solves the Lyapunov equation:
  \begin{align}
  (-S_{v,\mathrm{ricc}}^{(i)})^*(\hat{P}_{\mathrm{ricc}}^{(i)})^{-1}+(\hat{P}_{\mathrm{ricc}}^{(i)})^{-1}(-S_{v,\mathrm{ricc}}^{(i)})+(L_{v,\mathrm{ricc}}^{(i)})^\top L_{v,\mathrm{ricc}}^{(i)}+(\hat{C}_1^{(i)})^* \hat{C}_1^{(i)}=0.\label{lyap_p_ricc_recc}
  \end{align}
  \item\label{th_radi_2} $\hat{A}_1^{(i)}=\hat{E}_1^{(i)}S_{v,\mathrm{ricc}}^{(i)}-\hat{B}_1^{(i)}L_{v,\mathrm{ricc}}^{(i)}=\hat{P}_{\mathrm{ricc}}^{(i)}\Big[(-S_{v,\mathrm{ricc}}^{(i)})^*+(\hat{C}_1^{(i)})^* \hat{C}_1^{(i)}\hat{P}_{\mathrm{ricc}}^{(i)}\Big](\hat{P}_{\mathrm{ricc}}^{(i)})^{-1}$.
  \item The matrix $\hat{A}_1^{(i)}-\hat{P}_{\mathrm{ricc}}^{(i)}(\hat{C}_1^{(i)})^* \hat{C}_1^{(i)}$ is Hurwitz.
  \item\label{th_radi_3} $\hat{P}_{\mathrm{ricc}}^{(i)}$ is a stabilizing solution to the projected Riccati equation:
  \begin{align}
  \hat{A}_1^{(i)}\hat{P}_{\mathrm{ricc}}^{(i)}(\hat{E}_1^{(i)})^\top+\hat{E}_1^{(i)}\hat{P}_{\mathrm{ricc}}^{(i)}(\hat{A}_1^{(i)})^*+\hat{B}_1^{(i)}(\hat{B}_1^{(i)})^*
  -\hat{E}_1^{(i)}\hat{P}_{\mathrm{ricc}}^{(i)}(\hat{C}_1^{(i)})^*\hat{C}_1^{(i)}\hat{P}_{\mathrm{ricc}}^{(i)}(\hat{E}_1^{(i)})^\top=0.\label{proj_ricc}
  \end{align}
  \item $B_{\perp,i}^{\mathrm{ricc}}=B_1-E_1V_{\mathrm{radi}}^{(i)}\hat{P}_{\mathrm{ricc}}^{(i)}(L_{v,\mathrm{ricc}}^{(i)})^\top$.
  \item $V_\mathrm{radi}^{(i)}$ satisfies the Sylvester equations:
  \begin{align}
  A_1V_\mathrm{radi}^{(i)}-E_1V_\mathrm{radi}^{(i)} S_{v,\mathrm{ricc}}^{(i)}+B_1L_{v,\mathrm{ricc}}^{(i)}&=0,\label{radi_sylv1}\\
  A_1V_\mathrm{radi}^{(i)}-E_1V_\mathrm{radi}^{(i)} (\hat{E}_1^{(i)})^{-1}\hat{A}_1^{(i)}+B_{\perp,i}^{\mathrm{ricc}}L_{v,\mathrm{ricc}}^{(i)}&=0.\label{radi_sylv2}
  \end{align}
\end{enumerate}
\end{theorem}
\begin{proof}
The proof is given in Appendix C.
\end{proof}
Due to the block triangular structure of $S_{v,\mathrm{ricc}}^{(i)}$, its eigenvalues are $-\alpha_i$ with multiplicity $m_1$. Based on the connection between Sylvester equations and rational interpolation established in \cite{gallivan2004sylvester}, if $V_\mathrm{radi}^{(i)}$ has full column rank, it enforces interpolation at $-\alpha_i$ (see \cite{gallivan2004sylvester} for details). Therefore, RADI recursively enforces the interpolation conditions:
\[
C_1(-\alpha_iE_1-A_1)^{-1}B_1=\hat{C}_1^{(i)}(-\alpha_i\hat{E}_1^{(i)}-\hat{A}_1^{(i)})^{-1}\hat{B}_1^{(i)},
\]
where $\hat{E}_1^{(i)}$, $\hat{A}_1^{(i)}$, $\hat{B}_1^{(i)}$, and $\hat{C}_1^{(i)}$ are as defined in Theorem \ref{th_radi}.
\subsection{A Unified ADI Framework}
In this subsection, we show that CF-ADI can be used to solve the linear matrix equations \eqref{lyap_p}–\eqref{ricc_q_sf} simultaneously, yielding a unified ADI framework.
\subsubsection{Unifying CF-ADI and FADI}
Let $\{\alpha_i\}_{i=1}^{k}\in\mathbb{C}_{-}$ and $\{\beta_i\}_{i=1}^{k}\in\mathbb{C}_{-}$ be the ADI shifts used in CF-ADI to approximate \eqref{lyap_p} and \eqref{lyap_q}, respectively. Define $s_{v,\mathrm{lyap}}^{(i)}=-\alpha_i I$, $l_{v,\mathrm{lyap}}^{(i)}=-I$, $s_{w,\mathrm{lyap}}^{(i)}=-\overline{\beta}_i I$, and $l_{w,\mathrm{lyap}}^{(i)}=-I$. Then $v_i^{\mathrm{lyap}}$ and $w_i^{\mathrm{lyap}}$ are computed as:
\begin{align}
v_i^{\mathrm{lyap}}=(A_1+\alpha_iE_1)^{-1}B_{\perp,i-1}^{\mathrm{lyap}},\label{lin_solve_1}\\
A_1v_i^{\mathrm{lyap}}-E_1v_i^{\mathrm{lyap}}s_{v,\mathrm{lyap}}^{(i)}+B_{\perp,i-1}^{\mathrm{lyap}}l_{v,\mathrm{lyap}}^{(i)}=0,\label{v_lyap}\\
w_i^{\mathrm{lyap}}=(A_2^\top+\overline{\beta}_iE_2^\top)^{-1}(C_{\perp,i-1}^{\mathrm{lyap}})^*,\label{lin_solve_2}\\
A_2^\top w_i^{\mathrm{lyap}}-E_2^\top w_i^{\mathrm{lyap}}s_{w,\mathrm{lyap}}^{(i)}+(C_{\perp,i-1}^{\mathrm{lyap}})^*l_{w,\mathrm{lyap}}^{(i)}=0,\label{w_lyap}
\end{align}
where
\begin{align}
B_{\perp,i}^{\mathrm{lyap}}&=B_{\perp,i-1}^{\mathrm{lyap}}-2\mathrm{Re}(\alpha_i)E_1v_i^{\mathrm{lyap}},& B_{\perp,0}^{\mathrm{lyap}}&=B_1,\\
C_{\perp,i}^{\mathrm{lyap}}&=C_{\perp,i-1}^{\mathrm{lyap}}-2\mathrm{Re}(\beta_i)(w_i^{\mathrm{lyap}})^*E_2,& C_{\perp,0}^{\mathrm{lyap}}&=C_2.
\end{align}
Define $V_{\mathrm{lyap}}^{(i)}$, $\hat{P}_{\mathrm{lyap}}^{(i)}$, $S_{v,\mathrm{lyap}}^{(i)}$, $L_{v,\mathrm{lyap}}^{(i)}$, $W_{\mathrm{lyap}}^{(i)}$, $\hat{Q}_{\mathrm{lyap}}^{(i)}$, $S_{w,\mathrm{lyap}}^{(i)}$ and $L_{w,\mathrm{lyap}}^{(i)}$ as:
\begin{align}
V_{\mathrm{lyap}}^{(i)}&=\begin{bmatrix}v_1^{\mathrm{lyap}}&\cdots&v_i^{\mathrm{lyap}}\end{bmatrix},& \hat{P}_{\mathrm{lyap}}^{(i)}&=\mathrm{blkdiag}(-2\mathrm{Re}(\alpha_1) I,\cdots,-2\mathrm{Re}(\alpha_i) I),\label{V_lyap}\\
S_{v,\mathrm{lyap}}^{(i)}&=\begin{bmatrix}S_{v,\mathrm{lyap}}^{(i-1)}&\hat{P}_{\mathrm{lyap}}^{(i-1)}\big(L_{v,\mathrm{lyap}}^{(i-1)}\big)^\top l_{v,\mathrm{lyap}}^{(i)}\\0&s_{v,\mathrm{lyap}}^{(i)}\end{bmatrix},& L_{v,\mathrm{lyap}}^{(i)}&=\begin{bmatrix}l_{v,\mathrm{lyap}}^{(1)}&\cdots&l_{v,\mathrm{lyap}}^{(i)}\end{bmatrix},\\
W_{\mathrm{lyap}}^{(i)}&=\begin{bmatrix}w_1^{\mathrm{lyap}}&\cdots&w_i^{\mathrm{lyap}}\end{bmatrix},& \hat{Q}_{\mathrm{lyap}}^{(i)}&=\mathrm{blkdiag}(-2\mathrm{Re}(\beta_1) I,\cdots,-2\mathrm{Re}(\beta_i) I),\label{W_lyap}\\
S_{w,\mathrm{lyap}}^{(i)}&=\begin{bmatrix}S_{w,\mathrm{lyap}}^{(i-1)}&\hat{Q}_{\mathrm{lyap}}^{(i-1)}\big(L_{w,\mathrm{lyap}}^{(i-1)}\big)^\top l_{w,\mathrm{lyap}}^{(i)}\\0&s_{w,\mathrm{lyap}}^{(i)}\end{bmatrix},& L_{w,\mathrm{lyap}}^{(i)}&=\begin{bmatrix}l_{w,\mathrm{lyap}}^{(1)}&\cdots&l_{w,\mathrm{lyap}}^{(i)}\end{bmatrix}.
\end{align}
The projected reduced-order matrices are:
\begin{align}
\hat{E}_{1,i}^{\mathrm{lyap}}&=I,\quad
\hat{A}_{1,i}^{\mathrm{lyap}}=-\hat{P}_{\mathrm{lyap}}^{(i)}\big(S_{v,\mathrm{lyap}}^{(i)}\big)^*(\hat{P}_{\mathrm{lyap}}^{(i)})^{-1}
=\begin{bmatrix}\hat{A}_{1,i-1}^{\mathrm{lyap}}&0\\-2\mathrm{Re}(\alpha_i)L_{v,\mathrm{lyap}}^{(i-1)}&\overline{\alpha}_iI\end{bmatrix},\\
\hat{B}_{1,i}^{\mathrm{lyap}}&=\hat{P}_{\mathrm{lyap}}^{(i)}\big(L_{v,\mathrm{lyap}}^{(i)}\big)^\top=\begin{bmatrix}\hat{B}_{1,i-1}^{\mathrm{lyap}}\\2\mathrm{Re}(\alpha_i)I\end{bmatrix},\quad \hat{C}_{1,i}^{\mathrm{lyap}}=C_1V_{\mathrm{lyap}}^{(i)}=\begin{bmatrix}\hat{C}_{1,i-1}^{\mathrm{lyap}}&C_1v_i^{\mathrm{lyap}}\end{bmatrix},\\
\hat{E}_{2,i}^{\mathrm{lyap}}&=I,\quad \hat{A}_{2,i}^{\mathrm{lyap}}=-\big(\hat{Q}_{\mathrm{lyap}}^{(i)}\big)^{-1}S_{w,\mathrm{lyap}}^{(i)}\hat{Q}_{\mathrm{lyap}}^{(i)}=\begin{bmatrix}\hat{A}_{2,i-1}^{\mathrm{lyap}}&-2\mathrm{Re}(\beta_i)\big(L_{w,\mathrm{lyap}}^{(i-1)}\big)^\top\\0&\overline{\beta}_iI\end{bmatrix},\\ \hat{B}_{2,i}^{\mathrm{lyap}}&=\big(W_{\mathrm{lyap}}^{(i)}\big)^*B_2=\begin{bmatrix}\hat{B}_{2,i-1}^{\mathrm{lyap}}\\(w_i^{\mathrm{lyap}})^*B_2\end{bmatrix}, \quad \hat{C}_{2,i}^{\mathrm{lyap}}=L_{w,\mathrm{lyap}}^{(i)}\hat{Q}_{\mathrm{lyap}}^{(i)}=\begin{bmatrix}\hat{C}_{2,i-1}^{\mathrm{lyap}}&2\mathrm{Re}(\beta_i)I\end{bmatrix}.
\end{align}
As shown in \cite{wolfthesis}, $\hat{P}_{\mathrm{lyap}}^{(i)}$ and $\hat{Q}_{\mathrm{lyap}}^{(i)}$ solve the projected Lyapunov equations:
\begin{align}
 \hat{A}_{1,i}^{\mathrm{lyap}}\hat{P}_{\mathrm{lyap}}^{(i)}\big(\hat{E}_{1,i}^{\mathrm{lyap}}\big)^\top+\hat{E}_{1,i}^{\mathrm{lyap}}\hat{P}_{\mathrm{lyap}}^{(i)}\big( \hat{A}_{1,i}^{\mathrm{lyap}}\big)^*+\hat{B}_{1,i}^{\mathrm{lyap}}\big(\hat{B}_{1,i}^{\mathrm{lyap}}\big)^*&=0,\label{P_proj_adi}\\
\big(\hat{A}_{2,i}^{\mathrm{lyap}}\big)^*\hat{Q}_{\mathrm{lyap}}^{(i)}\hat{E}_{2,i}^{\mathrm{lyap}}+\big(\hat{E}_{2,i}^{\mathrm{lyap}}\big)^\top\hat{Q}_{\mathrm{lyap}}^{(i)}\hat{A}_{2,i}^{\mathrm{lyap}}+\big(\hat{C}_{2,i}^{\mathrm{lyap}}\big)^*\hat{C}_{2,i}^{\mathrm{lyap}}&=0.\label{Q_proj_adi}
\end{align}
Moreover, $V_{\mathrm{lyap}}^{(i)}$ and $W_{\mathrm{lyap}}^{(i)}$ uniquely satisfy the Sylvester equations:
\begin{align}
A_1V_{\mathrm{lyap}}^{(i)}-E_1V_{\mathrm{lyap}}^{(i)}S_{v,\mathrm{lyap}}^{(i)}+B_1L_{v,\mathrm{lyap}}^{(i)}&=0,\label{cfadi_v_sylv1}\\
A_1V_{\mathrm{lyap}}^{(i)}-E_1V_{\mathrm{lyap}}^{(i)}\big(\hat{E}_{1,i}^{\mathrm{lyap}}\big)^{-1} \hat{A}_{1,i}^{\mathrm{lyap}}+B_{\perp,i}^{\mathrm{lyap}}L_{v,\mathrm{lyap}}^{(i)}&=0,\label{cfadi_v_sylv2}\\
A_2^\top W_{\mathrm{lyap}}^{(i)}-E_2^\top W_{\mathrm{lyap}}^{(i)}S_{w,\mathrm{lyap}}^{(i)}+C_2^\top L_{w,\mathrm{lyap}}^{(i)}&=0,\label{cfadi_w_sylv1}\\
A_2^\top W_{\mathrm{lyap}}^{(i)}-E_2^\top W_{\mathrm{lyap}}^{(i)}\big(\hat{E}_{2,i}^{\mathrm{lyap}}\big)^{-\top}\big(\hat{A}_{2,i}^{\mathrm{lyap}}\big)^*+\big(C_{\perp,i}^{\mathrm{lyap}}\big)^*L_{w,\mathrm{lyap}}^{(i)}&=0.\label{cfadi_w_sylv2}
\end{align}
As shown in \cite{wolfthesis}, the projection-based approximations of $P_1$ and $Q_2$, identical to the ones produced by CF-ADI, are given by $P_1 \approx V_{\mathrm{lyap}}^{(i)} \hat{P}_{\mathrm{lyap}}^{(i)} \big(V_{\mathrm{lyap}}^{(i)}\big)^*$ and $Q_2 \approx W_{\mathrm{lyap}}^{(i)} \hat{Q}_{\mathrm{lyap}}^{(i)} \big(W_{\mathrm{lyap}}^{(i)}\big)^*$, with respective residuals $B_{\perp,i}^{\mathrm{lyap}}\big(B_{\perp,i}^{\mathrm{lyap}}\big)^*$ and $\big(C_{\perp,i}^{\mathrm{lyap}}\big)^* C_{\perp,i}^{\mathrm{lyap}}$.

Assuming that $m_1=p_2$, define
\begin{align}
s_{v,\mathrm{sylv}}^{(i)}&=-\alpha_i I,& l_{v,\mathrm{sylv}}^{(i)}&=-I, \nonumber\\
S_{v,\mathrm{sylv}}^{(i)}&=\begin{bmatrix} S_{v,\mathrm{sylv}}^{(i-1)}&D_{\mathrm{fadi}}^{(i-1)}(L_{w,\mathrm{sylv}}^{(i)})^\top l_{v,\mathrm{sylv}}^{(i)}\\0&s_{v,\mathrm{sylv}}^{(i)}\end{bmatrix}, & L_{v,\mathrm{sylv}}^{(i)}&=\begin{bmatrix}L_{v,\mathrm{sylv}}^{(i-1)}&l_{v,\mathrm{sylv}}^{(i)}\end{bmatrix},\\
s_{w,\mathrm{sylv}}^{(i)}&=-\overline{\beta}_iI,& l_{w,\mathrm{sylv}}^{(i)}&=-I,\nonumber\\
S_{w,\mathrm{sylv}}^{(i)}&=\begin{bmatrix}S_{w,\mathrm{sylv}}^{(i-1)}&(D_{\mathrm{fadi}}^{(i-1)})^*(L_{v,\mathrm{sylv}}^{(i)})^\top l_{w,\mathrm{sylv}}^{(i)}\\0&s_{w,\mathrm{sylv}}^{(i)}\end{bmatrix}, & L_{w,\mathrm{sylv}}^{(i)}&=\begin{bmatrix}L_{w,\mathrm{sylv}}^{(i-1)}&l_{w,\mathrm{sylv}}^{(i)}\end{bmatrix}.
\end{align}
Recall that $V_{\mathrm{fadi}}^{(i)}$ and $W_{\mathrm{fadi}}^{(i)}$ solve the following Sylvester equations:
\begin{align}
A_1V_{\mathrm{fadi}}^{(i)}-E_1V_{\mathrm{fadi}}^{(i)}S_{v,\mathrm{sylv}}^{(i)}+B_1L_{v,\mathrm{sylv}}^{(i)}&=0,\\
A_1V_{\mathrm{fadi}}^{(i)}+E_1V_{\mathrm{fadi}}^{(i)}D_{\mathrm{fadi}}^{(i)}\big(S_{w,\mathrm{sylv}}^{(i)}\big)^*\big(D_{\mathrm{fadi}}^{(i)}\big)^{-1}+B_{\perp,i}^{\mathrm{sylv}}L_{v,\mathrm{sylv}}^{(i)}&=0,\\
A_2^\top W_{\mathrm{fadi}}^{(i)}-E_2^\top W_{\mathrm{fadi}}^{(i)}S_{w,\mathrm{sylv}}^{(i)}+C_2^\top L_{w,\mathrm{sylv}}^{(i)}&=0,\\
A_2^\top W_{\mathrm{fadi}}^{(i)}+E_2^\top W_{\mathrm{fadi}}^{(i)}\big(D_{\mathrm{fadi}}^{(i)}\big)^*\big(S_{v,\mathrm{sylv}}^{(i)}\big)^*\big(D_{\mathrm{fadi}}^{(i)}\big)^{-*}+\big(C_{\perp,i}^{\mathrm{sylv}}\big)^*L_{w,\mathrm{sylv}}^{(i)}&=0.
\end{align} 

Define $T_{v,\mathrm{sylv}}^{(i)}$ and $T_{w,\mathrm{sylv}}^{(i)}$ by:
\begin{align}
T_{v,\mathrm{sylv}}^{(i)}&=\begin{bmatrix}T_{v,\mathrm{sylv}}^{(i-1)}&t_{1,v,\mathrm{sylv}}^{(i)}\\0&t_{2,v,\mathrm{sylv}}^{(i)}\end{bmatrix},&\quad\mathrm{and}&&
T_{w,\mathrm{sylv}}^{(i)}&=\begin{bmatrix}T_{w,\mathrm{sylv}}^{(i-1)}&t_{1,w,\mathrm{sylv}}^{(i)}\\0&t_{2,w,\mathrm{sylv}}^{(i)}\end{bmatrix},\label{TvTw_sylv}
\end{align}where
\begin{align}
t_{v,\mathrm{sylv}}^{(i)}&=\begin{bmatrix}t_{1,v,\mathrm{sylv}}^{(i)}\\t_{2,v,\mathrm{sylv}}^{(i)}\end{bmatrix}=\begin{bmatrix}\hat{A}_{1,i-1}^{\mathrm{lyap}}+\overline{\alpha}_iI&0\\-2\mathrm{Re}(\alpha_i)L_{v,\mathrm{lyap}}^{(i-1)}&2\mathrm{Re}(\alpha_i)I\end{bmatrix}^{-1}\begin{bmatrix}\hat{B}_{1,i-1}^{\mathrm{lyap}}-T_{v,\mathrm{sylv}}^{(i-1)}D_{\mathrm{fadi}}^{(i-1)}(L_{w,\mathrm{sylv}}^{(i-1)})^\top\\2\mathrm{Re}(\alpha_i)I\end{bmatrix},\\
t_{w,\mathrm{sylv}}^{(i)}&=\begin{bmatrix}t_{1,w,\mathrm{sylv}}^{(i)}\\t_{2,w,\mathrm{sylv}}^{(i)}\end{bmatrix}=\begin{bmatrix}(\hat{A}_{2,i-1}^{\mathrm{lyap}})^*+\overline{\beta}_iI&0\\-2\mathrm{Re}(\beta_i)L_{w,\mathrm{lyap}}^{(i-1)}&2\mathrm{Re}(\beta_i)I\end{bmatrix}^{-1}\begin{bmatrix}(\hat{C}_{2,i-1}^{\mathrm{lyap}})^*-T_{w,\mathrm{sylv}}^{(i-1)}(D_{\mathrm{fadi}}^{(i-1)})^*(L_{v,\mathrm{sylv}}^{(i-1)})^\top\\2\mathrm{Re}(\beta_i)I\end{bmatrix},
\end{align}
which uniquely solves the following Sylvester equations:
\begin{align}
\begin{bmatrix}\hat{A}_{1,i-1}^{\mathrm{lyap}}&0\\-2\mathrm{Re}(\alpha_i)L_{v,\mathrm{lyap}}^{(i-1)}&\overline{\alpha}_iI\end{bmatrix}\begin{bmatrix}t_{1,v,\mathrm{sylv}}^{(i)}\\t_{2,v,\mathrm{sylv}}^{(i)}\end{bmatrix}&-\begin{bmatrix}t_{1,v,\mathrm{sylv}}^{(i)}\\t_{2,v,\mathrm{sylv}}^{(i)}\end{bmatrix}s_{v,\mathrm{sylv}}^{(i)}\nonumber\\
&+\begin{bmatrix}\hat{B}_{1,i-1}^{\mathrm{lyap}}-T_{v,\mathrm{sylv}}^{(i-1)}D_{\mathrm{fadi}}^{(i-1)}(L_{w,\mathrm{sylv}}^{(i-1)})^\top\\2\mathrm{Re}(\alpha_i)I\end{bmatrix} l_{v,\mathrm{sylv}}^{(i)}=0,\label{tv_sylv}\\
\begin{bmatrix}\hat{A}_{2,i-1}^{\mathrm{lyap}}&-2\mathrm{Re}(\beta_i)(L_{w,\mathrm{lyap}}^{(i-1)})^\top\\0&\overline{\beta}_iI\end{bmatrix}^*\begin{bmatrix}t_{1,w,\mathrm{sylv}}^{(i)}\\t_{2,w,\mathrm{sylv}}^{(i)}\end{bmatrix}&-\begin{bmatrix}t_{1,w,\mathrm{sylv}}^{(i)}\\t_{2,w,\mathrm{sylv}}^{(i)}\end{bmatrix}s_{w,\mathrm{sylv}}^{(i)}\nonumber\\
&+\begin{bmatrix}(\hat{C}_{2,i-1}^{\mathrm{lyap}})^*-T_{w,\mathrm{sylv}}^{(i-1)}(D_{\mathrm{fadi}}^{(i-1)})^*(L_{v,\mathrm{sylv}}^{(i-1)})^\top\\2\mathrm{Re}(\beta_i)I\end{bmatrix}l_{w,\mathrm{sylv}}^{(i)}=0.\label{tw_sylv}
\end{align}
Alternatively, rather than computing recursively, one may obtain \( T_{v,\mathrm{sylv}}^{(i)} \) and \( T_{w,\mathrm{sylv}}^{(i)} \) directly by solving the Sylvester equations
\begin{align}
\hat{A}_{1,i}^{\mathrm{lyap}}T_{v,\mathrm{sylv}}^{(i)}-T_{v,\mathrm{sylv}}^{(i)}S_{v,\mathrm{sylv}}^{(i)}+\hat{B}_{1,i}^{\mathrm{lyap}}L_{v,\mathrm{sylv}}^{(i)}&=0,\\
\big(\hat{A}_{2,i}^{\mathrm{lyap}}\big)^*T_{w,\mathrm{sylv}}^{(i)}-T_{w,\mathrm{sylv}}^{(i)}S_{w,\mathrm{sylv}}^{(i)}+\big(\hat{C}_{2,i}^{\mathrm{lyap}}\big)^*L_{w,\mathrm{sylv}}^{(i)}&=0.
\end{align} 
The next theorem shows how \( v_i^{\mathrm{sylv}} \), \( w_i^{\mathrm{sylv}} \), \( V_{\mathrm{fadi}}^{(i)} \), and \( W_{\mathrm{fadi}}^{(i)} \) in FADI can be extracted from \( V_{\mathrm{lyap}}^{(i)} \) and \( W_{\mathrm{lyap}}^{(i)} \).
\begin{theorem}
Let $\{\alpha_i\}_{i=1}^{k} \subset \mathbb{C}_{-}$ and $\{\beta_i\}_{i=1}^{k} \subset \mathbb{C}_{-}$ be the ADI shifts used in CF-ADI to approximate \eqref{lyap_p} and \eqref{lyap_q}, respectively. Assume $m_1 = p_2$ and that all variables in Algorithm~\ref{fadi_alg} are well defined. Then $v_i^{\mathrm{sylv}}$, $w_i^{\mathrm{sylv}}$, $V_{\mathrm{fadi}}^{(i)}$, and $W_{\mathrm{fadi}}^{(i)}$ generated by FADI for the same shifts can be extracted from $V_{\mathrm{lyap}}^{(i)}$ and $W_{\mathrm{lyap}}^{(i)}$ as follows:
\begin{align}
v_i^{\mathrm{sylv}}=V_{\mathrm{lyap}}^{(i)}t_{v,\mathrm{sylv}}^{(i)},\quad w_i^{\mathrm{sylv}}=W_{\mathrm{lyap}}^{(i)}t_{w,\mathrm{sylv}}^{(i)},\quad V_{\mathrm{fadi}}^{(i)}=V_{\mathrm{lyap}}^{(i)}T_{v,\mathrm{sylv}}^{(i)},\quad W_{\mathrm{fadi}}^{(i)}=W_{\mathrm{lyap}}^{(i)}T_{w,\mathrm{sylv}}^{(i)}.
\end{align}
\end{theorem}
\begin{proof}
The proof is given in Appendix D.
\end{proof}
Note that \( V_{\mathrm{lyap}}^{(i)} \) and \( W_{\mathrm{lyap}}^{(i)} \) differ from \( V_{\mathrm{fadi}}^{(i)} \) and \( W_{\mathrm{fadi}}^{(i)} \) only by scaling factors: the former use \( 2\operatorname{Re}(\alpha_i) \) and \( 2\operatorname{Re}(\beta_i) \), whereas the latter use \( \alpha_i + \beta_i \). The matrices \( t_{v,\mathrm{sylv}}^{(i)} \) and \( t_{w,\mathrm{sylv}}^{(i)} \) effectively expunge the \( 2\operatorname{Re}(\alpha_i) \) and \( 2\operatorname{Re}(\beta_i) \) scalings from \( V_{\mathrm{lyap}}^{(i)} \) and \( W_{\mathrm{lyap}}^{(i)} \), respectively, and knit the \( \alpha_i + \beta_i \) scaling into them to produce \( v_i^{\mathrm{sylv}} \) and \( w_i^{\mathrm{sylv}} \). Moreover, \( v_i^{\mathrm{sylv}} \) and \( w_i^{\mathrm{sylv}} \) contain the full history of previous shifted linear solves—namely, \( (v_1^{\mathrm{sylv}}, \dots, v_{i-1}^{\mathrm{sylv}}) \) and \( (w_1^{\mathrm{sylv}}, \dots, w_{i-1}^{\mathrm{sylv}}) \)—due to the terms \( B_{\perp,i-1}^{\mathrm{sylv}} \) and \( C_{\perp,i-1}^{\mathrm{sylv}} \). Thus, all the previous history of \( v_i^{\mathrm{lyap}} \) and \( w_i^{\mathrm{lyap}} \) is used to retrieve \( v_i^{\mathrm{sylv}} \) and \( w_i^{\mathrm{sylv}} \). From a computational perspective, $t_{v,\mathrm{sylv}}^{(i)}$ is computed via one shifted linear solve with $\hat{A}_{1,i}^{\mathrm{lyap}}\in\mathbb{C}^{im_1\times im_1}$ and $m_1$ right-hand sides. Similarly, $t_{w,\mathrm{sylv}}^{(i)}$ is computed via one shifted linear solve with $\hat{A}_{2,i}^{\mathrm{lyap}}\in\mathbb{C}^{ip_2\times ip_2}$ and $p_2$ right-hand sides. Since $im_1\ll n_1$ and $ip_2\ll n_2$, these solves are much cheaper than computing $v_i^{\mathrm{sylv}}$ and $w_i^{\mathrm{sylv}}$ in FADI. Thus, extracting $v_i^{\mathrm{sylv}}$ and $w_i^{\mathrm{sylv}}$ from $V_{\mathrm{lyap}}^{(i)}$ and $W_{\mathrm{lyap}}^{(i)}$ offers notable memory savings for large-scale problems.

Due to the block triangular recursive structure of \( T_{v,\mathrm{sylv}}^{(i)} \), its invertibility is guaranteed provided that \( t_{2,v,\mathrm{sylv}}^{(i)} \) is invertible. Proposition \ref{tv_sylv_inv} gives the conditions for the invertibility of \( t_{2,v,\mathrm{sylv}}^{(i)} \), and hence for \( T_{v,\mathrm{sylv}}^{(i)} \).

\begin{proposition}\label{tv_sylv_inv}
Partition $t_{1,v,\mathrm{sylv}}^{(i)}$ as $\begin{bmatrix}t_{v,1}\\\vdots\\t_{v,i-1}\end{bmatrix}$, and define the matrix \( M = \sum_{k=1}^{i-1} t_{v,k} \). Then \( t_{2,v,\mathrm{sylv}}^{(i)} \) is invertible if the matrix \( M \) has no eigenvalue equal to 1.
\end{proposition}
\begin{proof}
The proof is given in Appendix E.
\end{proof}

The invertibility condition in Proposition \ref{tv_sylv_inv} is not restrictive in practice. In our extensive numerical experiments, we never encountered a case where \( T_{v,\mathrm{sylv}}^{(i)} \) was singular. Since \( T_{w,\mathrm{sylv}}^{(i)} \) is the dual of \( T_{v,\mathrm{sylv}}^{(i)} \), its invertibility follows by an analogous argument, which we omit for brevity.

Define  
\[
\tilde{X}_{\mathrm{sylv}}^{(i)} = \big(T_{w,\mathrm{sylv}}^{(i)}\big)^{-*} \big(D_{\mathrm{fadi}}^{(i)}\big)^{-1} \big(T_{v,\mathrm{sylv}}^{(i)}\big)^{-1}.
\]  
It follows directly that  
\[
X_{\mathrm{sylv}}\approx V_{\mathrm{fadi}}^{(i)} D_{\mathrm{fadi}}^{(i)} \big(W_{\mathrm{fadi}}^{(i)}\big)^* 
= V_{\mathrm{lyap}}^{(i)} \big(\tilde{X}_{\mathrm{sylv}}^{(i)}\big)^{-1} \big(W_{\mathrm{lyap}}^{(i)}\big)^*.
\]  
The next proposition shows that \( \tilde{X}_{\mathrm{sylv}}^{(i)} \) can be computed from \( S_{v,\mathrm{lyap}}^{(i)} \), \( L_{v,\mathrm{lyap}}^{(i)} \), \( S_{w,\mathrm{lyap}}^{(i)} \), and \( L_{w,\mathrm{lyap}}^{(i)} \), all of which are determined solely by the ADI shifts \( \alpha_i \) and \( \beta_i \). Consequently, the ADI shifts \( \alpha_i \) and \( \beta_i \) suffice to compute \( \tilde{X}_{\mathrm{sylv}}^{(i)} \).
\begin{proposition}\label{prop_x_sylv}
The following statements are true:
\begin{enumerate}
\item $S_{v,\mathrm{lyap}}^{(i)}=T_{v,\mathrm{sylv}}^{(i)}S_{v,\mathrm{sylv}}^{(i)}\big(T_{v,\mathrm{sylv}}^{(i)}\big)^{-1}$ and $L_{v,\mathrm{lyap}}^{(i)}=L_{v,\mathrm{sylv}}^{(i)}\big(T_{v,\mathrm{sylv}}^{(i)}\big)^{-1}$.
\item $S_{w,\mathrm{lyap}}^{(i)}=T_{w,\mathrm{sylv}}^{(i)}S_{w,\mathrm{sylv}}^{(i)}\big(T_{w,\mathrm{sylv}}^{(i)}\big)^{-1}$ and $L_{w,\mathrm{lyap}}^{(i)}=L_{w,\mathrm{sylv}}^{(i)}\big(T_{w,\mathrm{sylv}}^{(i)}\big)^{-1}$.
  \item $\tilde{X}_{\mathrm{sylv}}^{(i)}$ solves the following Sylvester equation:
\begin{align}
-\big(S_{w,\mathrm{lyap}}^{(i)}\big)^*\tilde{X}_{\mathrm{sylv}}^{(i)}-\tilde{X}_{\mathrm{sylv}}^{(i)}S_{v,\mathrm{lyap}}^{(i)}+(L_{w,\mathrm{lyap}}^{(i)})^\top L_{v,\mathrm{lyap}}^{(i)}=0.\label{X_t_sylv}
\end{align}
\item \label{prop_x_sylv_2}By setting the free parameters $\hat{B}_1^{(i)}=\big(\tilde{X}_{\mathrm{sylv}}^{(i)}\big)^{-1}\big(L_{w,\mathrm{lyap}}^{(i)}\big)^\top$ and $\hat{E}_1^{(i)}=I$, the matrix $\hat{A}_1^{(i)}=\hat{E}_1^{(i)}S_{v,\mathrm{lyap}}^{(i)}-\hat{B}_1^{(i)}L_{v,\mathrm{lyap}}^{(i)}=-\big(\tilde{X}_{\mathrm{sylv}}^{(i)}\big)^{-1}\big(S_{w,\mathrm{lyap}}^{(i)}\big)^*\tilde{X}_{\mathrm{sylv}}^{(i)}$.
\item \label{prop_x_sylv_3}By setting the free parameters $\hat{C}_2^{(i)}=L_{v,\mathrm{lyap}}^{(i)}\big(\tilde{X}_{\mathrm{sylv}}^{(i)}\big)^{-1}$ and $\hat{E}_2^{(i)}=I$, the matrix $\hat{A}_2^{(i)}=\big(S_{w,\mathrm{lyap}}^{(i)}\big)^*\hat{E}_2^{(i)}-\big(L_{w,\mathrm{lyap}}^{(i)}\big)^\top\hat{C}_2^{(i)}=-\tilde{X}_{\mathrm{sylv}}^{(i)}S_{v,\mathrm{lyap}}^{(i)}\big(\tilde{X}_{\mathrm{sylv}}^{(i)}\big)^{-1}$.
\item \label{prop_x_sylv_4}$\big(\tilde{X}_{\mathrm{sylv}}^{(i)}\big)^{-1}$ solves the following projected Sylvester equation:
\begin{align}
\hat{A}_1^{(i)}\big(\tilde{X}_{\mathrm{sylv}}^{(i)}\big)^{-1}\hat{E}_2^{(i)}+\hat{E}_1^{(i)}\big(\tilde{X}_{\mathrm{sylv}}^{(i)}\big)^{-1}\hat{A}_2^{(i)}+\hat{B}_1^{(i)}\hat{C}_2^{(i)}=0.
\end{align}
\end{enumerate}
\end{proposition}
\begin{proof}
The proof is given in Appendix F.
\end{proof}
It follows directly that
\begin{align}
B_{\perp,i}^{\mathrm{sylv}}&=B_1-E_1V_{\mathrm{fadi}}^{(i)}D_{\mathrm{fadi}}^{(i)}\big(L_{w,\mathrm{sylv}}^{(i)}\big)^\top=B_1-E_1V_{\mathrm{lyap}}^{(i)}\big(\tilde{X}_{\mathrm{sylv}}^{(i)}\big)^{-1}\big(L_{w,\mathrm{lyap}}^{(i)}\big)^\top,\\
C_{\perp,i}^{\mathrm{sylv}}&=C_2-L_{v,\mathrm{sylv}}^{(i)}D_{\mathrm{fadi}}^{(i)}\big(W_{\mathrm{fadi}}^{(i)}\big)^*E_2=C_2-L_{v,\mathrm{lyap}}^{(i)}\big(\tilde{X}_{\mathrm{sylv}}^{(i)}\big)^{-1}\big(W_{\mathrm{lyap}}^{(i)}\big)^*E_2.
\end{align}
\begin{remark}
From a control theory perspective, Proposition \ref{prop_x_sylv} tells us that the (implicit) free parameters $\hat{B}_1^{(i)}$ and $\hat{C}_2^{(i)}$ in CF-ADI (identified in \cite{wolf2016adi}) can always be reset to place the poles of $\hat{A}_1^{(i)}$ and $\hat{A}_2^{(i)}$ at $\beta_i$ and $\alpha_i$, respectively, since the pairs $\big(-S_{v,\mathrm{lyap}}^{(i)},L_{v,\mathrm{lyap}}^{(i)}\big)$ and $\big(-S_{w,\mathrm{lyap}}^{(i)},L_{w,\mathrm{lyap}}^{(i)}\big)$ are always observable, satisfying the pole-placement condition \cite{rugh}. By placing the poles at these locations in the $s$-plane, the FADI-based approximation $X_{\mathrm{sylv}}\approx V_{\mathrm{fadi}}^{(i)} D_{\mathrm{fadi}}^{(i)} \big(W_{\mathrm{fadi}}^{(i)}\big)^* $ can be obtained using CF-ADI.
\end{remark}
\subsubsection{Unifying CF-ADI and RADI}
Define \(T_{v,\mathrm{ricc}}^{(i)}\) as 
\begin{align}
T_{v,\mathrm{ricc}}^{(i)}&=\begin{bmatrix}T_{v,\mathrm{ricc}}^{(i-1)}&t_{1,v,\mathrm{ricc}}^{(i)}\\0&t_{2,v,\mathrm{ricc}}^{(i)}\end{bmatrix},\label{Tv_ricc}
\end{align}where
\begin{align}
t_{v,\mathrm{ricc}}^{(i)}&=\begin{bmatrix}t_{1,v,\mathrm{ricc}}^{(i)}\\t_{2,v,\mathrm{ricc}}^{(i)}\end{bmatrix}\nonumber\\
&=\begin{bmatrix}\hat{A}_{1,i-1}^{\mathrm{lyap}}+\alpha_iI-T_{v,\mathrm{ricc}}^{(i-1)}\hat{P}_{\mathrm{ricc}}^{(i-1)}(V_{\mathrm{radi}}^{(i-1)})^*C_1^\top C_1V_{\mathrm{lyap}}^{(i-1)}&-T_{v,\mathrm{ricc}}^{(i-1)}\hat{P}_{\mathrm{ricc}}^{(i-1)}(V_{\mathrm{radi}}^{(i-1)})^*C_1^\top C_1v_{i}^{\mathrm{lyap}}\\-2\mathrm{Re}(\alpha_i)L_{v,\mathrm{lyap}}^{(i-1)}&2\mathrm{Re}(\alpha_i)I\end{bmatrix}^{-1}\nonumber\\
&\hspace*{9cm}\begin{bmatrix}\hat{B}_{1,i-1}^{\mathrm{lyap}}-T_{v,\mathrm{ricc}}^{(i-1)}\hat{P}_{\mathrm{ricc}}^{(i-1)}(L_{v,\mathrm{lyap}}^{(i-1)})^\top\\2\mathrm{Re}(\alpha_i)I\end{bmatrix}.
\end{align}
$t_{v,\mathrm{ricc}}^{(i)}$ uniquely solves the Sylvester equation
\begin{align}
\begin{bmatrix}\hat{A}_{1,i-1}^{\mathrm{lyap}}-T_{v,\mathrm{ricc}}^{(i-1)}\hat{P}_{\mathrm{ricc}}^{(i-1)}(V_{\mathrm{radi}}^{(i-1)})^*C_1^\top C_1V_{\mathrm{lyap}}^{(i-1)}&-T_{v,\mathrm{ricc}}^{(i-1)}\hat{P}_{\mathrm{ricc}}^{(i-1)}(V_{\mathrm{radi}}^{(i-1)})^*C_1^\top C_1v_{i}^{\mathrm{lyap}}\\-2\mathrm{Re}(\alpha_i)L_{v,\mathrm{lyap}}^{(i-1)}&\overline{\alpha}_iI\end{bmatrix}\begin{bmatrix}t_{1,v,\mathrm{ricc}}^{(i)}\\t_{2,v,\mathrm{ricc}}^{(i)}\end{bmatrix}\nonumber\\
-\begin{bmatrix}t_{1,v,\mathrm{ricc}}^{(i)}\\t_{2,v,\mathrm{ricc}}^{(i)}\end{bmatrix}s_{v,\mathrm{lyap}}^{(i)}+\begin{bmatrix}\hat{B}_{1,i-1}^{\mathrm{lyap}}-T_{v,\mathrm{ricc}}^{(i-1)}\hat{P}_{\mathrm{ricc}}^{(i-1)}(L_{v,\mathrm{lyap}}^{(i-1)})^\top\\2\mathrm{Re}(\alpha_i)I\end{bmatrix}l_{v,\mathrm{lyap}}^{(i)}=0.
\label{tv_ricc}\end{align}
Rather than constructing \(T_{v,\mathrm{ricc}}^{(i)}\) recursively, one may also compute it directly by solving the Sylvester equation  
\begin{align}
\hat{A}_{1,i}^{\mathrm{lyap}}T_{v,\mathrm{ricc}}^{(i)}-T_{v,\mathrm{ricc}}^{(i)}S_{v,\mathrm{ricc}}^{(i)}+\hat{B}_{1,i}^{\mathrm{lyap}}L_{v,\mathrm{ricc}}^{(i)}&=0.\label{Tv_ricc_sylv}
\end{align}
However, unlike \(S_{v,\mathrm{sylv}}^{(i)}\), which depends only on the ADI shifts \(\alpha_i\) and \(\beta_i\), the matrix \(S_{v,\mathrm{ricc}}^{(i)}\) depends on \(v_i^{\mathrm{ricc}}\) itself. Consequently, \eqref{Tv_ricc_sylv} cannot be used to compute \(T_{v,\mathrm{ricc}}^{(i)}\) without already knowing \(v_i^{\mathrm{ricc}}\).

The shifted linear solves \(v_i^{\mathrm{ricc}}\) in Step (\ref{radi_step6}) of Algorithm \ref{radi_alg} may at first appear quite different from \(v_i^{\mathrm{lyap}}\). However, the Sylvester equations \eqref{radi_sylv1} and \eqref{radi_sylv2} share a very similar structure with \eqref{cfadi_v_sylv1} and \eqref{cfadi_v_sylv2}, suggesting that \(v_i^{\mathrm{ricc}}\) can in fact be obtained from \(V_{\mathrm{lyap}}^{(i)}\). This is confirmed by the following theorem.
\begin{theorem}\label{th_uni_radi}
Let $\{\alpha_i\}_{i=1}^{k}\in\mathbb{C}_{-}$ be the ADI shifts used in CF-ADI to approximate \eqref{lyap_p}, and assume all variables in Algorithm \ref{radi_alg} are well defined. Then $v_i^{\mathrm{ricc}}$ and $V_{\mathrm{radi}}^{(i)}$ in RADI for the same shifts can be extracted from $V_{\mathrm{lyap}}^{(i)}$ as:
\begin{align}
v_i^{\mathrm{ricc}}&=V_{\mathrm{lyap}}^{(i)}t_{v,\mathrm{ricc}}^{(i)},& V_{\mathrm{radi}}^{(i)}&=V_{\mathrm{lyap}}^{(i)}T_{v,\mathrm{ricc}}^{(i)}.\nonumber
\end{align}
\end{theorem}
\begin{proof}
The proof is given in Appendix G.
\end{proof}
Theorem \ref{th_uni_radi} reveals that computing \(v_i^{\mathrm{ricc}}\) via the SMW formula—as done in \cite{benner2018radi}, which increases the number of columns in the right-hand side of the linear solves from \(m_1\) to \(m_1 + p_1\)—is not strictly necessary. Instead, \(v_i^{\mathrm{ricc}}\) can be fetched from \(V_{\mathrm{lyap}}^{(i)}\) by computing \(t_{v,\mathrm{ricc}}^{(i)}\), which is inexpensive. Thus, by using the result of Theorem \ref{th_uni_radi}, the computation of \(v_i^{\mathrm{ricc}}\) in RADI could potentially be made as efficient as that of \(v_i^{\mathrm{lyap}}\) in CF-ADI. The benefit of using the SMW formula—rather than avoiding it as in \cite{lin2015new}, as noted in \cite{benner2018radi}—is that the Kalman filter gain  
\[
K_{\mathrm{kalman}}^{(i)} = E_1V_{\mathrm{radi}}^{(i)}\hat{P}_{\mathrm{radi}}^{(i)}\big(V_{\mathrm{radi}}^{(i)}\big)^*C_1^\top
\]  
can be updated recursively via  
\[
K_{\mathrm{kalman}}^{(i)} = \sum_{l=1}^{i} E_1 v_l^{\mathrm{ricc}} \hat{p}^{(l)}_{\mathrm{ricc}} (v_l^{\mathrm{ricc}})^* C_1^\top.
\]  
Consequently, if the sole aim of solving the Riccati equation \eqref{ricc_p} is to compute the Kalman filter gain, storing \(V_{\mathrm{radi}}^{(i)}\) and \(\hat{P}_{\mathrm{radi}}^{(i)}\) is unnecessary, which reduces memory usage. It is worth noting that Theorem \ref{th_uni_radi} shows the SMW formula can be avoided while still updating the Kalman gain recursively without storing \(V_{\mathrm{radi}}^{(i)}\) and \(\hat{P}_{\mathrm{radi}}^{(i)}\). However, implementing the RADI algorithm with CF-ADI still requires storing \( V_{\mathrm{lyap}}^{(i)} \).

Due to the block triangular recursive structure of \(T_{v,\mathrm{ricc}}^{(i)}\), the invertibility of \(T_{v,\mathrm{ricc}}^{(i)}\) is guaranteed if \(t_{2,v,\mathrm{ricc}}^{(i)}\) is invertible. The invertibility condition can be derived similarly to Proposition \ref{tv_sylv_inv}, which is omitted here for brevity.

Define \(\tilde{P}_{\mathrm{ricc}}^{(i)}\) as  
\[
\tilde{P}_{\mathrm{ricc}}^{(i)}=T_{v,\mathrm{ricc}}^{(i)}\hat{P}_{\mathrm{ricc}}^{(i)}\big(T_{v,\mathrm{ricc}}^{(i)}\big)^*.
\]
It follows directly that
\[
P_{\mathrm{ricc}}\approx V_{\mathrm{radi}}^{(i)}\hat{P}_{\mathrm{ricc}}^{(i)}\big(V_{\mathrm{radi}}^{(i)}\big)^*=V_{\mathrm{lyap}}^{(i)}\tilde{P}_{\mathrm{ricc}}^{(i)}\big(V_{\mathrm{lyap}}^{(i)}\big)^*.
\]
\begin{theorem}\label{prop_pt_ricc}
Assuming the invertibility of \(T_{v,\mathrm{ricc}}^{(i)}\), the following statements hold:
\begin{enumerate}
  \item $S_{v,\mathrm{lyap}}^{(i)}=T_{v,\mathrm{ricc}}^{(i)}S_{v,\mathrm{ricc}}^{(i)}\big(T_{v,\mathrm{ricc}}^{(i)}\big)^{-1}$ and $L_{v,\mathrm{lyap}}^{(i)}=L_{v,\mathrm{ricc}}^{(i)}\big(T_{v,\mathrm{ricc}}^{(i)}\big)^{-1}$.
\item $\big(\tilde{P}_{\mathrm{ricc}}^{(i)}\big)^{-1}$ solves the following Lyapunov equation:
\begin{align}
-\big(S_{v,\mathrm{lyap}}^{(i)}\big)^*\big(\tilde{P}_{\mathrm{ricc}}^{(i)}\big)^{-1}-\big(\tilde{P}_{\mathrm{ricc}}^{(i)}\big)^{-1}S_{v,\mathrm{lyap}}^{(i)}+\big(L_{v,\mathrm{lyap}}^{(i)}\big)^\top L_{v,\mathrm{lyap}}^{(i)}+\big(V_{\mathrm{lyap}}^{(i)}\big)^*C_1^\top C_1V_{\mathrm{lyap}}^{(i)}=0.\label{P_t_ricc}
\end{align}
\item \label{prop_pt_ricc_3}By setting the free parameters $\hat{B}_1^{(i)}=\tilde{P}_{\mathrm{ricc}}^{(i)}\big(L_{v,\mathrm{lyap}}^{(i)}\big)^\top$ and $\hat{E}_1^{(i)}=I$, the matrix $\hat{A}_1^{(i)}=\hat{E}_1^{(i)}S_{v,\mathrm{lyap}}^{(i)}-\hat{B}_1^{(i)}L_{v,\mathrm{lyap}}^{(i)}=
    \tilde{P}_{\mathrm{ricc}}^{(i)}\Big(-\big(S_{v,\mathrm{lyap}}^{(i)}\big)^*+\big(\hat{C}_1^{(i)}\big)^*\hat{C}_1^{(i)}\Big)\big(\tilde{P}_{\mathrm{ricc}}^{(i)}\big)^{-1}$, wherein $\hat{C}_1^{(i)}=C_1V_{\mathrm{lyap}}^{(i)}$.
\item \label{prop_pt_ricc_4}$\tilde{P}_{\mathrm{ricc}}^{(i)}$ solves the following projected Riccati equation:
\begin{align}
\hat{A}_1^{(i)}\tilde{P}_{\mathrm{ricc}}^{(i)}\big(\hat{E}_1^{(i)})^\top+\hat{E}_1^{(i)}\tilde{P}_{\mathrm{ricc}}^{(i)}\big(\hat{A}_1^{(i)}\big)^*+\hat{B}_1^{(i)}\big(\hat{B}_1^{(i)}\big)^*
-\hat{E}_1^{(i)}\tilde{P}_{\mathrm{ricc}}^{(i)}\big(\hat{C}_1^{(i)}\big)^*\hat{C}_1^{(i)}\tilde{P}_{\mathrm{ricc}}^{(i)}\big(\hat{E}_1^{(i)})^\top=0.\end{align}
\end{enumerate}
\end{theorem}
\begin{proof}
The proof is given in Appendix H.
\end{proof}
It follows directly that
\begin{align}
B_{\perp,i}^{\mathrm{ricc}}&=B_1-E_1V_{\mathrm{radi}}^{(i)}\hat{P}_{\mathrm{ricc}}^{(i)}\big(L_{v,\mathrm{ricc}}^{(i)}\big)^\top=B_1-E_1V_{\mathrm{lyap}}^{(i)}\tilde{P}_{\mathrm{ricc}}^{(i)}\big(L_{v,\mathrm{lyap}}^{(i)}\big)^\top.
\end{align}
\begin{remark}
Again, from a control theory perspective, Theorem \ref{prop_pt_ricc} tells us that the (implicit) free parameter \(\hat{B}_1^{(i)}\) in CF-ADI can always be reset to place the poles of \(\hat{A}_1^{(i)} - \tilde{P}_{\mathrm{ricc}}^{(i)} \big(\hat{C}_1^{(i)}\big)^* \hat{C}_1^{(i)}\) at \(\overline{\alpha}_i\), since the pair \(\big(-S_{v,\mathrm{lyap}}^{(i)}, L_{v,\mathrm{lyap}}^{(i)}\big)\) is always observable, satisfying the pole-placement condition. By placing the poles at these locations in the \(s\)-plane, the RADI-based approximation \(V_{\mathrm{radi}}^{(i)} \hat{P}_{\mathrm{ricc}}^{(i)} \big(V_{\mathrm{radi}}^{(i)}\big)^*\) can be obtained using CF-ADI.
\end{remark}
\subsubsection{Computing the Remaining Equations with CF-ADI}
The remaining linear matrix equations can first be rewritten as standard Lyapunov and Riccati equations and then solved using CF-ADI, as shown in the sequel.

Let us define \(A_{1,\mathrm{mp}}\) and \(B_{1,\mathrm{mp}}\) as $A_{1,\mathrm{mp}}=A_1-B_1D_1^{-1}C_1$ and $B_{1,\mathrm{mp}}=B_1D_1^{-1}$.
Then the Lyapunov equation \eqref{lyap_p_mp} can be rewritten as
\begin{align}
A_{1,\mathrm{mp}}P_{\mathrm{mp}}E_1^\top+E_1P_{\mathrm{mp}}A_{1,\mathrm{mp}}^\top+B_{1,\mathrm{mp}}B_{1,\mathrm{mp}}^\top=0.
\end{align}A low-rank approximation of \(P_{\mathrm{mp}}\) can then be obtained via CF-ADI using the ADI shifts \(\alpha_i\) as $P_{\mathrm{mp}}\approx V_{\mathrm{mp}}^{(i)}\hat{P}_{\mathrm{lyap}}^{(i)}\big(V_{\mathrm{mp}}^{(i)}\big)^*$, where $V_{\mathrm{mp}}^{(i)}=\begin{bmatrix}V_{\mathrm{mp}}^{(i-1)}&v_i^{\mathrm{mp}}\end{bmatrix}$.

Let us define \(A_{2,\mathrm{mp}}\) and \(C_{2,\mathrm{mp}}\) as $A_{2,\mathrm{mp}}=A_2-B_2D_2^{-1}C_2$ and $C_{2,\mathrm{mp}}=D_2^{-1}C_2$.
Then the Lyapunov equation \eqref{lyap_q_mp} can be rewritten as
\begin{align}
A_{1,\mathrm{mp}}^\top Q_{\mathrm{mp}}E_2+E_2^\top Q_{\mathrm{mp}}A_{2,\mathrm{mp}}+C_{2,\mathrm{mp}}^\top C_{2,\mathrm{mp}}=0.
\end{align}
A low-rank approximation of \(Q_{\mathrm{mp}}\) can then be obtained via CF-ADI using the ADI shifts \(\beta_i\) as $Q_{\mathrm{mp}}\approx W_{\mathrm{mp}}^{(i)}\hat{Q}_{\mathrm{lyap}}^{(i)}\big(W_{\mathrm{mp}}^{(i)}\big)^*$, where $W_{\mathrm{mp}}^{(i)}=\begin{bmatrix}W_{\mathrm{mp}}^{(i-1)}&w_i^{\mathrm{mp}}\end{bmatrix}$.

A low-rank approximation of \(Q_{\mathrm{ricc}}\) can be obtained using RADI with ADI shifts \(\beta_i\) as  $Q_{\mathrm{ricc}}=W_{\mathrm{radi}}^{(i)}\hat{Q}_{\mathrm{ricc}}^{(i)}\big(W_{\mathrm{radi}}^{(i)}\big)^*$, where $W_{\mathrm{radi}}^{(i)}=\begin{bmatrix}W_{\mathrm{radi}}^{(i-1)}&w_i^{\mathrm{ricc}}\end{bmatrix}$, and $\hat{Q}_{\mathrm{ricc}}^{(i)}=\mathrm{blkdiag}\big(\hat{Q}_{\mathrm{ricc}}^{(i-1)},\hat{q}_{\mathrm{ricc}}^{(i)}\big)$. The matrix \(\big(\hat{q}_{\mathrm{ricc}}^{(i)}\big)^{-1}\) is the solution to the following Lyapunov equation:
\begin{align}
-\big(s_{w,\mathrm{lyap}}^{(i)}\big)^*\big(\hat{q}_{\mathrm{ricc}}^{(i)}\big)^{-1}-\big(\hat{q}_{\mathrm{ricc}}^{(i)}\big)^{-1}s_{w,\mathrm{lyap}}^{(i)}+\big(l_{w,\mathrm{lyap}}^{(i)}\big)^\top l_{w,\mathrm{lyap}}^{(i)}+\big(w_i^{\mathrm{ricc}}\big)^* B_2 B_2^\top w_i^{\mathrm{ricc}}=0.\label{q_ricc}
\end{align}

Let us define \(C_{1,\infty} = \sqrt{1 - \gamma_1^{-2}}\, C_1\) and \(B_{2,\infty} = \sqrt{1 - \gamma_2^{-2}}\, B_2\). Then the Riccati equations \eqref{ricc_p_inf} and \eqref{ricc_q_inf} can be rewritten as
\begin{align}
A_1P_{\infty}E_1^\top +E_1P_{\infty}A_1^\top +B_1B_1^\top -E_1P_\infty C_{1,\infty}^\top C_{1,\infty}P_\infty E_1^\top=0,\\
A_2^\top Q_{\infty}E_2 +E_2^\top Q_{\infty}A_2 +C_2^\top C_2 -E_2^\top Q_\infty B_{2,\infty} B_{2,\infty}^\top P_\infty E_2=0.
\end{align}
Low-rank approximations of \(P_{\infty}\) and \(Q_{\infty}\) can then be obtained via RADI using ADI shifts \(\alpha_i\) and \(\beta_i\), respectively: $P_{\infty}=V_{\infty}^{(i)}\hat{P}_{\infty}^{(i)}\big(V_{\infty}^{(i)}\big)^*$ and $Q_{\infty}=W_{\infty}^{(i)}\hat{Q}_{\infty}^{(i)}\big(W_{\infty}^{(i)}\big)^*$, where $V_{\infty}^{(i)}=\begin{bmatrix}V_{\infty}^{(i-1)}&v_i^{\infty}\end{bmatrix}$, $\hat{P}_{\infty}^{(i)}=\mathrm{blkdiag}\big(\hat{P}_{\infty}^{(i-1)},\hat{p}_{\infty}^{(i)}\big)$, $W_{\infty}^{(i)}=\begin{bmatrix}W_{\infty}^{(i-1)}&w_i^{\infty}\end{bmatrix}$, and $\hat{Q}_{\infty}^{(i)}=\mathrm{blkdiag}\big(\hat{Q}_{\infty}^{(i-1)},\hat{q}_{\infty}^{(i)}\big)$. The matrices \(\big(\hat{p}_{\infty}^{(i)}\big)^{-1}\) and \(\big(\hat{q}_{\infty}^{(i)}\big)^{-1}\) are the solutions to the following Lyapunov equations:
\begin{align}
-\big(s_{v,\mathrm{lyap}}^{(i)}\big)^*\big(\hat{p}_{\infty}^{(i)}\big)^{-1}-\big(\hat{p}_{\infty}^{(i)}\big)^{-1}s_{v,\mathrm{lyap}}^{(i)}+\big(l_{v,\mathrm{lyap}}^{(i)}\big)^\top l_{v,\mathrm{lyap}}^{(i)}+(1-\gamma_1^{-2})\big(v_i^{\infty}\big)^* C_1^\top C_1v_i^{\infty}=0,\label{p_inf}\\
-\big(s_{w,\mathrm{lyap}}^{(i)}\big)^*\big(\hat{q}_{\infty}^{(i)}\big)^{-1}-\big(\hat{q}_{\infty}^{(i)}\big)^{-1}s_{w,\mathrm{lyap}}^{(i)}+\big(l_{w,\mathrm{lyap}}^{(i)}\big)^\top l_{w,\mathrm{lyap}}^{(i)}+(1-\gamma_2^{-2})\big(w_i^{\infty}\big)^* B_2 B_2^\top w_i^{\infty}=0.\label{q_inf}
\end{align}
Assuming that $D_1+D_1^\top >0$ and $D_2+D_2^\top >0$, define the following matrices:
\begin{align}
A_{1,\mathrm{pr}}=A_1-B_1(D_1+D_1^\top)^{-1}C_1,\quad B_{1,\mathrm{pr}}=B_1(D_1+D_1^\top)^{-\frac{1}{2}},\quad C_{1,\mathrm{pr}}=(D_1+D_1^\top)^{-\frac{1}{2}}C_1,\\
A_{2,\mathrm{pr}}=A_2-B_2(D_2+D_2^\top)^{-1}C_2,\quad B_{2,\mathrm{pr}}=B_2(D_2+D_2^\top)^{-\frac{1}{2}},\quad C_{2,\mathrm{pr}}=(D_2+D_2^\top)^{-\frac{1}{2}}C_2.
\end{align}
Then the Riccati equations \eqref{ricc_p_pr} and \eqref{ricc_q_pr} can be rewritten as
\begin{align}
A_{1,\mathrm{pr}}P_{\mathrm{pr}}E_1^\top+E_1P_{\mathrm{pr}}A_{1,\mathrm{pr}}^\top + B_{1,\mathrm{pr}}B_{1,\mathrm{pr}}^\top+E_1P_{\mathrm{pr}}C_{1,\mathrm{pr}}^\top C_{1,\mathrm{pr}}P_{\mathrm{pr}}E_1^\top=0,\\
A_{2,\mathrm{pr}}^\top Q_{\mathrm{pr}}E_2+E_2^\top Q_{\mathrm{pr}}A_{2,\mathrm{pr}} + C_{2,\mathrm{pr}}^\top C_{2,\mathrm{pr}}+E_2^\top Q_{\mathrm{pr}}B_{2,\mathrm{pr}} B_{2,\mathrm{pr}}^\top Q_{\mathrm{pr}}E_2=0.
\end{align}
Low-rank approximations of \(P_{\mathrm{pr}}\) and \(Q_{\mathrm{pr}}\) can then be obtained via RADI using ADI shifts \(\alpha_i\) and \(\beta_i\), respectively: $P_{\mathrm{pr}} = V_{\mathrm{pr}}^{(i)} \hat{P}_{\mathrm{pr}}^{(i)} \big(V_{\mathrm{pr}}^{(i)}\big)^*$ and
$Q_{\mathrm{pr}} = W_{\mathrm{pr}}^{(i)} \hat{Q}_{\mathrm{pr}}^{(i)} \big(W_{\mathrm{pr}}^{(i)}\big)^*$,
where $V_{\mathrm{pr}}^{(i)} = \begin{bmatrix} V_{\mathrm{pr}}^{(i-1)} & v_i^{\mathrm{pr}} \end{bmatrix}$, 
$\hat{P}_{\mathrm{pr}}^{(i)} = \mathrm{blkdiag}\big(\hat{P}_{\mathrm{pr}}^{(i-1)}, \hat{p}_{\mathrm{pr}}^{(i)}\big)$,
$W_{\mathrm{pr}}^{(i)} = \begin{bmatrix} W_{\mathrm{pr}}^{(i-1)} & w_i^{\mathrm{pr}} \end{bmatrix}$, and
$\hat{Q}_{\mathrm{pr}}^{(i)} = \mathrm{blkdiag}\big(\hat{Q}_{\mathrm{pr}}^{(i-1)}, \hat{q}_{\mathrm{pr}}^{(i)}\big)$. The matrices \(\big(\hat{p}_{\mathrm{pr}}^{(i)}\big)^{-1}\) and \(\big(\hat{q}_{\mathrm{pr}}^{(i)}\big)^{-1}\) are the solutions to the following Lyapunov equations:
\begin{align}
-\big(s_{v,\mathrm{lyap}}^{(i)}\big)^*\big(\hat{p}_{\mathrm{pr}}^{(i)}\big)^{-1}-\big(\hat{p}_{\mathrm{pr}}^{(i)}\big)^{-1}s_{v,\mathrm{lyap}}^{(i)}+\big(l_{v,\mathrm{lyap}}^{(i)}\big)^\top l_{v,\mathrm{lyap}}^{(i)}-\big(v_i^{\mathrm{pr}}\big)^* C_1^\top (D_1+D_1^\top)^{-1}C_1v_i^{\mathrm{pr}}=0,\label{p_pr}\\
-\big(s_{w,\mathrm{lyap}}^{(i)}\big)^*\big(\hat{q}_{\mathrm{pr}}^{(i)}\big)^{-1}-\big(\hat{q}_{\mathrm{pr}}^{(i)}\big)^{-1}s_{w,\mathrm{lyap}}^{(i)}+\big(l_{w,\mathrm{lyap}}^{(i)}\big)^\top l_{w,\mathrm{lyap}}^{(i)}-\big(w_i^{\mathrm{pr}}\big)^* B_2(D_2+D_2^\top)^{-1} B_2^\top w_i^{\mathrm{pr}}=0.\label{q_pr}
\end{align}
Assuming that $I-D_1D_1^\top> 0$ and $I-D_2^\top D_2 > 0$, define the following matrices:
\begin{align}
A_{1,\mathrm{br}}&=A_1+B_1D_1^\top (I-D_1D_1^\top)^{-1}C_1,\quad B_{1,\mathrm{br}}=B_1\big(I+D_1^\top(I-D_1D_1^\top)^{-1}D_1\big)^{\frac{1}{2}},\nonumber\\
C_{1,\mathrm{br}}&=(I-D_1D_1^\top)^{-\frac{1}{2}}C_1,\nonumber\\
A_{2,\mathrm{br}}&=A_2+B_2(I-D_2^\top D_2)^{-1}D_2^\top C_2,\quad B_{2,\mathrm{br}}=B_2\big(I-D_2^\top D_2\big)^{-\frac{1}{2}},\nonumber\\
C_{2,\mathrm{br}}&=\big(I+D_2(I-D_2^\top D_2)^{-1}D_2^\top\big)^{\frac{1}{2}}C_2.
\end{align}
Then the Riccati equations \eqref{ricc_p_br} and \eqref{ricc_q_br} can be rewritten as
\begin{align}
A_{1,\mathrm{br}}P_{\mathrm{br}}E_1^\top+E_1P_{\mathrm{br}}A_{1,\mathrm{br}}^\top + B_{1,\mathrm{br}}B_{1,\mathrm{br}}^\top+E_1P_{\mathrm{br}}C_{1,\mathrm{br}}^\top C_{1,\mathrm{br}}P_{\mathrm{br}}E_1^\top=0,\\
A_{2,\mathrm{br}}^\top Q_{\mathrm{br}}E_2+E_2^\top Q_{\mathrm{br}}A_{2,\mathrm{br}} + C_{2,\mathrm{br}}^\top C_{2,\mathrm{br}}+E_2^\top Q_{\mathrm{br}}B_{2,\mathrm{br}} B_{2,\mathrm{br}}^\top Q_{\mathrm{br}}E_2=0.
\end{align}
Low-rank approximations of \(P_{\mathrm{br}}\) and \(Q_{\mathrm{br}}\) can then be obtained via RADI using ADI shifts \(\alpha_i\) and \(\beta_i\), respectively: $P_{\mathrm{br}} = V_{\mathrm{br}}^{(i)} \hat{P}_{\mathrm{br}}^{(i)} \big(V_{\mathrm{br}}^{(i)}\big)^*$ and
$Q_{\mathrm{br}} = W_{\mathrm{br}}^{(i)} \hat{Q}_{\mathrm{br}}^{(i)} \big(W_{\mathrm{br}}^{(i)}\big)^*$, with  
$V_{\mathrm{br}}^{(i)} = \begin{bmatrix} V_{\mathrm{br}}^{(i-1)} & v_i^{\mathrm{br}} \end{bmatrix}$,
$\hat{P}_{\mathrm{br}}^{(i)} = \mathrm{blkdiag}\big(\hat{P}_{\mathrm{br}}^{(i-1)}, \hat{p}_{\mathrm{br}}^{(i)}\big)$,
$W_{\mathrm{br}}^{(i)} = \begin{bmatrix} W_{\mathrm{br}}^{(i-1)} & w_i^{\mathrm{br}} \end{bmatrix}$, and
$\hat{Q}_{\mathrm{br}}^{(i)} = \mathrm{blkdiag}\big(\hat{Q}_{\mathrm{br}}^{(i-1)}, \hat{q}_{\mathrm{br}}^{(i)}\big)$. The matrices \(\big(\hat{p}_{\mathrm{br}}^{(i)}\big)^{-1}\) and \(\big(\hat{q}_{\mathrm{br}}^{(i)}\big)^{-1}\) are the solutions to the following Lyapunov equations: 
\begin{align}
-\big(s_{v,\mathrm{lyap}}^{(i)}\big)^*\big(\hat{p}_{\mathrm{br}}^{(i)}\big)^{-1}-\big(\hat{p}_{\mathrm{br}}^{(i)}\big)^{-1}s_{v,\mathrm{lyap}}^{(i)}+\big(l_{v,\mathrm{lyap}}^{(i)}\big)^\top l_{v,\mathrm{lyap}}^{(i)}-\big(v_i^{\mathrm{br}}\big)^* C_1^\top (I-D_1D_1^\top)^{-1}C_1v_i^{\mathrm{br}}=0,\label{p_br}\\
-\big(s_{w,\mathrm{lyap}}^{(i)}\big)^*\big(\hat{q}_{\mathrm{br}}^{(i)}\big)^{-1}-\big(\hat{q}_{\mathrm{br}}^{(i)}\big)^{-1}s_{w,\mathrm{lyap}}^{(i)}+\big(l_{w,\mathrm{lyap}}^{(i)}\big)^\top l_{w,\mathrm{lyap}}^{(i)}-\big(w_i^{\mathrm{br}}\big)^* B_2(I-D_2^\top D_2)^{-1} B_2^\top w_i^{\mathrm{br}}=0.\label{q_br}
\end{align}
Assuming that $D_1^\top D_1>0$ and $D_2D_2^\top>0$, define the following matrices:
\begin{align}
A_{1,\mathrm{sf}}&=A_1-B_1(D_1^\top D_1)^{-1}\Big(B_1^\top Q_2+D_1^\top C_1\Big),\quad
B_{1,\mathrm{sf}}=B_1(D_1^\top D_1)^{-\frac{1}{2}},\nonumber\\
C_{1,\mathrm{sf}}&=(D_1^\top D_1)^{-\frac{1}{2}}\Big(B_1^\top Q_2+D_1^\top C_1\Big),\\
A_{2,\mathrm{sf}}&=A_2-\Big(P_1C_2^\top+B_2D_2^\top\Big)\big(D_2D_2^\top\big)^{-1}C_2,\quad B_{2,\mathrm{sf}}=\Big(P_1C_2^\top+B_2D_2^\top\Big)\big(D_2D_2^\top\big)^{-\frac{1}{2}},\nonumber\\
C_{2,\mathrm{sf}}&=\big(D_2D_2^\top\big)^{-\frac{1}{2}}C_2.
\end{align}
Then the Riccati equations \eqref{ricc_p_sf} and \eqref{ricc_q_sf} can be rewritten as
\begin{align}
A_{1,\mathrm{sf}}P_{\mathrm{sf}}E_1^\top+E_1P_{\mathrm{sf}}A_{1,\mathrm{sf}}^\top + B_{1,\mathrm{sf}}B_{1,\mathrm{sf}}^\top+E_1P_{\mathrm{sf}}C_{1,\mathrm{sf}}^\top C_{1,\mathrm{sf}}P_{\mathrm{sf}}E_1^\top=0,\\
A_{2,\mathrm{sf}}^\top Q_{\mathrm{sf}}E_2+E_2^\top Q_{\mathrm{sf}}A_{2,\mathrm{sf}} + C_{2,\mathrm{sf}}^\top C_{2,\mathrm{sf}}+E_2^\top Q_{\mathrm{sf}}B_{2,\mathrm{sf}} B_{2,\mathrm{sf}}^\top Q_{\mathrm{sf}}E_2=0.
\end{align}
Low-rank approximations of \(P_{\mathrm{sf}}\) and \(Q_{\mathrm{sf}}\) can then be obtained via RADI using ADI shifts \(\alpha_i\) and \(\beta_i\), respectively: $P_{\mathrm{sf}} = V_{\mathrm{sf}}^{(i)} \hat{P}_{\mathrm{sf}}^{(i)} \big(V_{\mathrm{sf}}^{(i)}\big)^*$ and
$Q_{\mathrm{sf}} = W_{\mathrm{sf}}^{(i)} \hat{Q}_{\mathrm{sf}}^{(i)} \big(W_{\mathrm{sf}}^{(i)}\big)^*$,  
with $V_{\mathrm{sf}}^{(i)} = \begin{bmatrix} V_{\mathrm{sf}}^{(i-1)} & v_i^{\mathrm{sf}} \end{bmatrix}$, 
$\hat{P}_{\mathrm{sf}}^{(i)} = \mathrm{blkdiag}\big(\hat{P}_{\mathrm{sf}}^{(i-1)}, \hat{p}_{\mathrm{sf}}^{(i)}\big)$,
$W_{\mathrm{sf}}^{(i)} = \begin{bmatrix} W_{\mathrm{sf}}^{(i-1)} & w_i^{\mathrm{sf}} \end{bmatrix}$, and
$\hat{Q}_{\mathrm{sf}}^{(i)} = \mathrm{blkdiag}\big(\hat{Q}_{\mathrm{sf}}^{(i-1)}, \hat{q}_{\mathrm{sf}}^{(i)}\big)$.
The matrices \(\big(\hat{p}_{\mathrm{sf}}^{(i)}\big)^{-1}\) and \(\big(\hat{q}_{\mathrm{sf}}^{(i)}\big)^{-1}\) are the solutions to the following Lyapunov equations:  
\begin{align}
-\big(s_{v,\mathrm{lyap}}^{(i)}\big)^*\big(\hat{p}_{\mathrm{sf}}^{(i)}\big)^{-1}-\big(\hat{p}_{\mathrm{sf}}^{(i)}\big)^{-1}s_{v,\mathrm{lyap}}^{(i)}+\big(l_{v,\mathrm{lyap}}^{(i)}\big)^\top l_{v,\mathrm{lyap}}^{(i)}-\big(v_i^{\mathrm{sf}}\big)^* C_{1,\mathrm{sf}}^\top C_{1,\mathrm{sf}}P_{\mathrm{sf}}v_i^{\mathrm{sf}}=0,\\
-\big(s_{w,\mathrm{lyap}}^{(i)}\big)^*\big(\hat{q}_{\mathrm{sf}}^{(i)}\big)^{-1}-\big(\hat{q}_{\mathrm{sf}}^{(i)}\big)^{-1}s_{w,\mathrm{lyap}}^{(i)}+\big(l_{w,\mathrm{lyap}}^{(i)}\big)^\top l_{w,\mathrm{lyap}}^{(i)}-\big(w_i^{\mathrm{sf}}\big)^* B_{2,\mathrm{sf}} B_{2,\mathrm{sf}}^\top w_i^{\mathrm{sf}}=0.
\end{align}
Define the following matrices:
\begin{align}
T_{v,\mathrm{mp}}^{(i)}&=\begin{bmatrix}T_{v,\mathrm{mp}}^{(i-1)}&t_{1,v,\mathrm{mp}}^{(i)}\\0&t_{2,v,\mathrm{mp}}^{(i)}\end{bmatrix},& T_{w,\mathrm{mp}}^{(i)}&=\begin{bmatrix}T_{w,\mathrm{mp}}^{(i-1)}&t_{1,w,\mathrm{mp}}^{(i)}\\0&t_{2,w,\mathrm{mp}}^{(i)}\end{bmatrix}, &
T_{w,\mathrm{ricc}}^{(i)}&=\begin{bmatrix}T_{w,\mathrm{ricc}}^{(i-1)}&t_{1,w,\mathrm{ricc}}^{(i)}\\0&t_{2,w,\mathrm{ricc}}^{(i)}\end{bmatrix},\label{TvTw1}\\
T_{v,\infty}^{(i)}&=\begin{bmatrix}T_{v,\infty}^{(i-1)}&t_{1,v,\infty}^{(i)}\\0&t_{2,v,\infty}^{(i)}\end{bmatrix},&
T_{w,\infty}^{(i)}&=\begin{bmatrix}T_{w,\infty}^{(i-1)}&t_{1,w,\infty}^{(i)}\\0&t_{2,w,\infty}^{(i)}\end{bmatrix},&
T_{v,\mathrm{pr}}^{(i)}&=\begin{bmatrix}T_{v,\mathrm{pr}}^{(i-1)}&t_{1,v,\mathrm{pr}}^{(i)}\\0&t_{2,v,\mathrm{pr}}^{(i)}\end{bmatrix},\label{TvTw2}\\
T_{w,\mathrm{pr}}^{(i)}&=\begin{bmatrix}T_{w,\mathrm{pr}}^{(i-1)}&t_{1,w,\mathrm{pr}}^{(i)}\\0&t_{2,w,\mathrm{pr}}^{(i)}\end{bmatrix},& T_{v,\mathrm{br}}^{(i)}&=\begin{bmatrix}T_{v,\mathrm{br}}^{(i-1)}&t_{1,v,\mathrm{br}}^{(i)}\\0&t_{2,v,\mathrm{br}}^{(i)}\end{bmatrix},&
T_{w,\mathrm{br}}^{(i)}&=\begin{bmatrix}T_{w,\mathrm{br}}^{(i-1)}&t_{1,w,\mathrm{br}}^{(i)}\\0&t_{2,w,\mathrm{br}}^{(i)}\end{bmatrix},\label{TvTw3}\\
T_{v,\mathrm{sf}}^{(i)}&=\begin{bmatrix}T_{v,\mathrm{sf}}^{(i-1)}&t_{1,v,\mathrm{sf}}^{(i)}\\0&t_{2,v,\mathrm{sf}}^{(i)}\end{bmatrix},&
T_{w,\mathrm{sf}}^{(i)}&=\begin{bmatrix}T_{w,\mathrm{sf}}^{(i-1)}&t_{1,w,\mathrm{sf}}^{(i)}\\0&t_{2,w,\mathrm{sf}}^{(i)}\end{bmatrix},\label{TvTw4}
\end{align}
where  
\begin{align}
t_{v,\mathrm{mp}}^{(i)}&=\begin{bmatrix}t_{1,v,\mathrm{mp}}^{(i)}\\t_{2,v,\mathrm{mp}}^{(i)}\end{bmatrix},&
t_{w,\mathrm{mp}}^{(i)}&=\begin{bmatrix}t_{1,w,\mathrm{mp}}^{(i)}\\t_{2,w,\mathrm{mp}}^{(i)}\end{bmatrix},&
t_{w,\mathrm{ricc}}^{(i)}&=\begin{bmatrix}t_{1,w,\mathrm{ricc}}^{(i)}\\t_{2,w,\mathrm{ricc}}^{(i)}\end{bmatrix},&
t_{v,\infty}^{(i)}&=\begin{bmatrix}t_{1,v,\infty}^{(i)}\\t_{2,v,\infty}^{(i)}\end{bmatrix},\nonumber\\
t_{v,\infty}^{(i)}&=\begin{bmatrix}t_{1,v,\infty}^{(i)}\\t_{2,v,\infty}^{(i)}\end{bmatrix},&
t_{w,\infty}^{(i)}&=\begin{bmatrix}t_{1,w,\infty}^{(i)}\\t_{2,w,\infty}^{(i)}\end{bmatrix},&
t_{v,\mathrm{pr}}^{(i)}&=\begin{bmatrix}t_{1,v,\mathrm{pr}}^{(i)}\\t_{2,v,\mathrm{pr}}^{(i)}\end{bmatrix},&
t_{w,\mathrm{pr}}^{(i)}&=\begin{bmatrix}t_{1,w,\mathrm{pr}}^{(i)}\\t_{2,w,\mathrm{pr}}^{(i)}\end{bmatrix},\nonumber\\
t_{v,\mathrm{br}}^{(i)}&=\begin{bmatrix}t_{1,v,\mathrm{br}}^{(i)}\\t_{2,v,\mathrm{br}}^{(i)}\end{bmatrix},&
t_{w,\mathrm{br}}^{(i)}&=\begin{bmatrix}t_{1,w,\mathrm{br}}^{(i)}\\t_{2,w,\mathrm{br}}^{(i)}\end{bmatrix},&
t_{v,\mathrm{sf}}^{(i)}&=\begin{bmatrix}t_{1,v,\mathrm{sf}}^{(i)}\\t_{2,v,\mathrm{sf}}^{(i)}\end{bmatrix},&
t_{w,\mathrm{sf}}^{(i)}&=\begin{bmatrix}t_{1,w,\mathrm{sf}}^{(i)}\\t_{2,w,\mathrm{sf}}^{(i)}\end{bmatrix},\nonumber
\end{align}
solve the following Sylvester equations:
\begin{align}
&\Big(\hat{A}_{1,i}^{\mathrm{lyap}}-\hat{B}_{1,i}^{\mathrm{lyap}}D_1^{-1}C_1V_{\mathrm{lyap}}^{(i)}\Big)t_{v,\mathrm{mp}}^{(i)}-t_{v,\mathrm{mp}}^{(i)}s_{v,\mathrm{lyap}}^{(i)}\nonumber\\ &\hspace*{6cm}+\Big(\hat{B}_{1,i}^{\mathrm{lyap}}D_1^{-1}-\begin{bmatrix}T_{v,\mathrm{mp}}^{(i-1)}\hat{P}_{\mathrm{lyap}}^{(i-1)}(L_{v,\mathrm{lyap}}^{(i-1)})^\top\\0\end{bmatrix}\Big)l_{v,\mathrm{lyap}}^{(i)}=0,\label{tv_mp}\\
&\Big(\big(\hat{A}_{2,i}^{\mathrm{lyap}}\big)^*-\big(\hat{C}_{2,i}^{\mathrm{lyap}}\big)^*D_2^{-\top}B_2^\top W_{\mathrm{lyap}}^{(i)}\Big)t_{w,\mathrm{mp}}^{(i)}-t_{w,\mathrm{mp}}^{(i)}s_{w,\mathrm{lyap}}^{(i)}\nonumber\\ 
&\hspace*{6cm}+\Big(\big(\hat{C}_{2,i}^{\mathrm{lyap}}\big)^*D_1^{-\top}-\begin{bmatrix}T_{w,\mathrm{mp}}^{(i-1)}\hat{Q}_{\mathrm{lyap}}^{(i-1)}(L_{w,\mathrm{lyap}}^{(i-1)})^\top\\0\end{bmatrix}\Big)l_{w,\mathrm{lyap}}^{(i)}=0,\label{tw_mp}\\
&\Bigg(\big(\hat{A}_{2,i}^{\mathrm{lyap}}\big)^\top-\begin{bmatrix}T_{w,\mathrm{ricc}}^{(i-1)}\hat{Q}_{\mathrm{ricc}}^{(i-1)}\big(W_{\mathrm{radi}}^{(i-1)}\big)^\top B_2B_2^\top W_{\mathrm{lyap}}^{(i)}\\0\end{bmatrix}\Bigg)t_{w,\mathrm{ricc}}^{(i)}-t_{w,\mathrm{ricc}}^{(i)}s_{w,\mathrm{lyap}}^{(i)}\nonumber\\
&\hspace*{6cm}+\Bigg((\hat{C}_{2,i}^{\mathrm{lyap}})^\top-\begin{bmatrix}T_{w,\mathrm{ricc}}^{(i-1)}\hat{Q}_{\mathrm{ricc}}^{(i-1)}\big(L_{w,\mathrm{lyap}}^{(i-1)}\big)^\top\\0\end{bmatrix}\Bigg)l_{w,\mathrm{lyap}}^{(i)}=0,\label{tw_ricc}\\
&\Bigg(\hat{A}_{1,i}^{\mathrm{lyap}}-\begin{bmatrix}(1-\gamma_1^2)T_{v,\infty}^{(i-1)}\hat{P}_{\infty}^{(i-1)}\big(V_{\infty}^{(i-1)}\big)^*C_1^\top C_1V_{\mathrm{lyap}}^{(i)}\\0\end{bmatrix}\Bigg)t_{v,\infty}^{(i)}-t_{v,\infty}^{(i)}s_{v,\mathrm{lyap}}^{(i)}\nonumber\\
&\hspace*{6cm}+\Bigg(\hat{B}_{1,i}^{\mathrm{lyap}}-\begin{bmatrix}T_{v,\infty}^{(i-1)}\hat{P}_{\infty}^{(i-1)}(L_{v,\mathrm{lyap}}^{(i-1)})^\top\\0\end{bmatrix}\Bigg)l_{v,\mathrm{lyap}}^{(i)}=0,\label{tv_inf}\\
&\Bigg(\big(\hat{A}_{2,i}^{\mathrm{lyap}}\big)^*-\begin{bmatrix}(1-\gamma_2^2)T_{w,\infty}^{(i-1)}\hat{Q}_{\infty}^{(i-1)}\big(W_{\infty}^{(i-1)}\big)^*B_2 B_2^\top W_{\mathrm{lyap}}^{(i)}\\0\end{bmatrix}\Bigg)t_{w,\infty}^{(i)}-t_{w,\infty}^{(i)}s_{w,\mathrm{lyap}}^{(i)}\nonumber\\
&\hspace*{6cm}+\Bigg(\big(\hat{C}_{2,i}^{\mathrm{lyap}}\big)^*-\begin{bmatrix}T_{w,\infty}^{(i-1)}\hat{Q}_{\infty}^{(i-1)}(L_{w,\mathrm{lyap}}^{(i-1)})^\top\\0\end{bmatrix}\Bigg)l_{w,\mathrm{lyap}}^{(i)}=0,\label{tw_inf}
\end{align}
\begin{align}
&\Bigg(\hat{A}_{1,i}^{\mathrm{lyap}}-\hat{B}_{1,i}^{\mathrm{lyap}}(D_1+D_1^\top)^{-1}C_1V_{\mathrm{lyap}}^{(i)}+
\begin{bmatrix}T_{v,\mathrm{pr}}^{(i-1)}\hat{P}_{\mathrm{pr}}^{(i-1)}\big(V_{\mathrm{pr}}^{(i-1)}\big)^*C_1^\top(D_1+D_1^\top)^{-1}C_1 V_{\mathrm{lyap}}^{(i)}\\0\end{bmatrix}\Bigg)t_{v,\mathrm{pr}}^{(i)}\nonumber\\
&\hspace*{3cm}-t_{v,\mathrm{pr}}^{(i)}s_{v,\mathrm{lyap}}^{(i)}+\Bigg(\hat{B}_{1,i}^{\mathrm{lyap}}(D_1+D_1^\top)^{-\frac{1}{2}}-\begin{bmatrix}T_{v,\mathrm{pr}}^{(i-1)}\hat{P}_{\mathrm{pr}}^{(i-1)}(L_{v,\mathrm{lyap}}^{(i-1)})^\top\\0\end{bmatrix}\Bigg)l_{v,\mathrm{lyap}}^{(i)}=0,\label{tv_pr}\\
&\Bigg(\big(\hat{A}_{2,i}^{\mathrm{lyap}}\big)^*-\big(\hat{C}_{2,i}^{\mathrm{lyap}}\big)^*(D_2+D_2^\top)^{-1}B_2^\top W_{\mathrm{lyap}}^{(i)}+
\begin{bmatrix}T_{w,\mathrm{pr}}^{(i-1)}\hat{Q}_{\mathrm{pr}}^{(i-1)}\big(W_{\mathrm{pr}}^{(i-1)}\big)^*B_2(D_2+D_2^\top)^{-1}B_2^\top W_{\mathrm{lyap}}^{(i)}\\0\end{bmatrix}\Bigg)t_{w,\mathrm{pr}}^{(i)}\nonumber\\
&\hspace*{2cm}-t_{w,\mathrm{pr}}^{(i)}s_{w,\mathrm{lyap}}^{(i)}+\Bigg(\big(\hat{C}_{2,i}^{\mathrm{lyap}}\big)^*(D_2+D_2^\top)^{-\frac{1}{2}}-\begin{bmatrix}T_{w,\mathrm{pr}}^{(i-1)}\hat{Q}_{\mathrm{pr}}^{(i-1)}(L_{w,\mathrm{lyap}}^{(i-1)})^\top\\0\end{bmatrix}\Bigg)l_{w,\mathrm{lyap}}^{(i)}=0,\label{tw_pr}\\
&\Bigg(\hat{A}_{1,i}^{\mathrm{lyap}}+\hat{B}_{1,i}^{\mathrm{lyap}}D_1^\top(I-D_1D_1^\top)^{-1}C_1V_{\mathrm{lyap}}^{(i)}
+\begin{bmatrix}T_{v,\mathrm{br}}^{(i-1)}\hat{P}_{\mathrm{br}}^{(i-1)}\big(V_{\mathrm{br}}^{(i-1)}\big)^*C_1^\top(I-D_1D_1^\top)^{-1} C_1V_{\mathrm{lyap}}^{(i)}\\0\end{bmatrix} \Bigg)t_{v,\mathrm{br}}^{(i)}\nonumber\\
&\hspace*{3cm}-t_{v,\mathrm{br}}^{(i)}s_{v,\mathrm{lyap}}^{(i)}+\Bigg(\hat{B}_{1,i}^{\mathrm{lyap}}\big(I+D_1^\top(I-D_1D_1^\top)^{-1}D_1\big)^{\frac{1}{2}}\nonumber\\
&\hspace*{9cm}-\begin{bmatrix}T_{v,\mathrm{br}}^{(i-1)}\hat{P}_{\mathrm{br}}^{(i-1)}(L_{v,\mathrm{lyap}}^{(i-1)})^\top\\0\end{bmatrix}\Bigg)l_{v,\mathrm{lyap}}^{(i)}=0,\label{tv_br}\\
&\Bigg(\big(\hat{A}_{2,i}^{\mathrm{lyap}}\big)^*+\big(\hat{C}_{2,i}^{\mathrm{lyap}}\big)^*D_2(I-D_2^\top D_2)^{-1}B_2^\top W_{\mathrm{lyap}}^{(i)}\nonumber\\
&\hspace*{1.5cm}+\begin{bmatrix}T_{w,\mathrm{br}}^{(i-1)}\hat{Q}_{\mathrm{br}}^{(i-1)}\big(W_{\mathrm{br}}^{(i-1)}\big)^*B_2(I-D_2^\top D_2)^{-1} B_2^\top W_{\mathrm{lyap}}^{(i)}\\0\end{bmatrix} \Bigg)t_{w,\mathrm{br}}^{(i)}-t_{w,\mathrm{br}}^{(i)}s_{w,\mathrm{lyap}}^{(i)}\nonumber\\
&\hspace*{2.5cm}+\Bigg(\big(\hat{C}_{2,i}^{\mathrm{lyap}}\big)^*\big(I+D_2(I-D_2^\top D_2)^{-1}D_2^\top\big)^{\frac{1}{2}}
-\begin{bmatrix}T_{w,\mathrm{br}}^{(i-1)}\hat{Q}_{\mathrm{br}}^{(i-1)}(L_{w,\mathrm{lyap}}^{(i-1)})^\top\\0\end{bmatrix}\Bigg)l_{w,\mathrm{lyap}}^{(i)}=0,\label{tw_br}\\
&\Bigg(\hat{A}_{1,i}^{\mathrm{lyap}}-\hat{B}_{1,i}^{\mathrm{lyap}}(D_1^\top D_1)^{-1}\big(B_1^\top Q_2+D_1^\top C_1\big)V_{\mathrm{lyap}}^{(i)}
+\begin{bmatrix}T_{v,\mathrm{sf}}^{(i-1)}\hat{P}_{\mathrm{sf}}^{(i-1)}\big(V_{\mathrm{sf}}^{(i-1)}\big)^*C_{1,\mathrm{sf}}^\top C_{1,\mathrm{sf}}V_{\mathrm{lyap}}^{(i)}\\0\end{bmatrix}\Bigg)t_{v,\mathrm{sf}}^{(i)}\nonumber\\
&\hspace*{3cm}-t_{v,\mathrm{sf}}^{(i)}s_{v,\mathrm{lyap}}^{(i)}+\Bigg(\hat{B}_{1,i}^{\mathrm{lyap}}(D_1^\top D_1)^{-\frac{1}{2}}-\begin{bmatrix}T_{v,\mathrm{sf}}^{(i-1)}\hat{P}_{\mathrm{sf}}^{(i-1)}(L_{v,\mathrm{lyap}}^{(i-1)})^\top\\0\end{bmatrix}\Bigg)l_{v,\mathrm{lyap}}^{(i)}=0,\label{tv_sf}\\
&\Bigg(\big(\hat{A}_{2,i}^{\mathrm{lyap}}\big)^*-\big(\hat{C}_{2,i}^{\mathrm{lyap}}\big)^*(D_2D_2^\top)^{-1}\big(C_2^\top P_1+D_2 B_2\big)W_{\mathrm{lyap}}^{(i)}
+\begin{bmatrix}T_{w,\mathrm{sf}}^{(i-1)}\hat{Q}_{\mathrm{sf}}^{(i-1)}\big(W_{\mathrm{sf}}^{(i-1)}\big)^*B_{2,\mathrm{sf}} B_{2,\mathrm{sf}}^\top W_{\mathrm{lyap}}^{(i)}\\0\end{bmatrix}\Bigg)t_{w,\mathrm{sf}}^{(i)}\nonumber\\
&\hspace*{3cm}-t_{w,\mathrm{sf}}^{(i)}s_{w,\mathrm{lyap}}^{(i)}+\Bigg(\big(\hat{B}_{2,i}^{\mathrm{lyap}}\big)^*(D_2 D_2^\top)^{-\frac{1}{2}}-\begin{bmatrix}T_{w,\mathrm{sf}}^{(i-1)}\hat{Q}_{\mathrm{sf}}^{(i-1)}(L_{w,\mathrm{lyap}}^{(i-1)})^\top\\0\end{bmatrix}\Bigg)l_{w,\mathrm{lyap}}^{(i)}=0.\label{tw_sf}
\end{align}
\begin{theorem}
Let $\{\alpha_i\}_{i=1}^{k}\in\mathbb{C}_{-}$ and $\{\beta_i\}_{i=1}^{k}\in\mathbb{C}_{-}$ be the ADI shifts used in CF-ADI to approximate \eqref{lyap_p} and \eqref{lyap_q}, respectively. Then $v_i^{\mathrm{mp}}$, $w_i^{\mathrm{mp}}$, $v_i^{\infty}$, $w_i^{\infty}$, $v_i^{\mathrm{pr}}$, $w_i^{\mathrm{pr}}$, $v_i^{\mathrm{br}}$, $w_i^{\mathrm{br}}$, $v_i^{\mathrm{sf}}$, $w_i^{\mathrm{sf}}$, $w_i^{\mathrm{ricc}}$ can be extracted from $V_{\mathrm{lyap}}^{(i)}$ and $W_{\mathrm{lyap}}^{(i)}$ as follows:
\begin{align}
v_i^{\mathrm{mp}}&=V_{\mathrm{lyap}}^{(i)}t_{v,\mathrm{mp}}^{(i)},& v_i^{\infty}&=V_{\mathrm{lyap}}^{(i)}t_{v,\infty}^{(i)},& v_i^{\mathrm{pr}}&=V_{\mathrm{lyap}}^{(i)}t_{v,\mathrm{pr}}^{(i)},& v_i^{\mathrm{br}}&=V_{\mathrm{lyap}}^{(i)}t_{v,\mathrm{br}}^{(i)},& v_i^{\mathrm{sf}}&=V_{\mathrm{lyap}}^{(i)}t_{v,\mathrm{sf}}^{(i)},\nonumber\\
w_i^{\mathrm{mp}}&=W_{\mathrm{lyap}}^{(i)}t_{w,\mathrm{mp}}^{(i)},& w_i^{\infty}&=W_{\mathrm{lyap}}^{(i)}t_{w,\infty}^{(i)},& w_i^{\mathrm{pr}}&=W_{\mathrm{lyap}}^{(i)}t_{w,\mathrm{pr}}^{(i)},& w_i^{\mathrm{br}}&=W_{\mathrm{lyap}}^{(i)}t_{w,\mathrm{br}}^{(i)},&  w_i^{\mathrm{sf}}&=W_{\mathrm{lyap}}^{(i)}t_{w,\mathrm{sf}}^{(i)},\nonumber\\
w_i^{\mathrm{ricc}}&=W_{\mathrm{lyap}}^{(i)}t_{w,\mathrm{ricc}}^{(i)}.\nonumber
\end{align}
Furthermore, $V_{\mathrm{mp}}^{(i)}$, $W_{\mathrm{mp}}^{(i)}$, $V_{\infty}^{(i)}$, $W_{\infty}^{(i)}$, $V_{\mathrm{pr}}^{(i)}$, $W_{\mathrm{pr}}^{(i)}$, $V_{\mathrm{br}}^{(i)}$, $W_{\mathrm{br}}^{(i)}$, $V_{\mathrm{sf}}^{(i)}$, $W_{\mathrm{sf}}^{(i)}$, and $W_{\mathrm{radi}}^{(i)}$ can be obtained from $V_{\mathrm{lyap}}^{(i)}$ and $W_{\mathrm{lyap}}^{(i)}$ as follows:
\begin{align}
V_{\mathrm{mp}}^{(i)}&=V_{\mathrm{lyap}}^{(i)}T_{v,\mathrm{mp}}^{(i)},& V_{\infty}^{(i)}&=V_{\mathrm{lyap}}^{(i)}T_{v,\infty}^{(i)},& V_{\mathrm{pr}}^{(i)}&=V_{\mathrm{lyap}}^{(i)}T_{v,\mathrm{pr}}^{(i)},& V_{\mathrm{br}}^{(i)}&=V_{\mathrm{lyap}}^{(i)}T_{v,\mathrm{br}}^{(i)},& V_{\mathrm{sf}}^{(i)}&=V_{\mathrm{lyap}}^{(i)}T_{v,\mathrm{sf}}^{(i)},\nonumber\\
W_{\mathrm{mp}}^{(i)}&=W_{\mathrm{lyap}}^{(i)}T_{w,\mathrm{mp}}^{(i)},& W_{\infty}^{(i)}&=W_{\mathrm{lyap}}^{(i)}T_{w,\infty}^{(i)},& W_{\mathrm{pr}}^{(i)}&=W_{\mathrm{lyap}}^{(i)}T_{w,\mathrm{pr}}^{(i)},& W_{\mathrm{br}}^{(i)}&=W_{\mathrm{lyap}}^{(i)}T_{w,\mathrm{br}}^{(i)},& W_{\mathrm{sf}}^{(i)}&=W_{\mathrm{lyap}}^{(i)}T_{w,\mathrm{sf}}^{(i)},\nonumber\\
W_{\mathrm{radi}}^{(i)}&=W_{\mathrm{lyap}}^{(i)}T_{w,\mathrm{ricc}}^{(i)}.\nonumber
\end{align}
\end{theorem}
\begin{proof}
The proof is similar to that of Theorem~\ref{th_uni_radi} and is therefore omitted for brevity.
\end{proof}
The Riccati equations \eqref{ricc_p_sf} and \eqref{ricc_q_sf} involve \(Q_2\) and \(P_1\), which are not feasible to compute in a large-scale setting. If, however, ADI-based approximations of \(Q_2\) and \(P_1\) are available—namely, $Q_2 \approx W_{\mathrm{lyap}}^{(i)}\hat{Q}_{\mathrm{lyap}}^{(i)}\big(W_{\mathrm{lyap}}^{(i)}\big)^*$ and
$P_1 \approx V_{\mathrm{lyap}}^{(i)}\hat{P}_{\mathrm{lyap}}^{(i)}\big(V_{\mathrm{lyap}}^{(i)}\big)^*$, then we can approximate \(P_{\mathrm{sf}}\) and \(Q_{\mathrm{sf}}\) as \(P_{\mathrm{sf}} \approx \mathcal{P}_{\mathrm{sf}}^{(i)}\) and \(Q_{\mathrm{sf}} \approx \mathcal{Q}_{\mathrm{sf}}^{(i)}\) by solving the following Riccati equations:
\begin{align}
&A_1 \mathcal{P}_{\mathrm{sf}}^{(i)}E_1^\top+E_1 \mathcal{P}_{\mathrm{sf}}^{(i)}A_1^\top+\Bigg(B_1-E_1\mathcal{P}_{\mathrm{sf}}^{(i)}\Big(W_{\mathrm{lyap}}^{(i)}\hat{Q}_{\mathrm{lyap}}^{(i)}\big(W_{\mathrm{lyap}}^{(i)}\big)^*B_1+C_1^\top D_1\Big)\Bigg)\big(D_1^\top D_1\big)^{-1}\nonumber\\
&\hspace*{6cm}\times\Bigg(B_1-E_1\mathcal{P}_{\mathrm{sf}}^{(i)}\Big(W_{\mathrm{lyap}}^{(i)}\hat{Q}_{\mathrm{lyap}}^{(i)}\big(W_{\mathrm{lyap}}^{(i)}\big)^*B_1+C_1^\top D_1\Big)\Bigg)^*=0,\label{sf_ricc_p}\\
&A_2^\top \mathcal{Q}_{\mathrm{sf}}^{(i)}E_2+E_2^\top \mathcal{Q}_{\mathrm{sf}}^{(i)}A_2+\Bigg(C_2-\Big(C_2V_{\mathrm{lyap}}^{(i)}\hat{P}_{\mathrm{lyap}}^{(i)}\big(V_{\mathrm{lyap}}^{(i)}\big)^*+D_2B_2^\top\Big)\mathcal{Q}_{\mathrm{sf}}^{(i)}E_2\Bigg)^*\big(D_2D_2^\top\big)^{-1}\nonumber\\
&\hspace*{6cm}\times\Bigg(C_2-\Big(C_2V_{\mathrm{lyap}}^{(i)}\hat{P}_{\mathrm{lyap}}^{(i)}\big(V_{\mathrm{lyap}}^{(i)}\big)^*+D_2B_2^\top\Big)\mathcal{Q}_{\mathrm{sf}}^{(i)}E_2\Bigg)=0.\label{sf_ricc_q}
\end{align}
RADI-based approximations of \(\mathcal{P}_{\mathrm{sf}}^{(i)}\) and \(\mathcal{Q}_{\mathrm{sf}}^{(i)}\) can still be extracted from \(V_{\mathrm{lyap}}^{(i)}\) and \(W_{\mathrm{lyap}}^{(i)}\) along similar lines by redefining \(A_{1,\mathrm{sf}}\), \(C_{1,\mathrm{sf}}\), \(A_{2,\mathrm{sf}}\), and \(B_{2,\mathrm{sf}}\) as follows:
\begin{align}
A_{1,\mathrm{sf}}&=A_1-B_1(D_1^\top D_1)^{-1}\Big(B_1^\top W_{\mathrm{lyap}}^{(i)}\hat{Q}_{\mathrm{lyap}}^{(i)}\big(W_{\mathrm{lyap}}^{(i)}\big)^*+D_1^\top C_1\Big),\nonumber\\
C_{1,\mathrm{sf}}&=\big(D_1^\top D_1\big)^{-\frac{1}{2}}\Big(B_1^\top W_{\mathrm{lyap}}^{(i)}\hat{Q}_{\mathrm{lyap}}^{(i)}\big(W_{\mathrm{lyap}}^{(i)}\big)^*+D_1^\top C_1\Big),\\
A_{2,\mathrm{sf}}&=A_2-\Big(V_{\mathrm{lyap}}^{(i)}\hat{P}_{\mathrm{lyap}}^{(i)}\big(V_{\mathrm{lyap}}^{(i)}\big)^*C_2^\top+B_2D_2^\top\Big)\big(D_2D_2^\top\big)^{-1}C_2,\nonumber\\ B_{2,\mathrm{sf}}&=\Big(V_{\mathrm{lyap}}^{(i)}\hat{P}_{\mathrm{lyap}}^{(i)}\big(V_{\mathrm{lyap}}^{(i)}\big)^*C_2^\top+B_2D_2^\top\Big)\big(D_2D_2^\top\big)^{-\frac{1}{2}}.
\end{align}
However, their simultaneous computation with \(Q_2\) and \(P_1\) is tricky: whenever the approximations $Q_2 \approx W_{\mathrm{lyap}}^{(i)}\hat{Q}_{\mathrm{lyap}}^{(i)}\big(W_{\mathrm{lyap}}^{(i)}\big)^*$ and
$P_1 \approx V_{\mathrm{lyap}}^{(i)}\hat{P}_{\mathrm{lyap}}^{(i)}\big(V_{\mathrm{lyap}}^{(i)}\big)^*$ are updated, \(\mathcal{P}_{\mathrm{sf}}^{(i)}\) and \(\mathcal{Q}_{\mathrm{sf}}^{(i)}\) change as well, so the previous approximations \(\mathcal{P}_{\mathrm{sf}}^{(i-1)}\) and \(\mathcal{Q}_{\mathrm{sf}}^{(i-1)}\) must be discarded. Moreover, the residuals of the Riccati equations \eqref{sf_ricc_p} and \eqref{sf_ricc_q} can only be used to assess the quality of the approximations \(P_{\mathrm{sf}} \approx \mathcal{P}_{\mathrm{sf}}^{(i)}\) and \(Q_{\mathrm{sf}} \approx \mathcal{Q}_{\mathrm{sf}}^{(i)}\) once the residuals of the Lyapunov equations \eqref{lyap_q} and \eqref{lyap_p} have dropped significantly—indicating that the approximations of \(Q_2\) and \(P_1\) are already accurate.

RADI-based approximations of \(\mathcal{P}_{\mathrm{sf}}^{(i)}\) and \(\mathcal{Q}_{\mathrm{sf}}^{(i)}\) can be obtained from \(V_{\mathrm{lyap}}^{(i)}\) and \(W_{\mathrm{lyap}}^{(i)}\) (similar to how the approximation of \(P_{\mathrm{ricc}}\) is obtained using Theorem~\ref{prop_pt_ricc}) as
\[\mathcal{P}_{\mathrm{sf}}^{(i)}\approx V_{\mathrm{lyap}}^{(i)} \mathcal{T}_{v,\mathrm{sf}}^{(i)}\hat{\mathcal{P}}_{\mathrm{sf}}^{(i)}\big(\mathcal{T}_{v,\mathrm{sf}}^{(i)}\big)^*\big(V_{\mathrm{lyap}}^{(i)}\big)^*\quad
\text{and}\quad\mathcal{Q}_{\mathrm{sf}}^{(i)}\approx W_{\mathrm{lyap}}^{(i)} \mathcal{T}_{w,\mathrm{sf}}^{(i)}\hat{\mathcal{Q}}_{\mathrm{sf}}^{(i)}\big(\mathcal{T}_{w,\mathrm{sf}}^{(i)}\big)^*\big(W_{\mathrm{lyap}}^{(i)}\big)^*,\]
where \(\mathcal{T}_{v,\mathrm{sf}}^{(i)}\), \(\big(\hat{\mathcal{P}}_{\mathrm{sf}}^{(i)}\big)^{-1}\), \(\mathcal{T}_{w,\mathrm{sf}}^{(i)}\), and \(\big(\hat{\mathcal{Q}}_{\mathrm{sf}}^{(i)}\big)^{-1}\) solve the following linear matrix equations:
\begin{align}
&\hat{A}_{1,i}^{\mathrm{lyap}}-\hat{B}_{1,i}^{\mathrm{lyap}}(D_1^\top D_1)^{-1}\Big(B_1^\top W_{\mathrm{lyap}}^{(i)}\hat{Q}_{\mathrm{lyap}}^{(i)}\big(W_{\mathrm{lyap}}^{(i)}\big)^*V_{\mathrm{lyap}}^{(i)}+D_1^\top C_1V_{\mathrm{lyap}}^{(i)}\Big)\mathcal{T}_{v,\mathrm{sf}}^{(i)}-\mathcal{T}_{v,\mathrm{sf}}^{(i)}S_{v,\mathrm{lyap}}^{(i)}\nonumber\\
&\hspace*{10cm}+\hat{B}_{1,i}^{\mathrm{lyap}}(D_1^\top D_1)^{-\frac{1}{2}}L_{v,\mathrm{lyap}}^{(i)}=0,\label{Tv_sf}
\end{align}
\begin{align}
&-\big(S_{v,\mathrm{lyap}}^{(i)}\big)^*\big(\hat{\mathcal{P}}_{\mathrm{sf}}^{(i)}\big)^{-1}-\big(\hat{\mathcal{P}}_{\mathrm{sf}}^{(i)}\big)^{-1}S_{v,\mathrm{lyap}}^{(i)}+\big(L_{v,\mathrm{lyap}}^{(i)}\big)^\top L_{v,\mathrm{lyap}}^{(i)}\nonumber\\
&\hspace*{3cm}-\big(\mathcal{T}_{v,\mathrm{sf}}^{(i)}\big)^*\big(V_{\mathrm{lyap}}^{(i)}\big)^*\Big(B_1^\top W_{\mathrm{lyap}}^{(i)}\hat{Q}_{\mathrm{lyap}}^{(i)}\big(W_{\mathrm{lyap}}^{(i)}\big)^*+D_1^\top C_1\Big)^*\big(D_1^\top D_1\big)^{-1}\nonumber\\
&\hspace*{7cm}\times\Big(B_1^\top W_{\mathrm{lyap}}^{(i)}\hat{Q}_{\mathrm{lyap}}^{(i)}\big(W_{\mathrm{lyap}}^{(i)}\big)^*+D_1^\top C_1\Big)V_{\mathrm{lyap}}^{(i)}\mathcal{T}_{v,\mathrm{sf}}^{(i)}=0,\label{P_sf}\\
&\big(\hat{A}_{2,i}^{\mathrm{lyap}}\big)^*-\big(\hat{C}_{2,i}^{\mathrm{lyap}}\big)^*(D_2 D_2^\top)^{-1}\Big(C_2 V_{\mathrm{lyap}}^{(i)}\hat{P}_{\mathrm{lyap}}^{(i)}\big(V_{\mathrm{lyap}}^{(i)}\big)^*W_{\mathrm{lyap}}^{(i)}+D_2 B_2^\top W_{\mathrm{lyap}}^{(i)}\Big)\mathcal{T}_{w,\mathrm{sf}}^{(i)}-\mathcal{T}_{w,\mathrm{sf}}^{(i)}S_{w,\mathrm{lyap}}^{(i)}\nonumber\\
&\hspace*{10cm}+\big(\hat{C}_{2,i}^{\mathrm{lyap}}\big)^*(D_2 D_2^\top)^{-\frac{1}{2}}L_{w,\mathrm{lyap}}^{(i)}=0,\label{Tw_sf}\\
&-\big(S_{w,\mathrm{lyap}}^{(i)}\big)^*\big(\hat{\mathcal{Q}}_{\mathrm{sf}}^{(i)}\big)^{-1}-\big(\hat{\mathcal{Q}}_{\mathrm{sf}}^{(i)}\big)^{-1}S_{w,\mathrm{lyap}}^{(i)}+\big(L_{w,\mathrm{lyap}}^{(i)}\big)^\top L_{w,\mathrm{lyap}}^{(i)}\nonumber\\
&\hspace*{3cm}-\big(\mathcal{T}_{w,\mathrm{sf}}^{(i)}\big)^*\big(W_{\mathrm{lyap}}^{(i)}\big)^*\Big(C_2 V_{\mathrm{lyap}}^{(i)}\hat{P}_{\mathrm{lyap}}^{(i)}\big(V_{\mathrm{lyap}}^{(i)}\big)^*+D_2 B_2^\top\Big)^*\big(D_2 D_2^\top\big)^{-1}\nonumber\\
&\hspace*{7cm}\times\Big(C_2 V_{\mathrm{lyap}}^{(i)}\hat{P}_{\mathrm{lyap}}^{(i)}\big(V_{\mathrm{lyap}}^{(i)}\big)^*+D_2 B_2^\top\Big)W_{\mathrm{lyap}}^{(i)}\mathcal{T}_{w,\mathrm{sf}}^{(i)}=0.\label{Q_sf}
\end{align}These approximations must be recomputed each time the approximations of \(Q_2\) and \(P_1\) are updated. Overall, the simultaneous approximation of \(P_{\mathrm{sf}}\) and \(Q_{\mathrm{sf}}\) together with \(Q_2\) and \(P_1\) is less elegant than the Riccati equations considered earlier, as it loses the recursive accumulation of matrices.
\subsection{Realification of Matrices in Unified ADI Framework}
In the recursive interpolation framework, there is no theoretical requirement to use one interpolation point at a time. As long as the pairs  
$\big(-(s_{v}^{(i)})^*,(l_v^{(i)})^*\big)$ and $\big(-(s_w^{(i)})^*,(l_w^{(i)})^*\big)$ are observable, $s_v^{(i)}$ and $s_w^{(i)}$ may contain multiple interpolation points simultaneously \cite{wolfthesis}. This property can be exploited in CF-ADI to keep the matrices $V_{\mathrm{lyap}}^{(i)}$ and $W_{\mathrm{lyap}}^{(i)}$ real-valued even when the shifts are complex-valued. Keeping most matrices in CF-ADI real is beneficial both in terms of practical utility and computational efficiency \cite{benner2013efficient}.

Assume that for every complex-valued shift $\alpha_i$ and $\beta_i$, the next shifts are $\alpha_{i+1} = \overline{\alpha}_i$ and $\beta_{i+1} = \overline{\beta}_i$. Then the simplest—and widely used—choice in rational interpolation literature \cite{beattie2017chapter} for the pairs $(s_{v,\mathrm{lyap}}^{(i)},l_{v,\mathrm{lyap}}^{(i)})$, $(s_{w,\mathrm{lyap}}^{(i)},l_{w,\mathrm{lyap}}^{(i)})$, $(s_{v,\mathrm{sylv}}^{(i)},l_{v,\mathrm{sylv}}^{(i)})$, $(s_{w,\mathrm{sylv}}^{(i)},l_{w,\mathrm{sylv}}^{(i)})$, $(s_{v,\mathrm{ricc}}^{(i)},l_{v,\mathrm{sylv}}^{(i)})$, and $(s_{w,\mathrm{ricc}}^{(i)},l_{w,\mathrm{ricc}}^{(i)})$ is as follows:
\begin{align}
s_{v,\mathrm{lyap}}^{(i)}&=s_{v,\mathrm{sylv}}^{(i)}=s_{v,\mathrm{ricc}}^{(i)}=\begin{bmatrix}-\mathrm{Re}(\alpha_i)I&\mathrm{Im}(-\alpha_i)I\\-\mathrm{Im}(-\alpha_i)I&-\mathrm{Re}(\alpha_i)I\end{bmatrix},&
l_{v,\mathrm{lyap}}^{(i)}&=l_{v,\mathrm{sylv}}^{(i)}=l_{v,\mathrm{ricc}}^{(i)}=\begin{bmatrix}-I&0\end{bmatrix},\nonumber\\
s_{w,\mathrm{lyap}}^{(i)}&=s_{w,\mathrm{sylv}}^{(i)}=s_{w,\mathrm{ricc}}^{(i)}=\begin{bmatrix}-\mathrm{Re}(\beta_i)I&\mathrm{Im}(\beta_i)I\\-\mathrm{Im}(\beta_i)I&-\mathrm{Re}(\beta_i)I\end{bmatrix},&
l_{w,\mathrm{lyap}}^{(i)}&=l_{w,\mathrm{sylv}}^{(i)}=l_{w,\mathrm{ricc}}^{(i)}=\begin{bmatrix}-I&0\end{bmatrix}.\nonumber
\end{align}
However, for consistency with existing literature—and to leverage certain analytical expressions found there—we adopt the same implicit choices used in \cite{benner2013reformulated,benner2014computing,benner2018radi}.

When $\alpha_i\in\mathbb{R}_{-}$ and $\beta_i\in\mathbb{R}_{-}$, the matrices $s_{v,\mathrm{lyap}}^{(i)}$, $l_{v,\mathrm{lyap}}^{(i)}$, $s_{w,\mathrm{lyap}}^{(i)}$, and $l_{w,\mathrm{lyap}}^{(i)}$ are (implicitly) set in \cite{benner2013reformulated} as follows:
\begin{align}
s_{v,\mathrm{lyap}}^{(i)}&=-\alpha_iI,& l_{v,\mathrm{lyap}}^{(i)}&=-\sqrt{-2\alpha_i}I,& s_{w,\mathrm{lyap}}^{(i)}&=-\beta_iI,& l_{w,\mathrm{lyap}}^{(i)}&=-\sqrt{-2\beta_i}I.\label{sv_cf_r}
\end{align}
When $\alpha_i\in\mathbb{C}_{-}$ and $\beta_i\in\mathbb{C}_{-}$, define  
$\phi_v = \sqrt{-\mathrm{Re}(\alpha_i)}$,  
$\delta_v = \frac{\mathrm{Re}(\alpha_i)}{\mathrm{Im}(\alpha_i)}$,  
$\phi_w = \sqrt{-\mathrm{Re}(\beta_i)}$, and  
$\delta_w = \frac{\mathrm{Re}(\beta_i)}{\mathrm{Im}(\beta_i)}$.  
Then, as in \cite{benner2013reformulated}, the matrices are (implicitly) set as:
\begin{align}
s_{v,\mathrm{lyap}}^{(i)}&=t_v^{-1}\begin{bmatrix}-\alpha_i&0\\0&-\overline{\alpha}_i\end{bmatrix}t_v\otimes I,& l_{v,\mathrm{lyap}}^{(i)}&=\begin{bmatrix}-1&-1\end{bmatrix}t_v\otimes I,\label{sv_cf_i}\\
s_{w,\mathrm{lyap}}^{(i)}&=t_w^{-1}\begin{bmatrix}-\beta_i&0\\0&-\overline{\beta}_i\end{bmatrix}t_w\otimes I,& l_{w,\mathrm{lyap}}^{(i)}&=\begin{bmatrix}-1&-1\end{bmatrix}t_w\otimes I,\label{sw_cf_i}
\end{align}where
\begin{align}
t_v&=\phi_v\begin{bmatrix}1&-j\\1&j\end{bmatrix}\begin{bmatrix}1&0\\\delta_v&\sqrt{1+\delta_v^2}\end{bmatrix},&
t_w&=\phi_w\begin{bmatrix}1&-j\\1&j\end{bmatrix}\begin{bmatrix}1&0\\\delta_w&\sqrt{1+\delta_w^2}\end{bmatrix}.
\end{align}
A notable advantage of this realification strategy is that the matrices $\hat{P}_{\mathrm{lyap}}^{(i)}$ and $\hat{Q}_{\mathrm{lyap}}^{(i)}$ become $\hat{P}_{\mathrm{lyap}}^{(i)} = I$ and $\hat{Q}_{\mathrm{lyap}}^{(i)} = I$, as noted in \cite{wolf2016adi}. Consequently, we obtain the simplified expressions:
\begin{align}
S_{v,\mathrm{lyap}}^{(i)}&=\begin{bmatrix}S_{v,\mathrm{lyap}}^{(i-1)}&(L_{v,\mathrm{lyap}}^{(i-1)})^\top l_{v,\mathrm{lyap}}^{(i)}\\0&s_{v,\mathrm{lyap}}^{(i)}\end{bmatrix},&
L_{v,\mathrm{lyap}}^{(i)}&=\begin{bmatrix}L_{v,\mathrm{lyap}}^{(i-1)}&l_{v,\mathrm{lyap}}^{(i)}\end{bmatrix},\label{SvLv}\\
\hat{A}_{1,i}^{\mathrm{lyap}} &= -\big(S_{v,\mathrm{lyap}}^{(i)}\big)^*,&  
\hat{B}_{1,i}^{\mathrm{lyap}} &= \big(L_{v,\mathrm{lyap}}^{(i)}\big)^\top,\label{A_lyapB_lyap}\\
S_{w,\mathrm{lyap}}^{(i)}&=\begin{bmatrix}S_{w,\mathrm{lyap}}^{(i-1)}&(L_{w,\mathrm{lyap}}^{(i-1)})^\top l_{w,\mathrm{lyap}}^{(i)}\\0&s_{w,\mathrm{lyap}}^{(i)}\end{bmatrix},& L_{w,\mathrm{lyap}}^{(i)}&=\begin{bmatrix}L_{w,\mathrm{lyap}}^{(i-1)}&l_{w,\mathrm{lyap}}^{(i)}\end{bmatrix},\label{SwLw}\\
\hat{A}_{2,i}^{\mathrm{lyap}} &= -S_{w,\mathrm{lyap}}^{(i)},&
\hat{C}_{2,i}^{\mathrm{lyap}} &= L_{w,\mathrm{lyap}}^{(i)}.\label{Aw_lyapBw_lyap}
\end{align}  
The same strategy is implicitly used for the realification of RADI in \cite{benner2018radi}.

For the realification of matrices in FADI, \cite{benner2014computing} proposes grouping the shifts as follows:
\begin{enumerate}
  \item Case I: $\alpha_i\in\mathbb{R}_{-}$ and $\beta_i\in\mathbb{R}_{-}$.
  \item Case II: $\alpha_i\in\mathbb{C}_{-}$, $\alpha_{i+1}=\overline{\alpha}_i$, $\beta_i\in\mathbb{C}_{-}$, and $\beta_{i+1}=\overline{\beta}_i$.
  \item Case III: $\alpha_i\in\mathbb{R}_{-}$, $\alpha_{i+1}\in\mathbb{R}_{-}$, $\beta_i\in\mathbb{C}_{-}$, and $\beta_{i+1}=\overline{\beta}_i$.
  \item Case IV: $\alpha_i\in\mathbb{C}_{-}$, $\alpha_{i+1}=\overline{\alpha}_i$, $\beta_i\in\mathbb{R}_{-}$, and $\beta_{i+1}\in\mathbb{R}_{-}$.
\end{enumerate}
For Case I, the matrices $s_{v,\mathrm{sylv}}^{(i)}$, $l_{v,\mathrm{sylv}}^{(i)}$, $s_{w,\mathrm{sylv}}^{(i)}$, and $l_{w,\mathrm{sylv}}^{(i)}$ are (implicitly) set in \cite{benner2014computing} as:
\begin{align}
s_{v,\mathrm{sylv}}^{(i)}&=-\alpha_iI,& l_{v,\mathrm{sylv}}^{(i)}&=-I,& s_{w,\mathrm{sylv}}^{(i)}&=-\beta_iI,& l_{w,\mathrm{sylv}}^{(i)}=-I.\label{sv_case1}
\end{align}
For Case II:
\begin{align}
s_{v,\mathrm{sylv}}^{(i)}&=\begin{bmatrix}-\mathrm{Re}(\alpha_i)I&\mathrm{Im}(-\alpha_i)I\\-\mathrm{Im}(-\alpha_i)I&-\mathrm{Re}(\alpha_i)I\end{bmatrix},&
l_{v,\mathrm{sylv}}^{(i)}&=\begin{bmatrix}-I&0\end{bmatrix},\label{sv_case2}\\
s_{w,\mathrm{sylv}}^{(i)}&=\begin{bmatrix}-\mathrm{Re}(\beta_i)I&\mathrm{Im}(\beta_i)I\\-\mathrm{Im}(\beta_i)I&-\mathrm{Re}(\beta_i)I\end{bmatrix},&
l_{w,\mathrm{sylv}}^{(i)}&=\begin{bmatrix}-I&0\end{bmatrix}.\label{sw_case2}
\end{align}
For Case III:
\begin{align}
s_{v,\mathrm{sylv}}^{(i)}&=\begin{bmatrix}-\alpha_iI&I\\0&-\alpha_{i+1}\end{bmatrix},& l_{v,\mathrm{sylv}}^{(i)}&=\begin{bmatrix}-I&0\end{bmatrix},\label{sv_case3}\\
s_{w,\mathrm{sylv}}^{(i)}&=\begin{bmatrix}-\mathrm{Re}(\beta_i)I&\mathrm{Im}(\beta_i)I\\-\mathrm{Im}(\beta_i)I&-\mathrm{Re}(\beta_i)I\end{bmatrix},&
l_{w,\mathrm{sylv}}^{(i)}&=\begin{bmatrix}-I&0\end{bmatrix}.\label{sw_case3}
\end{align}
For Case IV:
\begin{align}
s_{v,\mathrm{sylv}}^{(i)}&=\begin{bmatrix}-\mathrm{Re}(\alpha_i)I&\mathrm{Im}(-\alpha_i)I\\-\mathrm{Im}(-\alpha_i)I&-\mathrm{Re}(\alpha_i)I\end{bmatrix},&
l_{v,\mathrm{sylv}}^{(i)}&=\begin{bmatrix}-I&0\end{bmatrix},\label{sv_case4}\\
s_{w,\mathrm{sylv}}^{(i)}&=\begin{bmatrix}-\beta_iI&I\\0&-\beta_{i+1}I\end{bmatrix},&
l_{w,\mathrm{sylv}}^{(i)}&=\begin{bmatrix}-I&0\end{bmatrix}.\label{sw_case4}
\end{align}
The pseudo-code for the proposed Unified low-rank ADI algorithm (UADI), incorporating these realification strategies to simultaneously solve \eqref{lyap_p}–\eqref{ricc_q_sf} with only two linear solves per iteration, is given in Algorithm 4. For clarity, we solve the small-scale projected Lyapunov/Sylvester equations satisfied by the variables  
$\big(\hat{d}_{\mathrm{sylv}}^{(i)}\big)^{-1}$, $\big(\hat{p}_{\mathrm{ricc}}^{(i)}\big)^{-1}$,  
$\big(\hat{q}_{\mathrm{ricc}}^{(i)}\big)^{-1}$,  
$\big(\hat{p}_{\infty}^{(i)}\big)^{-1}$,  
$\big(\hat{q}_{\infty}^{(i)}\big)^{-1}$,  
$\big(\hat{p}_{\mathrm{pr}}^{(i)}\big)^{-1}$,  
$\big(\hat{q}_{\mathrm{pr}}^{(i)}\big)^{-1}$,
$\big(\hat{p}_{\mathrm{br}}^{(i)}\big)^{-1}$, and
$\big(\hat{q}_{\mathrm{br}}^{(i)}\big)^{-1}$ in UADI. Analytical expressions for these matrices can be found in \cite{benner2013reformulated,benner2014computing,benner2018radi}, allowing them to be computed directly without solving the corresponding projected Lyapunov/Sylvester equations. Nevertheless, even if those equations are solved, the computations remain small-scale and inexpensive. The remaining Sylvester equations in UADI can be solved efficiently using the algorithm proposed in \cite{MPIMD11-11}. The main computational cost in UADI is concentrated in Steps (\ref{step_main_1}) and (\ref{step_main_2}), which involve shifted linear solves to compute $v_i^{\mathrm{lyap}}$ and $w_i^{\mathrm{lyap}}$. All other steps are small-scale and can be executed efficiently.

\noindent\rule{\textwidth}{0.8pt}

\noindent\textbf{Algorithm 4:} UADI

\noindent\rule{\textwidth}{0.8pt}

\textbf{Inputs:} Matrices of linear matrix equations: \(A_1\), \(B_1\), \(C_1\), \(D_1\), \(E_1\), \(S_1\), \(\gamma_1\), \(A_2\), \(B_2\), \(C_2\), \(D_2\), \(E_2\), \(S_2\), \(\gamma_2\); \\
       ADI shifts: \(\{\alpha_i\}_{i=1}^k \in \mathbb{C}_{-}\) and \(\{\beta_i\}_{i=1}^k \in \mathbb{C}_{-}\).

\textbf{Outputs:} Approximations:  
\(P_1 \approx V_{\mathrm{lyap}}^{(k)} (V_{\mathrm{lyap}}^{(k)})^\top\),  
\(Q_2 \approx W_{\mathrm{lyap}}^{(k)} (W_{\mathrm{lyap}}^{(k)})^\top\),  
\(P_s \approx V_{\mathrm{lyap}}^{(k)} (I\otimes S_1) (V_{\mathrm{lyap}}^{(k)})^\top\),  
\(Q_s \approx W_{\mathrm{lyap}}^{(k)} (I\otimes S_2) (W_{\mathrm{lyap}}^{(k)})^\top\),  
\(P_{\mathrm{mp}} \approx V_{\mathrm{mp}}^{(i)}(V_{\mathrm{mp}}^{(i)})^\top\),  
\(Q_{\mathrm{mp}} \approx W_{\mathrm{mp}}^{(i)}(W_{\mathrm{mp}}^{(i)})^\top\),  
\(X_{\mathrm{sylv}} \approx V_{\mathrm{fadi}}^{(i)}D_{\mathrm{fadi}}^{(i)}(W_{\mathrm{fadi}}^{(i)})^\top\),  
\(P_{\mathrm{ricc}} \approx V_{\mathrm{radi}}^{(i)}\hat{P}_{\mathrm{ricc}}^{(i)}(V_{\mathrm{radi}}^{(i)})^\top\),  
\(Q_{\mathrm{ricc}} \approx W_{\mathrm{radi}}^{(i)}\hat{Q}_{\mathrm{ricc}}^{(i)}(W_{\mathrm{radi}}^{(i)})^\top\),  
\(P_{\infty} \approx V_{\infty}^{(i)}\hat{P}_{\infty}^{(i)}(V_{\infty}^{(i)})^\top\),  
\(Q_{\infty} \approx W_{\infty}^{(i)}\hat{Q}_{\infty}^{(i)}(W_{\infty}^{(i)})^\top\),  
\(P_{\mathrm{pr}} \approx V_{\mathrm{pr}}^{(i)}\hat{P}_{\mathrm{pr}}^{(i)}(V_{\mathrm{pr}}^{(i)})^\top\),  
\(Q_{\mathrm{pr}} \approx W_{\mathrm{pr}}^{(i)}\hat{Q}_{\mathrm{pr}}^{(i)}(W_{\mathrm{pr}}^{(i)})^\top\),  
\(P_{\mathrm{br}} \approx V_{\mathrm{br}}^{(i)}\hat{P}_{\mathrm{br}}^{(i)}(V_{\mathrm{br}}^{(i)})^\top\),  
\(Q_{\mathrm{br}} \approx W_{\mathrm{br}}^{(i)}\hat{Q}_{\mathrm{br}}^{(i)}(W_{\mathrm{br}}^{(i)})^\top\),  
\(P_{\mathrm{sf}} \approx \mathcal{V}_{\mathrm{sf}}^{(i)}\hat{\mathcal{P}}_{\mathrm{sf}}^{(i)}(\mathcal{V}_{\mathrm{sf}}^{(i)})^\top\),  
\(Q_{\mathrm{sf}} \approx \mathcal{W}_{\mathrm{sf}}^{(i)}\hat{\mathcal{Q}}_{\mathrm{sf}}^{(i)}(\mathcal{W}_{\mathrm{sf}}^{(i)})^\top\);
Residuals:  
\(R_{p,\mathrm{lyap}} = B_{\perp,k}^{\mathrm{lyap}} (B_{\perp,k}^{\mathrm{lyap}})^\top\),  
\(R_{q,\mathrm{lyap}} = (C_{\perp,k}^{\mathrm{lyap}})^\top C_{\perp,k}^{\mathrm{lyap}}\),  
\(R_{p,\mathrm{s}} = B_{\perp,k}^{\mathrm{lyap}} S_1(B_{\perp,k}^{\mathrm{lyap}})^\top\),  
\(R_{q,\mathrm{s}} = (C_{\perp,k}^{\mathrm{lyap}})^\top S_2C_{\perp,k}^{\mathrm{lyap}}\),  
\(R_{p,\mathrm{mp}} = B_{\perp,k}^{\mathrm{mp}} (B_{\perp,k}^{\mathrm{mp}})^\top\),  
\(R_{q,\mathrm{mp}} = (C_{\perp,k}^{\mathrm{mp}})^\top C_{\perp,k}^{\mathrm{mp}}\),  
\(R_{\mathrm{sylv}} = B_{\perp,k}^{\mathrm{sylv}} C_{\perp,k}^{\mathrm{sylv}}\),  
\(R_{p,\mathrm{ricc}} = B_{\perp,k}^{\mathrm{ricc}} (B_{\perp,k}^{\mathrm{ricc}})^\top\),  
\(R_{q,\mathrm{ricc}} = (C_{\perp,k}^{\mathrm{ricc}})^\top C_{\perp,k}^{\mathrm{ricc}}\),  
\(R_{p,\infty} = B_{\perp,k}^{\infty} (B_{\perp,k}^{\infty})^\top\),  
\(R_{q,\infty} = (C_{\perp,k}^{\infty})^\top C_{\perp,k}^{\infty}\),  
\(R_{p,\mathrm{pr}} = B_{\perp,k}^{\mathrm{pr}} (B_{\perp,k}^{\mathrm{pr}})^\top\),  
\(R_{q,\mathrm{pr}} = (C_{\perp,k}^{\mathrm{pr}})^\top C_{\perp,k}^{\mathrm{pr}}\),  
\(R_{p,\mathrm{br}} = B_{\perp,k}^{\mathrm{br}} (B_{\perp,k}^{\mathrm{br}})^\top\),  
\(R_{q,\mathrm{br}} = (C_{\perp,k}^{\mathrm{br}})^\top C_{\perp,k}^{\mathrm{br}}\),  
\(R_{p,\mathrm{sf}} = B_{\perp,k}^{\mathrm{sf}} (B_{\perp,k}^{\mathrm{sf}})^\top\),  
\(R_{q,\mathrm{sf}} = (C_{\perp,k}^{\mathrm{sf}})^\top C_{\perp,k}^{\mathrm{sf}}\).

\begin{enumerate}
\item Initialize \(i = 1\), \(B_{\perp,0}^{\mathrm{lyap}} = B_1\), \(C_{\perp,0}^{\mathrm{lyap}} = C_2\), \(B_{\perp,0}^{\mathrm{mp}} = B_1D_1^{-1}\), \(C_{\perp,0}^{\mathrm{mp}} = D_2^{-1}C_2\), \(B_{\perp,0}^{\mathrm{sylv}} = B_1\), \(C_{\perp,0}^{\mathrm{sylv}} = C_2\), \(B_{\perp,0}^{\mathrm{ricc}} = B_1\), \(C_{\perp,0}^{\mathrm{ricc}} = C_2\), \(B_{\perp,0}^{\infty} = B_1\), \(C_{\perp,0}^{\infty} = C_2\), \(B_{\perp,0}^{\mathrm{pr}} = B_1(D_1 + D_1^\top)^{-\frac{1}{2}}\), \(C_{\perp,0}^{\mathrm{pr}} = (D_2 + D_2^\top)^{-\frac{1}{2}}C_2\), \(B_{\perp,0}^{\mathrm{br}} = B_1\big(I + D_1^\top(I - D_1D_1^\top)^{-1}D_1\big)^{\frac{1}{2}}\), and \(C_{\perp,0}^{\mathrm{br}} = \big(I + D_2(I - D_2^\top D_2)^{-1}D_2^\top\big)^{\frac{1}{2}}C_2\).

\item \textbf{while} \(i \leq k\) \textbf{do}

\item\label{step_main_1} Solve \((A_1 + \alpha_i E_1) v_i^{\mathrm{lyap}} = B_{\perp,i-1}^{\mathrm{lyap}}\) for \(v_i^{\mathrm{lyap}}\).

\item\label{step_main_2} Solve \((A_2^\top + \beta_i E_2^\top) w_i^{\mathrm{lyap}} = (C_{\perp,i-1}^{\mathrm{lyap}})^\top\) for \(w_i^{\mathrm{lyap}}\).

\item \textbf{if} \(\alpha_i \in \mathbb{R}_{-} \land \beta_i \in \mathbb{R}_{-}\) \textbf{then}

\item Expand \(V_{\mathrm{lyap}}^{(i)} = \big[ V_{\mathrm{lyap}}^{(i-1)} \;\; \sqrt{-2\alpha_i}\, v_i^{\mathrm{lyap}} \big]\) and update \(B_{\perp,i}^{\mathrm{lyap}} = B_{\perp,i-1}^{\mathrm{lyap}} - 2\alpha_i E_1 v_i^{\mathrm{lyap}}\).

\item\label{uadi_step7} Set \(s_{v,\mathrm{lyap}}^{(i)}\) and \(l_{v,\mathrm{lyap}}^{(i)}\) as in \eqref{sv_cf_r} and expand \(S_{v,\mathrm{lyap}}^{(i)}\), \(L_{v,\mathrm{lyap}}^{(i)}\), \(\hat{A}_{1,i}^{\mathrm{lyap}}\), and \(\hat{B}_{1,i}^{\mathrm{lyap}}\) as in \eqref{SvLv} and \eqref{A_lyapB_lyap}.

\item Compute \(t_{v,\mathrm{ricc}}^{(i)}\) by solving the Sylvester equation \eqref{tv_ricc} and set \(v_i^{\mathrm{ricc}} = V_{\mathrm{lyap}}^{(i)}t_{v,\mathrm{ricc}}^{(i)}\).

\item Set \(s_{v,\mathrm{ricc}}^{(i)} = s_{v,\mathrm{lyap}}^{(i)}\) and \(l_{v,\mathrm{ricc}}^{(i)} = l_{v,\mathrm{lyap}}^{(i)}\) in \eqref{lyap_p_ricc}, and solve the Lyapunov equation to compute \(\big(\hat{p}_{\mathrm{ricc}}^{(i)}\big)^{-1}\).

\item Expand \(V_{\mathrm{radi}}^{(i)} = \big[ V_{\mathrm{radi}}^{(i-1)} \;\; v_i^{\mathrm{ricc}} \big]\) and \(\hat{P}_{\mathrm{ricc}}^{(i)} = \mathrm{blkdiag}\big(\hat{P}_{\mathrm{ricc}}^{(i-1)},\hat{p}_{\mathrm{ricc}}^{(i)}\big)\).

\item Update \(B_{\perp,i}^{\mathrm{ricc}} = B_{\perp,i-1}^{\mathrm{ricc}} - E_1 v_i^{\mathrm{ricc}}\hat{p}_{\mathrm{ricc}}^{(i)}(l_{v,\mathrm{lyap}}^{(i)})^\top\).

\item Compute \(t_{v,\infty}^{(i)}\) by solving the Sylvester equation \eqref{tv_inf} and set \(v_i^{\infty} = V_{\mathrm{lyap}}^{(i)}t_{v,\infty}^{(i)}\).

\item Compute \(\big(\hat{p}_{\infty}^{(i)}\big)^{-1}\) by solving the Lyapunov equation \eqref{p_inf}.

\item Expand \(V_{\infty}^{(i)} = \big[ V_{\infty}^{(i-1)} \;\; v_i^{\infty} \big]\) and \(\hat{P}_{\infty}^{(i)} = \mathrm{blkdiag}\big(\hat{P}_{\infty}^{(i-1)},\hat{p}_{\infty}^{(i)}\big)\).

\item Update \(B_{\perp,i}^{\infty} = B_{\perp,i-1}^{\infty} - E_1 v_i^{\infty}\hat{p}_{\infty}^{(i)}(l_{v,\mathrm{lyap}}^{(i)})^\top\).

\item Compute \(t_{v,\mathrm{br}}^{(i)}\) by solving the Sylvester equation \eqref{tv_br} and set \(v_i^{\mathrm{br}} = V_{\mathrm{lyap}}^{(i)}t_{v,\mathrm{br}}^{(i)}\).

\item Compute \(\big(\hat{p}_{\mathrm{br}}^{(i)}\big)^{-1}\) by solving the Lyapunov equation \eqref{p_br}.

\item Expand \(V_{\mathrm{br}}^{(i)} = \big[ V_{\mathrm{br}}^{(i-1)} \;\; v_i^{\mathrm{br}} \big]\) and \(\hat{P}_{\mathrm{br}}^{(i)} = \mathrm{blkdiag}\big(\hat{P}_{\mathrm{br}}^{(i-1)},\hat{p}_{\mathrm{br}}^{(i)}\big)\).

\item Update \(B_{\perp,i}^{\mathrm{br}} = B_{\perp,i-1}^{\mathrm{br}} - E_1 v_i^{\mathrm{br}}\hat{p}_{\mathrm{br}}^{(i)}(l_{v,\mathrm{lyap}}^{(i)})^\top\).

\item \textbf{if} \(m_1 = p_1\) \textbf{then}

\item Compute \(t_{v,\mathrm{mp}}^{(i)}\) and \(t_{v,\mathrm{pr}}^{(i)}\) by solving Sylvester equations \eqref{tv_mp} and \eqref{tv_pr}.

\item Set \(v_i^{\mathrm{mp}} = V_{\mathrm{lyap}}^{(i)}t_{v,\mathrm{mp}}^{(i)}\), \(v_i^{\mathrm{pr}} = V_{\mathrm{lyap}}^{(i)}t_{v,\mathrm{pr}}^{(i)}\), and compute \(\big(\hat{p}_{\mathrm{pr}}^{(i)}\big)^{-1}\) by solving the Lyapunov equation \eqref{p_pr}.

\item Expand: \(V_{\mathrm{mp}}^{(i)} = \big[ V_{\mathrm{mp}}^{(i-1)} \;\; v_i^{\mathrm{mp}} \big]\), \(V_{\mathrm{pr}}^{(i)} = \big[ V_{\mathrm{pr}}^{(i-1)} \;\; v_i^{\mathrm{pr}} \big]\), and \(\hat{P}_{\mathrm{pr}}^{(i)} = \mathrm{blkdiag}\big(\hat{P}_{\mathrm{pr}}^{(i-1)},\hat{p}_{\mathrm{pr}}^{(i)}\big)\).

\item Update: \(B_{\perp,i}^{\mathrm{mp}} = B_{\perp,i-1}^{\mathrm{mp}} - E_1v_i^{\mathrm{mp}}(l_{v,\mathrm{lyap}}^{(i)})^\top\), \(B_{\perp,i}^{\mathrm{pr}} = B_{\perp,i-1}^{\mathrm{pr}} - E_1v_i^{\mathrm{pr}}\hat{p}_{\mathrm{pr}}^{(i)}(l_{v,\mathrm{lyap}}^{(i)})^\top\).

\item \textbf{end if}

\item\label{uadi_step26} Expand \(T_{v,\mathrm{ricc}}^{(i)}\), \(T_{v,\mathrm{mp}}^{(i)}\), \(T_{v,\infty}^{(i)}\), \(T_{v,\mathrm{pr}}^{(i)}\), and \(T_{v,\mathrm{br}}^{(i)}\) as in \eqref{Tv_ricc} and \eqref{TvTw1}–\eqref{TvTw3}.

\item\label{uadi_step27} Expand \(W_{\mathrm{lyap}}^{(i)} = \big[ W_{\mathrm{lyap}}^{(i-1)} \;\; \sqrt{-2\beta_i}\, w_i^{\mathrm{lyap}} \big]\) and update \(C_{\perp,i}^{\mathrm{lyap}} = C_{\perp,i-1}^{\mathrm{lyap}} - 2\beta_i (w_i^{\mathrm{lyap}})^\top E_2\).

\item Set \(s_{w,\mathrm{lyap}}^{(i)}\) and \(l_{w,\mathrm{lyap}}^{(i)}\) as in \eqref{sv_cf_r} and expand \(S_{w,\mathrm{lyap}}^{(i)}\), \(L_{w,\mathrm{lyap}}^{(i)}\), \(\hat{A}_{2,i}^{\mathrm{lyap}}\), and \(\hat{C}_{2,i}^{\mathrm{lyap}}\) as in \eqref{SwLw} and \eqref{Aw_lyapBw_lyap}.

\item\label{uadi_step29} Compute \(t_{w,\mathrm{ricc}}^{(i)}\) by solving the Sylvester equation \eqref{tw_ricc} and set \(w_i^{\mathrm{ricc}} = W_{\mathrm{lyap}}^{(i)}t_{w,\mathrm{ricc}}^{(i)}\).

\item Compute \(\big(\hat{q}_{\mathrm{ricc}}^{(i)}\big)^{-1}\) by solving the Lyapunov equation \eqref{q_ricc}.

\item Expand \(W_{\mathrm{radi}}^{(i)} = \big[ W_{\mathrm{radi}}^{(i-1)} \;\; w_i^{\mathrm{ricc}} \big]\) and \(\hat{Q}_{\mathrm{ricc}}^{(i)} = \mathrm{blkdiag}\big(\hat{Q}_{\mathrm{ricc}}^{(i-1)},\hat{q}_{\mathrm{ricc}}^{(i)}\big)\).

\item Update \(C_{\perp,i}^{\mathrm{ricc}} = C_{\perp,i-1}^{\mathrm{ricc}} - l_{w,\mathrm{lyap}}^{(i)}\hat{q}_{\mathrm{ricc}}^{(i)}(w_i^{\mathrm{ricc}})^\top E_2\).

\item Compute \(t_{w,\infty}^{(i)}\) by solving the Sylvester equation \eqref{tw_inf} and set \(w_i^{\infty} = W_{\mathrm{lyap}}^{(i)}t_{w,\infty}^{(i)}\).

\item Compute \(\big(\hat{q}_{\infty}^{(i)}\big)^{-1}\) by solving the Lyapunov equation \eqref{q_inf}.

\item Expand \(W_{\infty}^{(i)} = \big[ W_{\infty}^{(i-1)} \;\; w_i^{\infty} \big]\) and \(\hat{Q}_{\infty}^{(i)} = \mathrm{blkdiag}\big(\hat{Q}_{\infty}^{(i-1)},\hat{q}_{\infty}^{(i)}\big)\).

\item Update \(C_{\perp,i}^{\infty} = C_{\perp,i-1}^{\infty} - l_{w,\mathrm{lyap}}^{(i)}\hat{q}_{\infty}^{(i)}(w_i^{\infty})^\top E_2\).

\item Compute \(t_{w,\mathrm{br}}^{(i)}\) by solving the Sylvester equation \eqref{tw_br} and set \(w_i^{\mathrm{br}} = W_{\mathrm{lyap}}^{(i)}t_{w,\mathrm{br}}^{(i)}\).

\item Compute \(\big(\hat{q}_{\mathrm{br}}^{(i)}\big)^{-1}\) by solving the Lyapunov equation \eqref{q_br}.

\item Expand \(W_{\mathrm{br}}^{(i)} = \big[ W_{\mathrm{br}}^{(i-1)} \;\; w_i^{\mathrm{br}} \big]\) and \(\hat{Q}_{\mathrm{br}}^{(i)} = \mathrm{blkdiag}\big(\hat{Q}_{\mathrm{br}}^{(i-1)},\hat{q}_{\mathrm{br}}^{(i)}\big)\).

\item Update \(C_{\perp,i}^{\mathrm{br}} = C_{\perp,i-1}^{\mathrm{br}} - l_{w,\mathrm{lyap}}^{(i)}\hat{q}_{\mathrm{br}}^{(i)}(w_i^{\mathrm{br}})^\top E_2\).

\item \textbf{if} \(m_2 = p_2\) \textbf{then}

\item Compute \(t_{w,\mathrm{mp}}^{(i)}\) and \(t_{w,\mathrm{pr}}^{(i)}\) by solving Sylvester equations \eqref{tw_mp} and \eqref{tw_pr}.

\item Set \(w_i^{\mathrm{mp}} = W_{\mathrm{lyap}}^{(i)}t_{w,\mathrm{mp}}^{(i)}\), \(w_i^{\mathrm{pr}} = W_{\mathrm{lyap}}^{(i)}t_{w,\mathrm{pr}}^{(i)}\), and compute \(\big(\hat{q}_{\mathrm{pr}}^{(i)}\big)^{-1}\) by solving the Lyapunov equation \eqref{q_pr}.

\item Expand: \(W_{\mathrm{mp}}^{(i)} = \big[ W_{\mathrm{mp}}^{(i-1)} \;\; w_i^{\mathrm{mp}} \big]\), \(W_{\mathrm{pr}}^{(i)} = \big[ W_{\mathrm{pr}}^{(i-1)} \;\; w_i^{\mathrm{pr}} \big]\).

\item Update: \(C_{\perp,i}^{\mathrm{mp}} = C_{\perp,i-1}^{\mathrm{mp}} - l_{w,\mathrm{lyap}}^{(i)}(w_i^{\mathrm{mp}})^\top E_2\), \(C_{\perp,i}^{\mathrm{pr}} = C_{\perp,i-1}^{\mathrm{pr}} - l_{w,\mathrm{lyap}}^{(i)}\hat{q}_{\mathrm{pr}}^{(i)}(w_i^{\mathrm{pr}})^\top E_2\).

\item \textbf{end if}

\item Expand \(T_{w,\mathrm{ricc}}^{(i)}\), \(T_{w,\mathrm{mp}}^{(i)}\), \(T_{w,\infty}^{(i)}\), \(T_{w,\mathrm{pr}}^{(i)}\), and \(T_{w,\mathrm{br}}^{(i)}\) as in \eqref{TvTw1}–\eqref{TvTw3}.

\item \textbf{if} \(A_1 = A_2\) \(\land\) \(B_1 = B_2\) \(\land\) \(C_1 = C_2\) \(\land\) \(D_1 = D_2\) \(\land\) \(E_1 = E_2\) \textbf{then}

\item Compute \(\mathcal{T}_{v,\mathrm{sf}}^{(i)}\) and \(\mathcal{T}_{w,\mathrm{sf}}^{(i)}\) by solving the Sylvester equations \eqref{Tv_sf} and \eqref{Tw_sf}.

\item Compute \(\mathcal{V}_{\mathrm{sf}}^{(i)} = V_{\mathrm{lyap}}^{(i)}\mathcal{T}_{v,\mathrm{sf}}^{(i)}\) and \(\mathcal{W}_{\mathrm{sf}}^{(i)} = W_{\mathrm{lyap}}^{(i)}\mathcal{T}_{w,\mathrm{sf}}^{(i)}\).

\item Compute \(\hat{\mathcal{P}}_{\mathrm{sf}}\) and \(\hat{\mathcal{Q}}_{\mathrm{sf}}\) by solving the Lyapunov equations \eqref{P_sf} and \eqref{Q_sf}.

\item Update \(B_{\perp,i}^{\mathrm{sf}} = B_1(D_1^\top D_1)^{-\frac{1}{2}} - E_1\mathcal{V}_{\mathrm{sf}}^{(i)}\hat{\mathcal{P}}_{\mathrm{sf}}\big(L_{v,\mathrm{lyap}}^{(i)}\big)^\top\) and \(C_{\perp,i}^{\mathrm{sf}} = (D_2 D_2^\top)^{-\frac{1}{2}}C_2 - L_{v,\mathrm{lyap}}^{(i)}\hat{\mathcal{Q}}_{\mathrm{sf}}(\mathcal{W}_{\mathrm{sf}}^{(i)})^\top E_2\).

\item\label{uadi_step53} \textbf{end if}

\item \textbf{if} \(m_1 = p_2\) \textbf{then}

\item Set \(s_{v,\mathrm{sylv}}^{(i)}\), \(l_{v,\mathrm{sylv}}^{(i)}\), \(s_{w,\mathrm{sylv}}^{(i)}\), and \(l_{w,\mathrm{sylv}}^{(i)}\) as in \eqref{sv_case1}.

\item\label{uadi_step56} Set \(s_v^{(i)} = s_{v,\mathrm{sylv}}^{(i)}\), \(l_v^{(i)} = l_{v,\mathrm{sylv}}^{(i)}\), \(s_w^{(i)} = s_{w,\mathrm{sylv}}^{(i)}\), and \(l_w^{(i)} = l_{w,\mathrm{sylv}}^{(i)}\) in \eqref{eq_sylv} and solve the Sylvester equation to compute \(d_{\mathrm{sylv}}^{(i)}\).

\item Compute \(t_{v,\mathrm{sylv}}^{(i)}\) and \(t_{w,\mathrm{sylv}}^{(i)}\) by solving the Sylvester equations \eqref{tv_sylv} and \eqref{tw_sylv}.

\item Set \(v_i^{\mathrm{sylv}} = V_{\mathrm{lyap}}^{(i)}t_{v,\mathrm{sylv}}^{(i)}\) and \(w_i^{\mathrm{sylv}} = W_{\mathrm{lyap}}^{(i)}t_{w,\mathrm{sylv}}^{(i)}\).

\item Expand \(V_{\mathrm{fadi}}^{(i)} = \big[ V_{\mathrm{fadi}}^{(i-1)} \;\; v_i^{\mathrm{sylv}} \big]\), \(W_{\mathrm{fadi}}^{(i)} = \big[ W_{\mathrm{fadi}}^{(i-1)} \;\; w_i^{\mathrm{sylv}} \big]\), and \(D_{\mathrm{fadi}}^{(i)} = \mathrm{blkdiag}\big(D_{\mathrm{fadi}}^{(i-1)},(d_{\mathrm{sylv}}^{(i)})^{-1}\big)\).

\item Update \(B_{\perp,i}^{\mathrm{sylv}} = B_{\perp,i-1}^{\mathrm{sylv}} - E_1 v_i^{\mathrm{sylv}}(d_{\mathrm{sylv}}^{(i)})^{-1}(l_{w,\mathrm{sylv}}^{(i)})^\top\), and \(C_{\perp,i}^{\mathrm{sylv}} = C_{\perp,i-1}^{\mathrm{sylv}} - l_{v,\mathrm{sylv}}^{(i)}(d_{\mathrm{sylv}}^{(i)})^{-1}(w_i^{\mathrm{sylv}})^\top E_2\).

\item\label{uadi_step61} Expand \(T_{v,\mathrm{sylv}}^{(i)}\) and \(T_{w,\mathrm{sylv}}^{(i)}\) as in \eqref{TvTw_sylv}.

\item \textbf{end if}

\item \(i \gets i + 1\)

\item \textbf{else if} \(\alpha_i \in \mathbb{C}_{-} \land \alpha_{i+1} = \overline{\alpha_i} \land \beta_i \in \mathbb{C}_{-} \land \beta_{i+1} = \overline{\beta_i}\) \textbf{then}

\item\label{uadi_step65} Set \(\phi_v = \sqrt{-\operatorname{Re}(\alpha_i)}\) and \(\delta_v = \dfrac{\operatorname{Re}(\alpha_i)}{\operatorname{Im}(\alpha_i)}\).

\item Expand \(V_{\mathrm{lyap}}^{(i)} = \big[ V_{\mathrm{lyap}}^{(i-1)} \;\; 2\phi_v\big(\operatorname{Re}(v_i^{\mathrm{lyap}}) + \delta_v \operatorname{Im}(v_i^{\mathrm{lyap}})\big) \;\; 2\phi_v \sqrt{1 + \delta_v^2}\, \operatorname{Im}(v_i^{\mathrm{lyap}}) \big]\).

\item Update \(B_{\perp,i}^{\mathrm{lyap}} = B_{\perp,i-1}^{\mathrm{lyap}} - 4\operatorname{Re}(\alpha_i) E_1 \big(\operatorname{Re}(v_i^{\mathrm{lyap}}) + \delta_v \operatorname{Im}(v_i^{\mathrm{lyap}})\big)\).

\item Set \(s_{v,\mathrm{lyap}}^{(i)}\) and \(l_{v,\mathrm{lyap}}^{(i)}\) as in \eqref{sv_cf_i}, and expand \(S_{v,\mathrm{lyap}}^{(i)}\), \(L_{v,\mathrm{lyap}}^{(i)}\), \(\hat{A}_{1,\mathrm{lyap}}^{(i)}\), and \(\hat{C}_{1,\mathrm{lyap}}^{(i)}\) as in \eqref{SvLv} and \eqref{A_lyapB_lyap}.

\item\label{uadi_step69} \textbf{Repeat} Steps~\ref{uadi_step7} to~\ref{uadi_step26}.

\item\label{uadi_step70} Set \(\phi_w = \sqrt{-\operatorname{Re}(\beta_i)}\) and \(\delta_w = \dfrac{\operatorname{Re}(\beta_i)}{\operatorname{Im}(\beta_i)}\).

\item Expand \(W_{\mathrm{lyap}}^{(i)} = \big[ W_{\mathrm{lyap}}^{(i-1)} \;\; 2\phi_w\big(\operatorname{Re}(w_i^{\mathrm{lyap}}) + \delta_w \operatorname{Im}(w_i^{\mathrm{lyap}})\big) \;\; 2\phi_w \sqrt{1 + \delta_w^2}\, \operatorname{Im}(w_i^{\mathrm{lyap}}) \big]\).

\item Update \(C_{\perp,i}^{\mathrm{lyap}} = C_{\perp,i-1}^{\mathrm{lyap}} - 4\operatorname{Re}(\beta_i) \big(\operatorname{Re}(w_i^{\mathrm{lyap}}) + \delta_w \operatorname{Im}(w_i^{\mathrm{lyap}})\big)^\top E_2\).

\item Set \(s_{w,\mathrm{lyap}}^{(i)}\) and \(l_{w,\mathrm{lyap}}^{(i)}\) as in \eqref{sw_cf_i} and expand \(S_{w,\mathrm{lyap}}^{(i)}\), \(L_{w,\mathrm{lyap}}^{(i)}\), \(\hat{A}_{2,\mathrm{lyap}}^{(i)}\), and \(\hat{C}_{2,\mathrm{lyap}}^{(i)}\) as in \eqref{SwLw} and \eqref{Aw_lyapBw_lyap}.

\item\label{uadi_step74} \textbf{Repeat} Steps~\ref{uadi_step29} to~\ref{uadi_step53}.

\item \textbf{if} \(m_1 = p_2\) \textbf{then}

\item Set \(s_{v,\mathrm{sylv}}^{(i)}\), \(l_{v,\mathrm{sylv}}^{(i)}\), \(s_{w,\mathrm{sylv}}^{(i)}\), and \(l_{w,\mathrm{sylv}}^{(i)}\) as in \eqref{sv_case2} and \eqref{sw_case2}.

\item \textbf{Repeat} Steps~\ref{uadi_step56} to~\ref{uadi_step61}.

\item \textbf{end if}

\item \(i \gets i + 2\)

\item \textbf{else if} \(\alpha_i \in \mathbb{R}_{-} \land \alpha_{i+1} \in \mathbb{R}_{-} \land \beta_i \in \mathbb{C}_{-} \land \beta_{i+1} = \overline{\beta_i}\) \textbf{then}

\item \textbf{Repeat} Steps~\ref{uadi_step7} to~\ref{uadi_step26}.

\item \(i \gets i + 1\)

\item \textbf{Repeat} Steps~\ref{uadi_step7} to~\ref{uadi_step26}.

\item \textbf{Repeat} Steps~\ref{uadi_step70} to~\ref{uadi_step74}.

\item \textbf{if} \(m_1 = p_2\) \textbf{then}

\item Set \(s_{v,\mathrm{sylv}}^{(i)}\), \(l_{v,\mathrm{sylv}}^{(i)}\), \(s_{w,\mathrm{sylv}}^{(i)}\), and \(l_{w,\mathrm{sylv}}^{(i)}\) as in \eqref{sv_case3} and \eqref{sw_case3}.

\item \textbf{Repeat} Steps~\ref{uadi_step56} to~\ref{uadi_step61}.

\item \textbf{end if}

\item \(i \gets i + 1\)

\item \textbf{else if} \(\alpha_i \in \mathbb{C}_{-} \land \alpha_{i+1} = \overline{\alpha_i} \land \beta_i \in \mathbb{R}_{-} \land \beta_{i+1} \in \mathbb{R}_{-}\) \textbf{then}

\item \textbf{Repeat} Steps~\ref{uadi_step65} to~\ref{uadi_step69}.

\item \textbf{Repeat} Steps~\ref{uadi_step27} to~\ref{uadi_step53}.

\item \(i \gets i + 1\)

\item \textbf{Repeat} Steps~\ref{uadi_step27} to~\ref{uadi_step53}.

\item \textbf{if} \(m_1 = p_2\) \textbf{then}

\item Set \(s_{v,\mathrm{sylv}}^{(i)}\), \(l_{v,\mathrm{sylv}}^{(i)}\), \(s_{w,\mathrm{sylv}}^{(i)}\), and \(l_{w,\mathrm{sylv}}^{(i)}\) as in \eqref{sv_case4} and \eqref{sw_case4}.

\item \textbf{Repeat} Steps~\ref{uadi_step56} to~\ref{uadi_step61}.

\item \textbf{end if}

\item \(i \gets i + 1\)

\item \textbf{end if}

\item \textbf{end while}

\end{enumerate}

\noindent\rule{\textwidth}{0.8pt}
\subsection{Pole-placement Property of ADI Methods}
Let $V_1$ be computed as in \eqref{v_kry} by solving the Sylvester equations \eqref{v_kry_sylv} and \eqref{v_kry_sylv_2}, respectively. Let the interpolation points be the mirror images of the ADI shifts used to approximate $P_1$, i.e., $\{\sigma_i\}_{i=1}^{k} = \{-\alpha_i\}_{i=1}^{k}$. Furthermore, assume that the pair $(-S_v, L_v)$ is observable and that $Q_s$ solves the following Lyapunov equation:
\begin{align}
-S_v^* Q_s - Q_s S_v + L_v^\top L_v = 0.
\end{align}
It is shown in \cite{wolf2016adi} that the approximations $P_1 \approx V_{\mathrm{lyap}}^{(k)}(V_{\mathrm{lyap}}^{(k)})^\top$ and $P_1 \approx V_1 Q_s^{-1} V_1^*$ are identical. Thus, in theory, UADI can be implemented using standard rational interpolation. However, a major disadvantage of standard interpolation is that, as the number of interpolation points increases, the pair $(-S_v, L_v)$ loses observability, and consequently $Q_s$ ceases to be invertible.

Recall that in both UADI and the realified version of CF-ADI, $\hat{P}_{\mathrm{lyap}}^{(i)} = I$, meaning the pair $(-S_{v,\mathrm{lyap}}^{(i)},$ $L_{v,\mathrm{lyap}}^{(i)})$ remains observable regardless of the number of ADI shifts. Moreover, since CF-ADI is a recursive interpolation algorithm, the projected matrix $\hat{A}_{1,i}^{\mathrm{lyap}}$ can be parameterized in terms of $\hat{B}_{1,i}^{\mathrm{lyap}}$ as
\[
\hat{A}_{1,i}^{\mathrm{lyap}} = S_{v,\mathrm{lyap}}^{(i)} - \hat{B}_{1,i}^{\mathrm{lyap}} L_{v,\mathrm{lyap}}^{(i)}.
\]
From the pole-placement perspective in control theory, the free parameter $\hat{B}_{1,i}^{\mathrm{lyap}}$ acts as a gain that places the poles of $\hat{A}_{1,i}^{\mathrm{lyap}}$ at desired locations. This pole-placement property holds only if the pair $(-S_{v,\mathrm{lyap}}^{(i)}, L_{v,\mathrm{lyap}}^{(i)})$ is observable \cite{rugh}. As explained in \cite{wolf2016adi}, CF-ADI sets $\hat{B}_{1,i}^{\mathrm{lyap}} = (L_{v,\mathrm{lyap}}^{(i)})^\top$, which places the poles of $\hat{A}_{1,i}^{\mathrm{lyap}}$ at $\{\overline{\alpha}_i\}_{i=1}^{k} \subset \mathbb{C}_{-}$, ensuring that $\hat{A}_{1,i}^{\mathrm{lyap}}$ remains Hurwitz. This property allows CF-ADI to avoid a common pitfall of other projection methods—namely, that the projected Lyapunov equation may have no solution because the projected $A$-matrix is not Hurwitz. By exploiting pole placement, CF-ADI guarantees that the projected Lyapunov equation \eqref{P_proj_adi} always admits the solution $\hat{P}_{\mathrm{lyap}}^{(i)} = I$.

As shown earlier, FADI is also a projection algorithm, where $D_{\mathrm{fadi}}^{(i)}$ solves the projected Sylvester equation \eqref{D_fadi_proj}. FADI places the poles of $\hat{A}_1^{(i)}$ at $\{\beta_i\}_{i=1}^{k}$ and those of $\hat{A}_2^{(i)}$ at $\{\alpha_i\}_{i=1}^{k}$ by choosing the free parameters as $\hat{B}_1^{(i)} = D_{\mathrm{fadi}}^{(i)} (L_{w,\mathrm{sylv}}^{(i)})^\top$ and $\hat{C}_2^{(i)} = L_{v,\mathrm{sylv}}^{(i)} D_{\mathrm{fadi}}^{(i)}$. Since $\alpha_i \neq -\beta_i$ in FADI, the projected Sylvester equation \eqref{D_fadi_proj} always has a unique solution. Like CF-ADI, FADI thus leverages the pole-placement property within a recursive interpolation framework to avoid a common failure mode of other projection methods—the lack of a unique solution to the projected Sylvester equation.

Similarly, RADI sets the free parameter $\hat{B}_1^{(i)} = \hat{P}_{\mathrm{ricc}}^{(i)} (L_{v,\mathrm{lyap}}^{(i)})^\top$ to place the poles of $\hat{A}_1^{(i)}$ at 
\[
\{\overline{\alpha}_i\}_{i=1}^{k} + \lambda_i\!\left(\hat{P}_{\mathrm{ricc}}^{(i)} (\hat{C}_1^{(i)})^\top \hat{C}_1^{(i)}\right),
\]
so that $\hat{A}_1^{(i)} - \hat{P}_{\mathrm{ricc}}^{(i)} (\hat{C}_1^{(i)})^\top \hat{C}_1^{(i)}$ has eigenvalues in $\{\overline{\alpha}_i\}_{i=1}^{k} \subset \mathbb{C}_{-}$. This ensures that the projected Riccati equation \eqref{proj_ricc} always admits a stabilizing solution. Again, RADI cleverly exploits the pole-placement capability inherent in ADI-based methods.

The pole-placement properties of UADI used to solve the linear matrix equations \eqref{lyap_p}–\eqref{ricc_q_sf} are summarized in Table \ref{tab01}. The matrix triplets \((\hat{A}_1^{(i)},\hat{B}_1^{(i)},\hat{C}_1^{(i)})\) and \((\hat{A}_2^{(i)},\hat{B}_2^{(i)},\hat{C}_2^{(i)})\) in Table \ref{tab01} define the respective projected linear matrix equations that ADI methods implicitly solve. The matrices $\hat{A}_{1,\mathrm{pr}}^{(i)}$, $\hat{C}_{1,\mathrm{pr}}^{(i)}$, $\hat{A}_{2,\mathrm{pr}}^{(i)}$, $\hat{B}_{2,\mathrm{pr}}^{(i)}$, $\hat{A}_{1,\mathrm{br}}^{(i)}$, $\hat{C}_{1,\mathrm{br}}^{(i)}$, $\hat{A}_{2,\mathrm{br}}^{(i)}$, $\hat{B}_{2,\mathrm{br}}^{(i)}$, $\hat{A}_{1,\mathrm{sf}}^{(i)}$, $\hat{C}_{1,\mathrm{sf}}^{(i)}$, $\hat{A}_{2,\mathrm{sf}}^{(i)}$, and $\hat{B}_{2,\mathrm{sf}}^{(i)}$ in Table \ref{tab01} are given below:
\begin{align}
\hat{A}_{1,\mathrm{pr}}^{(i)}&=\hat{A}_1^{(i)}-\hat{B}_1^{(i)}(D_1+D_1^\top)^{-1}\hat{C}_1^{(i)},& \hat{C}_{1,\mathrm{pr}}^{(i)}&=(D_1+D_1^\top)^{-\frac{1}{2}}\hat{C}_1^{(i)},\nonumber\\
\hat{A}_{2,\mathrm{pr}}^{(i)}&=\hat{A}_2^{(i)}-\hat{B}_2^{(i)}(D_2+D_2^\top)^{-1}\hat{C}_2^{(i)},& \hat{B}_{2,\mathrm{pr}}^{(i)}&=\hat{B}_2^{(i)}(D_2+D_2^\top)^{-\frac{1}{2}},\nonumber\\
\hat{A}_{1,\mathrm{br}}^{(i)}&=\hat{A}_1^{(i)}-\hat{B}_1^{(i)}D_1^\top(I-D_1D_1^\top)^{-1}\hat{C}_1^{(i)},& \hat{C}_{1,\mathrm{br}}^{(i)}&=(I-D_1D_1^\top)^{-\frac{1}{2}}\hat{C}_1^{(i)},\nonumber\\
\hat{A}_{2,\mathrm{br}}^{(i)}&=\hat{A}_2^{(i)}-\hat{B}_2^{(i)}(I-D_2^\top D_2)^{-1}D_2^\top\hat{C}_2^{(i)},& \hat{B}_{2,\mathrm{br}}^{(i)}&=\hat{B}_2^{(i)}(I-D_2^\top D_2)^{-\frac{1}{2}},\nonumber\\
\hat{A}_{1,\mathrm{sf}}^{(i)}&=\hat{A}_1^{(i)}-\hat{B}_1^{(i)}(D_1^\top D_1)^{-\frac{1}{2}}\hat{C}_{1,\mathrm{sf}}^{(i)},& \hat{C}_{1,\mathrm{sf}}^{(i)}&=(D_1^\top D_1)^{-\frac{1}{2}}\Big(B_1^\top W_{\mathrm{lyap}}^{(i)}(W_{\mathrm{lyap}}^{(i)})^\top \mathcal{V}_{\mathrm{sf}}^{(i)}+D_1^\top\hat{C}_1^{(i)}\Big),\nonumber\\
\hat{A}_{2,\mathrm{sf}}^{(i)}&=\hat{A}_2^{(i)}-\hat{B}_{2,\mathrm{sf}}^{(i)}(D_2 D_2^\top)^{-\frac{1}{2}}\hat{C}_2^{(i)},& \hat{B}_{2,\mathrm{sf}}^{(i)}&=\Big((\mathcal{W}_{\mathrm{sf}}^{(i)})^\top V_{\mathrm{lyap}}^{(i)}(V_{\mathrm{lyap}}^{(i)})^\top C_2^\top+\hat{B}_2^{(i)}D_2^\top\Big)(D_2 D_2^\top)^{-\frac{1}{2}}.\nonumber
\end{align}
\begin{table}[!h]
\centering
\caption{Pole Placement Properties of UADI}\label{tab01}
\begin{tabular}{|c|c|c|}\hline
Linear Matrix Equation&Projected Matrix & Poles Location\\\hline
$P_1$&$\hat{A}_1^{(i)}$&$\overline{\alpha}_i$\\
$Q_2$&$\hat{A}_2^{(i)}$&$\beta_i$\\
$P_s$&$\hat{A}_1^{(i)}$&$\overline{\alpha}_i$\\
$Q_s$&$\hat{A}_2^{(i)}$&$\beta_i$\\
$P_{\mathrm{mp}}$&$\hat{A}_1^{(i)}-\hat{B}_1^{(i)}D_1^{-1}\hat{C}_1^{(i)}$&$\overline{\alpha}_i$\\
$Q_{\mathrm{mp}}$&$\hat{A}_2^{(i)}-\hat{B}_2^{(i)}D_2^{-1}\hat{C}_2^{(i)}$&$\beta_i$\\
$X_\mathrm{sylv}$&$\hat{A}_1^{(i)}$&$\beta_i$\\
$X_\mathrm{sylv}$&$\hat{A}_2^{(i)}$&$\alpha_i$\\
$P_{\mathrm{ricc}}$&$\hat{A}_1^{(i)}-\hat{P}_{\mathrm{ricc}}^{(i)}(\hat{C}_1^{(i)})^\top \hat{C}_1^{(i)}$&$\overline{\alpha}_i$\\
$Q_{\mathrm{ricc}}$&$\hat{A}_2^{(i)}-\hat{B}_2^{(i)}(\hat{B}_2^{(i)})^\top \hat{Q}_{\mathrm{ricc}}^{(i)}$&$\beta_i$\\
$P_{\infty}$&$\hat{A}_1^{(i)}-\hat{P}_{\infty}^{(i)}(\hat{C}_1^{(i)})^\top \hat{C}_1^{(i)}$&$\overline{\alpha}_i$\\
$Q_{\infty}$&$\hat{A}_2^{(i)}-\hat{B}_2^{(i)}(\hat{B}_2^{(i)})^\top \hat{Q}_{\infty}^{(i)}$&$\beta_i$\\
$P_{\mathrm{pr}}$&$\hat{A}_{1,\mathrm{pr}}^{(i)}+\hat{P}_{\mathrm{pr}}^{(i)}(\hat{C}_{1,\mathrm{pr}}^{(i)})^\top \hat{C}_{1,\mathrm{pr}}^{(i)}$&$\overline{\alpha}_i$\\
$Q_{\mathrm{pr}}$&$\hat{A}_{2,\mathrm{pr}}^{(i)}+\hat{B}_{2,\mathrm{pr}}^{(i)}(\hat{B}_{2,\mathrm{pr}}^{(i)})^\top \hat{Q}_{\mathrm{pr}}^{(i)}$&$\beta_i$\\
$P_{\mathrm{br}}$&$\hat{A}_{1,\mathrm{br}}^{(i)}+\hat{P}_{\mathrm{br}}^{(i)}(\hat{C}_{1,\mathrm{br}}^{(i)})^\top \hat{C}_{1,\mathrm{br}}^{(i)}$&$\overline{\alpha}_i$\\
$Q_{\mathrm{br}}$&$\hat{A}_{2,\mathrm{br}}^{(i)}+\hat{B}_{2,\mathrm{br}}^{(i)}(\hat{B}_{2,\mathrm{br}}^{(i)})^\top \hat{Q}_{\mathrm{br}}^{(i)}$&$\beta_i$\\
$P_{\mathrm{sf}}$&$\hat{A}_{1,\mathrm{sf}}^{(i)}+\hat{P}_{\mathrm{sf}}^{(i)}(\hat{C}_{1,\mathrm{sf}}^{(i)})^\top \hat{C}_{1,\mathrm{sf}}^{(i)}$&$\overline{\alpha}_i$\\
$Q_{\mathrm{sf}}$&$\hat{A}_{2,\mathrm{sf}}^{(i)}+\hat{B}_{2,\mathrm{sf}}^{(i)}(\hat{B}_{2,\mathrm{sf}}^{(i)})^\top \hat{Q}_{\mathrm{sf}}^{(i)}$&$\beta_i$\\\hline
\end{tabular}
\end{table}

CF-ADI, FADI, and RADI preserve the pole-placement property due to the block-triangular structures of $S_{v,\mathrm{lyap}}^{(i)}$, $S_{v,\mathrm{sylv}}^{(i)}$, $S_{w,\mathrm{sylv}}^{(i)}$, and $S_{v,\mathrm{ricc}}^{(i)}$, together with a judicious choice of off-diagonal blocks. These off-diagonal blocks eliminate the observability dependence of the new pairs $(-s_{v,\mathrm{lyap}}^{(i)}, l_{v,\mathrm{lyap}}^{(i)})$, $(-s_{v,\mathrm{sylv}}^{(i)}, l_{v,\mathrm{sylv}}^{(i)})$, $(-s_{w,\mathrm{sylv}}^{(i)}, l_{w,\mathrm{sylv}}^{(i)})$, and $(-s_{v,\mathrm{ricc}}^{(i)}, l_{v,\mathrm{ricc}}^{(i)})$ on the previous blocks $(-S_{v,\mathrm{lyap}}^{(i-1)}, L_{v,\mathrm{lyap}}^{(i-1)})$, $(-S_{v,\mathrm{sylv}}^{(i-1)}, L_{v,\mathrm{sylv}}^{(i-1)})$, $(-S_{w,\mathrm{sylv}}^{(i-1)}$, $L_{w,\mathrm{sylv}}^{(i-1)})$, and $(-S_{v,\mathrm{ricc}}^{(i-1)}, L_{v,\mathrm{ricc}}^{(i-1)})$. Consequently, as long as each new pair is observable, the augmented pairs remain observable. In contrast, the diagonal structure of $S_v$ in standard interpolation is less robust: as the number of interpolation points grows, observability of $(-S_v, L_v)$ is typically lost, causing the pole-placement property to vanish and severely limiting applicability.

Beyond preserving observability, the off-diagonal blocks in $S_{v,\mathrm{lyap}}^{(i)}$, $S_{v,\mathrm{sylv}}^{(i)}$, $S_{w,\mathrm{sylv}}^{(i)}$, and $S_{v,\mathrm{ricc}}^{(i)}$ serve another purpose: they ensure that $\hat{P}_{\mathrm{lyap}}^{(i)}$, $D_{\mathrm{fadi}}^{(i)}$, and $\hat{P}_{\mathrm{ricc}}^{(i)}$ remain block-diagonal. This is a direct consequence of decoupling the new pairs from the previous ones, as described above. As a result, these quantities can be updated recursively—block by block—rather than requiring a full recomputation of the projected linear matrix equations at every iteration.

In \cite{zulfiqar2025data}, a workaround is proposed to enable standard interpolation to mimic the properties of CF-ADI and RADI. It is shown that when the ADI shifts are lightly damped, i.e.,
\begin{align}
\alpha_i = \left( \frac{\zeta_{i,\alpha} \omega_{i,\alpha}}{\sqrt{1 - \zeta_{i,\alpha}^2}} \right) + j \omega_{i,\alpha}, \quad
\beta_i = \left( \frac{\zeta_{i,\beta} \omega_{i,\beta}}{\sqrt{1 - \zeta_{i,\beta}^2}} \right) + j \omega_{i,\beta},
\label{adi_shifts}
\end{align}
with $\zeta_{i,\alpha} \ll 1$ and $\zeta_{i,\beta} \ll 1$, the pairs $(-S_v, L_v)$ and $(-S_w, L_w)$ remain observable. Moreover, the projected linear matrix equations become approximately block-diagonal, with off-diagonal blocks tending to zero as $\zeta_{i,\alpha}, \zeta_{i,\beta} \to 0$. This makes recursive, block-wise accumulation feasible. A unified ADI framework is then proposed, which can be viewed as a standard-interpolation-based approximation of UADI. While this shift selection is reasonable in the context of data-driven model order reduction—the focus of \cite{zulfiqar2025data}—it is extremely restrictive in most other settings. In contrast, UADI adopts a recursive interpolation framework and does not require such restrictive assumptions on the ADI shifts. Nevertheless, the results in \cite{zulfiqar2025data} motivated the research presented in this paper.
\subsection{UADI as a MOR Framework}
UADI can be used to implement low-rank versions of several BT algorithms, such as standard BT \cite{moore2003principal}, self-weighted BT \cite{zhou1995frequency}, LQG BT \cite{jonckheere2003new}, $\mathcal{H}_{\infty}$ BT \cite{mustafa2002controller}, positive-real BT \cite{phillips2002guaranteed}, bounded-real BT \cite{phillips2002guaranteed}, and stochastic BT \cite{green1988balanced}. However, the focus of this subsection is to show that UADI can also be viewed directly as a standalone MOR framework. UADI implicitly performs MOR and generates ROMs with important properties—such as guaranteed stability, minimum phase behavior, positive realness, and bounded realness.  

The results in this subsection are straightforward (yet important) generalizations of those in \cite{zulfiqar2025data}, which assume that the pairs $(-S_v,L_v)$ and $(-S_w,L_w)$ are observable. In UADI, the pairs $(-S_{v,\mathrm{lyap}}^{(i)},L_{v,\mathrm{lyap}}^{(i)})$ and $(-S_{w,\mathrm{lyap}}^{(i)},L_{w,\mathrm{lyap}}^{(i)})$ are always observable. Therefore, the results from \cite{zulfiqar2025data} for the standard interpolation-based unified ADI framework also apply to UADI.

Let $\hat{G}_1^{(i)}(s)=\hat{C}_1^{(i)}(sI-\hat{A}_1^{(i)})^{-1}\hat{B}_1^{(i)}+D_1$ be a ROM produced by UADI as follows:
\[
\hat{A}_1^{(i)}=S_{v,\mathrm{lyap}}^{(i)}-\hat{B}_1^{(i)}L_{v,\mathrm{lyap}}^{(i)}=-\big(S_{v,\mathrm{lyap}}^{(i)}\big)^\top,\quad \hat{B}_1^{(i)}=\big(L_{v,\mathrm{lyap}}^{(i)}\big)^\top, \quad \hat{C}_1^{(i)}=C_1V_{\mathrm{lyap}}^{(i)}.
\]
If $V_{\mathrm{lyap}}^{(i)}$ has full column rank, then $\hat{G}_1^{(i)}(s)$ interpolates $G_1(s)$ at the points $(-\alpha_1,\dots,-\alpha_i)$, regardless of the choice of the free parameter $\hat{B}_1^{(i)}$.

If the free parameter $\hat{B}_1^{(i)}$ is set to 
\[
\hat{B}_1^{(i)}=\big(\tilde{X}_{\mathrm{sylv}}^{(i)}\big)^{-1}\big(L_{w,\mathrm{lyap}}^{(i)}\big)^\top
=T_{v,\mathrm{sylv}}^{(i)}D_{\mathrm{fadi}}^{(i)}\big(T_{w,\mathrm{sylv}}^{(i)}\big)^\top\big(L_{w,\mathrm{lyap}}^{(i)}\big)^\top
\]
according to Proposition \ref{prop_x_sylv}, then $\hat{A}_1^{(i)}$ becomes
\[
\hat{A}_1^{(i)}=-\big(\tilde{X}_{\mathrm{sylv}}^{(i)}\big)^{-1}\big(S_{w,\mathrm{lyap}}^{(i)}\big)^\top \tilde{X}_{\mathrm{sylv}}^{(i)}.
\]
Thus, the poles of $\hat{A}_1^{(i)}$ are placed at $\{\overline{\beta}_l\}_{l=1}^{i}$. Note that $\tilde{X}_{\mathrm{sylv}}^{(i)}$ and $L_{w,\mathrm{lyap}}^{(i)}$ depend only on the ADI shifts $(\alpha_1,\dots,\alpha_i)$ and $(\beta_1,\dots,\beta_i)$. Consequently, UADI can construct an interpolant of $G_1(s)$ with prescribed poles $\overline{\beta}_i\in\mathbb{C}_{-}$. If these prescribed poles are a subset of the poles of $G_1(s)$, UADI preserves them in the ROM $\hat{G}_1^{(i)}(s)$. Pole-preserving MOR is particularly useful in power systems and structural dynamics \cite{scarciotti2016low,gawronski2004dynamics,werner2021structure}, and UADI is well-suited for such applications.

By choosing $\hat{B}_1^{(i)} = T_{v,\mathrm{ricc}}^{(i)}\hat{P}_{\mathrm{ricc}}^{(i)}\big(L_{v,\mathrm{lyap}}^{(i)}\big)^\top$, the reduced-order Kalman filter for state estimation becomes
\[
\dot{\hat{x}}(t)=\hat{A}_1^{(i)}\hat{x}(t)+\hat{B}_1^{(i)}u(t)+\tilde{P}_{\mathrm{ricc}}^{(i)}\big(\hat{C}_1^{(i)}\big)^\top\Big(y(t)-\hat{C}_1^{(i)}\hat{x}(t)-D_1u(t)\Big),
\]
where $u(t)$ and $y(t)$ are the input and output of $G_1(s)$, respectively, and $\hat{x}(t)$ estimates the most controllable states of the realization $(A_1,B_1,C_1,D_1,E_1)$ \cite{rugh}. Thus, UADI can also be used for reduced-order observer design in large-scale systems.

Similarly, setting $\hat{B}_1^{(i)} = T_{v,\infty}^{(i)}\hat{P}_{\infty}^{(i)}\big(L_{v,\mathrm{lyap}}^{(i)}\big)^\top$ yields the reduced-order $\mathcal{H}_\infty$ filter:
\[
\dot{\hat{x}}(t)=\hat{A}_1^{(i)}\hat{x}(t)+\hat{B}_1^{(i)}u(t)+T_{v,\infty}^{(i)}\hat{P}_{\infty}^{(i)}\big(T_{v,\infty}^{(i)}\big)^\top\big(\hat{C}_1^{(i)}\big)^\top\Big(y(t)-\hat{C}_1^{(i)}\hat{x}(t)-D_1u(t)\Big);
\]
cf. \cite{rugh}.

Assume now that $G_1(s)$ is a stable minimum-phase transfer function. Then its minimum-phase property can be preserved in $\hat{G}_1^{(i)}(s)$ by setting $\hat{B}_1^{(i)} = T_{v,\mathrm{mp}}^{(i)}\big(L_{v,\mathrm{lyap}}^{(i)}\big)^\top D_1$; see Theorem 2.7 of \cite{zulfiqar2025data}. This choice places the poles of $\hat{A}_1^{(i)} - \hat{B}_1^{(i)}D_1^{-1}\hat{C}_1^{(i)}$—which are the zeros of $\hat{G}_1^{(i)}(s)$—at $\{\overline{\alpha}_l\}_{l=1}^{i}$. Moreover, this ROM ensures a small relative error $\big(G_1(s)-\hat{G}_1^{(i)}(s)\big)\big(\hat{G}_1^{(i)}(s)\big)^{-1}$ \cite{zhou1995frequency}. Hence, UADI enables relative-error MOR while preserving the minimum-phase property of $G_1(s)$.

If $G_1(s)$ is positive-real, this property can be preserved in $\hat{G}_1^{(i)}(s)$ by choosing  
\[
\hat{B}_1^{(i)} = T_{v,\mathrm{pr}}^{(i)}\hat{P}_{\mathrm{pr}}^{(i)}\big(L_{v,\mathrm{lyap}}^{(i)}\big)^\top\big(D_1+D_1^\top\big)^{\frac{1}{2}},
\]
as stated in Theorem 2.3 of \cite{zulfiqar2025data}. Preserving positive realness is essential for passive interconnected circuits \cite{phillips2002guaranteed}, and UADI is effective for this purpose.

Similarly, if $G_1(s)$ is bounded-real, setting  
\[
\hat{B}_1^{(i)} = T_{v,\mathrm{br}}^{(i)}\hat{P}_{\mathrm{br}}^{(i)}\big(L_{v,\mathrm{lyap}}^{(i)}\big)^\top\big(I + D_1^\top(I - D_1D_1^\top)^{-1}D_1\big)^{-\frac{1}{2}}
\]
preserves bounded realness in the ROM (see Theorem 2.5 of \cite{zulfiqar2025data}), which is again critical in passive circuit modeling \cite{phillips2002guaranteed}.

Dually, let $\hat{G}_2^{(i)}(s)=\hat{C}_2^{(i)}(sI-\hat{A}_2^{(i)})^{-1}\hat{B}_2^{(i)}+D_2$ be a ROM produced by UADI as follows:
\[
\hat{A}_2^{(i)}=\big(S_{w,\mathrm{lyap}}^{(i)}\big)^\top-\big(L_{w,\mathrm{lyap}}^{(i)}\big)^\top\hat{C}_2^{(i)}=-S_{w,\mathrm{lyap}}^{(i)},\quad 
\hat{B}_2^{(i)}=\big(W_{\mathrm{lyap}}^{(i)}\big)^\top B_2, \quad 
\hat{C}_2^{(i)}=L_{w,\mathrm{lyap}}^{(i)}.
\]
If $W_{\mathrm{lyap}}^{(i)}$ has full column rank, then $\hat{G}_2^{(i)}(s)$ interpolates $G_2(s)$ at the points $(-\beta_1,\dots,-\beta_i)$, irrespective of the choice of the free parameter $\hat{C}_2^{(i)}$.

Choosing $\hat{C}_2^{(i)} = L_{v,\mathrm{lyap}}^{(i)}\big(\tilde{X}_{\mathrm{sylv}}^{(i)}\big)^{-1}$ (per Proposition \ref{prop_x_sylv}) yields
\[
\hat{A}_2^{(i)} = -\tilde{X}_{\mathrm{sylv}}^{(i)}S_{v,\mathrm{lyap}}^{(i)}\big(\tilde{X}_{\mathrm{sylv}}^{(i)}\big)^{-1},
\]
so the poles of $\hat{A}_2^{(i)}$ are placed at $\{\alpha_l\}_{l=1}^{i}$. Thus, UADI can construct an interpolant of $G_2(s)$ with prescribed poles $\alpha_i \in \mathbb{C}_{-}$.

Setting $\hat{C}_2^{(i)} = L_{w,\mathrm{lyap}}^{(i)}\hat{Q}_{\mathrm{ricc}}^{(i)}\big(T_{w,\mathrm{ricc}}^{(i)}\big)^\top$ allows computation of the reduced-order LQR gain as  
\[
K = \big(\hat{B}_2^{(i)}\big)^\top T_{w,\mathrm{ricc}}^{(i)}\hat{Q}_{\mathrm{ricc}}^{(i)}\big(T_{w,\mathrm{ricc}}^{(i)}\big)^\top
\]
showing that UADI can also be used for reduced-order controller design \cite{rugh}.

Similarly, with $\hat{C}_2^{(i)} = L_{w,\mathrm{lyap}}^{(i)}\hat{Q}_{\infty}^{(i)}\big(T_{w,\infty}^{(i)}\big)^\top$, the reduced-order $\mathcal{H}_\infty$ regulator gain is  
\[
K = \big(\hat{B}_2^{(i)}\big)^\top T_{w,\infty}^{(i)}\hat{Q}_{\infty}^{(i)}\big(T_{w,\infty}^{(i)}\big)^\top.
\]

Assume $G_2(s)$ is stable and minimum-phase. Then setting $\hat{C}_2^{(i)} = D_2L_{w,\mathrm{lyap}}^{(i)}\big(T_{w,\mathrm{mp}}^{(i)}\big)^\top$ preserves the minimum-phase property (Theorem 2.7 of \cite{zulfiqar2025data}) and places the zeros of $\hat{G}_2^{(i)}(s)$ at $\{\beta_l\}_{l=1}^{i}$.

For a positive-real $G_2(s)$, choosing  
\[
\hat{C}_2^{(i)} = \big(D_2 + D_2^\top\big)^{\frac{1}{2}}L_{w,\mathrm{lyap}}^{(i)}\hat{Q}_{\mathrm{pr}}^{(i)}\big(T_{w,\mathrm{pr}}^{(i)}\big)^\top
\]
preserves positive realness (Theorem 2.4 of \cite{zulfiqar2025data}).

For a bounded-real $G_2(s)$, setting  
\[
\hat{C}_2^{(i)} = \big(I + D_2(I - D_2^\top D_2)^{-1}D_2^\top\big)^{-\frac{1}{2}}L_{w,\mathrm{lyap}}^{(i)}\hat{Q}_{\mathrm{br}}^{(i)}\big(T_{w,\mathrm{br}}^{(i)}\big)^\top
\]
preserves bounded realness (Theorem 2.6 of \cite{zulfiqar2025data}).

All the ROMs mentioned above can be constructed non-intrusively, without knowledge of the state-space realizations \((A_1,B_1,C_1,D_1,E_1)\) and \((A_2,B_2,C_2,D_2,E_2)\), using the signal generator–based frameworks proposed in \cite{scarciotti2017data} and \cite{mao2022data}. A signal generator built from the pair \((S_{v,\mathrm{lyap}}^{(i)},L_{v,\mathrm{lyap}}^{(i)})\), which contains only the ADI shifts \(\alpha_i\), can be used to generate an input signal \(u_1(t)\). The corresponding steady-state response \(y_1(t)\) can then be measured and used—via the method in \cite{scarciotti2017data}, without modification—to estimate \(C_1V_{\mathrm{lyap}}^{(i)}\). Similarly, the matrix \(D_1\) can be approximated using a signal generator based on the pair \((\sigma I, I)\), where \(\sigma\) is a large positive real number. Dually, the signal generator setup proposed in \cite{mao2022data}, which uses the pair \((S_{w,\mathrm{lyap}}^{(i)},L_{w,\mathrm{lyap}}^{(i)})\) containing only the ADI shifts \(\beta_i\), can be applied without modification to approximate \(\big(W_{\mathrm{lyap}}^{(i)}\big)^\top B_2\) from time-domain data. Likewise, \(D_2\) can be approximated using a signal generator with the pair \((\sigma I, I)\), again with \(\sigma\) a large positive real number. Once the quantities \(C_1V_{\mathrm{lyap}}^{(i)}\), \(D_1\), \(\big(W_{\mathrm{lyap}}^{(i)}\big)^\top B_2\), and \(D_2\) are obtained from time-domain measurements, all other quantities needed to construct the ROMs in this subsection can be computed using only the ADI shifts \(\alpha_i\) and \(\beta_i\). Thus, a data-driven implementation of UADI is readily achieved within the signal generator–based framework of \cite{scarciotti2024interconnection}. Furthermore, by employing the signal generator framework in \cite{mao2024data}, UADI-based, data-driven, non-intrusive implementations of BT, self-weighted BT \cite{zhou1995frequency}, LQG BT \cite{jonckheere2003new}, \(\mathcal{H}_{\infty}\) BT \cite{mustafa2002controller}, positive-real BT \cite{phillips2002guaranteed}, bounded-real BT \cite{phillips2002guaranteed}, and stochastic BT \cite{green1988balanced} can also be obtained using the pairs \((S_{v,\mathrm{lyap}}^{(i)},L_{v,\mathrm{lyap}}^{(i)})\) and \((S_{w,\mathrm{lyap}}^{(i)},L_{w,\mathrm{lyap}}^{(i)})\) to construct signal generators; see Section 3.3 of \cite{zulfiqar2025compression} for details.

It is important to note that UADI is built on the CF-ADI algorithm. Thus, CF-ADI can in fact solve a wide range of problems—including Lyapunov, Sylvester, and Riccati equations—as well as various MOR tasks. However, these capabilities are not widely recognized, leading to underutilization of CF-ADI in practice. The key reason CF-ADI is so versatile is that the pair $(-S_{v,\mathrm{lyap}}^{(i)},L_{v,\mathrm{lyap}}^{(i)})$ remains observable at every iteration, allowing the free parameter $\hat{B}_1^{(i)}$ to be chosen freely to place poles as needed. An interesting future direction is to approach additional problems as pole-placement problems and solve them using CF-ADI. Moreover, CF-ADI (and by extension, UADI) can serve as a rational interpolation framework for other large-scale problems. For example, the matrix exponential product $e^{E_1^{-1}A_1t}E_1^{-1}B_1$ and the matrix logarithm product $\mathrm{ln}(E_1^{-1}A_1)E_1^{-1}B_1$ can be approximated as  
\[
e^{E_1^{-1}A_1t}E_1^{-1}B_1 \approx V_{\mathrm{lyap}}^{(i)}e^{-(S_{v,\mathrm{lyap}}^{(i)})^\top t}\big(L_{v,\mathrm{lyap}}^{(i)}\big)^\top,
\]
\[
\mathrm{ln}(E_1^{-1}A_1)E_1^{-1}B_1 \approx V_{\mathrm{lyap}}^{(i)}\mathrm{ln}\big(-(S_{v,\mathrm{lyap}}^{(i)})^\top\big)\big(L_{v,\mathrm{lyap}}^{(i)}\big)^\top.
\]
A key advantage is that $-\big(S_{v,\mathrm{lyap}}^{(i)}\big)^\top$ is always Hurwitz, which ensures better numerical stability for direct small-scale computations of $e^{-(S_{v,\mathrm{lyap}}^{(i)})^\top t}$ and $\mathrm{ln}\big(-(S_{v,\mathrm{lyap}}^{(i)})^\top\big)$ \cite{al2010new,fasi2018multiprecision}.
\subsection{Automatic Shift Generation}
There are three self-generating shift strategies proposed in \cite{benner2014self}. ``Self-generating shifts'' means that ADI shifts are not precomputed; instead, the ADI method generates the next shifts automatically during subsequent iterations once it has started. Among these strategies, two are computationally efficient; however, no theoretical justification was provided in \cite{benner2014self}. In this subsection, we provide a theoretical justification for these strategies from the perspective of rational interpolation. Furthermore, we revisit a shift generation strategy that was briefly mentioned but dismissed in \cite{benner2014self} due to concerns about increasing computational cost and other issues in shift selection. We note that the issues highlighted in \cite{benner2014self} can be easily addressed.

Define \( H_1(s) \) and \( H_{\perp,i}^{\mathrm{lyap}}(s) \) as follows:
\begin{align}
H_1(s)=C_1(sE_1-A_1)^{-1}B_1\quad \text{and} \quad H_{\perp,i}^{\mathrm{lyap}}(s)=C_1(sE_1-A_1)^{-1}B_{\perp,i}^{\mathrm{lyap}}.
\end{align}
Let the ROMs \( \hat{H}_1^{(i)}(s) \), \( \hat{H}_{\perp,i}^{\mathrm{lyap}}(s) \), and \( \tilde{H}_{\perp,i-1}^{\mathrm{lyap}}(s) \) be obtained by projecting \( H_1(s) \), \( H_{\perp,i}^{\mathrm{lyap}}(s) \), and \( H_{\perp,i-1}^{\mathrm{lyap}}(s) \) using \( V_1 = V_{\mathrm{lyap}}^{(i)} \) as follows:
\begin{align}
\hat{H}_1^{(i)}(s)&=C_1V_1\Big(sV_1^\top E_1V_1 - V_1^\top A_1V_1\Big)^{-1}V_1^\top B_1,\\
\hat{H}_{\perp,i}^{\mathrm{lyap}}(s)&=C_1V_1\Big(sV_1^\top E_1V_1 - V_1^\top A_1V_1\Big)^{-1}V_1^\top B_{\perp,i}^{\mathrm{lyap}},\\
\tilde{H}_{\perp,i-1}^{\mathrm{lyap}}(s)&=C_1V_1\Big(sV_1^\top E_1V_1 - V_1^\top A_1V_1\Big)^{-1}V_1^\top B_{\perp,i-1}^{\mathrm{lyap}}.
\end{align}
Assuming that \( V_{\mathrm{lyap}}^{(i)} \) has full column rank, the Sylvester equations \eqref{v_lyap}, \eqref{cfadi_v_sylv1}, and \eqref{cfadi_v_sylv2} imply the following interpolatory properties from a rational interpolation viewpoint:
\begin{enumerate}
  \item \( \hat{H}_1^{(i)}(s) \) interpolates \( H_1(s) \) at the points \( (-\alpha_1,\dots,-\alpha_i) \).
  \item \( \hat{H}_{\perp,i}^{\mathrm{lyap}}(s) \) interpolates \( H_{\perp,i}^{\mathrm{lyap}}(s) \) at the points \( (\alpha_1,\dots,\alpha_i) \).
  \item \( \tilde{H}_{\perp,i-1}^{\mathrm{lyap}}(s) \) interpolates \( H_{\perp,i-1}^{\mathrm{lyap}}(s) \) at the points \( (\alpha_1,\dots,\alpha_{i-1}, -\alpha_i) \).
\end{enumerate}
Moreover, note that \( H_1(s) \), \( H_{\perp,i}^{\mathrm{lyap}}(s) \), and \( H_{\perp,i-1}^{\mathrm{lyap}}(s) \) share the same poles. Similarly, their reduced counterparts \( \hat{H}_1^{(i)}(s) \), \( \hat{H}_{\perp,i}^{\mathrm{lyap}}(s) \), and \( \tilde{H}_{\perp,i-1}^{\mathrm{lyap}}(s) \) also share the same poles. The transfer functions \( H_{\perp,i}^{\mathrm{lyap}}(s) \) and \( H_{\perp,i-1}^{\mathrm{lyap}}(s) \) differ from \( H_1(s) \) only in their \( B \)-matrices. Although this was not recognized in \cite{wolf2016adi}, replacing \( B_1 \) with \( B_{\perp,i}^{\mathrm{lyap}} \) or \( B_{\perp,i-1}^{\mathrm{lyap}} \) has a rich history in the literature on large-scale eigenvalue problems \cite{saad2011numerical,rommes2008modal}, where it is known as “deflation.” Deflation effectively flattens the peaks in the frequency-domain plot of \( H_1(s) \) that correspond to dominant poles already captured by the ROM \( \hat{H}_1^{(i)}(s) \). Consequently, \( H_{\perp,i}^{\mathrm{lyap}}(s) \) retains only the peaks of \( H_1(s) \) that have not yet been captured by \( \hat{H}_1^{(i)}(s) \). In other words, the dominant poles of \( H_1(s) \) captured by \( \hat{H}_1^{(i)}(s) \) become poorly controllable in \( H_{\perp,i}^{\mathrm{lyap}}(s) \); see Theorem 3.3.1 of \cite{rommes2007methods}. Dominant pole approximation methods \cite{rommes2008modal} use deflation to avoid repeatedly capturing the same pole, since deflation ensures that an already captured pole is no longer dominant. These methods typically apply deflation only after a dominant pole has been captured. In contrast, CF-ADI applies deflation in every iteration after the first. Nevertheless, conceptually, replacing \( B_1 \) with \( B_{\perp,i-1}^{\mathrm{lyap}} \) in CF-ADI serves the same purpose as in dominant pole approximation methods \cite{rommes2008modal}.
\subsection{Existing Self-Generating Shift Strategies}
In \cite{benner2014self}, one self-generating shift strategy—proposed without theoretical justification—involves projecting the matrices \( A_1 \) and \( E_1 \) using \( V_1 = \mathrm{orth}\big(B_{\perp,i-1}^{\mathrm{lyap}}\big) \) and selecting the Ritz values with respect to \( A_1 \) and \( E_1 \) as the next ADI shifts \( \alpha_i \). We refer to this approach as the ``Projection-I'' shift generation strategy throughout the remainder of this paper. If \( E_1 = I \), projecting \( H_{\perp,i-1}^{\mathrm{lyap}}(s) \) via \( V_1 = \mathrm{orth}\big(B_{\perp,i-1}^{\mathrm{lyap}}\big) \) to obtain \( \tilde{H}_{\perp,i-1}^{\mathrm{lyap}}(s) \) enforces interpolation at \( s \to \infty \) \cite{wolfthesis}. Thus, CF-ADI with the Projection-I strategy can be interpreted as a pole approximation algorithm for \( H_1(s) \) (since \( H_{\perp,i-1}^{\mathrm{lyap}}(s) \) is a deflated version of \( H_1(s) \) with the same poles), but it uses a restrictive choice of interpolation point. Note that dominant poles are a good choice as ADI shifts; see Section 4.3.3 of \cite{saak2009efficient}. Moreover, most pole approximation algorithms are interpolatory in nature and interpolate at Ritz values until convergence; see Table 7.1 of \cite{mengi2022large}. When \( E_1 \neq I \), however, the Projection-I strategy loses the interpolation property at \( s \to \infty \), making its theoretical foundation questionable. Even when the interpolatory property holds, repeatedly using the same interpolation point is overly restrictive. It is therefore not surprising that its performance is not consistent—even in the numerical results reported in \cite{benner2014self}.

Another self-generating shift strategy proposed in \cite{benner2014self}—also without theoretical justification—involves projecting \( A_1 \) and \( E_1 \) using \( V_1 = \mathrm{orth}\big(v_i^{\mathrm{lyap}}\big) \) and taking the Ritz values with respect to \( A_1 \) and \( E_1 \) as the next ADI shifts \( \alpha_i \). We refer to this as the ``Projection-II'' shift generation strategy. Projecting \( H_1^{(i)}(s) \), \( H_{\perp,i}^{\mathrm{lyap}}(s) \), and \( H_{\perp,i-1}^{\mathrm{lyap}}(s) \) via \( V_1 = \mathrm{orth}\big(v_i^{\mathrm{lyap}}\big) \) to obtain \( \hat{H}_1^{(i)}(s) \), \( \hat{H}_{\perp,i}^{\mathrm{lyap}}(s) \), and \( \tilde{H}_{\perp,i-1}^{\mathrm{lyap}}(s) \) enforces interpolation at \( -\alpha_i \), \( \alpha_i \), and \( -\alpha_i \), respectively. This interpolatory behavior closely resembles that of the deflation-based dominant pole approximation (DPA) algorithm \cite{rommes2008modal,martins1996computing}. Hence, CF-ADI with the Projection-II strategy can be viewed as a variant of DPA (with deflation) applied to \( H_1(s) \), since both \( H_{\perp,i-1}^{\mathrm{lyap}}(s) \) and \( H_{\perp,i}^{\mathrm{lyap}}(s) \) are deflated versions of \( H_1(s) \) with identical poles. Thus, the Projection-II strategy is theoretically sound, which explains its strong numerical performance in \cite{benner2014self}.

A significant limitation of the Shift‑I and Shift‑II strategies is that, when \( m_1 = 1 \), they keep producing real‑valued shifts. This means that if the dominant poles are complex‑valued, they cannot be captured by these strategies. The proposed self‑generating shift, explained in the sequel, does not face this problem.
\subsection{Proposed Self-generating Shift Strategy}
In \cite{benner2014self}, the choice of setting the projection matrix \( V_1 = \mathrm{orth}\big(V_{\mathrm{lyap}}^{(i)}\big) \) for projecting \( A_1 \) and \( E_1 \) was dismissed for the following reasons: as the number of columns of \( V_{\mathrm{lyap}}^{(i)} \) grows,
\begin{enumerate}
  \item The eigenvalue decomposition of \( (V_1^\top E_1 V_1)^{-1} V_1^\top A_1 V_1 \) to compute Ritz values becomes computationally demanding.
  \item Orthogonalizing \( V_{\mathrm{lyap}}^{(i)} \) becomes computationally demanding.
  \item Many Ritz values are produced, with no clear criterion for selecting the most relevant ones.
\end{enumerate}
At the time, the interpolatory and deflation properties of ADI methods were not recognized. Consequently, the choice \( V_1 = \mathrm{orth}\big(V_{\mathrm{lyap}}^{(i)}\big) \) was dismissed, despite the observed steep decline in residuals reported in \cite{saak2009efficient,benner2010galerkin}.

In the large-scale eigenvalue problem literature, using the history of previous shifted linear solves together with the current one is known as ``subspace acceleration'', which significantly accelerates convergence to dominant poles. This idea forms the core innovation of the Subspace Accelerated DPA (SADPA) algorithm proposed in \cite{rommes2006efficient}. To keep the number of columns of the projection matrix within an allowable limit, SADPA employs ``implicit restart''—a technique widely used in large-scale eigenvalue computations \cite{saad2011numerical}. In our context, this means discarding the accumulated history of shifted linear solves $(v_1^{\mathrm{lyap}},\cdots,v_{i-1}^{\mathrm{lyap}})$ once \( V_1 \) exceeds the maximum allowable size and beginning to collect a new history from the current shifted linear solve $v_i^{\mathrm{lyap}}$ onward. Pole selection then becomes straightforward: the most dominant pole of the deflated transfer function \( \hat{H}_{\perp,i}^{\mathrm{lyap}}(s) \) is chosen as the next shift, as done in both SADPA and the algorithm in \cite{mengi2022large}. Thus, the lack of clarity noted in \cite{benner2014self} regarding Ritz value selection is resolved by the results presented in this subsection.

In the sequel, we propose a self-generating shift strategy that leverages the interpolatory properties of CF-ADI to mimic those of SADPA identified in \cite{mengi2022large}, as well as the shift generation approach from \cite{rommes2006efficient}. The resulting CF-ADI method—equipped with this subspace-accelerated shift strategy and implicit restart—can be viewed as a variant of SADPA.
\subsubsection{For Lyapunov Equations}\label{3.9.1}
Let us project the deflated transfer function \( H_{\perp,i}^{\mathrm{lyap}}(s) \) using  
\( V_1 = \mathrm{orth}\big(\begin{bmatrix}v_1^{\mathrm{lyap}} & \cdots & v_i^{\mathrm{lyap}}\end{bmatrix}\big) \), together with an implicit restart scheme to keep the number of columns of \( V_1 \) within an allowable limit, yielding:
\begin{align}
\hat{E}_1 = V_1^\top E_1 V_1,\quad \hat{A}_1 = V_1^\top A_1 V_1,\quad \hat{B}_{\perp,i}^{\mathrm{lyap}} = V_1^\top B_{\perp,i}^{\mathrm{lyap}}.
\end{align}

It was mentioned earlier in this subsection that this ROM interpolates the deflated transfer function \( H_{\perp,i}^{\mathrm{lyap}}(s) \) at the interpolation points \( (\alpha_1, \dots, \alpha_i) \), where the ADI shifts \( \alpha_i \) are approximate dominant poles of \( H_{\perp,i}^{\mathrm{lyap}}(s) \). As shown in \cite{mengi2022large}, SADPA similarly interpolates its corresponding deflated transfer function at its approximate dominant poles; see Table 7.1 of \cite{mengi2022large}.

Now, let us compute the eigenvalue decomposition of \( \hat{E}_1^{-1}\hat{A}_1 \) as  
\( \hat{E}_1^{-1}\hat{A}_1 = \hat{T}\,\mathrm{diag}\big(\hat{\lambda}_1,\dots,\hat{\lambda}_k\big)\hat{T}^{-1} \).  
Define \( \hat{r}_{b,l} = \hat{T}^{-1}(l,:)\,\hat{B}_{\perp,i}^{\mathrm{lyap}} \).  
The most controllable pole of \( \hat{E}_1^{-1}\hat{A}_1 \) is the pole \( \hat{\lambda}_l \) associated with the largest value of  
\[
\hat{\phi}_l = \frac{\|\hat{r}_{b,l}\|_2^2}{|\mathrm{Re}(\hat{\lambda}_l)|},
\]
as in \cite{rommes2007methods,mengi2022large}.  
We then sort the columns of \( \hat{T} \) and the eigenvalues \( \hat{\lambda}_l \) in descending order of \( \hat{\phi}_l \).  
Consequently, the pole \( \hat{\lambda}_1 \) with the largest \( \hat{\phi}_l \) can be used as the next ADI shift.

Dually, the ADI shifts \( \beta_i \) can be chosen as the most observable approximate poles of the deflated transfer function \( C_{\perp,i}^{\mathrm{lyap}}(sE_2 - A_2)^{-1}B_2 \).  
Let us set the projection matrix \( W_2 \) as  
\( W_2 = \mathrm{orth}\big(\begin{bmatrix}w_1^{\mathrm{lyap}} & \cdots & w_i^{\mathrm{lyap}}\end{bmatrix}\big) \),  
again using an implicit restart scheme to limit the number of columns, resulting in:
\begin{align}
\hat{E}_2 = W_2^\top E_2 W_2,\quad \hat{A}_2 = W_2^\top A_2 W_2,\quad \hat{C}_{\perp,i}^{\mathrm{lyap}} = C_{\perp,i}^{\mathrm{lyap}} W_2. \label{proj_mat_2}
\end{align}

Next, consider the eigenvalue decomposition of \( \hat{A}_2 \hat{E}_2^{-1} \):  
\( \hat{A}_2 \hat{E}_2^{-1} = \bar{T}\,\mathrm{diag}\big(\bar{\lambda}_1,\dots,\bar{\lambda}_r\big)\bar{T}^{-1} \).  
Define \( \bar{r}_{c,l} = \hat{C}_{\perp,i}^{\mathrm{lyap}}\,\bar{T}(:,l) \).  
The most observable pole of \( \hat{A}_2 \hat{E}_2^{-1} \) is the pole \( \bar{\lambda}_l \) corresponding to the largest  
\[
\bar{\phi}_l = \frac{\|\bar{r}_{c,l}\|_2^2}{|\mathrm{Re}(\bar{\lambda}_l)|},
\]
as in \cite{rommes2007methods,mengi2022large}. We then sort the columns of \( \bar{T} \) and the eigenvalues \( \bar{\lambda}_l \) in descending order of \( \bar{\phi}_l \). Thus, the pole \( \bar{\lambda}_1 \) with the largest \( \bar{\phi}_l \) can be used as the next ADI shift.
\subsubsection{For Low-rank BT}\label{3.9.2}
When the CF-ADI method is used to implement low-rank BT, the residuals \( R_{p,\mathrm{lyap}} \) and \( R_{q,\mathrm{lyap}} \) are often not good indicators of the accuracy of the ROM; see \cite{saak2012goal} and \cite{zulfiqar2024balanced}. The main reason is that a pole which is strongly controllable may be poorly observable, and vice versa, as illustrated by an example in \cite{zulfiqar2024balanced}. Thus, the Galerkin projection strategy discussed so far may not be suitable if the goal is to use the low-rank solutions of \( P_1 \) and \( Q_2 \) to implement BT—which aims to preserve states that are simultaneously the most controllable and the most observable.

Assume that \( A_1 = A_2 \), \( B_1 = B_2 \), \( C_1 = C_2 \), \( D_1 = D_2 \), and \( E_1 = E_2 \). Then \( V_1 \) and \( W_2 \) can be set as follows:
\[
V_1 = \mathrm{orth}\big(\begin{bmatrix}v_1^{\mathrm{lyap}} & \cdots & v_i^{\mathrm{lyap}}\end{bmatrix}\big)
\quad\text{and}\quad
W_2 = \mathrm{orth}\big(\begin{bmatrix}w_1^{\mathrm{lyap}} & \cdots & w_i^{\mathrm{lyap}}\end{bmatrix}\big),
\]
with an implicit restart mechanism, yielding the following ROM:
\[
\hat{E}_1 = W_2^\top E_1 V_1,\quad 
\hat{A}_1 = W_2^\top A_1 V_1,\quad 
\hat{B}_{\perp,i}^{\mathrm{lyap}} = W_2^\top B_{\perp,i}^{\mathrm{lyap}},\quad  
\hat{C}_{\perp,i}^{\mathrm{lyap}} = C_{\perp,i}^{\mathrm{lyap}} V_1.
\]

Next, solve \( \hat{E}_1 \tilde{A} = \hat{A}_1 \) and \( \hat{E}_1 \tilde{B}_{\perp} = \hat{B}_{\perp,i}^{\mathrm{lyap}} \) for \( \tilde{A} \) and \( \tilde{B}_{\perp} \), respectively. If $m_1=p_2$, $\tilde{A}=(\hat{E}_1)^{-1}\hat{A}_1$ and $\tilde{B}_{\perp}=(\hat{E}_1)^{-1}\hat{B}_{\perp,i}^{\mathrm{lyap}}$. Then compute the eigenvalue decomposition of \( \tilde{A} \): $\tilde{A} = \tilde{T}\,\mathrm{diag}\big(\tilde{\lambda}_1, \dots, \tilde{\lambda}_r\big)\tilde{T}^{-1}$. Define \( \tilde{r}_{b,l} = \tilde{T}^{-1}(l,:) \tilde{B}_{\perp} \) and \( \tilde{r}_{c,l} = \hat{C}_{\perp,i}^{\mathrm{lyap}} \tilde{T}(:,l) \).  
The dominant pole of \( \tilde{A} \) is the pole \( \tilde{\lambda}_l \) associated with the largest value of
\[
\tilde{\phi}_l = \frac{\|\tilde{r}_{c,l}\|_2 \, \|\tilde{r}_{b,l}\|_2}{|\mathrm{Re}(\tilde{\lambda}_l)|},
\]
as defined in \cite{rommes2007methods,mengi2022large}.  
Sort the columns of \( \tilde{T} \) and the eigenvalues \( \tilde{\lambda}_l \) in descending order of \( \tilde{\phi}_l \).  
Then, the pole \( \tilde{\lambda}_1 \) with the largest \( \tilde{\phi}_l \) can be used as the ADI shift, setting \( \alpha_i = \beta_i \).
\begin{remark}\label{remark5}
While the proposed subspace-accelerated shift generation strategy is inspired by SADPA~\cite{rommes2006efficient} and closely mimics it, it is not yet clear whether this strategy also offers a good chance of rapid convergence to the dominant poles. An in-depth convergence analysis would constitute a research topic in its own right and lies beyond the scope of this paper—though it is certainly an interesting direction to pursue. If this strategy indeed proves effective at capturing dominant poles, the applicability of UADI could be extended to dominant pole estimation as well. On the other hand, successive deflation in UADI might negatively affect convergence to dominant poles, as the peaks associated with those poles are flattened prematurely. In many cases, a peak is effectively captured by interpolation in a neighborhood close to the dominant pole; afterward, the peak is flattened, and the proposed shift strategy may cease to target that pole further. Repeated shifts in ADI methods generally lead to slower residual decay, so this premature deflation is actually beneficial from that perspective—it avoids using numerically very close shifts repeatedly. However, from a pole estimation standpoint, UADI with the proposed subspace-accelerated shift generation strategy appears unlikely to compete with SADPA~\cite{rommes2006efficient}, as the inherent premature deflation may prevent further refinement of the pole estimate. A detailed mathematical analysis could determine whether these initial qualitative assessments are rigorous. In the next section, numerical experiments demonstrate that the proposed shift generation strategy was indeed able to capture the dominant poles in the examples considered.
\end{remark}
\subsubsection{For Sylvester Equations}\label{3.9.3}
In \cite{benner2014self}, the Projection-I and Projection-II shift generation strategies are generalized for FADI by suggesting the choices  
\( V_1 = \mathrm{orth}\big(B_{\perp,i-1}^{\mathrm{sylv}}\big) \), \( W_2 = \mathrm{orth}\big(C_{\perp,i-1}^{\mathrm{sylv}}\big) \) for Projection-I, and  
\( V_1 = \mathrm{orth}\big(v_i^{\mathrm{sylv}}\big) \), \( W_2 = \mathrm{orth}\big(w_i^{\mathrm{sylv}}\big) \) for Projection-II.  
The theoretical rationale for this choice is the same as for Lyapunov equations.  
Along similar lines, we can generalize the proposed subspace-accelerated Galerkin projection scheme—originally developed for Lyapunov equations—to Sylvester equations by setting the projection matrices as
\[
V_1 = \mathrm{orth}\big(\begin{bmatrix}v_1^{\mathrm{sylv}} & \cdots & v_i^{\mathrm{sylv}}\end{bmatrix}\big)
\quad\text{and}\quad
W_2 = \mathrm{orth}\big(\begin{bmatrix}w_1^{\mathrm{sylv}} & \cdots & w_i^{\mathrm{sylv}}\end{bmatrix}\big).
\]
However, we now argue that these extensions of Projection-I, Projection-II, and the subspace-accelerated shift generation strategies may not be suitable for most Sylvester equations.

From an approximate integration perspective, FADI approximates \( X_{\mathrm{sylv}} \) as follows:
\begin{align}
X_{\mathrm{sylv}} \approx \frac{j}{2\pi} V_{\mathrm{fadi}}^{(i)} \Bigg[ \int_{-\infty}^{\infty} \Big(j\omega I + (S_{w,\mathrm{sylv}}^{(i)})^* \Big)^{-1} (L_{w,\mathrm{sylv}}^{(i)})^\top L_{v,\mathrm{sylv}}^{(i)} \Big(j\omega I + S_{v,\mathrm{sylv}}^{(i)} \Big)^{-*} d\omega \Bigg]^{-1} (W_{\mathrm{fadi}}^{(i)})^*. \label{fadi_int}
\end{align}
Assume for simplicity that \( A_2 E_2^{-1} = \mathrm{diag}(\lambda_1, \dots, \lambda_{n_2}) \), and define \( C_2 E_2^{-1} = [c_1, \dots, c_{n_2}] \).  
Then \( X_{\mathrm{sylv}} \) can be expressed as samples of \( (s E_1 - A_1)^{-1} B_1 \) at \( -\lambda_i^* \) in the direction \( c_i \), i.e.,
\[
X_{\mathrm{sylv}} = \begin{bmatrix}
(-\lambda_1^* E_1 - A_1)^{-1} B_1 c_1 & \cdots & (-\lambda_{n_2}^* E_1 - A_1)^{-1} B_1 c_{n_2}
\end{bmatrix};
\]
see \cite{MPIMD11-11}.  
From an interpolation theory standpoint, interpolating at the mirror images of the most observable poles of \( A_2 E_2^{-1} \) can yield a good low-rank approximation of \( X_{\mathrm{sylv}} \), as suggested by the integral expression \eqref{fadi_int}. After all, the quality of the approximation of the integrand directly affects the accuracy of the numerical integration. Thus, in the FADI method, the shifts \( \alpha_i = \beta_i \) should be chosen as the most observable poles of \( A_2 E_2^{-1} \)—that is, the poles \( \lambda_i \) with the largest residuals \( c_i \)—to obtain a good approximation according to \eqref{fadi_int} and the numerical integration perspective.

Dually, assume for simplicity that \( E_1^{-1} A_1 = \mathrm{diag}(\lambda_1, \dots, \lambda_{n_1}) \), and define \( E_1^{-1} B_1 = [b_1, \dots, b_{n_1}] \).  
Then \( X_{\mathrm{sylv}} \) can be written as samples of \( (s E_2 - A_2)^{-1} C_1 \) at \( -\lambda_i^* \) in the direction \( b_i \), i.e.,
\[
X_{\mathrm{sylv}} = 
\begin{bmatrix}
b_1 C_2 (-\lambda_1^* E_2 - A_2)^{-1} \\
\vdots \\
b_{n_1} C_2 (-\lambda_{n_1}^* E_2 - A_2)^{-1}
\end{bmatrix}.
\]
Again, from an interpolation perspective, interpolating at the mirror images of the most controllable poles of \( E_1^{-1} A_1 \) can lead to a good low-rank approximation of \( X_{\mathrm{sylv}} \).  
Hence, in FADI, the shifts \( \alpha_i = \beta_i \) should be the most controllable poles of \( E_1^{-1} A_1 \)—that is, the poles \( \lambda_i \) with the largest residuals \( b_i \)—to achieve a good approximation.

Note, however, that choosing \( \alpha_i \) as the most observable poles of \( A_2 E_2^{-1} \) and \( \beta_i \) as the most controllable poles of \( E_1^{-1} A_1 \) may not be a good strategy from a numerical integration standpoint if \( \alpha_i \neq \beta_i \), because the integrand
\[
(sE_1 - A_1)^{-1} B_1 C_2 (sE_2 - A_2)^{-*}
\]
is not jointly interpolated when \( \alpha_i \neq \beta_i \). Instead, only \( (sE_1 - A_1)^{-1} B_1 \) is interpolated at \( -\alpha_i \), and \( C_2 (sE_2 - A_2)^{-1} \) is interpolated at \( -\beta_i \). This issue is illustrated in the following example.\\

\noindent\textbf{Illustrative Example:} Consider the following state-space realizations:
\begin{align}
E_1 &=\begin{bsmallmatrix}0.2498 &  0  &  0.0002  &  0 &    0.0001  &  0\\
   0   & 0.2498   & 0 &   0.0002  &  0 &    0.0001\\
    0.0002 &   0 &   0.2499  &  0.0001  &  0   0\\
    0 &   0.0002  &  0.0001  &  0.2499  &  0&   0 \\
    0.0001  &  0 &   0 &   0 &    0.2500 &   0 \\
    0 &    0.0001 &   0 &   0 &   0 &   0.2500\end{bsmallmatrix},&
A_1 &= \begin{bsmallmatrix}-0.2508 &  24.4816 &  -0.5095 &  -0.4723  & -0.5112 &  -0.4814\\
  -25.4822  & -0.2508  & -0.5283  & -0.4908 &  -0.5188  & -0.4888\\
   -0.4908 &  -0.4723  & -0.2497  & 49.4833 &  -0.5046  & -0.4861\\
   -0.5283 &  -0.5095 & -50.4837  & -0.2497  & -0.5141 &  -0.4952\\
   -0.4888 &  -0.4814 &  -0.4952 &  -0.4861  & -0.2499 &  99.4954\\
   -0.5188 &  -0.5112 &  -0.5141&   -0.5046 &-100.4955 &  -0.2499\end{bsmallmatrix},\nonumber\\
B_1&=\begin{bsmallmatrix}0.5025& 0.4965 & -0.0051& 0.0035& 0.0133& -0.0116\end{bsmallmatrix}^\top,&
C_1 &=\begin{bsmallmatrix}0.4965&    0.5025 &   0.0035  & -0.0051 &  -0.0116  &  0.0133\end{bsmallmatrix},\quad D_1=0,\nonumber\\
E_2 &=\begin{bsmallmatrix}0.2498 &  0 &    0.0002 &   0 &   0.0001 &   0\\
   0&    0.2498 &   0 &    0.0002  &  0 &   0.0001\\
    0.0002 &   0 &  0.2499 &   0.0001  &  0 &    0\\
    0 &    0.0002 &   0.0001  &  0.2499 &   0 &    0\\
    0.0001 &   0 &   0 &  0 &    0.2500  &  0\\
    0 &   0.0001  &  0 &   0 &   0 &    0.2500\end{bsmallmatrix},&
A_2 &= \begin{bsmallmatrix}-0.2508 &  24.4816 &  -0.5095 &  -0.4723 &  -0.5112  & -0.4814\\
  -25.4822  & -0.2508 &  -0.5283 &  -0.4908  & -0.5188 &  -0.4888\\
   -0.4908  & -0.4723  & -0.2497 &  49.4833 &  -0.5046 &  -0.4861\\
   -0.5283  & -0.5095 & -50.4837 &  -0.2497 &  -0.5141  & -0.4952\\
   -0.4888  & -0.4814  & -0.4952 &  -0.4861 &  -0.2499  & 99.4954\\
   -0.5188  & -0.5112  & -0.5141 &  -0.5046 & -100.4955 &  -0.2499\end{bsmallmatrix},\nonumber\\
B_2 &=\begin{bsmallmatrix}-0.0029& 0.0042& 0.0118& -0.0129 & 0.4866& 0.5131\end{bsmallmatrix}^\top,&
C_2&=\begin{bsmallmatrix}0.0042&-0.0029 &  -0.0129  &  0.0118 &   0.5131 &   0.4866\end{bsmallmatrix},\quad D_2=0.\nonumber
\end{align} Both \( G_1(s) \) and \( G_2(s) \) have identical poles at \( -1 \pm j100 \), \( -1 \pm j200 \), and \( -1 \pm j400 \). Their frequency responses are shown in Figure~\ref{fig1}.
\begin{figure}[!h]
  \centering
  \includegraphics[width=12cm]{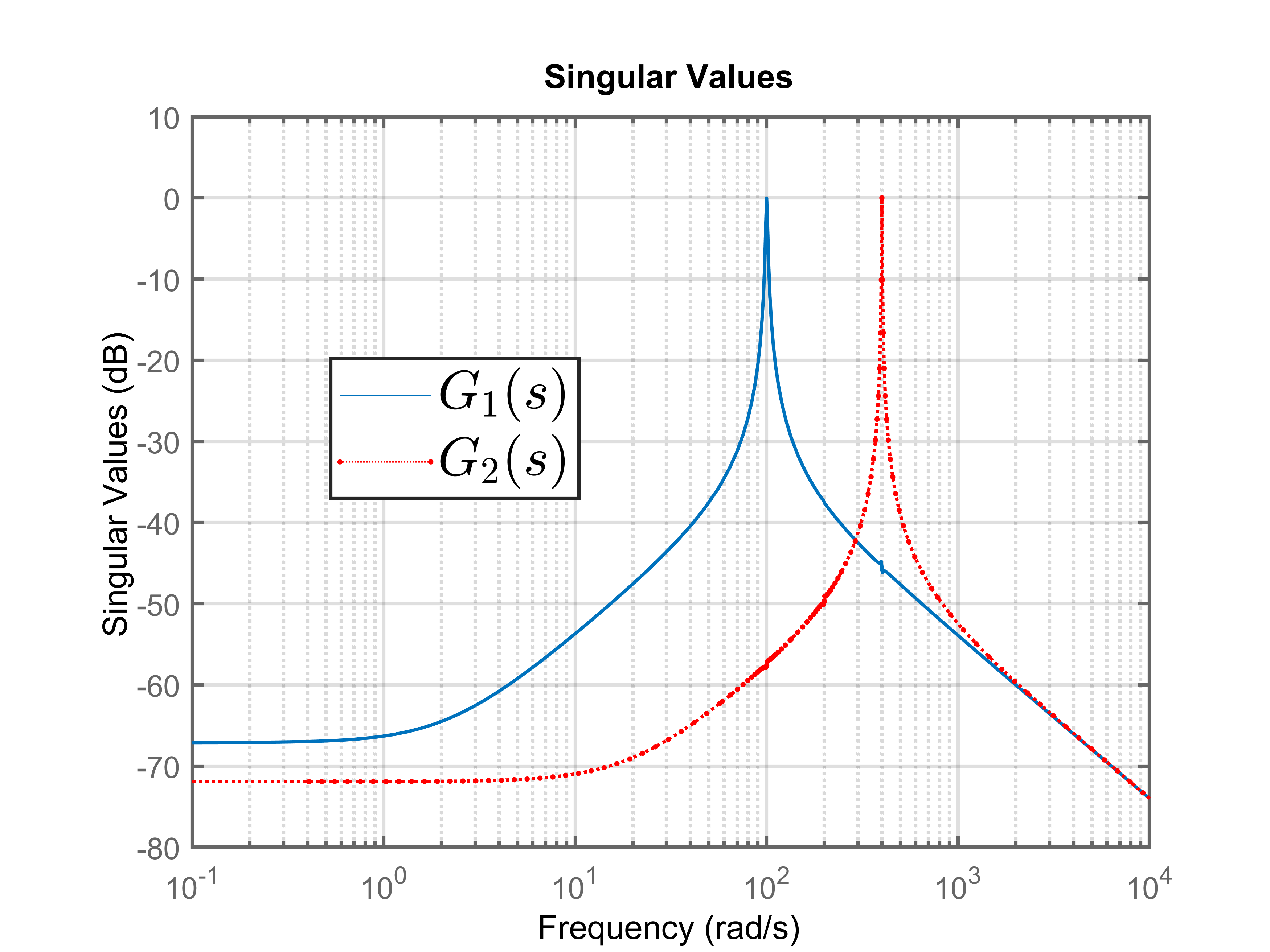}
  \caption{Frequency domain plot of $G_1(s)$ and $G_2(s)$}\label{fig1}
\end{figure}
\( G_1(s) \) exhibits a peak at 100 rad/sec (corresponding to \( -1 \pm j100 \)), while \( G_2(s) \) peaks at 400 rad/sec (corresponding to \( -1 \pm j400 \)). Table~\ref{tab1} lists various choices of ADI shifts and the corresponding normalized residuals \( \frac{\|R_{\mathrm{sylv}}\|_2}{\|B_1 C_2\|_2} \).
\begin{table}[!h]
\centering
\caption{Normalized Residual for Various ADI Shifts}\label{tab1}
\begin{tabular}{|c|c|c|}\hline
$\alpha_i$ & $\beta_i$ & $\frac{\|R_{\mathrm{sylv}}\|_2}{\|B_1C_2\|_2}$\\\hline
$-1\pm j100$& $-1\pm j400$ & $3.51\times 10^{4}$ \\
$-1\pm j400$& $-1\pm j100$ & $12.2839$ \\
$-1\pm j100$& $-1\pm j100$ & $0.0412$ \\
$-1\pm j400$& $-1\pm j400$ & $0.0411$ \\
\hline
\end{tabular}
\end{table}
Although setting \( \alpha_i = -1 \pm j100 \) captures the peak of \( G_1(s) \) and \( \beta_i = -1 \pm j400 \) captures the peak of \( G_2(s) \), this combination yields the worst approximation. Choosing \( \alpha_i \) as the most observable poles and \( \beta_i \) as the most controllable poles gives a better result, but the best approximations are obtained when \( \alpha_i = \beta_i \)—either both equal to the most controllable poles of \( E_1^{-1} A_1 \) or both equal to the most observable poles of \( A_2 E_2^{-1} \). This observation aligns with the earlier discussion. Consequently, the shift strategy based on Ritz values of \( (E_1, A_1) \) and \( (E_2, A_2) \) proposed in \cite{benner2014self} is unlikely to yield good numerical accuracy when the most controllable poles of \( E_1^{-1} A_1 \) and the most observable poles of \( A_2 E_2^{-1} \) are very different. Therefore, we propose a different shift generation strategy for Sylvester equations.

We propose selecting  
\( V_1 = \mathrm{orth}\big(\begin{bmatrix} v_1^{\mathrm{sylv}} & \cdots & v_i^{\mathrm{sylv}} \end{bmatrix}\big) \)  
with an implicit restart mechanism, and projecting as follows:
\[
\hat{E}_1 = V_1^\top E_1 V_1, \quad
\hat{A}_1 = V_1^\top A_1 V_1, \quad
\hat{B}_{\perp,i}^{\mathrm{sylv}} = V_1^\top B_{\perp,i}^{\mathrm{sylv}}.
\]
We then use the most controllable pole of these projected matrices as the ADI shifts \( \alpha_i \) and \( \beta_i \), ensuring \( \alpha_i = \beta_i \). In the subsequent iteration, we select  
\( W_2 = \mathrm{orth}\big(\begin{bmatrix} w_1^{\mathrm{sylv}} & \cdots & w_i^{\mathrm{sylv}} \end{bmatrix}\big) \)  
with implicit restart, and project as:
\[
\hat{E}_2 = W_2^\top E_2 W_2, \quad
\hat{A}_2 = W_2^\top A_2 W_2, \quad
\hat{C}_{\perp,i}^{\mathrm{sylv}} = C_{\perp,i}^{\mathrm{sylv}} W_2,
\]
and use the most observable pole of these projected matrices as the common ADI shift \( \alpha_i = \beta_i \).

An additional advantage of enforcing \( \alpha_i = \beta_i \) is that the FADI-based approximation \( V_{\mathrm{fadi}}^{(i)} D_{\mathrm{fadi}}^{(i)} (W_{\mathrm{fadi}}^{(i)})^\top \) coincides with \( V_{\mathrm{lyap}}^{(i)} (W_{\mathrm{lyap}}^{(i)})^\top \), and the residual \( R_{\mathrm{sylv}} = B_{\perp,i}^{\mathrm{sylv}} C_{\perp,i}^{\mathrm{sylv}} \) matches \( B_{\perp,i}^{\mathrm{lyap}} C_{\perp,i}^{\mathrm{lyap}} \).  
Thus, when \( \alpha_i = \beta_i \), there is no need to extract \( V_{\mathrm{fadi}}^{(i)} \) and \( W_{\mathrm{fadi}}^{(i)} \) separately from \( V_{\mathrm{lyap}}^{(i)} \) and \( W_{\mathrm{lyap}}^{(i)} \).
\subsubsection{For Riccati Equations}
Just as for Lyapunov equations, \( (sE_1 - A_1)^{-1} B_{\perp,i}^{\mathrm{ricc}} \) and \( C_{\perp,i}^{\mathrm{ricc}} (sE_2 - A_2)^{-*} \) can be projected to generate subsequent shifts for RADI. The projection matrix \( V_1 \) can be set as  
\( V_1 = \mathrm{orth}\big(\begin{bmatrix} v_1^{\mathrm{ricc}} & \cdots & v_i^{\mathrm{ricc}} \end{bmatrix}\big) \) with an implicit restart, and similarly, \( W_2 \) can be set as \( W_2 = \mathrm{orth}\big(\begin{bmatrix} w_1^{\mathrm{ricc}} & \cdots & w_i^{\mathrm{ricc}} \end{bmatrix}\big)\) with an implicit restart. Again, if the goal is to use the low-rank solution to implement generalizations of BT—such as positive-real BT or bounded-real BT—then shifts should be generated using Petrov–Galerkin projection rather than Galerkin projection, as discussed earlier in Section~\ref{3.9.2}.
\subsubsection{For UADI Framework}
Note that all ADI methods interpolate at the mirror images of the ADI shifts and therefore share similar interpolatory properties.  
Consequently, shift generation using  
\( V_1 = \mathrm{orth}\big(\begin{bmatrix} v_1^{\mathrm{lyap}} & \cdots & v_i^{\mathrm{lyap}} \end{bmatrix}\big) \),  
\( V_1 = \mathrm{orth}\big(\begin{bmatrix} v_1^{\mathrm{sylv}} & \cdots & v_i^{\mathrm{sylv}} \end{bmatrix}\big) \), or  
\( V_1 = \mathrm{orth}\big(\begin{bmatrix} v_1^{\mathrm{ricc}} & \cdots & v_i^{\mathrm{ricc}} \end{bmatrix}\big) \)  
yields nearly identical numerical performance in terms of residual decay. Therefore, within the UADI framework, we may set the projection matrices as  
\( V_1 = \mathrm{orth}\big(\begin{bmatrix} v_1^{\mathrm{lyap}} & \cdots & v_i^{\mathrm{lyap}} \end{bmatrix}\big) \) and  
\( W_2 = \mathrm{orth}\big(\begin{bmatrix} w_1^{\mathrm{lyap}} & \cdots & w_i^{\mathrm{lyap}} \end{bmatrix}\big) \),  
both equipped with implicit restart mechanisms. However, the actual generation of ADI shifts using these projection matrices depends on the goal of applying the UADI framework. For instance, if the aim is to implement low-rank versions of BT and its variants, shifts should be generated via the Petrov–Galerkin projection discussed earlier. If, on the other hand, the goal is to solve a Sylvester equation, an alternating shift generation strategy should be used—one that ensures \( \alpha_i = \beta_i \) via the Galerkin projection approach described previously.
\subsection{A MATLAB-based Implementation of UADI}
The MATLAB implementation of UADI is publicly available at \cite{mycode}. It is not optimized for memory usage, as it is written with the readers of this paper in mind, allowing them to verify most of the mathematical results. Consequently, several intermediate matrices are stored and returned as outputs, even though they could be omitted to reduce memory usage and improve computational speed. The initial ADI shifts $\alpha_1$ and $\beta_1$ are set to $0.001$ by default. If any shift-generation strategy produces shifts with positive real parts, they are made negative by multiplying them by $-1$. Although academic in nature, the implementation is still efficient enough to handle problems where $n_1$ is in the millions without memory issues; we have tested it for sizes up to $10^7$.

For practical applications where computational speed is a priority, and where best practices for selecting initial shifts and treating shifts with positive real parts are relevant, we recommend implementing UADI using the latest version of the Matrix Equation Sparse Solver (M-M.E.S.S.) library \cite{saak2021mm}. Such an implementation will be more efficient and effective than the academic-oriented version we provide at \cite{mycode}. For example, although our sparse–dense Sylvester solver in \cite{mycode} is considerably faster than MATLAB’s built-in ``\textit{lyap}'' command, it is still not as fast as the corresponding M-M.E.S.S. routine, which is state-of-the-art and applies best coding practices to exploit sparsity and save memory. However, for academic purposes, the source code of M-M.E.S.S. is not straightforward to follow. Therefore, for reproducibility, verification of the mathematical results in this paper, and a clearer understanding of the ADI properties explored here, we recommend using our implementation at \cite{mycode}, whose source code is written specifically with these goals in mind. 
\section{Numerical Simulation}
The numerical performance and computational efficiency of ADI methods are well documented in the literature. Therefore, we avoid reporting similar results in this section. The main focus of this section is to test the self-generating shift strategies proposed in the previous section. All tests are performed using MATLAB R2021b on a laptop with 16 GB of random access memory (RAM) and a 2 GHz Intel i7 processor. The shifted linear systems are solved using MATLAB’s backslash operator `$\backslash$'. The MATLAB codes and data to reproduce the results of this section are available at \cite{mycode}.

Three numerical examples are considered. The model in the first example is of modest order, allowing a posteriori error analysis to be performed on the ROMs produced by BT. The models in the remaining two examples are large-scale. The first test focuses on evaluating the shift strategy proposed in the previous section for BT. The second test focuses on the shift strategy for the Sylvester equation. The third example shows that subspace acceleration improves convergence speed, repaying the initial effort through shorter overall computation time.
\subsection{RLC Circuit Network}
This example considers a two-port passive RLC network (shown in Figure \ref{fig2}) with 2 inputs and 2 outputs. Each RLC ladder consists of 400 repeated segments. The resistance, capacitance, and inductance values in the segments are as follows:  
\(\bar{R}_1 = 0.1\,\Omega\), \(\bar{R}_2 = 1\,\Omega\), \(\bar{L}_1 = 0.1\,\text{H}\), \(\bar{C}_1 = 0.1\,\text{F}\), \(\bar{R}_3 = 0.5\,\Omega\), \(\bar{R}_4 = 3\,\Omega\), \(\bar{L}_2 = 0.2\,\text{H}\), \(\bar{C}_2 = 0.2\,\text{F}\), and \(\bar{R}_s = 0.2\,\Omega\).
\begin{figure}[!h]
  \centering
  \includegraphics[width=10cm]{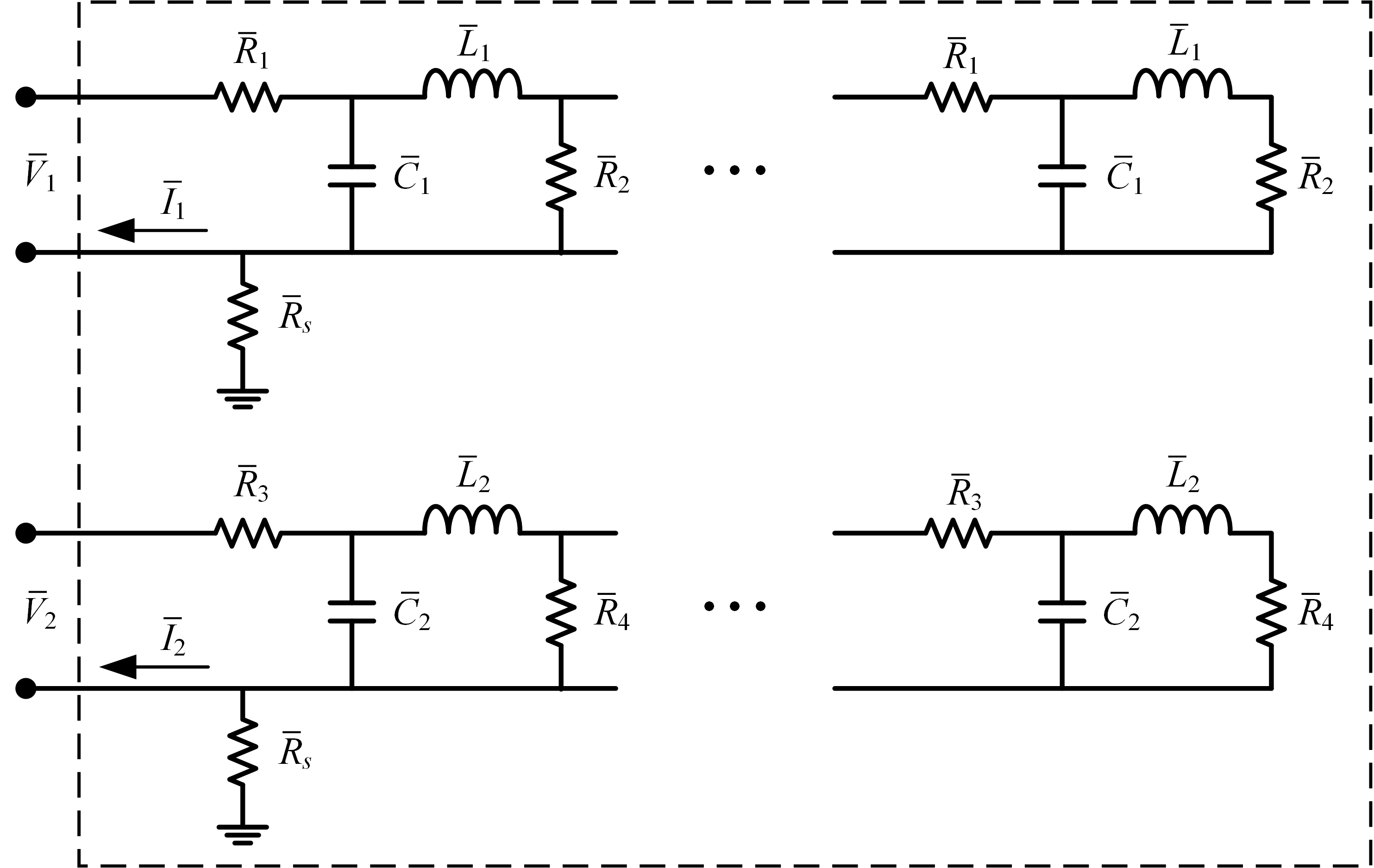}
  \caption{A Two-port Passive RLC Network}\label{fig2}
\end{figure}
The state-space matrices \(A_1 \in \mathbb{R}^{1600 \times 1600}\), \(B_1 \in \mathbb{R}^{1600 \times 2}\), \(C_1 \in \mathbb{R}^{2 \times 1600}\), \(D_1 \in \mathbb{R}^{2 \times 2}\), and \(E_1 \in \mathbb{R}^{1600 \times 1600}\) of this $1600^{th}$-order model are available in \cite{mycode}. In this example, \(G_1(s)\) and \(G_2(s)\) are identical, as the focus is on the applicability of UADI in implementing the BT family. Since the model is square, stable, minimum phase, positive-real, and bounded-real, all the linear matrix equations \eqref{lyap_p}-\eqref{ricc_q_sf} considered in the paper are well defined, and UADI can compute their low-rank solutions simultaneously. The remaining matrices and parameters for this example are:  
\(S_1 = \begin{bmatrix}1 & 1 \\ 0 & -1\end{bmatrix}\), \(S_2 = \begin{bmatrix}1 & 1 \\ 1 & -1\end{bmatrix}\), \(\gamma_1 = 2\), and \(\gamma_2 = 3\). The maximum number of ADI iterations is set to \(k = 50\). The shift generation strategy proposed in Section \ref{3.9.2} is used. The projection matrices \(V_1\) and \(W_2\) are implicitly restarted once their number of columns reaches $10$. Consequently, the projected eigenvalue problem never involves computing more than $10$ eigenvalues—ensuring that shift generation remains an inexpensive operation.

Figure \ref{fig3} shows that the proposed shift generation strategy yields excellent performance in terms of the steep decline of the normalized residuals (in terms of $L_2$ norm) for the linear matrix equations \eqref{lyap_p}–\eqref{ricc_q_sf}.
\begin{figure}[!h]
  \centering
  \includegraphics[width=12cm]{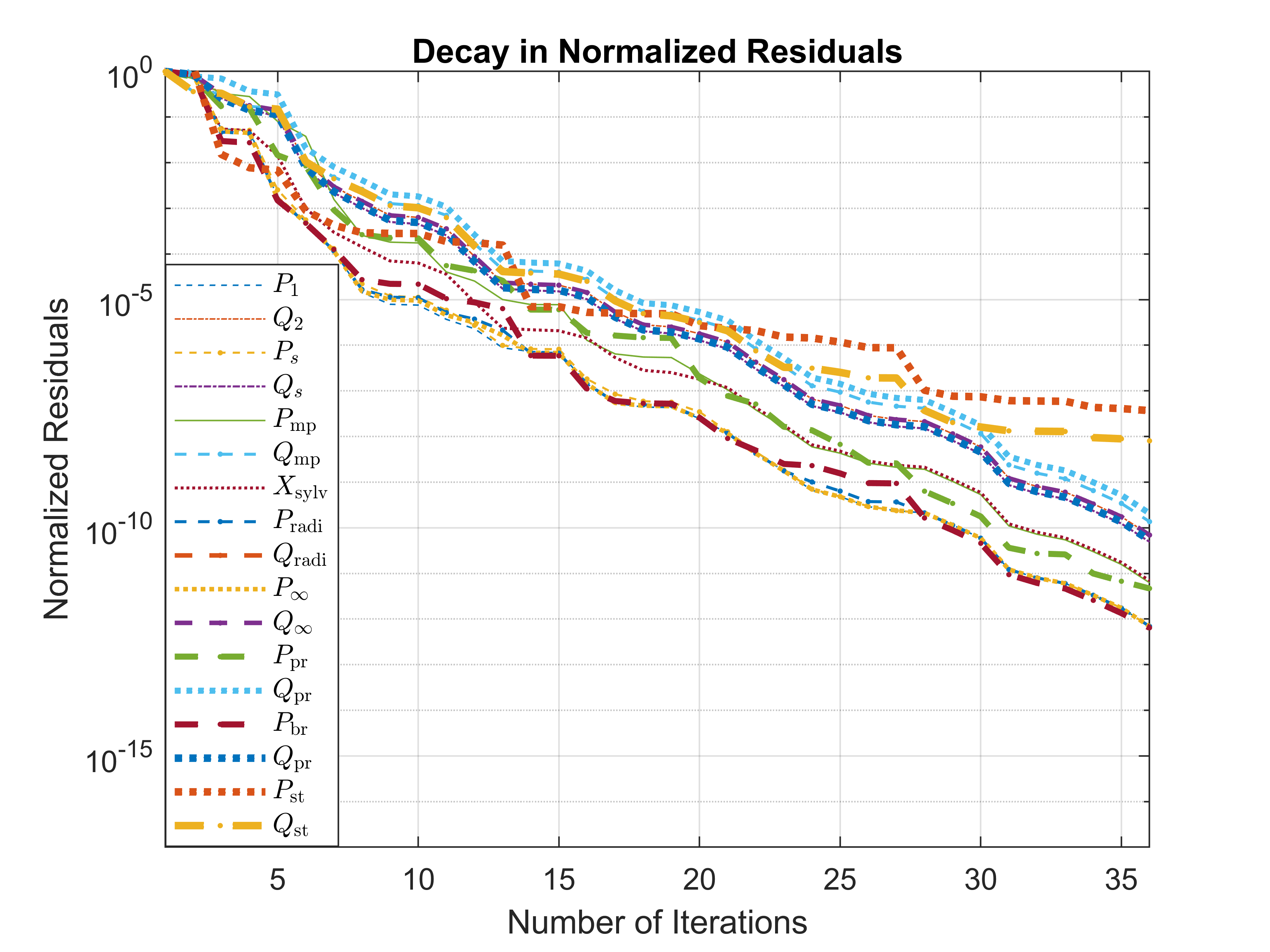}
  \caption{Decay in Normalized Residuals}\label{fig3}
\end{figure}
The most dominant pole of \(G_1(s)\) is \(-98.9898\). By examining the 49 shifts generated by the proposed strategy, we found that the fourth shift is \(-98.3214\). Note that the initial shift \(\alpha_1 = -0.001\) was not in the neighbourhood of \(-98.9898\), yet the proposed strategy quickly identified a shift near the dominant pole within just four iterations. The fifth shift produced by the strategy is \(-3.9099 + 6.9746i\), indicating that the dominant pole \(-98.9898\) was prematurely deflated by UADI—confirming the expectation expressed in Remark \ref{remark5}.

Next, low-rank BT is performed using the approximations \(V_{\mathrm{lyap}}^{(i)}\) and \(W_{\mathrm{lyap}}^{(i)}\). As shown in Figure \ref{fig4}, the low-rank BT—based on low-rank square-root factors of \(P_1\) and \(Q_2\)—preserves the most significant Hankel singular values of the original system.
\begin{figure}[!h]
  \centering
  \includegraphics[width=12cm]{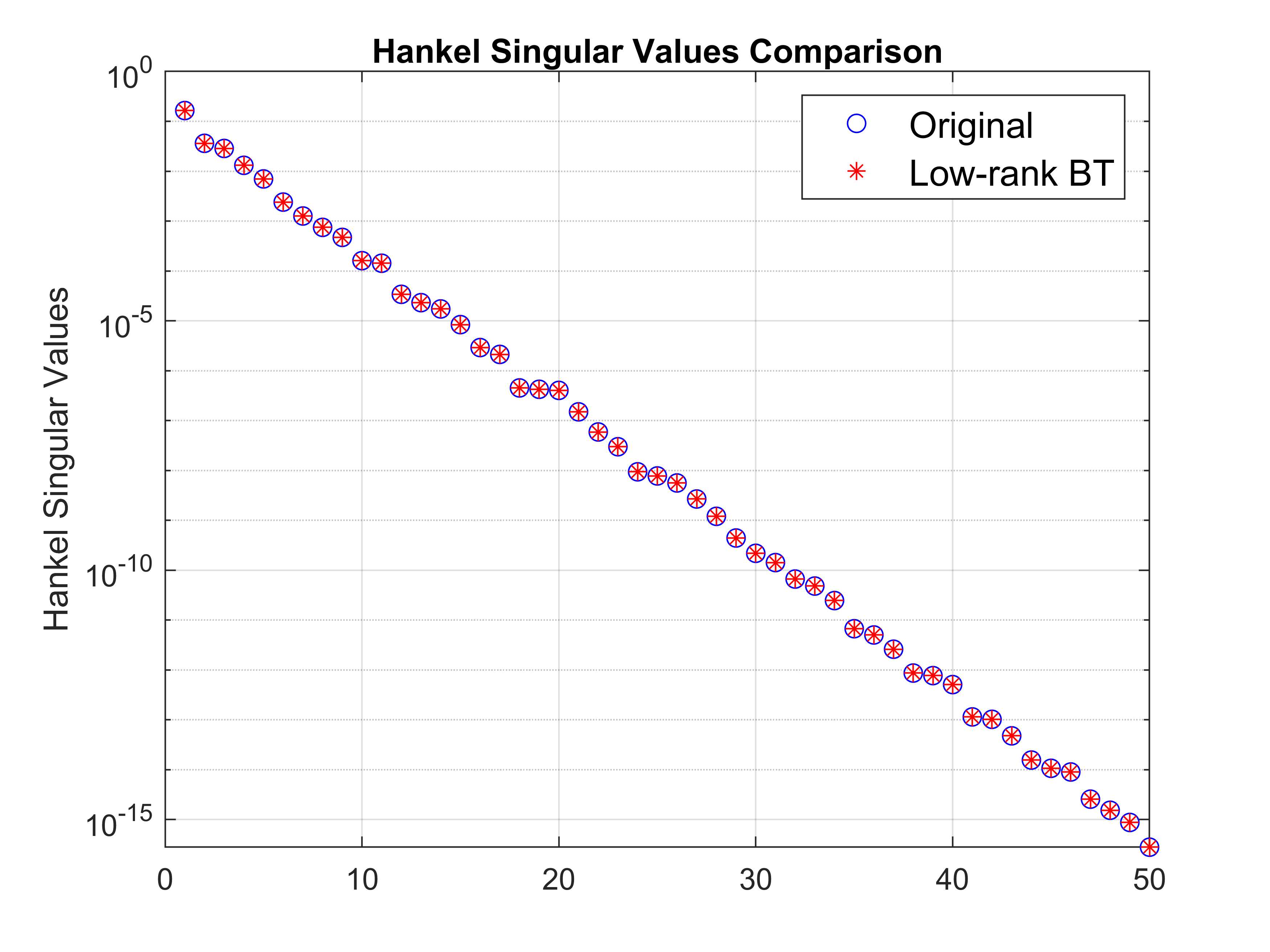}
  \caption{Hankel Singular Values Comparison}\label{fig4}
\end{figure}
Finally, Figure \ref{fig5} shows the frequency-domain responses of the original system and the $10^{th}$-order ROMs produced by full-rank and low-rank BT. The frequency responses of the original system and the ROMs are indistinguishable, confirming that the shift strategy successfully ensured an accurate ROM. Thus, the small residuals in \(P_1\) and \(Q_2\) were indeed translated into high ROM accuracy.
\begin{figure}[!h]
  \centering
  \includegraphics[width=12cm]{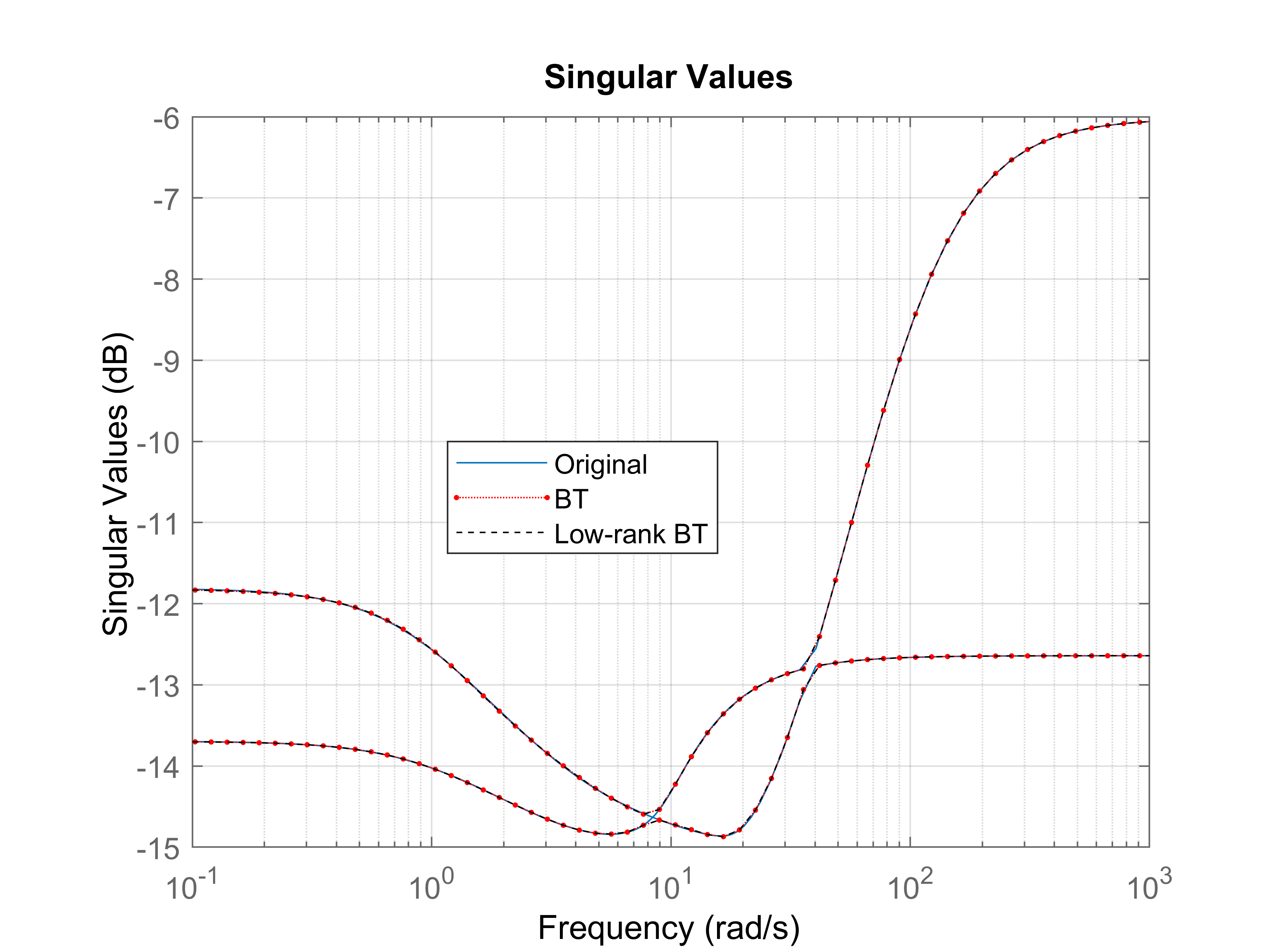}
  \caption{Singular values of $G_1(s)$ and ROM}\label{fig5}
\end{figure}
\subsection{Penzl's Triple Peak Model}
This constructive example generalizes Penzl’s procedure for building a dynamical system with three peaks in the frequency-domain plot \cite{morPen06,slicot_fom,chahlaoui2005benchmark}. The MATLAB function to generate such a system of arbitrary order is provided in Appendix I, and the frequencies at which the peaks occur can be specified as inputs. Using this function, \(G_1(s)\) and \(G_2(s)\) are constructed with the following state-space matrix dimensions: \(E_1, A_1 \in \mathbb{R}^{10^6 \times 10^6}\), \(B_1 \in \mathbb{R}^{10^6 \times 1}\), \(C_1 \in \mathbb{R}^{1 \times 10^6}\), \(D_1 = 0\), \(E_2, A_2 \in \mathbb{R}^{10^6 \times 10^6}\), \(B_2 \in \mathbb{R}^{10^6 \times 1}\), \(C_2 \in \mathbb{R}^{1 \times 10^6}\), \(D_2 = 0\). Unlike in the original Penzl procedure \cite{morPen06}, the matrices \(E_1\) and \(E_2\) are not identity matrices. \(G_1(s)\) has peaks at $10$ rad/s, $20$ rad/s, and $30$ rad/s, corresponding to the most controllable poles \(-1 \pm j10\), \(-1 \pm j20\), and \(-1 \pm j30\), respectively. The remaining poles are all real and located at \(-1, -2, \dots, -9, 99, 994\); by construction in Penzl’s triple-peak models, these are significantly less controllable \cite{morPen06,slicot_fom,chahlaoui2005benchmark}. Similarly, \(G_2(s)\) has peaks at $40$ rad/s, $50$ rad/s, and $60$ rad/s, corresponding to the most observable poles \(-1 \pm j40\), \(-1 \pm j50\), and \(-1 \pm j60\), respectively. Its remaining poles are the same real poles \(-1, -2, \dots, -9, 99, 994\), which are significantly less observable by design. The purpose of using this constructive example is that, despite its large scale, the pole locations in Penzl’s triple-peak models are known. This helps us assess the success or failure of various shift strategies. Importantly, this pole information is not used to pre-specify shifts; rather, we investigate whether the proposed shift generation strategy can automatically capture these dominant poles.

For the implicit restart in the shift generation strategy of Section \ref{3.9.1}, the maximum number of columns in \(V_1\) and \(W_2\) is set to $20$. The maximum number of ADI iterations is set to \(k = 70\). The normalized residuals for \(P_1\) and \(Q_2\) are plotted in Figure \ref{fig6}.
\begin{figure}[!h]
  \centering
  \includegraphics[width=12cm]{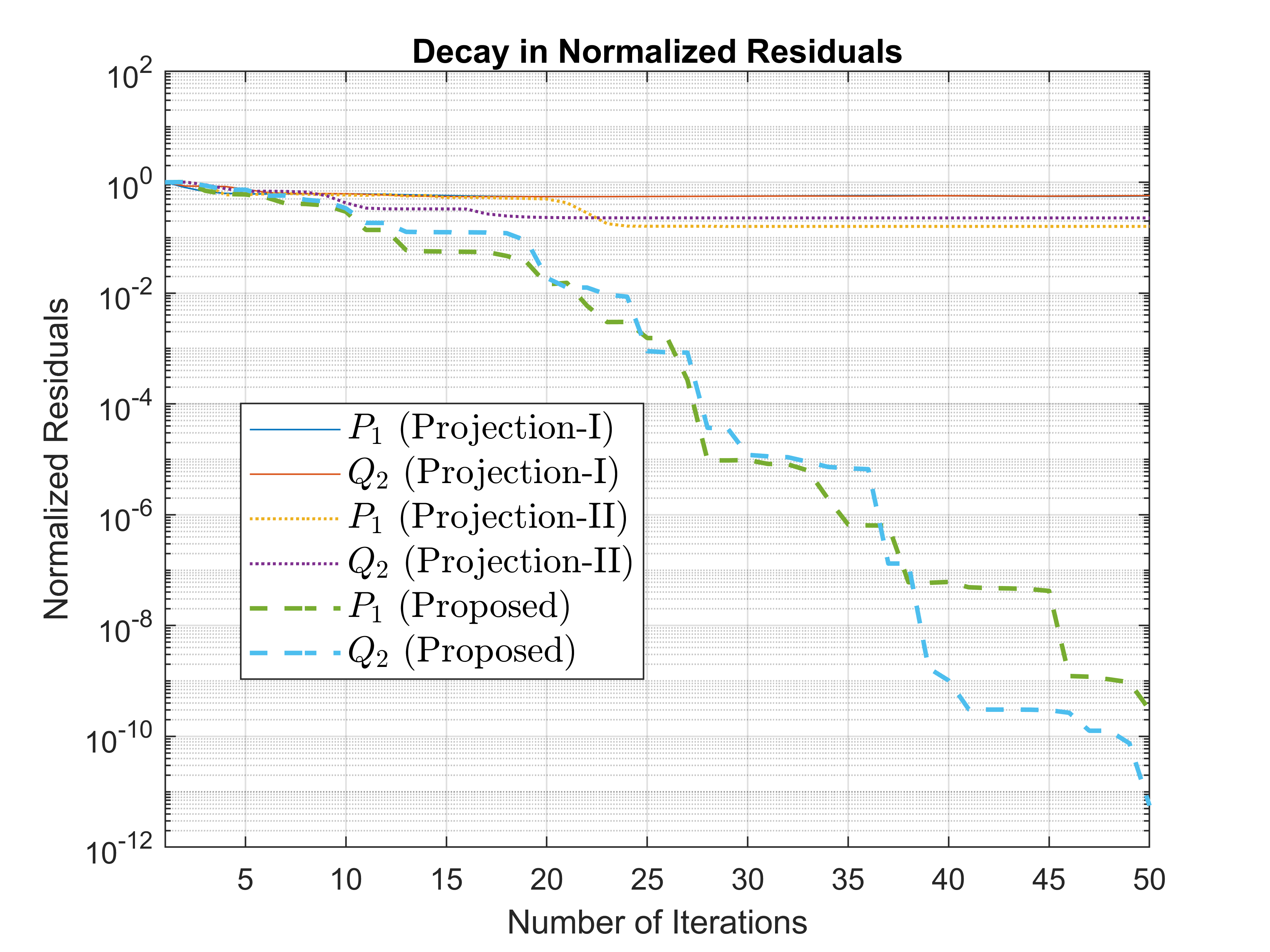}
  \caption{Decay in Normalized Residuals}\label{fig6}
\end{figure}
The proposed shift generation strategy yields a steep decline in the normalized residuals, significantly outperforming the Projection-I and Projection-II strategies. Figure \ref{fig7} shows the shifts produced by the proposed strategy. It clearly captures the three most dominant poles of \(G_1(s)\) and the three most dominant poles of \(G_2(s)\), which explains the high accuracy achieved.
\begin{figure}[!h]
  \centering
  \includegraphics[width=12cm]{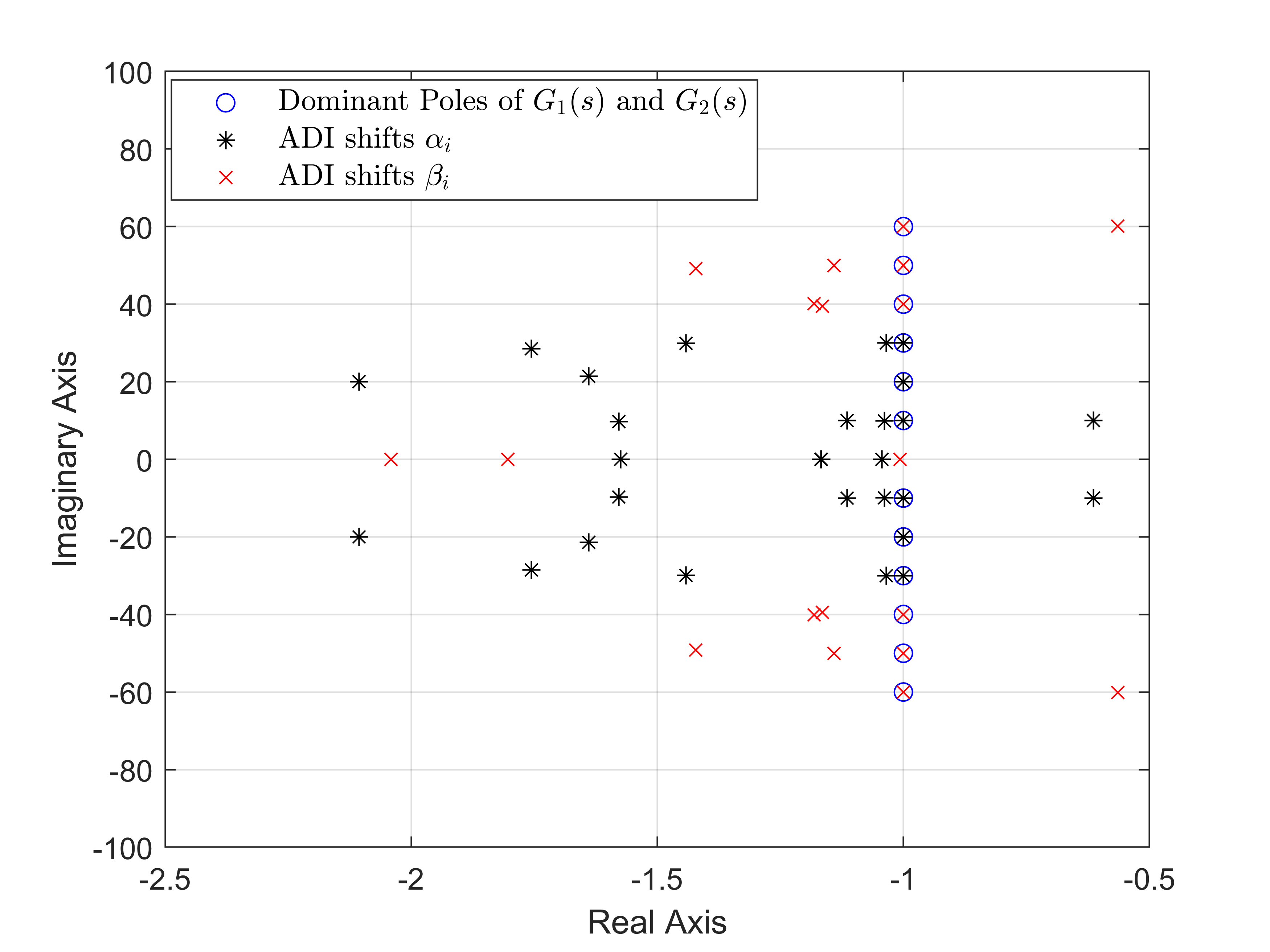}
  \caption{Dominant Poles and ADI Shifts}\label{fig7}
\end{figure}
Figure \ref{fig8} shows the normalized residual \(\frac{\|R_{\mathrm{sylv}}^{(i)}\|_2}{\|B_1 C_2\|_2}\) for these shifts.
\begin{figure}[!h]
  \centering
  \includegraphics[width=12cm]{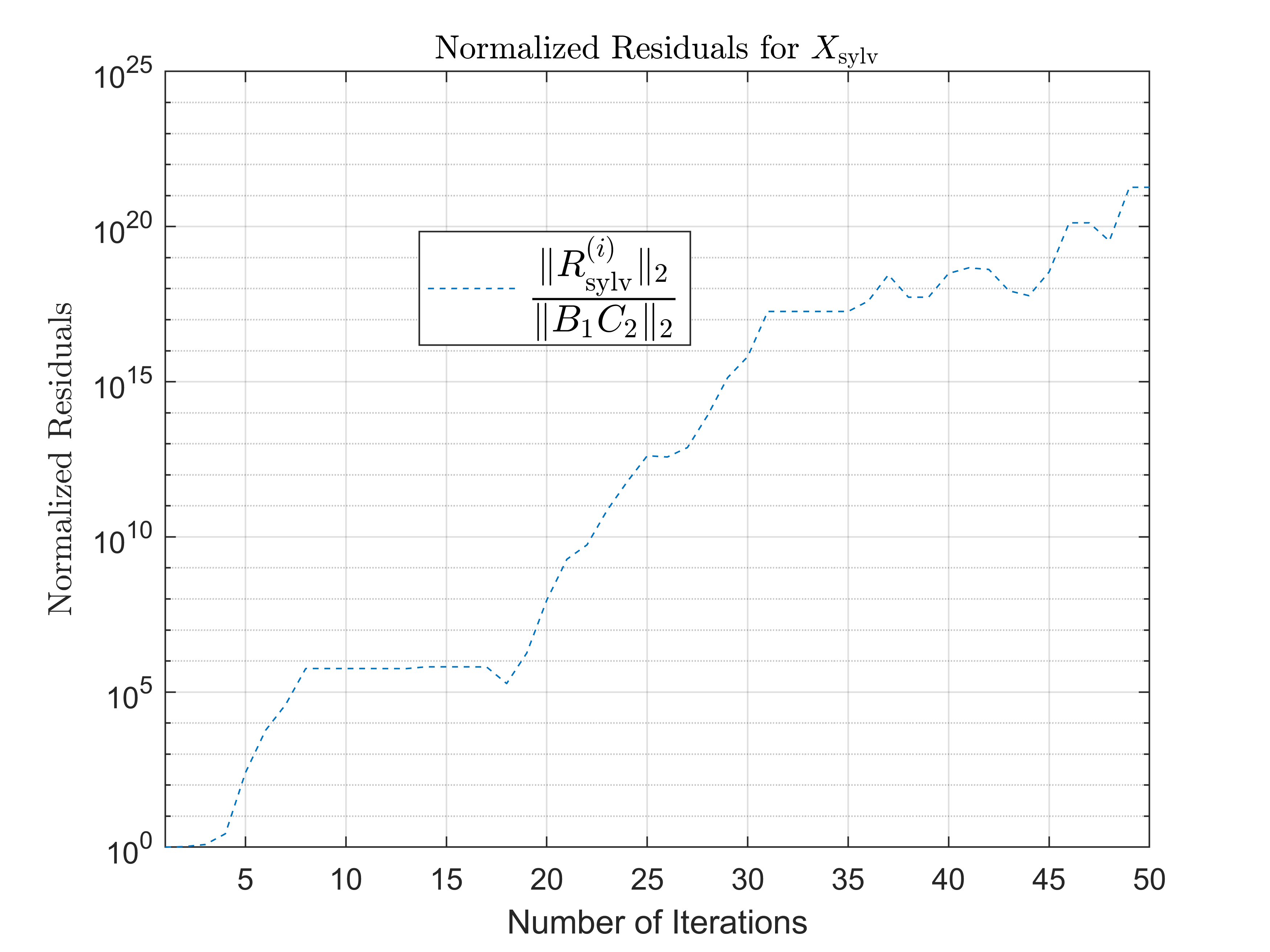}
  \caption{Normalized Residual for \(X_{\mathrm{sylv}}\)}\label{fig8}
\end{figure}
Interestingly, the shifts that produced excellent approximations of \(P_1\) and \(Q_2\) yielded a poor approximation of \(X_{\mathrm{sylv}}\); in fact, the residual \(R_{\mathrm{sylv}}\) increased rather than decreased—consistent with the observation in Section \ref{3.9.3}.

To prioritize the approximation of \(X_{\mathrm{sylv}}\), we instead use the shift generation strategy from Section \ref{3.9.3}, which enforces \(\alpha_i = \beta_i\). Figure \ref{fig9} shows the normalized residuals for \(P_1\), \(Q_2\), and \(X_{\mathrm{sylv}}\).
\begin{figure}[!h]
  \centering
  \includegraphics[width=12cm]{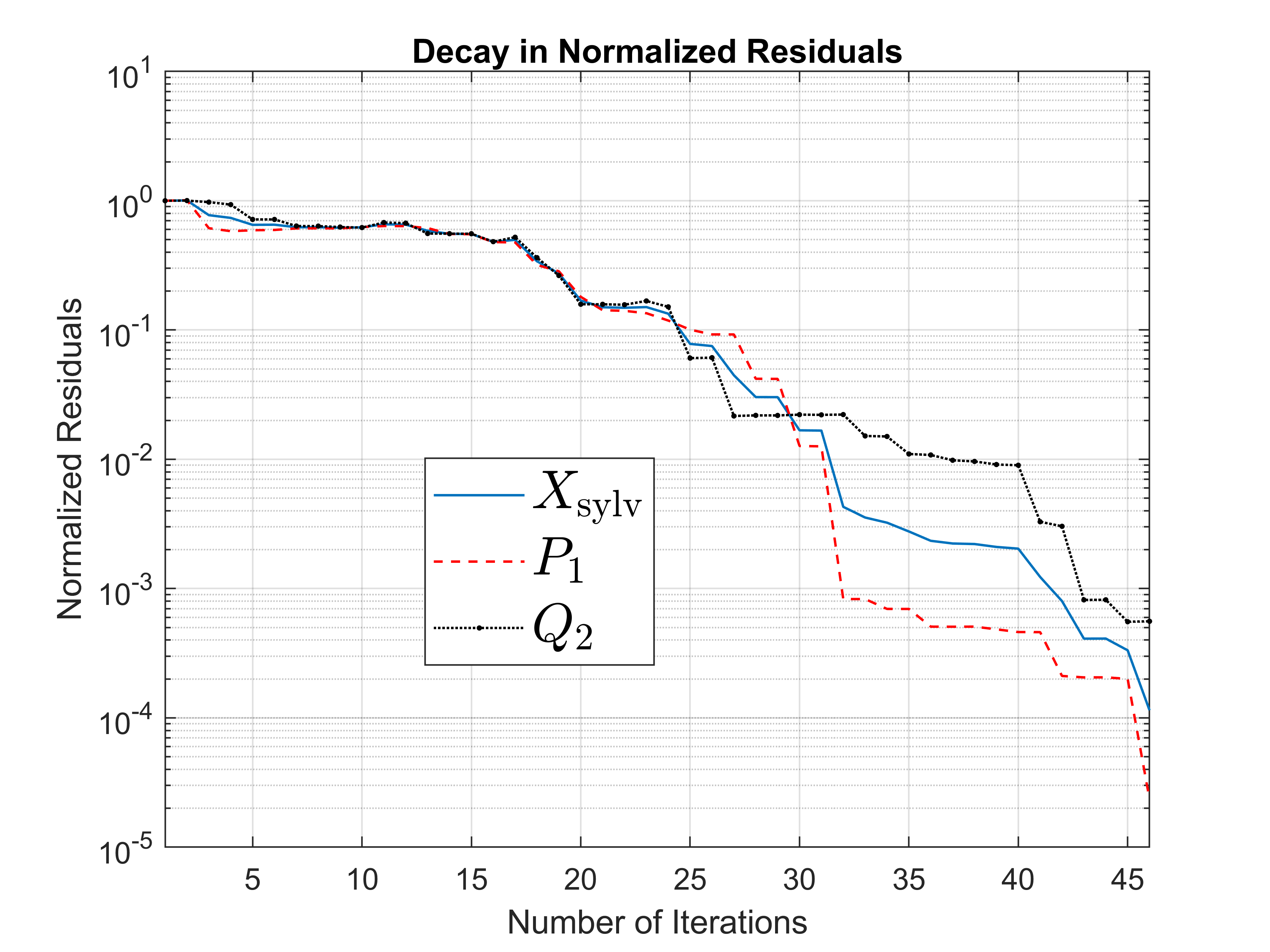}
  \caption{Normalized Residuals for \(P_1\), \(Q_2\), and \(X_{\mathrm{sylv}}\)}\label{fig9}
\end{figure}
This strategy delivers an excellent approximation for \(X_{\mathrm{sylv}}\) as well as for \(P_1\) and \(Q_2\), in agreement with the findings in Section \ref{3.9.3}. However, because the approximation of \(X_{\mathrm{sylv}}\) is prioritized, the decay of the residuals for \(P_1\) and \(Q_2\) is less steep compared to the strategy in Section \ref{3.9.1}.
\subsection{Steel Profile Model}
This benchmark is a semi-discretized heat transfer problem for optimal cooling of steel profiles \cite{benner2005semi}, also known as the rail model. The dynamical system \(G_1(s)\) is a $317,377^{th}$-order steel profile model taken from \cite{saak2021mm}. The dimensions of the matrices are as follows: \(E_1 \in \mathbb{R}^{317{,}377 \times 317{,}377}\), \(A_1 \in \mathbb{R}^{317{,}377 \times 317{,}377}\), \(B_1 \in \mathbb{R}^{317{,}377 \times 7}\), and \(C_1 \in \mathbb{R}^{6 \times 317{,}377}\). For implicit restart in the proposed shift generation strategy, the maximum number of columns in \(V_1\) is set to 21. The maximum number of ADI shifts \(k\) is set to $100$. When the normalized residual \(\frac{\|R_1^{(i)}\|_2}{\|B_1 B_1^\top\|_2}\) drops below the tolerance \(10^{-8}\), the approximation \(P_1 \approx V_{\mathrm{lyap}}^{(i)} (V_{\mathrm{lyap}}^{(i)})^\top\) is considered converged, and UADI stops. For visual clarity, only \(P_1\) is approximated in this experiment to avoid clutter in the figure. Figure \ref{fig10} shows that the proposed subspace-accelerated shift generation strategy reaches convergence in the fewest iterations.
\begin{figure}[!h]
  \centering
  \includegraphics[width=12cm]{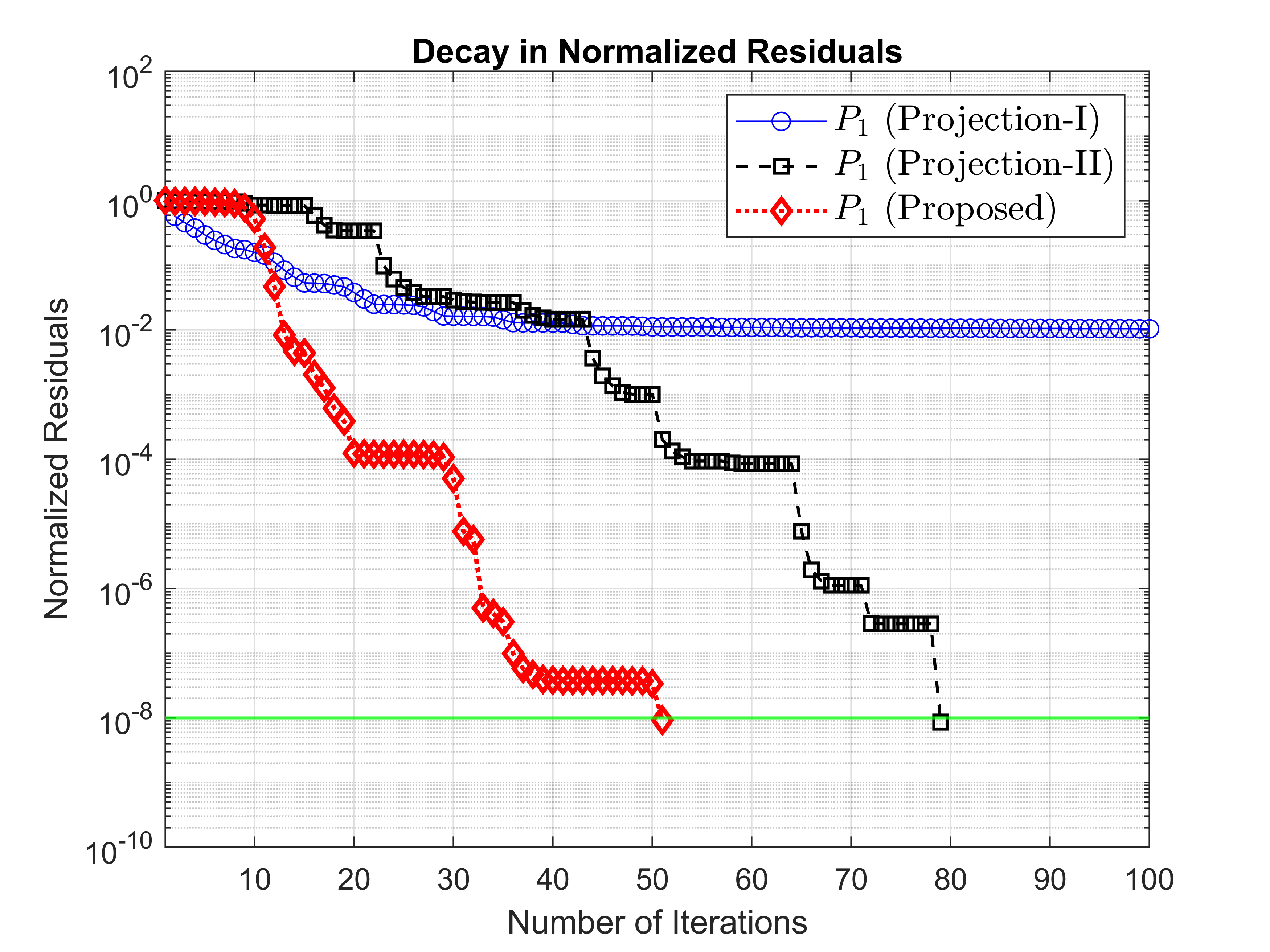}
  \caption{Normalized Residual for \(P_1\)}\label{fig10}
\end{figure}
In contrast, the Projection-I strategy fails to converge within the maximum allowable number of iterations, while the Projection-II strategy converges but requires more iterations than the proposed approach. The elapsed time for this experiment is tabulated in Table \ref{tab2}.
\begin{table}[!h]
\centering
\caption{Simulation Time Comparison}\label{tab2}
\begin{tabular}{|c|c|}\hline
Shift Generation Strategy & Elapsed Time (sec)\\\hline
Projection-I& $293.4014$\\
Projection-II&$219.9676$\\
Proposed& $189.8744$\\
\hline
\end{tabular}
\end{table}
The proposed strategy achieves convergence in the least amount of time overall. Note that the Projection-I and Projection-II strategies generate shifts more cheaply and, in this example, produce $7$ ADI shifts at a time—thus they are used more sparingly. In contrast, the proposed strategy generates only one or two shifts per step, requiring orthogonalization and eigenvalue decomposition more frequently. Nevertheless, this additional upfront cost pays off by significantly reducing the total computational time through faster convergence and fewer ADI iterations.
\section{Conclusion}
The paper shows that low-rank ADI methods for Lyapunov, Sylvester, and Riccati equations perform recursive interpolation at the mirror images of the ADI shifts. The methods differ only in their pole-placement behavior. By placing the poles of the projected matrices at chosen locations, ADI ensures that the projected linear matrix equation it implicitly solves has a unique solution. Thus, unlike several Krylov subspace methods—which may yield projected equations without unique solutions—ADI avoids this issue while still being projection-based. Unlike standard rational Krylov interpolation, ADI retains its pole-placement property as the number of interpolation points grows. Yet, as in standard interpolation, the projected matrices admit parametrizations, and the free parameters can be adjusted to relocate poles. This allows low-rank solutions of various linear matrix equations to be recovered from the low-rank solution of a standard Lyapunov equation by modifying the free parameter. No additional shifted linear solves are required; the solves from the Lyapunov ADI iterations can be reused. Since the shifted solves dominate the cost of low-rank ADI, this yields significant savings. Computing the free parameter requires only small-scale operations and is inexpensive, enabling multiple Lyapunov, Sylvester, and Riccati equations to be solved with the same linear solves. The paper further shows that the proposed UADI algorithm implicitly performs MOR and preserves several key properties of the original transfer function. The ROMs, which interpolate the transfer function at the mirror images of the ADI shifts, can be accumulated recursively. Two existing self-generating shift strategies are reviewed, and theoretical justification—not provided in the original work—is given using the interpolation properties of ADI methods identified in this paper. A subspace-accelerated self-generating shift strategy is proposed, inspired by SADPA, which estimates dominant poles and uses them as ADI shifts. This makes UADI fully automatic. Several problem-specific modifications to shift generation are also proposed. Numerical results show that the proposed strategy outperforms existing approaches and provides an effective and efficient way to automate ADI methods.
\section*{Acknowledgement}
The first author thanks Prof. Peter Benner of the Max Planck Institute for Dynamics of Complex Technical Systems, Magdeburg, Germany, for his time and valuable feedback. The first author also thanks Prof. Patrick Kürschner of Leipzig University of Applied Sciences for his patient and detailed responses to our many questions about the ADI method. This work was supported by the National Natural Science Foundation of China under Grant No. 62350410484.
\section*{Appendix A}
\begin{enumerate}
  \item Pre-multiplying \eqref{eq_sylv} by $(d_{\mathrm{sylv}}^{(i)})^{-1}$ yields:
  \begin{align}
  (d_{\mathrm{sylv}}^{(i)})^{-1}(-s_w^{(i)})^* d_{\mathrm{sylv}}^{(i)} -s_v^{(i)} +(d_{\mathrm{sylv}}^{(i)})^{-1}(l_w^{(i)})^\top l_v^{(i)}  &= 0\nonumber\\
  (d_{\mathrm{sylv}}^{(i)})^{-1}(-s_w^{(i)})^* d_{\mathrm{sylv}}^{(i)}&=s_v^{(i)} -(d_{\mathrm{sylv}}^{(i)})^{-1}(l_w^{(i)})^\top l_v^{(i)}\nonumber\\
   -(d_{\mathrm{sylv}}^{(i)})^{-1}(s_w^{(i)})^* d_{\mathrm{sylv}}^{(i)}&=\hat{a}_1^{(i)}.\nonumber
  \end{align}
  \item Post-multiplying \eqref{eq_sylv} by $(d_{\mathrm{sylv}}^{(i)})^{-1}$ gives:
  \begin{align}
     (-s_w^{(i)})^* + d_{\mathrm{sylv}}^{(i)} (-s_v^{(i)})(d_{\mathrm{sylv}}^{(i)})^{-1} +(l_w^{(i)})^\top l_v^{(i)}(d_{\mathrm{sylv}}^{(i)})^{-1}  &= 0\nonumber\\
  d_{\mathrm{sylv}}^{(i)} (-s_v^{(i)})(d_{\mathrm{sylv}}^{(i)})^{-1}   &= (s_w^{(i)})^*-(l_w^{(i)})^\top l_v^{(i)}(d_{\mathrm{sylv}}^{(i)})^{-1}\nonumber\\
    -d_{\mathrm{sylv}}^{(i)} s_v^{(i)}(d_{\mathrm{sylv}}^{(i)})^{-1}   &= \hat{a}_2^{(i)}.\nonumber
     \end{align}
\item Observe that:
\begin{align}
&\hat{a}_1^{(i)}(d_{\mathrm{sylv}}^{(i)})^{-1}+(d_{\mathrm{sylv}}^{(i)})^{-1}\hat{a}_2^{(i)}+\hat{b}_1^{(i)}\hat{c}_2^{(i)}\nonumber\\
&= -(d_{\mathrm{sylv}}^{(i)})^{-1}(s_w^{(i)})^*-s_v^{(i)}(d_{\mathrm{sylv}}^{(i)})^{-1}+(d_{\mathrm{sylv}})^{-1}(l_w^{(i)})^\top l_v^{(i)}(d_{\mathrm{sylv}})^{-1}\nonumber\\
&=(d_{\mathrm{sylv}}^{(i)})^{-1}\Big((-s_w^{(i)})^* d_{\mathrm{sylv}}^{(i)} + d_{\mathrm{sylv}}^{(i)} (-s_v^{(i)}) +(l_w^{(i)})^\top l_v^{(i)}\Big)(d_{\mathrm{sylv}}^{(i)})^{-1}\nonumber\\
&=0.\nonumber
\end{align}
\item Substituting $s_v^{(i)}=-\alpha_iI$, $l_v^{(i)}=-I$, $s_w^{(i)}=-\overline{\beta_i}I$, and $l_w^{(i)}=-I$ into \eqref{eq_sylv} gives:
 \begin{align}
 \beta_i d_{\mathrm{sylv}}^{(i)} + d_{\mathrm{sylv}}^{(i)}\alpha_i +I  &= 0\nonumber\\
 (\alpha_i+\beta_i)d_{\mathrm{sylv}}^{(i)}&=-I\nonumber\\
 d_{\mathrm{sylv}}^{(i)}&=-\frac{1}{\alpha_i+\beta_i}I\nonumber\\
( d_{\mathrm{sylv}}^{(i)})^{-1}&=-(\alpha_i+\beta_i)I.\nonumber
 \end{align}
\end{enumerate}
\section*{Appendix B}
\begin{enumerate}
  \item Observe that:
  \begin{align}
  \hat{B}_1^{(i)}&=\begin{bmatrix}\hat{b}_1^{(1)}\\\vdots\\\hat{b}_1^{(i)}\end{bmatrix}=\begin{bmatrix}(d_{\mathrm{sylv}}^{(1)})^{-1}(l_w^{(1)})^\top\\\vdots\\(d_{\mathrm{sylv}}^{(i)})^{-1}(l_w^{(i)})^\top\end{bmatrix}
  =\Big(\mathrm{blkdiag}\big(d_{\mathrm{sylv}}^{(1)},\cdots,d_{\mathrm{sylv}}^{(i)}\big)\Big)^{-1}\begin{bmatrix}l_w^{(1)}&\cdots&l_w^{(i)}\end{bmatrix}^\top\nonumber\\
  &=(D_{\mathrm{sylv}}^{(i)})^{-1}(L_w^{(i)})^\top.\nonumber
  \end{align}
 Similarly:
  \begin{align}
  \hat{C}_2^{(i)}&=\begin{bmatrix}\hat{c}_2^{(1)}&\cdots&\hat{c}_2^{(i)}\end{bmatrix}=\begin{bmatrix}l_v^{(1)}(d_\mathrm{sylv}^{(1)})^{-1}&\cdots&l_v^{(i)}(d_\mathrm{sylv}^{(i)})^{-1}\end{bmatrix}\nonumber\\
  &=\begin{bmatrix}l_v^{(1)}&\cdots&l_v^{(i)}\end{bmatrix}\Big(\mathrm{blkdiag}\big(d_{\mathrm{sylv}}^{(1)},\cdots,d_{\mathrm{sylv}}^{(i)}\big)\Big)^{-1}=L_v^{(i)}(D_{\mathrm{sylv}})^{-1}.\nonumber
  \end{align}
  \item Note the recursive structure:
  \begin{align}
  S_v^{(i)}&=\begin{bmatrix}S_v^{(i-1)}&(D_{\mathrm{sylv}}^{(i-1)})^{-1}(L_w^{(i-1)})^\top l_v^{(i)}\\0&s_v^{(i)}\end{bmatrix},& L_v^{(i)}&=\begin{bmatrix}L_v^{(i-1)}&l_v^{(i)}\end{bmatrix},\nonumber\\
  S_w^{(i)}&=\begin{bmatrix}S_w^{(i-1)}&(D_{\mathrm{sylv}}^{(i-1)})^{-*}(L_v^{(i-1)})^\top l_w^{(i)}\\0&s_w^{(i)}\end{bmatrix},& L_w^{(i)}&=\begin{bmatrix}L_w^{(i-1)}&l_w^{(i)}\end{bmatrix}.\nonumber
  \end{align}
  Then:
  \begin{align}
  &\begin{bmatrix}-(S_w^{(i-1)})^*&0\\-(l_w^{(i)})^\top L_v^{(i-1)}(D_{\mathrm{sylv}}^{(i-1)})^{-1}&-(s_w^{(i)})^*\end{bmatrix}\begin{bmatrix}D_{\mathrm{sylv}}^{(i-1)}&0\\0&d_{\mathrm{sylv}}^{(i)}\end{bmatrix}+\nonumber\\
  &\hspace*{2cm}\begin{bmatrix}D_{\mathrm{sylv}}^{(i-1)}&0\\0&d_{\mathrm{sylv}}^{(i)}\end{bmatrix}\begin{bmatrix}-S_v^{(i-1)}&-(D_{\mathrm{sylv}}^{(i-1)})^{-1}(L_w^{(i-1)})^\top l_v^{(i)}\\0&-s_v^{(i)}\end{bmatrix}+\begin{bmatrix}(L_w^{(i-1)})^\top\\(l_w^{(i)})^\top\end{bmatrix}\begin{bmatrix}L_v^{(i-1)}&l_v^{(i)}\end{bmatrix}\nonumber
  \end{align}
  \begin{align}
  &=\begin{bmatrix}-(S_w^{(i-1)})^*D_{\mathrm{sylv}}^{(i-1)}&0\\-(l_w^{(i)})^\top L_v^{(i-1)}&-(s_w^{(i)})^*d_{\mathrm{sylv}}^{(i)}\end{bmatrix}+
  \begin{bmatrix}-D_{\mathrm{sylv}}^{(i-1)}S_v^{(i-1)}&-(L_w^{(i-1)})^\top l_v^{(i)}\\0&-d_{\mathrm{sylv}}^{(i)}s_v^{(i)}\end{bmatrix}+\nonumber\\
  &\hspace*{8cm}\begin{bmatrix}(L_w^{(i-1)})^\top L_v^{(i-1)}&(L_w^{(i-1)})^\top l_v^{(i)}\\(l_w^{(i)})^\top L_v^{(i-1)}&(l_w^{(i)})^\top l_v^{(i)}\end{bmatrix}\nonumber\\
  &=\begin{bmatrix}-(S_w^{(i-1)})^*D_{\mathrm{sylv}}^{(i-1)}-D_{\mathrm{sylv}}^{(i-1)}S_v^{(i-1)}+(L_w^{(i-1)})^\top L_v^{(i-1)}&0\\0&-(s_w^{(i)})^*d_{\mathrm{sylv}}^{(i)}-d_{\mathrm{sylv}}^{(i)}s_v^{(i)}+(l_w^{(i)})^\top l_v^{(i)}\end{bmatrix}.\nonumber
  \end{align}
 Increasing \(i\) recursively from \(i=1\), it follows by induction that
  \begin{align}
  \begin{bmatrix}-(S_w^{(i-1)})^*D_{\mathrm{sylv}}^{(i-1)}-D_{\mathrm{sylv}}^{(i-1)}S_v^{(i-1)}+(L_w^{(i-1)})^\top L_v^{(i-1)}&0\\0&-(s_w^{(i)})^*d_{\mathrm{sylv}}^{(i)}-d_{\mathrm{sylv}}^{(i)}s_v^{(i)}+(l_w^{(i)})^\top l_v^{(i)}\end{bmatrix}=0.\nonumber
  \end{align}
  \item The proofs of (\ref{3_of_prop2}), (\ref{4_of_prop2}), and (\ref{5_of_prop2}) follow the same steps as the proofs of (\ref{1_of_prop1}), (\ref{2_of_prop1}), and (\ref{3_of_prop1}) in Proposition~\ref{prop1}, and are omitted for brevity.
\end{enumerate}
\section*{Appendix C}
\begin{enumerate}
  \item Substituting $s_{v,\mathrm{ricc}}^{(i)}$, $l_{v,\mathrm{ricc}}^{(i)}$, and $\hat{c}_1^{(i)}$ from \eqref{radi_free} into \eqref{lyap_p_ricc} yields:
  \begin{align}
    (\alpha_i I)^*(\hat{p}_{\mathrm{ricc}}^{(i)})^{-1}+(\hat{p}_{\mathrm{ricc}}^{(i)})^{-1}(\alpha_i I)&=-I-(v_i^{\mathrm{ricc}})^*C_1^\top C_1v_i^{\mathrm{ricc}}\nonumber\\
    2\mathrm{Re}(\alpha_i)(\hat{p}_{\mathrm{ricc}}^{(i)})^{-1}&=-\big[I+(v_i^{\mathrm{ricc}})^*C_1^\top C_1v_i^{\mathrm{ricc}}\big]\nonumber\\
    \hat{p}_{\mathrm{ricc}}^{(i)}&=-2\mathrm{Re}(\alpha_i)\big[I+(v_i^{\mathrm{ricc}})^*C_1^\top C_1v_i^{\mathrm{ricc}}\big]^{-1}.\nonumber
  \end{align}
  \item Note that:
  \begin{align}
  &-\begin{bmatrix}
        S_{v,\mathrm{ricc}}^{(i-1)} & \hat{P}_{\mathrm{ricc}}^{(i-1)}\Big(\big(L_{v,\mathrm{ricc}}^{(i-1)}\big)^\top l_{v,\mathrm{ricc}}^{(i)}\ + \big(\hat{C}_1^{(i-1)}\big)^* \hat{c}_1^{(i)}\Big) \\
        0 & s_{v,\mathrm{ricc}}^{(i)}\end{bmatrix}^*\begin{bmatrix}\big(\hat{P}_{\mathrm{ricc}}^{(i-1)}\big)^{-1}&0\\0&(\hat{p}_{\mathrm{ricc}}^{(i)})^{-1}\end{bmatrix}\nonumber\\
        &\hspace*{2cm}-\begin{bmatrix}\big(\hat{P}_{\mathrm{ricc}}^{(i-1)}\big)^{-1}&0\\0&(\hat{p}_{\mathrm{ricc}}^{(i)})^{-1}\end{bmatrix}\begin{bmatrix}
        S_{v,\mathrm{ricc}}^{(i-1)} & \hat{P}_{\mathrm{ricc}}^{(i-1)}\Big(\big(L_{v,\mathrm{ricc}}^{(i-1)}\big)^\top l_{v,\mathrm{ricc}}^{(i)}\ + \big(\hat{C}_1^{(i-1)}\big)^* \hat{c}_1^{(i)}\Big) \\
        0 & s_{v,\mathrm{ricc}}^{(i)}\end{bmatrix}\nonumber\\
        &\hspace*{2cm}+\begin{bmatrix}\big(L_{v,\mathrm{ricc}}^{(i-1)}\big)^\top\\(l_{v,\mathrm{ricc}}^{(i)})^\top\end{bmatrix}\begin{bmatrix}L_{v,\mathrm{ricc}}^{(i-1)}&l_{v,\mathrm{ricc}}^{(i)}\end{bmatrix}+\begin{bmatrix}\big(\hat{C}_1^{(i-1)}\big)^*\\(\hat{c}_1^{(i)})^*\end{bmatrix}\begin{bmatrix}\hat{C}_1^{(i-1)}&\hat{c}_1^{(i)}\end{bmatrix}\nonumber\\
        &=\begin{bmatrix}P_{11}&0\\0& p_{22}\end{bmatrix},\nonumber
        \end{align}where
        \begin{align}
        P_{11}&=-\big(S_{v,\mathrm{ricc}}^{(i-1)}\big)^*\big(\hat{P}_{\mathrm{ricc}}^{(i-1)}\big)^{-1}-\big(\hat{P}_{\mathrm{ricc}}^{(i-1)}\big)^{-1}S_{v,\mathrm{ricc}}^{(i-1)}+(L_{v,\mathrm{ricc}}^{(i-1)})^\top L_{v,\mathrm{ricc}}^{(i-1)} +\big(\hat{C}_1^{(i-1)}\big)^*\hat{C}_1^{(i-1)},\nonumber\\
        p_{22}&=(-s_{v,\mathrm{ricc}}^{(i)})^*(\hat{p}_{\mathrm{ricc}}^{(i)})^{-1}+(\hat{p}_{\mathrm{ricc}}^{(i)})^{-1}(-s_{v,\mathrm{ricc}}^{(i)})+(l_{v,\mathrm{ricc}}^{(i)})^\top l_{v,\mathrm{ricc}}^{(i)}+(\hat{c}_1^{(i)})^*\hat{c}_1^{(i)}.\nonumber
        \end{align}
        By increasing $i$ recursively from $i = 1$, it follows by induction that $P_{11}=0$ and $p_{22}=0$.
        \item Pre-multiplying \eqref{lyap_p_ricc_recc} by $\hat{P}_{\mathrm{ricc}}^{(i)}$ gives:
        \begin{align}
         \hat{P}_{\mathrm{ricc}}^{(i)}\big(-S_{v,\mathrm{ricc}}^{(i)}\big)^*(\hat{P}_{\mathrm{ricc}}^{(i)})^{-1}-S_{v,\mathrm{ricc}}^{(i)}+\hat{P}_{\mathrm{ricc}}^{(i)}\big(L_{v,\mathrm{ricc}}^{(i)}\big)^\top L_{v,\mathrm{ricc}}^{(i)}+\hat{P}_{\mathrm{ricc}}^{(i)}\big(\hat{C}_1^{(i)}\big)^* \hat{C}_1^{(i)}=0\nonumber\\
             \hat{P}_{\mathrm{ricc}}^{(i)}\big(-S_{v,\mathrm{ricc}}^{(i)}\big)^*\big(\hat{P}_{\mathrm{ricc}}^{(i)}\big)^{-1}+\hat{P}_{\mathrm{ricc}}^{(i)}\big(\hat{C}_1^{(i)}\big)^* \hat{C}_1^{(i)}=S_{v,\mathrm{ricc}}^{(i)}-\hat{P}_{\mathrm{ricc}}^{(i)}\big(L_{v,\mathrm{ricc}}^{(i)}\big)^\top L_{v,\mathrm{ricc}}^{(i)}\nonumber\\
             \hat{P}_{\mathrm{ricc}}^{(i)}\Big(\big(-S_{v,\mathrm{ricc}}^{(i)}\big)^*+\big(\hat{C}_1^{(i)}\big)^* \hat{C}_1^{(i)}\hat{P}_{\mathrm{ricc}}^{(i)}\Big)(\hat{P}_{\mathrm{ricc}}^{(i)})^{-1} =\hat{A}_1^{(i)}.\nonumber
        \end{align}
        \item Since $\hat{A}_1^{(i)}-\hat{P}_{\mathrm{ricc}}^{(i)}\big(\hat{C}_1^{(i)}\big)^* \hat{C}_1^{(i)}=\hat{P}_{\mathrm{ricc}}^{(i)}\big(-S_{v,\mathrm{ricc}}^{(i)}\big)^*\big(\hat{P}_{\mathrm{ricc}}^{(i)}\big)^{-1}$ and $S_{v,\mathrm{ricc}}^{(i)}$ is block triangular, the eigenvalues of $\hat{A}_1^{(i)}-\hat{P}_{\mathrm{ricc}}^{(i)}\big(\hat{C}_1^{(i)}\big)^* \hat{C}_1^{(i)}$ are $\alpha_1^*,\dots,\alpha_i^*$. As $\alpha_i$ have negative real parts, $\hat{A}_1^{(i)}-\hat{P}_{\mathrm{ricc}}^{(i)}\big(\hat{C}_1^{(i)}\big)^* \hat{C}_1^{(i)}$ is Hurwitz.
        \item Substituting $\hat{A}_1^{(i)}=\hat{P}_{\mathrm{ricc}}^{(i)}\Big(\big(-S_{v,\mathrm{ricc}}^{(i)}\big)^*+\big(\hat{C}_1^{(i)}\big)^*\hat{C}_1^{(i)}\hat{P}_{\mathrm{ricc}}^{(i)}\Big)\big(\hat{P}_{\mathrm{ricc}}^{(i)}\big)^{-1}$,
            and $\hat{B}_1^{(i)}=\hat{P}_{\mathrm{ricc}}^{(i)}\big(L_{v,\mathrm{ricc}}^{(i)}\big)^\top$ into the left-hand side of \eqref{proj_ricc} yields:
        \begin{align}
 \hat{P}_{\mathrm{ricc}}^{(i)}\big(-S_{v,\mathrm{ricc}}^{(i)}\big)^*-S_{v,\mathrm{ricc}}^{(i)}\hat{P}_{\mathrm{ricc}}^{(i)}+\hat{P}_{\mathrm{ricc}}^{(i)}\big(L_{v,\mathrm{ricc}}^{(i)}\big)^\top L_{v,\mathrm{ricc}}^{(i)}\hat{P}_{\mathrm{ricc}}^{(i)}+\hat{P}_{\mathrm{ricc}}^{(i)}\big(\hat{C}_1^{(i)}\big)^*\hat{C}_1^{(i)}\hat{P}_{\mathrm{ricc}}^{(i)}.\label{interim_ricc}
        \end{align}
        Pre- and post-multiplying \eqref{interim_ricc} by $\big(\hat{P}_{\mathrm{ricc}}^{(i)}\big)^{-1}$ gives:
        \begin{align}
 \big(-S_{v,\mathrm{ricc}}^{(i)}\big)^*\big(\hat{P}_{\mathrm{ricc}}^{(i)}\big)^{-1}-\big(\hat{P}_{\mathrm{ricc}}^{(i)}\big)^{-1}S_{v,\mathrm{ricc}}^{(i)}+\big(L_{v,\mathrm{ricc}}^{(i)}\big)^\top L_{v,\mathrm{ricc}}^{(i)}+\big(\hat{C}_1^{(i)}\big)^*\hat{C}_1^{(i)}=0.\nonumber
        \end{align}
        \item Note that:
        \begin{align}
        B_{\perp,1}^{\mathrm{ricc}}&=B_1-2\mathrm{Re}(\alpha_1)E_1v_1^{\mathrm{ricc}}\big[I+(v_1^{\mathrm{ricc}})^*C_1^\top C_1v_1^{\mathrm{ricc}}\big]^{-1}\nonumber\\
        &=B_1-E_1v_1^{\mathrm{ricc}}\hat{p}_{\mathrm{ricc}}^{(1)}(l_{v,\mathrm{ricc}}^{(1)})^\top\nonumber\\
        B_{\perp,2}^{\mathrm{ricc}}&=B_{\perp,1}^{\mathrm{ricc}}-2\mathrm{Re}(\alpha_2)E_1v_2^{\mathrm{ricc}}\big[I+(v_2^{\mathrm{ricc}})^*C_1^\top C_1v_2^{\mathrm{ricc}}\big]^{-1}\nonumber\\
        &=B_1-E_1v_1^{\mathrm{ricc}}\hat{p}_{\mathrm{ricc}}^{(1)}(l_{v,\mathrm{ricc}}^{(1)})^\top-E_1v_2^{\mathrm{ricc}}\hat{p}_{\mathrm{ricc}}^{(2)}(l_{v,\mathrm{ricc}}^{(2)})^\top\nonumber\\
        &=B_1-E_1V_{\mathrm{ricc}}^{(2)}\hat{P}_{\mathrm{ricc}}^{(2)}\big(L_{v,\mathrm{ricc}}^{(2)}\big)^\top.\nonumber
        \end{align}
        Continuing recursively yields $B_{\perp,i}^{\mathrm{ricc}}=B_1-E_1V_{\mathrm{ricc}}^{(i)}\hat{P}_{\mathrm{ricc}}^{(i)}\big(L_{v,\mathrm{ricc}}^{(i)}\big)^\top$.
        \item Consider the left-hand side of \eqref{radi_sylv1}:
        \begin{align}
        &A_1V_\mathrm{radi}^{(i)}-E_1V_\mathrm{radi}^{(i)} S_{v,\mathrm{ricc}}^{(i)}+B_1L_{v,\mathrm{ricc}}^{(i)}\nonumber\\
        &=A_1\begin{bmatrix}V_\mathrm{radi}^{(i-1)}&v_i^{\mathrm{ricc}}\end{bmatrix}-E_1\begin{bmatrix}V_\mathrm{radi}^{(i-1)}&v_i^{\mathrm{ricc}}\end{bmatrix} \begin{bmatrix}
        S_{v,\mathrm{ricc}}^{(i-1)} & \hat{P}_{\mathrm{ricc}}^{(i-1)}\Big(\big(L_{v,\mathrm{ricc}}^{(i-1)}\big)^\top l_{v,\mathrm{ricc}}^{(i)}\ + \big(\hat{C}_1^{(i-1)}\big)^* \hat{c}_1^{(i)}\Big) \\
        0 & s_{v,\mathrm{ricc}}^{(i)}\end{bmatrix}\nonumber\\
        &\hspace*{2.9cm}+B_1\begin{bmatrix}L_{v,\mathrm{ricc}}^{(i-1)}&l_{v,\mathrm{ricc}}^{(i)}\end{bmatrix}\nonumber\\
        &=\begin{bmatrix}p_{v,1}&p_{v,2}\end{bmatrix},\nonumber
        \end{align}where
        \begin{align}
        p_{v,1}&=A_1V_\mathrm{radi}^{(i)}-E_1V_\mathrm{radi}^{(i)} S_{v,\mathrm{ricc}}^{(i)}+B_1L_{v,\mathrm{ricc}}^{(i)},\nonumber\\
        p_{v,2}&=A_1v_i^{\mathrm{ricc}}-E_1V_{\mathrm{radi}}^{(i-1)}\hat{P}_{\mathrm{ricc}}^{(i-1)}\big(L_{v,\mathrm{ricc}}^{(i-1)}\big)^\top l_{v,\mathrm{ricc}}^{(i)}\ - E_1V_{\mathrm{radi}}^{(i-1)}\hat{P}_{\mathrm{ricc}}^{(i-1)}\big(\hat{C}_1^{(i-1)}\big)^*C_1v_i^{\mathrm{ricc}}\nonumber\\
        &\hspace*{3cm}-E_1v_i^{\mathrm{ricc}}s_{v,\mathrm{ricc}}^{(i)}+B_1l_{v,\mathrm{ricc}}^{(i)}\nonumber\\
        &=\Big(A_1-E_1V_{\mathrm{radi}}^{(i-1)}\hat{P}_{\mathrm{ricc}}^{(i-1)}\big(V_{\mathrm{radi}}^{(i-1)}\big)^*C_1^\top C_1\Big)v_i^{\mathrm{ricc}}-E_1v_i^{\mathrm{ricc}}s_{v,\mathrm{ricc}}^{(i)}+B_{\perp,i-1}^{\mathrm{ricc}}l_{v,\mathrm{ricc}}^{(i)}\nonumber.
        \end{align}
        Since $v_i^{\mathrm{ricc}}=\Big(A_1-E_1V_{\mathrm{radi}}^{(i-1)}\hat{P}_{\mathrm{ricc}}^{(i-1)}\big(V_{\mathrm{radi}}^{(i-1)}\big)^*C_1^\top C_1+\alpha_iE_1\Big)^{-1}B_{\perp,i-1}^{\mathrm{ricc}}$, we have $p_{v,2}=0$. Increasing $i$ recursively shows $p_{v,1}=0$ and $p_{v,2}=0$.
        
        Now consider the left-hand side of \eqref{radi_sylv2}:
        \begin{align}
         &A_1V_\mathrm{radi}^{(i)}-E_1V_\mathrm{radi}^{(i)} \big(\hat{E}_1^{(i)}\big)^{-1}\hat{A}_1^{(i)}+B_{\perp,i}^{\mathrm{ricc}}L_{v,\mathrm{ricc}}^{(i)}\nonumber\\
         &=A_1\begin{bmatrix}V_\mathrm{radi}^{(i-1)}&v_i^{\mathrm{ricc}}\end{bmatrix}-E_1\begin{bmatrix}V_\mathrm{radi}^{(i-1)}&v_i^{\mathrm{ricc}}\end{bmatrix}\begin{bmatrix}\hat{A}_1^{(i-1)}&\hat{P}_{\mathrm{ricc}}^{(i-1)}\big(\hat{C}_1^{(i-1)}\big)^* \hat{c}_1^{(i)}\\-\hat{b}_1^{(i)}L_{v,\mathrm{ricc}}^{(i-1)}&\hat{a}_1^{(i)}\end{bmatrix}+B_{\perp,i}^{\mathrm{ricc}}\begin{bmatrix}L_{v,\mathrm{ricc}}^{(i-1)}&l_{v,\mathrm{ricc}}^{(i)}\end{bmatrix}\nonumber\\
          &=\begin{bmatrix}p_{v,3}&p_{v,4}\end{bmatrix},\nonumber
        \end{align}
        where
        \begin{align}
        p_{v,3}&=A_1V_\mathrm{radi}^{(i-1)}-E_1V_\mathrm{radi}^{(i-1)}\hat{A}_1^{(i-1)}+E_1v_i^{\mathrm{ricc}}\hat{p}_{\mathrm{ricc}}(l_{v,\mathrm{ricc}}^{(i)})^\top L_{v,\mathrm{ricc}}^{(i-1)}+B_{\perp,i}^{\mathrm{ricc}}L_{v,\mathrm{ricc}}^{(i-1)}\nonumber\\
        &=A_1V_\mathrm{radi}^{(i-1)}-E_1V_\mathrm{radi}^{(i-1)}\hat{A}_1^{(i-1)}+B_{\perp,i-1}^{\mathrm{ricc}}L_{v,\mathrm{ricc}}^{(i-1)},\nonumber\\
        p_{v,4}&=A_1v_i^{\mathrm{ricc}}-E_1V_\mathrm{radi}^{(i-1)}\hat{P}_{\mathrm{ricc}}^{(i-1)}\big(V_\mathrm{radi}^{(i-1)}\big)^*C_1^\top C_1v_i^{\mathrm{ricc}}-E_1v_i^{\mathrm{ricc}}\hat{a}_1^{(i)}+B_{\perp,i}^{\mathrm{ricc}}l_{v,\mathrm{ricc}}^{(i)}\nonumber\\
        &=A_1v_i^{\mathrm{ricc}}-E_1V_\mathrm{radi}^{(i-1)}\hat{P}_{\mathrm{ricc}}^{(i-1)}\big(V_\mathrm{radi}^{(i-1)}\big)^*C_1^\top C_1v_i^{\mathrm{ricc}}-E_1v_i^{\mathrm{ricc}}\hat{a}_1^{(i)}+B_{\perp,i-1}^{\mathrm{ricc}}l_{v,\mathrm{ricc}}^{(i)}\nonumber\\
        &-E_1v_i^{\mathrm{ricc}}\hat{p}_{\mathrm{ricc}}^{(i)}(l_{v,\mathrm{ricc}}^{(i)})^\top l_{v,\mathrm{ricc}}^{(i)}\nonumber\\
        &=\Big(A_1-E_1V_\mathrm{radi}^{(i-1)}\hat{P}_{\mathrm{ricc}}^{(i-1)}\big(V_\mathrm{radi}^{(i-1)}\big)^*C_1^\top C_1\Big)v_i^{\mathrm{ricc}}-E_1v_i^{\mathrm{ricc}}s_{v,\mathrm{ricc}}^{(i)}+B_{\perp,i-1}^{\mathrm{ricc}}l_{v,\mathrm{ricc}}^{(i)}\nonumber\\
        &=0.\nonumber
        \end{align}
       Recursion over $i$ shows $p_{v,3}=0$ and $p_{v,4}=0$.
\end{enumerate}
\section*{Appendix D}
Since \( w_i^{\mathrm{sylv}} = W_{\mathrm{lyap}}^{(i)} t_{w,\mathrm{sylv}}^{(i)} \) and \( W_{\mathrm{fadi}}^{(i)} = W_{\mathrm{lyap}}^{(i)} T_{w,\mathrm{sylv}}^{(i)} \) are duals of \( v_i^{\mathrm{sylv}} = V_{\mathrm{lyap}}^{(i)} t_{v,\mathrm{sylv}}^{(i)} \) and \( V_{\mathrm{fadi}}^{(i)} = V_{\mathrm{lyap}}^{(i)} T_{v,\mathrm{sylv}}^{(i)} \), respectively, we restrict ourselves to proving  
\( v_i^{\mathrm{sylv}} = V_{\mathrm{lyap}}^{(i)} t_{v,\mathrm{sylv}}^{(i)} \) and \( V_{\mathrm{fadi}}^{(i)} = V_{\mathrm{lyap}}^{(i)} T_{v,\mathrm{sylv}}^{(i)} \).

Note that
\begin{align}
\begin{bmatrix}\hat{A}_{1,i-1}^{\mathrm{lyap}}+\overline{\alpha}_iI&0\\-2\mathrm{Re}(\alpha_i)L_{v,\mathrm{lyap}}^{(i-1)}&2\mathrm{Re}(\alpha_i)I\end{bmatrix}
\begin{bmatrix}t_{1,v,\mathrm{sylv}}^{(i)}\\t_{2,v,\mathrm{sylv}}^{(i)}\end{bmatrix}
&=\begin{bmatrix}\hat{B}_{1,i-1}^{\mathrm{lyap}}-T_{v,\mathrm{sylv}}^{(i-1)}D_{\mathrm{fadi}}^{(i-1)}(L_{w,\mathrm{sylv}}^{(i-1)})^\top\\2\mathrm{Re}(\alpha_i)I\end{bmatrix}.\nonumber
\end{align}
Thus, \( t_{2,v,\mathrm{sylv}}^{(i)} = I + L_{v,\mathrm{lyap}}^{(i-1)} t_{1,v,\mathrm{sylv}}^{(i)} \). To prove \( v_i^{\mathrm{sylv}} = V_{\mathrm{lyap}}^{(i)} t_{v,\mathrm{sylv}}^{(i)} \), it suffices to show that
\begin{align}
(A_1+\alpha_iE_1)v_i^{\mathrm{sylv}}=(A_1+\alpha_iE_1)V_{\mathrm{lyap}}^{(i)}t_{v,\mathrm{sylv}}^{(i)}=B_{\perp,i-1}^{\mathrm{sylv}}=B_1-E_1V_{\mathrm{fadi}}^{(i-1)}D_{\mathrm{fadi}}^{(i-1)}(L_{w,\mathrm{sylv}}^{(i-1)})^\top.\nonumber
\end{align}
Consider the following:
\begin{align}
(A_1&+\alpha_iE_1)V_{\mathrm{lyap}}^{(i)}t_{v,\mathrm{sylv}}^{(i)}\nonumber\\
&=(A_1+\alpha_iE_1)\begin{bmatrix}V_{\mathrm{lyap}}^{(i-1)}&v_i^{\mathrm{lyap}}\end{bmatrix}\begin{bmatrix}t_{1,v,\mathrm{sylv}}^{(i)}\\t_{2,v,\mathrm{sylv}}^{(i)}\end{bmatrix}\nonumber\\
&=(A_1+\alpha_iE_1)\big(V_{\mathrm{lyap}}^{(i-1)}t_{1,v,\mathrm{sylv}}^{(i)}+v_i^{\mathrm{lyap}}t_{2,v,\mathrm{sylv}}^{(i)}\big)\nonumber\\
&=(A_1+\alpha_iE_1)V_{\mathrm{lyap}}^{(i-1)}t_{1,v,\mathrm{sylv}}^{(i)}+(A_1+\alpha_iE_1)v_i^{\mathrm{lyap}}t_{2,v,\mathrm{sylv}}^{(i)}\nonumber\\
&=A_1V_{\mathrm{lyap}}^{(i-1)}t_{1,v,\mathrm{sylv}}^{(i)}+\alpha_iE_1V_{\mathrm{lyap}}^{(i-1)}t_{1,v,\mathrm{sylv}}^{(i)}+B_{\perp,i-1}^{\mathrm{lyap}}\big(I+L_{v,\mathrm{lyap}}^{(i-1)}t_{1,v,\mathrm{sylv}}^{(i)}\big)\nonumber\\
&=\big(A_1V_{\mathrm{lyap}}^{(i-1)}+B_{\perp,i-1}^{\mathrm{lyap}}L_{v,\mathrm{lyap}}^{(i-1)}\big)t_{1,v,\mathrm{sylv}}^{(i)}+\alpha_iE_1V_{\mathrm{lyap}}^{(i-1)}t_{1,v,\mathrm{sylv}}^{(i)}+B_{\perp,i-1}^{\mathrm{lyap}}\nonumber\\
&=E_1V_{\mathrm{lyap}}^{(i-1)}\hat{A}_{1,i-1}^{\mathrm{lyap}}t_{1,v,\mathrm{sylv}}^{(i)}+\alpha_iE_1V_{\mathrm{lyap}}^{(i-1)}t_{1,v,\mathrm{sylv}}^{(i)}+B_1-E_1V_{\mathrm{lyap}}^{(i-1)}\hat{B}_{1,i-1}^{\mathrm{lyap}}\nonumber\\
&=E_1V_{\mathrm{lyap}}^{(i-1)}\Big(\hat{A}_{1,i-1}^{\mathrm{lyap}}t_{1,v,\mathrm{sylv}}^{(i)}-s_v^{(i)}t_{1,v,\mathrm{sylv}}^{(i)}+\hat{B}_{1,i-1}^{\mathrm{lyap}}l_v^{(i)}\Big)+B_1\nonumber\\
&=B_1-E_1V_{\mathrm{lyap}}^{(i-1)}T_{v,\mathrm{sylv}}^{(i-1)}D_{\mathrm{fadi}}^{(i-1)}(L_{w,\mathrm{sylv}}^{(i-1)})^\top.\nonumber
\end{align}
For $i=1$, $V_{\mathrm{lyap}}^{(i)}=v_1^{\mathrm{lyap}}$, $T_{v,\mathrm{sylv}}^{(i-1)}=[$ $]$, $D_{\mathrm{fadi}}^{(i-1)}=[$ $]$, $L_{w,\mathrm{sylv}}^{(i-1)}=[$ $]$. Hence,
\begin{align}
(A_1+\alpha_1E_1)v_1^{\mathrm{lyap}}t_{v,\mathrm{sylv}}^{(1)}=(A_1+\alpha_1E_1)v_1^{\mathrm{lyap}}=(A_1+\alpha_1E_1)v_1^{\mathrm{sylv}}=B_1.\nonumber
\end{align} Clearly, \( t_{v,\mathrm{sylv}}^{(1)} = I \), which is consistent since \( v_1^{\mathrm{lyap}} = v_1^{\mathrm{sylv}} \).

For \( i = 2 \), \( V_{\mathrm{lyap}}^{(i)} = \begin{bmatrix} v_1^{\mathrm{lyap}} & v_2^{\mathrm{lyap}} \end{bmatrix} \), \( T_{v,\mathrm{sylv}}^{(i-1)} = I \), \( D_{\mathrm{fadi}}^{(i-1)} = -(\alpha_1 + \beta_1) I \), and \( L_{w,\mathrm{sylv}}^{(i-1)} = -I \). Thus,
\[
(A_1+\alpha_2E_1)V_{\mathrm{lyap}}^{(2)}t_{v,\mathrm{sylv}}^{(2)}=B_1-(\alpha_1+\beta_1)E_1v_1^{\mathrm{sylv}}=B_{\perp,i-1}^{\mathrm{sylv}}.
\] It follows immediately that \( V_{\mathrm{lyap}}^{(2)} t_{v,\mathrm{sylv}}^{(2)} = v_2^{\mathrm{sylv}} \). By induction, the identity \( v_i^{\mathrm{sylv}} = V_{\mathrm{lyap}}^{(i)} t_{v,\mathrm{sylv}}^{(i)} \) holds for \( i = 3, 4, \dots \).

Next, observe that
\begin{align}
V_{\mathrm{lyap}}^{(i)}T_{v,\mathrm{sylv}}^{(i)}&=\begin{bmatrix}V_{\mathrm{lyap}}^{(i-1)}&v_i^{\mathrm{lyap}}\end{bmatrix}\begin{bmatrix}T_{v,\mathrm{sylv}}^{(i-1)}&t_{1,v,\mathrm{sylv}}^{(i)}\\
0&t_{2,v,\mathrm{sylv}}^{(i)}\end{bmatrix}=\begin{bmatrix}V_{\mathrm{lyap}}^{(i-1)}T_{v,\mathrm{sylv}}^{(i-1)}&V_{\mathrm{lyap}}^{(i)}t_{v,\mathrm{sylv}}^{(i)}\end{bmatrix}.\nonumber
\end{align}
For $i=1$, $V_{\mathrm{lyap}}^{(i-1)}=[$ $]$ and $T_{v,\mathrm{sylv}}^{(i-1)}=[$ $]$, so 
\begin{align}
V_{\mathrm{lyap}}^{(1)}T_{v,\mathrm{sylv}}^{(1)}=V_{\mathrm{lyap}}^{(1)}t_{v,\mathrm{sylv}}^{(1)}=v_1^{\mathrm{sylv}}=V_{\mathrm{fadi}}^{(1)}.\nonumber
\end{align}
For \( i = 2 \), \( T_{v,\mathrm{sylv}}^{(1)} = t_{v,\mathrm{sylv}}^{(1)} \), and thus
\begin{align}
V_{\mathrm{lyap}}^{(2)}T_{v,\mathrm{sylv}}^{(2)}&=\begin{bmatrix}V_{\mathrm{lyap}}^{(1)}t_{v,\mathrm{sylv}}^{(1)}&V_{\mathrm{lyap}}^{(2)}t_{v,\mathrm{sylv}}^{(2)}\end{bmatrix}
=\begin{bmatrix}v_1^{\mathrm{sylv}}&v_2^{\mathrm{sylv}}\end{bmatrix}=V_{\mathrm{fadi}}^{(2)}.\nonumber
\end{align}
By induction, \( V_{\mathrm{fadi}}^{(i)} = V_{\mathrm{lyap}}^{(i)} T_{v,\mathrm{sylv}}^{(i)} \) holds for all \( i = 3, 4, \dots \).
\section*{Appendix E}
Recall that $t_{2,v,\mathrm{sylv}}^{(i)}=I+L_{v,\mathrm{lyap}}^{(i-1)}t_{1,v,\mathrm{sylv}}^{(i)}$. Thus, $t_{2,v,\mathrm{sylv}}^{(i)}=I-\sum_{k=1}^{i-1}t_{v,k}=I-M$. If the matrix \(M\) does not have an eigenvalue equal to \(1\), then \(\big(t_{2,v,\mathrm{sylv}}^{(i)}\big)^{-1}\) has no eigenvalue at \(0\), and hence \(t_{2,v,\mathrm{sylv}}^{(i)}\) is invertible.
\section*{Appendix F}
\begin{enumerate}
  \item Note that
\begin{align}
A_1V_{\mathrm{lyap}}^{(i)}T_{v,\mathrm{sylv}}^{(i)}-E_1V_{\mathrm{lyap}}^{(i)}T_{v,\mathrm{sylv}}^{(i)}S_{v,\mathrm{sylv}}^{(i)}+B_1L_{v,\mathrm{sylv}}^{(i)}=0.\label{e147}
\end{align}
Post-multiplying \eqref{e147} by $\big(T_{v,\mathrm{sylv}}^{(i)}\big)^{-1}$ yields:
\begin{align}
A_1V_{\mathrm{lyap}}^{(i)}-E_1V_{\mathrm{lyap}}^{(i)}T_{v,\mathrm{sylv}}^{(i)}S_{v,\mathrm{sylv}}^{(i)}\big(T_{v,\mathrm{sylv}}^{(i)}\big)^{-1}+B_1L_{v,\mathrm{sylv}}^{(i)}\big(T_{v,\mathrm{sylv}}^{(i)}\big)^{-1}=0.\nonumber
\end{align}
By uniqueness of the solution to \eqref{cfadi_v_sylv1}, we have $S_{v,\mathrm{lyap}}^{(i)}=T_{v,\mathrm{sylv}}^{(i)}S_{v,\mathrm{sylv}}^{(i)}\big(T_{v,\mathrm{sylv}}^{(i)}\big)^{-1}$ and $L_{v,\mathrm{lyap}}^{(i)}=L_{v,\mathrm{sylv}}^{(i)}\big(T_{v,\mathrm{sylv}}^{(i)}\big)^{-1}$.
\item Note that
\begin{align}
A_2^\top W_{\mathrm{lyap}}^{(i)}T_{w,\mathrm{sylv}}^{(i)}-E_2^\top W_{\mathrm{lyap}}^{(i)}T_{w,\mathrm{sylv}}^{(i)} S_{w,\mathrm{sylv}}^{(i)}+C_2^\top L_{w,\mathrm{sylv}}^{(i)}=0.\label{e148}
\end{align} Post-multiplying \eqref{e148} by $\big(T_{w,\mathrm{sylv}}^{(i)}\big)^{-1}$ yields:
\begin{align}
A_2^\top W_{\mathrm{lyap}}^{(i)}-E_2^\top W_{\mathrm{lyap}}^{(i)}T_{w,\mathrm{sylv}}^{(i)} S_{w,\mathrm{sylv}}^{(i)}\big(T_{w,\mathrm{sylv}}^{(i)}\big)^{-1}+C_2^\top L_{w,\mathrm{sylv}}^{(i)}\big(T_{w,\mathrm{sylv}}^{(i)}\big)^{-1}=0.\nonumber
\end{align}
By uniqueness of the solution to \eqref{cfadi_w_sylv1}, we have $S_{w,\mathrm{lyap}}^{(i)}=T_{w,\mathrm{sylv}}^{(i)}S_{w,\mathrm{sylv}}^{(i)}\big(T_{w,\mathrm{sylv}}^{(i)}\big)^{-1}$ and $L_{w,\mathrm{lyap}}^{(i)}=L_{w,\mathrm{sylv}}^{(i)}\big(T_{w,\mathrm{sylv}}^{(i)}\big)^{-1}$.
\item Now consider
\begin{align}
-\big(S_{w,\mathrm{sylv}}^{(i)}\big)^*\big(D_{\mathrm{fadi}}^{(i)}\big)^{-1}-\big(D_{\mathrm{fadi}}^{(i)}\big)^{-1}S_{v,\mathrm{sylv}}^{(i)}+(L_{w,\mathrm{sylv}}^{(i)})^\top L_{v,\mathrm{sylv}}^{(i)}=0.\label{e149}
\end{align}
Pre-multiplying \eqref{e149} by $\big(T_{w,\mathrm{sylv}}^{(i)}\big)^{-*}$ and post-multiplying by $\big(T_{v,\mathrm{sylv}}^{(i)}\big)^{-1}$ gives
\begin{align}
-\big(S_{w,\mathrm{lyap}}^{(i)}\big)^*\big(T_{w,\mathrm{sylv}}^{(i)}\big)^{-*}\big(D_{\mathrm{fadi}}^{(i)}\big)^{-1}\big(T_{v,\mathrm{sylv}}^{(i)}\big)^{-1}-\big(T_{w,\mathrm{sylv}}^{(i)}\big)^{-*}\big(D_{\mathrm{fadi}}^{(i)}\big)^{-1}\big(T_{v,\mathrm{sylv}}^{(i)}\big)^{-1}S_{v,\mathrm{lyap}}^{(i)}\nonumber\\
+(L_{w,\mathrm{lyap}}^{(i)})^\top L_{v,\mathrm{lyap}}^{(i)}=0\nonumber\\
-\big(S_{w,\mathrm{lyap}}^{(i)}\big)^*\tilde{X}_{\mathrm{sylv}}^{(i)}-\tilde{X}_{\mathrm{sylv}}^{(i)}S_{v,\mathrm{lyap}}^{(i)}+(L_{w,\mathrm{lyap}}^{(i)})^\top L_{v,\mathrm{lyap}}^{(i)}=0.\nonumber
\end{align}
The proofs of \eqref{prop_x_sylv_2}, \eqref{prop_x_sylv_3}, and \eqref{prop_x_sylv_4} in Proposition~\ref{prop_x_sylv} follow similarly to \eqref{1_of_prop1}, \eqref{2_of_prop1}, and \eqref{3_of_prop1} in Proposition~\ref{prop1}, respectively, and are therefore omitted for brevity.
\end{enumerate}
\section*{Appendix G}
First, observe that
\begin{align}
&\begin{bmatrix}\hat{A}_{1,i-1}^{\mathrm{lyap}}+\alpha_iI-T_{v,\mathrm{ricc}}^{(i-1)}\hat{P}_{\mathrm{ricc}}^{(i-1)}(V_{\mathrm{radi}}^{(i-1)})^*C_1^\top C_1V_{\mathrm{lyap}}^{(i-1)}&-T_{v,\mathrm{ricc}}^{(i-1)}\hat{P}_{\mathrm{ricc}}^{(i-1)}(V_{\mathrm{radi}}^{(i-1)})^*C_1^\top C_1v_{i}^{\mathrm{lyap}}\\-2\mathrm{Re}(\alpha_i)L_{v,\mathrm{lyap}}^{(i-1)}&2\mathrm{Re}(\alpha_i)I\end{bmatrix}\begin{bmatrix}t_{1,v,\mathrm{ricc}}^{(i)}\\t_{2,v,\mathrm{ricc}}^{(i)}\end{bmatrix}\nonumber\\
&=\begin{bmatrix}\hat{B}_{1,i-1}^{\mathrm{lyap}}-T_{v,\mathrm{ricc}}^{(i-1)}\hat{P}_{\mathrm{ricc}}^{(i-1)}(L_{v,\mathrm{ricc}}^{(i-1)})^\top\\2\mathrm{Re}(\alpha_i)I\end{bmatrix},\nonumber
\end{align}
which implies $t_{2,v,\mathrm{ricc}}^{(i)}=I+L_{v,\mathrm{lyap}}^{(i-1)}t_{1,v,\mathrm{ricc}}^{(i)}$.

To show that \(v_i^{\mathrm{ricc}} = V_{\mathrm{lyap}}^{(i)} t_{v,\mathrm{ricc}}^{(i)}\), it suffices to verify that  
\begin{align}
\Big(A_1&+\alpha_iE_1-E_1V_{\mathrm{radi}}^{(i-1)}\hat{P}_{\mathrm{ricc}}^{(i-1)}(V_{\mathrm{radi}}^{(i-1)})^*C_1^\top C_1\Big)v_i^{\mathrm{ricc}}\nonumber\\
&=\Big(A_1+\alpha_iE_1-E_1V_{\mathrm{radi}}^{(i-1)}\hat{P}_{\mathrm{ricc}}^{(i-1)}(V_{\mathrm{radi}}^{(i-1)})^*C_1^\top C_1\Big)V_{\mathrm{lyap}}^{(i)}t_{v,\mathrm{ricc}}^{(i)}\nonumber\\
&=B_1-E_1V_{\mathrm{radi}}^{(i-1)}\hat{P}_{\mathrm{ricc}}^{(i-1)}(L_{v,\mathrm{ricc}}^{(i-1)})^\top\nonumber\\
&=B_{\perp,i-1}^{\mathrm{ricc}}.\nonumber
\end{align}
Consider the expression
\begin{align}
\Big(A_1&+\alpha_iE_1-E_1V_{\mathrm{lyap}}^{(i-1)}T_{v,\mathrm{ricc}}^{(i-1)}\hat{P}_{\mathrm{ricc}}^{(i-1)}(V_{\mathrm{radi}}^{(i-1)})^*C_1^\top C_1\Big)V_{\mathrm{lyap}}^{(i)}t_{v,\mathrm{ricc}}^{(i)}\nonumber\\
&=\Big(A_1+\alpha_iE_1-E_1V_{\mathrm{lyap}}^{(i-1)}T_{v,\mathrm{ricc}}^{(i-1)}\hat{P}_{\mathrm{ricc}}^{(i-1)}(V_{\mathrm{radi}}^{(i-1)})^*C_1^\top C_1\Big)\begin{bmatrix}V_{\mathrm{lyap}}^{(i-1)}&v_i^{\mathrm{lyap}}\end{bmatrix}\begin{bmatrix}t_{1,v,\mathrm{ricc}}^{(i)}\\t_{2,v,\mathrm{ricc}}^{(i)}\end{bmatrix}\nonumber\\
&=\Big(A_1+\alpha_iE_1-E_1V_{\mathrm{lyap}}^{(i-1)}T_{v,\mathrm{ricc}}^{(i-1)}\hat{P}_{\mathrm{ricc}}^{(i-1)}(V_{\mathrm{radi}}^{(i-1)})^*C_1^\top C_1\Big)\Big(V_{\mathrm{lyap}}^{(i-1)}t_{1,v,\mathrm{ricc}}^{(i)}+v_i^{\mathrm{lyap}}t_{2,v,\mathrm{ricc}}^{(i)}\Big)\nonumber\\
&=(A_1+\alpha_iE_1)v_i^{\mathrm{lyap}}t_{2,v,\mathrm{ricc}}^{(i)}-E_1V_{\mathrm{lyap}}^{(i-1)}T_{v,\mathrm{ricc}}^{(i-1)}\hat{P}_{\mathrm{ricc}}^{(i-1)}(V_{\mathrm{radi}}^{(i-1)})^*C_1^\top C_1v_i^{\mathrm{lyap}}t_{2,v,\mathrm{ricc}}^{(i)}\nonumber\\
&+(A_1+\alpha_iE_1)V_{\mathrm{lyap}}^{(i-1)}t_{1,v,\mathrm{ricc}}^{(i)}-E_1V_{\mathrm{lyap}}^{(i-1)}T_{v,\mathrm{ricc}}^{(i-1)}\hat{P}_{\mathrm{ricc}}^{(i-1)}(V_{\mathrm{radi}}^{(i-1)})^*C_1^\top C_1V_{\mathrm{lyap}}^{(i-1)}t_{1,v,\mathrm{ricc}}^{(i)}\nonumber\\
&=B_{\perp,i-1}^{\mathrm{lyap}}t_{2,v,\mathrm{ricc}}^{(i)}-E_1V_{\mathrm{lyap}}^{(i-1)}T_{v,\mathrm{ricc}}^{(i-1)}\hat{P}_{\mathrm{ricc}}^{(i-1)}(V_{\mathrm{radi}}^{(i-1)})^*C_1^\top C_1v_i^{\mathrm{lyap}}t_{2,v,\mathrm{ricc}}^{(i)}\nonumber\\
&+E_1V_{\mathrm{lyap}}^{(i-1)}\hat{A}_{1,i-1}^{\mathrm{lyap}}t_{1,v,\mathrm{ricc}}^{(i)}-B_{\perp,i-1}^{\mathrm{lyap}}L_{v,i-1}^{\mathrm{lyap}}t_{1,v,\mathrm{ricc}}^{(i)}+\alpha_iE_1V_{\mathrm{lyap}}^{(i-1)}t_{1,v,\mathrm{ricc}}^{(i)}\nonumber\\
&-E_1V_{\mathrm{lyap}}^{(i-1)}T_{v,\mathrm{ricc}}^{(i-1)}\hat{P}_{\mathrm{ricc}}^{(i-1)}(V_{\mathrm{radi}}^{(i-1)})^*C_1^\top C_1V_{\mathrm{lyap}}^{(i-1)}t_{1,v,\mathrm{ricc}}^{(i)}\nonumber
\end{align}
\begin{align}
&=B_{\perp,i-1}^{\mathrm{lyap}}\big(I+L_{v,\mathrm{lyap}}^{(i-1)}t_{1,v,\mathrm{ricc}}^{(i)}\big)+E_1V_{\mathrm{lyap}}^{(i-1)}\Big(-T_{v,\mathrm{ricc}}^{(i-1)}\hat{P}_{\mathrm{ricc}}^{(i-1)}(V_{\mathrm{radi}}^{(i-1)})^*C_1^\top C_1v_i^{\mathrm{lyap}}t_{2,v,\mathrm{ricc}}^{(i)}\nonumber\\
&+\hat{A}_{1,i-1}^{\mathrm{lyap}}t_{1,v,\mathrm{ricc}}^{(i)}-T_{v,\mathrm{ricc}}^{(i-1)}\hat{P}_{\mathrm{ricc}}^{(i-1)}(V_{\mathrm{radi}}^{(i-1)})^*C_1^\top C_1V_{\mathrm{lyap}}^{(i-1)}t_{1,v,\mathrm{ricc}}^{(i)}-t_{1,v,\mathrm{ricc}}^{(i)}s_{v,\mathrm{ricc}}^{(i)}\Big)-B_{\perp,i-1}^{\mathrm{lyap}}L_{v,i-1}^{\mathrm{lyap}}t_{1,v,\mathrm{ricc}}^{(i)}\nonumber\\
&=B_1-E_1V_{\mathrm{lyap}}^{(i-1)}\hat{B}_{1,i-1}^{\mathrm{lyap}}+E_1V_{\mathrm{lyap}}^{(i-1)}\hat{B}_{1,i-1}^{\mathrm{lyap}}-E_1V_{\mathrm{lyap}}^{(i-1)}T_{v,\mathrm{ricc}}^{(i-1)}\hat{P}_{\mathrm{ricc}}^{(i-1)}(L_{v,\mathrm{ricc}}^{(i-1)})^\top\nonumber\\
&=B_1-E_1V_{\mathrm{lyap}}^{(i-1)}T_{v,\mathrm{ricc}}^{(i-1)}\hat{P}_{\mathrm{ricc}}^{(i-1)}(L_{v,\mathrm{ricc}}^{(i-1)})^\top.\nonumber
\end{align}
For $i=1$, we have $V_{\mathrm{lyap}}^{(i-1)}=[$ $]$, $T_{v,\mathrm{ricc}}^{(i-1)}=[$ $]$, $\hat{P}_{\mathrm{ricc}}^{(i-1)}=[$ $]$, so
\begin{align}
(A_1+\alpha_1E_1)v_1^{\mathrm{lyap}}t_{v,\mathrm{ricc}}^{(1)}=(A_1+\alpha_1E_1)v_1^{\mathrm{lyap}}=(A_1+\alpha_1E_1)v_1^{\mathrm{ricc}}=B_1,\nonumber
\end{align}
which gives \(t_{v,\mathrm{ricc}}^{(1)} = I\), consistent with \(v_1^{\mathrm{lyap}} = v_1^{\mathrm{ricc}}\).

For $i=2$, note that $V_{\mathrm{lyap}}^{(i-1)}=v_1^{\mathrm{ricc}}$, $T_{v,\mathrm{ricc}}^{(i-1)}=I$, $\hat{P}_{\mathrm{ricc}}^{(i-1)}=-2\mathrm{Re}(\alpha_1)\big[I+(v_1^{\mathrm{ricc}})^*C_1^\top C_1v_1^{\mathrm{ricc}}\big]^{-1}$, so
\begin{align}
\Big(A_1+\alpha_iE_1-E_1v_1^{\mathrm{ricc}}\hat{P}_{\mathrm{ricc}}^{(1)}(v_1^{\mathrm{ricc}})^*C_1^\top C_1\Big)V_{\mathrm{lyap}}^{(i)}t_{v,\mathrm{ricc}}^{(i)}&=B_1-2\mathrm{Re}(\alpha_1)E_1v_1^{\mathrm{ricc}}\big[I+(v_1^{\mathrm{ricc}})^*C_1^\top C_1v_1^{\mathrm{ricc}}\big]^{-1}=B_{\perp,1}^{\mathrm{ricc}},\nonumber
\end{align}
and thus \(V_{\mathrm{lyap}}^{(2)} t_{v,\mathrm{ricc}}^{(2)} = v_2^{\mathrm{ricc}}\). By induction, the identity \(v_i^{\mathrm{ricc}} = V_{\mathrm{lyap}}^{(i)} t_{v,\mathrm{ricc}}^{(i)}\) holds for all \(i = 3, 4, \dots\).

Finally, note that 
\begin{align}
V_{\mathrm{lyap}}^{(i)}T_{v,\mathrm{ricc}}^{(i)}&=\begin{bmatrix}V_{\mathrm{lyap}}^{(i-1)}&v_i^{\mathrm{lyap}}\end{bmatrix}\begin{bmatrix}T_{v,\mathrm{ricc}}^{(i-1)}&t_{1,v,\mathrm{ricc}}^{(i)}\\
0&t_{2,v,\mathrm{ricc}}^{(i)}\end{bmatrix}=\begin{bmatrix}V_{\mathrm{lyap}}^{(i-1)}T_{v,\mathrm{ricc}}^{(i-1)}&V_{\mathrm{lyap}}^{(i)}t_{v,\mathrm{ricc}}^{(i)}\end{bmatrix}.\nonumber
\end{align}
For \(i = 1\), this reduces to \(V_{\mathrm{lyap}}^{(1)} T_{v,\mathrm{ricc}}^{(1)} = v_1^{\mathrm{ricc}}\).  
For \(i = 2\), since \(T_{v,\mathrm{ricc}}^{(1)} = t_{v,\mathrm{ricc}}^{(1)}\), we obtain  
\begin{align}
V_{\mathrm{lyap}}^{(2)}T_{v,\mathrm{ricc}}^{(2)}&=\begin{bmatrix}V_{\mathrm{lyap}}^{(1)}t_{v,\mathrm{ricc}}^{(1)}&V_{\mathrm{lyap}}^{(2)}t_{v,\mathrm{ricc}}^{(2)}\end{bmatrix}
=\begin{bmatrix}v_1^{\mathrm{ricc}}&v_2^{\mathrm{ricc}}\end{bmatrix}.\nonumber
\end{align}
By induction, it follows that \(V_{\mathrm{radi}}^{(i)} = V_{\mathrm{lyap}}^{(i)} T_{v,\mathrm{ricc}}^{(i)}\) for all \(i \geq 1\).
\section*{Appendix H}
Note that
\begin{align}
A_1V_{\mathrm{lyap}}^{(i)}T_{v,\mathrm{ricc}}^{(i)}-E_1V_{\mathrm{lyap}}^{(i)}T_{v,\mathrm{ricc}}^{(i)}S_{v,\mathrm{ricc}}^{(i)}+B_1L_{v,\mathrm{ricc}}^{(i)}=0.\label{e158}
\end{align}
Post-multiplying \eqref{e158} by $\big(T_{v,\mathrm{ricc}}^{(i)}\big)^{-1}$ gives:
\begin{align}
A_1V_{\mathrm{lyap}}^{(i)}-E_1V_{\mathrm{lyap}}^{(i)}T_{v,\mathrm{ricc}}^{(i)}S_{v,\mathrm{ricc}}^{(i)}\big(T_{v,\mathrm{ricc}}^{(i)}\big)^{-1}+B_1L_{v,\mathrm{ricc}}^{(i)}\big(T_{v,\mathrm{ricc}}^{(i)}\big)^{-1}=0.\nonumber
\end{align}
By the uniqueness of the solution to \eqref{cfadi_v_sylv1}, it follows that $S_{v,\mathrm{lyap}}^{(i)}=T_{v,\mathrm{ricc}}^{(i)}S_{v,\mathrm{ricc}}^{(i)}\big(T_{v,\mathrm{ricc}}^{(i)}\big)^{-1}$ and $L_{v,\mathrm{lyap}}^{(i)}=L_{v,\mathrm{ricc}}^{(i)}\big(T_{v,\mathrm{ricc}}^{(i)}\big)^{-1}$.

Now, pre-multiplying \eqref{lyap_p_ricc_recc} by $\big(T_{v,\mathrm{ricc}}^{(i)}\big)^{-*}$ and post-multiplying by $\big(T_{v,\mathrm{ricc}}^{(i)}\big)^{-1}$ yields:
\begin{align}
-\big(T_{v,\mathrm{ricc}}^{(i)}\big)^{-*}\big(S_{v,\mathrm{ricc}}^{(i)}\big)^*\big(\hat{P}_{\mathrm{ricc}}^{(i)}\big)^{-1}\big(T_{v,\mathrm{ricc}}^{(i)}\big)^{-1}-\big(T_{v,\mathrm{ricc}}^{(i)}\big)^{-*}\big(\hat{P}_{\mathrm{ricc}}^{(i)}\big)^{-1}S_{v,\mathrm{ricc}}^{(i)}\big(T_{v,\mathrm{ricc}}^{(i)}\big)^{-1}\nonumber\\
+\big(L_{v,\mathrm{lyap}}^{(i)}\big)^\top L_{v,\mathrm{lyap}}^{(i)}+\big(V_{\mathrm{lyap}}^{(i)}\big)^*C_1^\top C_1 V_{\mathrm{lyap}}^{(i)}=0\nonumber\\
-\big(S_{v,\mathrm{lyap}}^{(i)}\big)^*\big(T_{v,\mathrm{ricc}}^{(i)}\big)^{-*}\big(\hat{P}_{\mathrm{ricc}}^{(i)}\big)^{-1}\big(T_{v,\mathrm{ricc}}^{(i)}\big)^{-1}-\big(T_{v,\mathrm{ricc}}^{(i)}\big)^{-*}\big(\hat{P}_{\mathrm{ricc}}^{(i)}\big)^{-1}\big(T_{v,\mathrm{ricc}}^{(i)}\big)^{-1}S_{v,\mathrm{lyap}}^{(i)}\nonumber\\
+\big(L_{v,\mathrm{lyap}}^{(i)}\big)^\top L_{v,\mathrm{lyap}}^{(i)}+\big(V_{\mathrm{lyap}}^{(i)}\big)^*C_1^\top C_1 V_{\mathrm{lyap}}^{(i)}=0\nonumber\\
-\big(S_{v,\mathrm{lyap}}^{(i)}\big)^*\big(\tilde{P}_{\mathrm{ricc}}^{(i)}\big)^{-1}-\big(\tilde{P}_{\mathrm{ricc}}^{(i)}\big)^{-1}S_{v,\mathrm{lyap}}^{(i)}+\big(L_{v,\mathrm{lyap}}^{(i)}\big)^\top L_{v,\mathrm{lyap}}^{(i)}+\big(V_{\mathrm{lyap}}^{(i)}\big)^*C_1^\top C_1 V_{\mathrm{lyap}}^{(i)}=0.\nonumber
\end{align}
The proofs of statements \eqref{prop_pt_ricc_3} and \eqref{prop_pt_ricc_4} in Theorem \ref{prop_pt_ricc} follow similarly to those of \eqref{th_radi_2} and \eqref{th_radi_3} in Theorem \ref{th_radi}, and are therefore omitted for brevity.
\section*{Appendix I}
\begin{verbatim}
function [E,A,B,C] = tripple_peak(n,w1,w2,w3)

a1=[-1 w1; -w1 -1]; a2=[-1 w2; -w2 -1]; a3=[-1 w3; -w3 -1];
aa=blkdiag(a1,a2,a3); bb=10*ones(6,1); cc=10*ones(1,6);
pp=lyap(aa,bb*bb'); qq=lyap(aa',cc'*cc);
e=qq*pp; a=qq*aa*pp; b=qq*bb; c=cc*pp;
e4=speye(n-6); a4=-spdiags((1:n-6)',0,n-6,n-6); b4=ones(n-6,1); c4=ones(1,n-6);
E=blkdiag(e,e4); A=blkdiag(a,a4); B=[b;b4]; C=[c c4];

end
\end{verbatim}

\end{document}